\documentclass[12pt,twoside,openright,leqno]{report}
\usepackage[online]{suthesis-2e}

\usepackage{amsmath}

\usepackage{hyperref}
\usepackage[ruled,vlined,linesnumbered]{algorithm2e}
\usepackage{algorithmicx}
\usepackage{amsfonts}
\usepackage{amstext}
\usepackage{amsthm}
\usepackage{amssymb}
\usepackage{bm}
\usepackage{booktabs}
\usepackage{caption}
\usepackage{color}
\usepackage{enumerate}
\usepackage{enumitem}
\usepackage{graphicx}
\usepackage{hyperref}
\usepackage{listings}
\usepackage{microtype} 
\usepackage{multirow}
\usepackage{rotating}
\usepackage{siunitx}
\usepackage{standalone}
\usepackage{subcaption}
\usepackage{tabularx}
\usepackage{url}
\usepackage[usenames,dvipsnames]{xcolor}
\usepackage{xspace}

\usepackage{tikz,nicefrac,amsmath,pifont,tkz-graph,makecell,tkz-graph}
\usetikzlibrary{arrows,backgrounds,patterns,matrix,shapes,fit,calc,shadows,plotmarks,arrows.meta,positioning,shapes.geometric}

\usepackage[T1]{fontenc}

\definecolor{mylinkcolor}{RGB}{0,0,140}
\hypersetup{colorlinks,allcolors=mylinkcolor,citecolor=mylinkcolor}
\usepackage[noabbrev,capitalize]{cleveref}

\usepackage{paralist}
\renewenvironment{itemize}[1]{\begin{compactitem}#1}{\end{compactitem}}
\renewenvironment{enumerate}[1]{\begin{compactenum}#1}{\end{compactenum}}

\SetAlgoCaptionSeparator{ --}

\usepackage[square,semicolon,sort]{rnatbib}
\renewcommand{\cite}{\citep}

\graphicspath{{./introduction/FIG/},{./higher-order-network-clustering/FIG/},{./ccfs/FIG/}{./temporal-motifs/FIG/}}

\usepackage[osf,sc]{mathpazo}
\usepackage[scaled=0.90]{helvet}
\usepackage[scaled=0.85]{beramono}

\newcommand{\phantomsubfigure}[1]{\begin{subfigure}[b]{0.1\textwidth}\phantomcaption\label{#1}\end{subfigure}}
\newcommand{\xhdr}[1]{\vspace{1.7mm}\noindent{{\bf #1.}}}
\newcommand{\hide}[1]{}
\newcommand{\dataset}[1]{\textsc{#1}}
\newcommand{\given}{\;\vert\;}

\newcommand{\mycomment}[1]{{\sffamily /* #1 */}}
\newcommand{\algoptions}{
  \SetKw{KwTo}{in}\SetKwFor{For}{for}{\string:}{}%
  \SetFuncSty{emph}
  \SetKwProg{myproc}{Procedure}{}{}
  \DontPrintSemicolon
}  

\newtheorem{theorem}{Theorem}
\newtheorem{lemma}[theorem]{Lemma}
\newtheorem{proposition}[theorem]{Proposition}
\newtheorem{corollary}[theorem]{Corollary}

\newtheorem{definition}[theorem]{Definition}

\newcommand*\dualcaption[2]{\caption[#1]{\textit{\sffamily #1.} #2}}

\newcommand{\honctitle}{Higher-order network clustering}

\newcommand{\cond}[1]{\phi(#1)}
\newcommand{\mcond}[2]{\phi_{#1}(#2)}
\newcommand{\mmcond}[1]{\phi_{M}(#1)}
\newcommand{\gcond}[2]{\phi^{(#1)}(#2)}
\newcommand{\gmmcond}[2]{\phi^{(#1)}_{M}(#2)}

\newcommand{\vol}[1]{\textnormal{vol}(#1)}
\newcommand{\mvol}[2]{\textnormal{vol}_{#1}(#2)}
\newcommand{\mmvol}[1]{\textnormal{vol}_{M}(#1)}
\newcommand{\gvol}[2]{\textnormal{vol}^{(#1)}(#2)}
\newcommand{\gmmvol}[2]{\textnormal{vol}^{(#1)}_{M}(#2)}

\newcommand{\cut}[1]{\textnormal{cut}(#1)}
\newcommand{\mcut}[2]{\textnormal{cut}_{#1}(#2)}
\newcommand{\mmcut}[1]{\textnormal{cut}_{M}(#1)}
\newcommand{\gcut}[2]{\textnormal{cut}^{(#1)}(#2)}
\newcommand{\gmmcut}[2]{\textnormal{cut}^{(#1)}_{M}(#2)}

\newcommand{\anchorset}{\mathcal{A}}
\newcommand{\anchornodes}{\chi_{\anchorset}(v)}
\newcommand{\minstance}{(v, \anchornodes)}

\newcommand{\motifweightedij}{(W_M)_{ij}}
\newcommand{\normmotiflap}{N_M}

\newcommand{\diag}[1]{\text{diag}\left( #1 \right)}
\newcommand{\allones}{e}

\newcommand{\mbifan}{M_{\textnormal{bifan}}}
\newcommand{\medge}{M_{\textnormal{edge}}}
\newcommand{\setof}[1]{\textnormal{set}(#1)}
\DeclareMathOperator{\truth}{\textnormal{Ind}}
\newcommand{\indicator}[1]{\truth[#1]}

\newcommand{\ppr}{\texttt{PPR}}

\newcommand{\dblp}{\dataset{com-DBLP}}
\newcommand{\amazon}{\dataset{com-Amazon}}
\newcommand{\youtube}{\dataset{com-Youtube}}
\newcommand{\lj}{\dataset{com-LiveJournal}}
\newcommand{\orkut}{\dataset{com-Orkut}}
\newcommand{\friendster}{\dataset{com-Friendster}}

\newcommand{\wikicat}{\dataset{wiki-cats}}

\newcommand{\mtri}{M_{\textnormal{tri}}}
\newcommand{\mcyc}{M_{\textnormal{cyc}}}
\newcommand{\mffl}{M_{\textnormal{ffl}}}

\newcommand{\hoccfstitle}{Higher-order clustering coefficients}
\newcommand{\gccf}[1]{C_{#1}}   
\newcommand{\lccf}[2]{C_{#1}(#2)}    
\newcommand{\accf}[1]{\bar{C}_{#1}}  
\newcommand{\clique}[1]{K_{#1}}  
\newcommand{\cliqueL}[2]{K_{#1}(#2)}  
\newcommand{\onehopnou}{N_u}
\newcommand{\onehopu}{N'_u}
\newcommand{\onehopv}{N'_v}
\newcommand{\p}[2]{p_{#1}(#2)}  

\newcommand{\condmat}{\dataset{ca-CondMat}}
\newcommand{\harvard}{\dataset{fb-Harvard1}}
\newcommand{\enron}{\dataset{email-Enron}}
\newcommand{\google}{\dataset{web-Google}}
\newcommand{\emaileucore}{\dataset{email-Eu-core}}

\newcommand{\astroph}{\dataset{ca-AstroPh}}

\newcommand{\stanford}{\dataset{fb-Stanford3}}

\newcommand{\er}{\text{Erd\H{o}s-R\'enyi}}

\newcommand{\expect}[1]{\mathbb{E}\left[#1\right]}
\newcommand{\expectover}[2]{\mathbb{E}_{#1}\left[#2\right]}
\newcommand{\prob}[1]{\mathbb{P}\left[#1\right]}

\newcommand{\tmtitle}{Motifs in temporal networks}

\newcommand{\minuseq}{\mathrel{-}=}
\newcommand{\pluseq}{\mathrel{+}=}
\newcommand{\triangleenum}{\Delta_{\textnormal{enum}}}
\newcommand{\multigraph}{K}

\newcommand{\concat}{\textnormal{concat}}
\newcommand{\counts}{\textnormal{counts}}
\newcommand{\keys}{\textnormal{keys}}

\newcommand{\prefix}{\textnormal{prefix}}
\newcommand{\reverse}{\textnormal{reverse}}
\newcommand{\suffix}{\textnormal{suffix}}
\newcommand{\tend}{\textnormal{end}}
\newcommand{\tstart}{\textnormal{start}}

\newcommand{\emaileu}{\dataset{email-Eu}}

\newcommand{\phone}{\dataset{Phonecall-Eu}}

\newcommand{\sms}{\dataset{SMS-A}}
\newcommand{\messages}{\dataset{CollegeMsg}}
\newcommand{\stackoverflow}{\dataset{StackOverflow}}
\newcommand{\mathoverflow}{\dataset{MathOverflow}}
\newcommand{\superuser}{\dataset{SuperUser}}
\newcommand{\askubuntu}{\dataset{AskUbuntu}}

\newcommand{\fbwall}{\dataset{FBWall}}
\newcommand{\bitcoin}{\dataset{Bitcoin}}
\newcommand{\wikitalk}{\dataset{WikiTalk}}
\newcommand{\phonecallme}{\dataset{Phonecall-ME}}
\newcommand{\smsme}{\dataset{SMS-ME}}

\newcommand{\dir}{\textnormal{dir}}
\newcommand{\globalall}{\textnormal{count}}
\newcommand{\globalmid}{\textnormal{peri\_count}}
\newcommand{\globalpost}{\textnormal{post\_count}}
\newcommand{\globalpre}{\textnormal{pre\_count}}
\newcommand{\middlesum}{\textnormal{peri\_sum}}
\newcommand{\muorv}{\textnormal{1-u\_or\_v}}

\newcommand{\nbr}{\textnormal{nbr}}
\newcommand{\nodecount}{\textnormal{node\_count}}
\newcommand{\postnodecount}{\textnormal{post\_nodes}}
\newcommand{\postsum}{\textnormal{post\_sum}}
\newcommand{\prenodecount}{\textnormal{pre\_nodes}}
\newcommand{\presum}{\textnormal{pre\_sum}}
\newcommand{\sumtxt}{\textnormal{sum}}
\newcommand{\uorv}{\textnormal{u\_or\_v}}
\newcommand{\utov}{\textnormal{u\_to\_v}}

\begin{document}


\title{Tools for higher-order network analysis}
\author{Austin Reilley Benson}
\icmethesis
\dept{Institute for Computational and Mathematical Engineering}
\principaladviser{Jure Leskovec}
\firstreader{David F.~Gleich}
\secondreader{Johan Ugander}


\submitdate{June 6, 2017}

\beforepreface

\prefacesection{Abstract}

Networks are a fundamental model of complex systems throughout the sciences, and
network datasets are typically analyzed through lower-order connectivity
patterns described at the level of individual nodes and edges.  However,
higher-order connectivity patterns captured by small subgraphs, also called
network motifs, describe the fundamental structures that control and mediate the
behavior of many complex systems.  We develop three tools for network analysis
that use higher-order connectivity patterns to gain new insights into network
datasets: (1) a framework to cluster nodes into modules based on joint
participation in network motifs; (2) a generalization of the clustering
coefficient measurement to investigate higher-order closure patterns; and (3) a
definition of network motifs for temporal networks and fast algorithms for
counting them.  Using these tools, we analyze data from
biology,
ecology,
economics,
neuroscience,
online social networks,
scientific collaborations,
telecommunications,
transportation,
and the World Wide Web.

\prefacesection{Acknowledgments}

I owe a great debt of gratitude to many folks for my ability to finally complete
this dissertation.
First, the content of this dissertation is comprised of collaborative research.  Thank
you to my co-authors
Ashwin Paranjape,
David Gleich,
Hao Yin, and
Jure Leskovec.
Second, the research would not have been possible without generous funding from
an Office of Technology Licensing Stanford Graduate Fellowship and from DARPA
SIMPLEX.  The SGF program was especially beneficial during my scientific
wayfinding in the early years at Stanford.

Throughout my time at Stanford, I was fortunate to receive guidance from a number of
more senior folks.  Thanks especially to
Andrew Tomkins,
David Gleich,
Grey Ballard,
Johan Ugander,
Jack Poulson,
Jure Leskovec,
Lexing Ying,
Margot Gerritsen,
Michael Saunders,
Moses Charikar,
Ravi Kumar, and
Rob Schreiber.
Additional thanks to Johan for reading and signing off on this thesis, to Moses
for serving on my defense committee, to Lexing for chairing my defense committee,
and to Margot for making ICME such a great place.

Jure has been an amazingly supportive and dedicated advisor.  He is always
pushing me to turn algorithmic ideas into meaningful applications, which has
shaped the way that I think about research.  I value his ability to
provide guidance while simultaneously letting me pursue my own interests and ideas.

David also deserves special recognition.  He has been an incredible collaborator
and even better mentor ever since I started working with him while I was an
undergraduate at Berkeley.  My career and this thesis would not have been
possible without his support and guidance.

Next, thank you to all of the ICME and SNAP students for all of the wonderful
interactions.  Special thanks to
Anil,
Brad,
David,
Eileen,
Hima,
Jason,
Milinda,
Nolan,
Rikel,
Ron,
Ryan,
Sven,
Tim,
Victor,
Will,
Xiaotong,
Yingzhou, and
Yuekai.

My family has been a tremendous support network.  Thank you to my parents---Karen
and Craig Benson---and to my Davis family---Marianne, Brian, Oak, and Bob Hallet---for
all of your help.

Finally, thank you to the amazing Catherine Hallet Benson for her continued
support, which keeps me going.

\afterpreface

\chapter{Higher-order thinking in network analysis}
\label{ch:intro}

\newcommand*{\intropath}{introduction}

A network is a model for a system of connections between things.  The core of a
network model is a graph, a mathematical object consisting of nodes (modeling
the things) and edges (modeling the connections between things).  This
thesis provides new graph theoretic methods to analyze datasets associated with
network models.  Such network datasets show up in a wide variety of scientific
disciplines, and this thesis alone analyzes networks from the following domains.

\begin{itemize}
\item \emph{Ecological systems --} Nodes are species and edges are who-eats-whom
  relationships in food web models.  (Here, the edges are \emph{directed} to
  represent asymmetry---sharks regularly consume sardines, but sharks are not a
  part of a sardine's diet!)  We examine the graph structure behind aquatic
  layers in a food web of Florida Bay.
  
\item \emph{Human communication systems --} Nodes are people and edges are
  messages between people.  Again, there is a natural direction to the edges.
  Our network datasets also have temporal information associated with the graph
  to denote when messages are sent.  We analyze the differences between communication
  behaviors in e-mail, SMS texting, and phone call networks.
  
\item \emph{Neural systems --} Nodes are neurons and edges are synapses.  We
  examine the patterns of connections in the complete neural network of the
  nematode worm \emph{C.~elegans}.
  
\item \emph{Payment systems --} Nodes are people, businesses, or accounts and
  edges are payments between them.  We analyze transactions between addresses on
  Bitcoin.  We again have temporal information corresponding to when payments
  are made.
  
\item \emph{Transcription regulation networks --} Nodes are groups of genes and
  edges represent which groups of genes regulate genetic transcription in which
  other groups of genes.  The edges in the network data also carry information
  about whether the regulation type is activation or suppression (in this case,
  we say that the edges are \emph{signed}).  We analyze the transcription regulation network
  of the yeast \emph{S.~cerevisiae}.
  
\item \emph{Transportation systems --} Nodes are locations and edges represent
  connectivity.  We analyze an air travel reachability network, where nodes are
  cities and edges represent how long it takes to travel from one city to
  another using commercial airline flights.  Our analysis also incorporates
  additional geographical data available for the nodes.
    
\item \emph{Scientific collaborations --} Nodes are scientists and there is an
  edge between two scientists if they have co-authored a paper.  We look at
  collaboration patterns between physicists.

\item \emph{Social networks --} Nodes are people and edges are social
  relationships.  We examine Facebook friendships, Twitter followers, and Stack
  Overflow question-and-answer interactions.

\item \emph{The World Wide Web --} Nodes are web pages and edges are hyperlinks.
  We analyze the structure of the Stanford web graph and Wikipedia.
\end{itemize}

The above descriptions illustrate how graphs, which consist of nodes and edges,
are a natural mathematical structure for network models.  The graph is the
backbone of the network datasets, even though we may have additional information
such as timestamps, signed edges, geographical data, etc.  Consequently, we often frame our
analysis of and questions about networks in terms of nodes and edges.  For
example, we might be interested in the number of edges that a node is in (how
many different species does the shark consume in the food web?), the number of
edges between subsets of the nodes (how many Facebook friendships are there
between Stanford students?), or whether or not there exists a sequence of edges
to traverse to go from one node to another (can I get from Palo Alto to Berkeley
on public transit?).


\begin{figure}[t]
\centering
\phantomsubfigure{fig:intro_fflA}
\phantomsubfigure{fig:intro_fflB}
\phantomsubfigure{fig:intro_fflC}
\includegraphics[width=\textwidth]{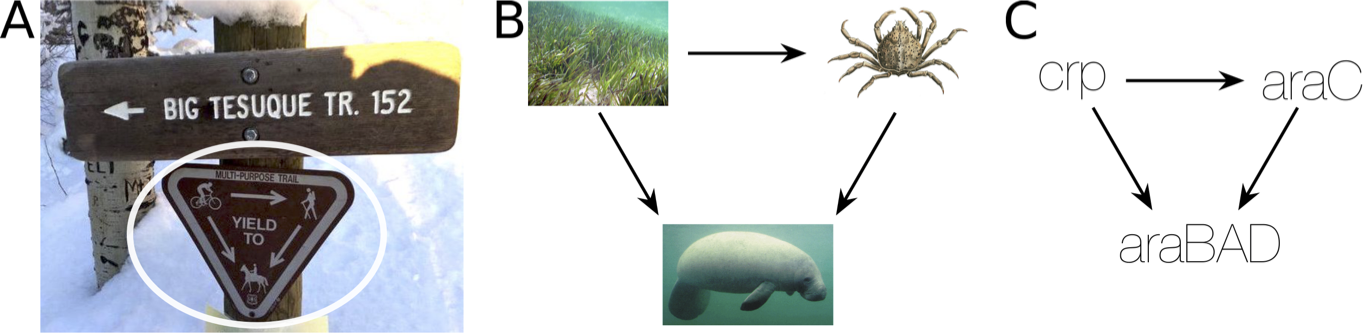}
\dualcaption{Example higher-order network structure (network motif) appearing
in different domains}{%
{\bf A:}
Traffic yielding pattern on the Big Tesuque Trail---bikers yield to hikers and
equestrians, and hikers also yield to equestrians (photo taken by the author
after presenting the research in \Cref{ch:honc} at the Santa Fe Institute in
December, 2015).
{\bf B:} 
The same pattern is called an omnivory chain in
ecology~\cite{camacho2007quantitative,bascompte2009disentangling}.  Here we show
an example in the Florida Bay food web, where the
seagrass producer (thalassia) is a food source for crabs and manatees, and crabs
are also a food source for the manatees~\cite{ulanowicz1998network}.
{\bf C:}
In transcription regulation networks, the pattern is called a feedforward loop (figure reproduced from
\cite{shen2002network}).  In \emph{E. coli}, the operon \textsf{crp} encodes
transcription factors that regulate operons \textsf{araC} and \textsf{araBAD}, and
\textsf{araC} also encodes a transcription factor that regulates
\textsf{araBAD}.  Feedforward loops occur more frequently than one would expect
by chance in the \emph{E. coli} and \emph{S.~cerevisiae} transcription networks~\cite{milo2002network} and function as delay units in
transcription~\cite{mangan2003coherent,mangan2003structure}.
}
\label{fig:intro_ffl}
\end{figure}

While analysis in terms of nodes and edges is natural, there is substantial
evidence that \emph{higher-order structures}, or small subgraph patterns between
a few nodes (\cref{fig:intro_ffl}), are essential to the behavior of many
complex systems modeled by
networks~\cite{milo2002network,yaverouglu2014revealing}.  This idea has a long
history in sociology.  In the early 1900s, Simmel theorized that triangles (a
complete graph on three nodes) form in social networks because two friends of
the same individual get opportunities to meet and become friends
themselves~\cite{simmel1908sociology}, an idea later popularized by
\citet{granovetter1973strength} as \emph{triadic closure}.
\Citet{bavelas1950communication} surveyed several studies that examined how
small communication patterns used between individuals in a group affects their
ability to jointly perform tasks, and \citet{davis1971structure} found that
certain directed triads were less common than predicted by a null model,
providing numerical evidence for the ranked groups social network theory of
\citet{homans1950human}.  In more recent work, \citet{ugander2013subgraph} use
small subgraph frequencies within induced Facebook friendship subgraphs to
identify social processes.

In network analysis more broadly, higher-order structure is often described
through the idea of a \emph{network motif}.  \Citet{shen2002network}
introduced the term when they analyzed the frequency of higher-order
interactions through subgraph patterns, which they called network motifs, in
the transcription regulation network of \emph{E. coli}.\footnote{\Citet{shen2002network} borrowed the phrase motif from gene
  sequence analysis, where motifs are ``short, recurring patterns in DNA that
  are presumed to have a biological function''~\cite{d2006dna}.}  This style of
analysis was popularized in a landmark paper of \citet{milo2002network}
(published in \emph{Science} the same year), where they dubbed network motifs the
``building blocks of complex networks'' and identified motifs common in gene
regulation networks, food webs, and electronic circuits.  In follow-up work,
\citet{milo2004superfamilies} showed that the profile of z-scores (with
respect to samples of a configuration model with the same degree distribution)
for the frequencies of the 13 connected, 3-node directed subgraphs were
sufficient to distinguish biological, neurological, social, and linguistic
networks.  In subsequent research, important motifs have been identified in a
variety of domains, including brain science~\cite{sporns2004motifs},
ecology~\cite{camacho2007quantitative}, biology~\cite{wuchty2003evolutionary},
and social network analysis~\cite{leskovec2010signed}.

\Citet{milo2002network} defined network motifs as ``patterns of interconnections
occurring in complex networks at numbers that are significantly higher than
those in randomized networks.''  The ``significant'' qualifier in this
definition has caused some confusion around the terminology.  Some take the
qualifier to heart and use the term \emph{graphlet} to specify any small network
pattern and reserve the motif designation for graphlets that occur significantly
more frequently with respect to some null model and some measure of statistical
significance~\cite{prvzulj2004modeling}.\footnote{Presumably, the term graphlets
  is a play on wavelets, which are used in signal processing to localize time
  and frequency, but I have not found this explicitly mentioned in the
  literature.}  However, one could just as well call a graphlet a small
subgraph, and indeed, many
do~\cite{ugander2013subgraph,demeyer2013index,lahiri2007structure,pinar2017escape}.
Usually, there is an implicit assumption that motifs and graphlets are
\emph{connected} patterns, whereas general subgraphs do not carry this
assumption.  In this thesis, we just use ``motifs'' to specify patterns in
networks and sometimes use the term subgraph when discussing the more
mathematical points.  In \cref{ch:honc}, we also add to the vernacular by
defining an ``anchored motif'', which provides a formalism for specifying a
subset of the nodes in a motif relevant to certain computations.  We also use
``higher-order structures'' as a catch-all phrase.

Operating under the assumption that higher-order structures are important to
networks, the algorithmic data mining community has developed a litany of
methods for efficiently finding and counting them.  These include, to name a
few, algorithms for enumerating triangles that perform especially well given the
power-law degree distributions typical of many real-world
networks~\cite{schank2005finding,latapy2008main,berry2014why}, fast algorithms
for enumerating general motifs via principled
heuristics~\cite{wernicke2006fanmod,demeyer2013index,houbraken2014index,ahmed2015efficient},
algorithms that approximate the total number of
triangles~\cite{tsourakakis2009doulion,seshadhri2014wedge,lim2015mascot,de2016triest},
and algorithms that approximate the total number of motifs composed of 4 or more
nodes~\cite{jha2015path,bressan2017counting,wang2014efficiently,slota2013fast}.
For the computations in this thesis, we only need the method of \citet[Algorithm 1]{schank2005finding}
for enumerating triangles and the method of
\citet{chiba1985arboricity} for enumerating cliques.  We use these algorithms
because they are straightforward to implement and are fast for the datasets
that we analyze.

We summarize the setup---networks consisting of nodes and edges model a broad
range of systems, higher-order interactions (patterns between a small number of
nodes) are important to many networks, and we have algorithms adept at finding
and counting these higher-order structures.  However, higher-order structures
have not been well integrated into the analyses, models, and algorithms that we
actually use to study the structure of complex networks.  This thesis develops
new methods, or \emph{tools}, for analyzing network data based on higher-order
structures, i.e., \emph{for higher-order network analysis}.  Consequently, the
network analyst is able to examine data in terms of
the higher-order interactions that are important to him or her.

We demonstrate how our tools use higher-order interactions as a primitive in order
to gain new insights into complex systems.  For example, our analysis of the neural
network for the nematode worm \emph{C. elegans} in \cref{sec:honc_celegans}
finds a cluster of 20 neurons in the frontal section of the worm that is a
plausible control of nictation, a type of worm movement.  We discover this group
through a new algorithm that looks for subsets of the graph in which a certain
4-node ``bi-fan'' motif is contained.  The higher-order structure is key to the
discovery---we need to optimize an objective function for sets of nodes that
models the bi-fan motif---and \emph{not} simply edges, as is typically done---in
order to find this group.  Later, in \cref{sec:ccfs_empirical}, we find that
under a generalized ``closure'' model for the presence of edges in the network,
3-cliques (triangles) in \emph{C. elegans} are more common than expected, while
4-cliques are less common than expected, which challenges a longstanding view of
how edges tend to cluster together in network models of complex
systems~\cite{watts1998collective}.  In this case, the consideration of
higher-order structure at the level of 4-cliques was necessary for the finding.

Higher-order analysis is a small part of network science, but
higher-order thinking has still shown up in several contexts.  For example, motifs
have been used to generate features for both prediction
problems~\cite{milenkovic2008uncovering,bonato2014dimensionality,ugander2013subgraph,chakraborty2014automatic}
and unsupervised learning
tasks~\cite{yaverouglu2014revealing,henderson2012rolx}, motif frequencies are
used to fit parameters of random graph
models~\cite{gleich2012moment,benson2014learning,bickel2011method,wasserman1996logit},
and motif modeling is used to improve network alignment
algorithms~\cite{milenkovic2010optimal,mohammadi2016triangular}.  In addition,
there are several higher-order approaches that do not necessarily use motif
structure but are similar in spirit.  Examples include ``meta paths'' to
represent multi-relational networked data~\cite{zhang2014meta,sun2011pathsim},
higher-order and variable-order Markov chain
models~\cite{rosvall2014memory,chierichetti2012web,benson2017spacey,wu2016general}, and
``higher-order network'' models for aggregating temporal paths in networks where
edges have
timestamps~\cite{scholtes2014causality,scholtes2017network,scholtes2016higher}.

The new ideas in \cref{ch:honc,ch:hoccfs} also align with recent higher-order
network analyses that directly generalize classical ideas to account for
higher-order structures in the network.  For example, \citet{tsourakakis2015k}
generalized the (edge) densest subgraph problem to the $k$-clique densest
subgraph problem, \citet{zhang2012extracting} generalized (edge-based) $k$-cores
to triangle-based $k$-cores, and \citet{sariyuce2015finding} generalized
$k$-core and $k$-truss decompositions to account for clique containment with the
nucleus decomposition.  We describe our contributions in more detail in the
following section.

\section{Contributions}

The core contributions of this thesis are three new tools for higher-order
network analysis.  We contextualize and summarize them in
\cref{tab:contributions}.

\newcolumntype{P}[1]{>{\raggedright\let\newline\\\arraybackslash\hspace{0pt}}p{#1}}
\begin{table}[h]
\centering
\dualcaption{Tools for higher-order network analysis developed in this thesis}{}
\scalebox{0.92}{
\begin{tabular}{P{0.25\columnwidth} P{0.4\columnwidth} P{0.35\columnwidth}}
\toprule
{\small
Concept
}
& 
{\small
Classical tools
}
&
{\small
New tool
}
\\ \midrule
{\footnotesize
Discover modules (clusters of nodes) in a network using only the network topology.
}
&
{\footnotesize
Algorithms for optimizing edge-based objective functions such as modularity,
conductance, and normalized cut.

\cite{newman2004finding,von2007tutorial,leskovec2009community,andersen2006local}
}
&
{\footnotesize
Algorithms for optimizing \textbf{higher-order}, motif-based objective functions. 

\Cref{ch:honc}; \cite{benson2016higher,yin2017local}
}
\\ \midrule
{\footnotesize
Measure the tendency of edges in a network to cluster.
}
&
{\footnotesize
The clustering coefficient based on relative triangle frequencies.

\cite{watts1998collective,barrat2000properties}
}
&
{\footnotesize
\textbf{Higher-order} clustering coefficients based on
relative clique frequencies.

\Cref{ch:hoccfs}; \cite{yin2017higher,yin2017local}
}
\\ \midrule
{\footnotesize
Count the frequency of \textbf{higher-order} structures
in a network.
}
&
{\footnotesize
Motifs in static networks.

\cite{milo2002network,prvzulj2004modeling}
}
&
{\footnotesize
Motifs in temporal networks.

\cref{ch:tm}; \cite{paranjape2017motifs}
}
\\
\bottomrule
\end{tabular}
}
\label{tab:contributions}
\end{table}

In addition to developing the tools, we make a concerted effort to use the tools
to gain new insights into a variety of network datasets.  Thus, there is a
considerable data mining component to the chapters ahead, which are also
meaningful contributions.  The next few subsections provide additional details
on the tools and the insights that they provide.

\subsection{\cref{ch:honc} -- Higher-order network clustering}

The first tool is a framework for graph clustering based on motifs.  Graph
clustering is a broad research problem that assigns nodes in a graph to
clusters, where the clusters are meant to represent some module in the network.
Typically, the objective function modeling what it means to have good clusters
involves a combination of the number of \emph{edges} contained within the
clusters and the number of \emph{edges} that go between clusters.  We re-define
what it means to be a good cluster with an objective function that considers the
number of \emph{motifs} contained within the clusters and the number of
\emph{motifs} that go between clusters.  Thus, if a particular motif is
important for some domain, we can run an algorithm to optimize an objective
function that accounts for that motif.

Our main theoretical contributions include (i) defining the new objective
function, \emph{motif conductance}, which is a generalization of the classical
conductance measure of cluster quality; and (ii) a spectral algorithm for
finding sets with small motif conductance.  The algorithm is accompanied by an
approximation guarantee on the quality of the output clusters in terms of motif
conductance.  The algorithm and quality guarantee are generalizations of the
Fiedler partition or sweep cut procedure and the Cheeger inequality.

Using this framework we show that
\begin{itemize}
\item a particular directed triangular motif reveals aquatic layers in a food
  web;
\item in the \emph{C. elegans} neural system, a cluster we find based
  on a previously studied 4-node motif (the ``bi-fan'')
  is a plausible control mechanism for nictation, a type of movement in the worm;
\item clustering with length-2 paths automatically reveals the hub structure and
  geography of an air travel transportation reachability network;
\item clustering with the feedforward loop (\cref{fig:intro_fflC}) reveals known
  modules in transcription regulation networks with higher accuracy than
  edge-based methods; and
\item clustering with particular directed triangular motifs reveal anomalous clusters in
  the English Wikipedia web graph and the Twitter follower network.
\end{itemize}

We then extend the algorithm to handle the problem of \emph{localized
  clustering} or \emph{targeted clustering}, where the goal is to find a (small) cluster
of nodes containing a given seed node.  Again, existing approaches optimize the
classical conductance, and we instead optimize motif conductance.  Our
approach is a generalization of the approximate Personalized PageRank
method~\cite{andersen2006local}.

The algorithmic framework and applications appeared as a publication in
\emph{Science}~\cite{benson2016higher},\footnote{See the accompanying
  perspective piece by \citet{prvzulj2016network} for a broader context for this
  work.} which was joint work with David Gleich and Jure Leskovec.  The
extension to localized clustering will appear in a paper in the proceedings of
the 2017 KDD conference~\cite{yin2017local}, which was joint work with Hao Yin,
David, and Jure.  Hao implemented the localized algorithm, ran the experiments for the
results in \cref{sec:local}, and came up with the idea to look at recovery in the
planted partition model (\cref{sec:local_synth}).

\subsection{\cref{ch:hoccfs} -- \hoccfstitle}

The second tool is a measurement for the extent to which nodes in a network
cluster together.  This is a generalization of the classical clustering
coefficient that measures the fraction of length-2 paths that induce a triangle.
We reinterpret the clustering coefficient as a clique expansion and closure
process---a 2-clique (an edge) \emph{expands} with an adjacent edge and we check
if this structure \emph{closes} by forming a 3-clique (a triangle).  We
generalize this by considering an $\ell$-clique that expands with an adjacent
edge and checking if this structure closes by forming an $(\ell + 1)$-clique.
We call the fraction of ($\ell$-clique, adjacent edge) pairs that induce an
$(\ell + 1)$-clique the $\ell$th-order clustering coefficient.

We theoretically analyze the higher-order clustering coefficient in the
Erd\H{o}s-R\'enyi and small-world random graph models and empirically analyze
the higher-order clustering coefficient on three real-world networks.  We find that
although the \emph{C. elegans} neural network has high clustering in the
traditional sense, it does not have high third-order clustering (3-cliques tend
not to expand into 4-cliques).  This could arise from the fact that 4-cliques represent
redundant processing units and their absence leads to a more efficient neural
architecture.

We also make a connection between higher-order clustering coefficients and motif
conductance.  We show that if a network has a large $\ell$th-order clustering
coefficient, then there is a node whose 1-hop neighborhood subgraph has small motif
conductance for the $\ell$-clique motif.  This is a generalization of the $\ell = 2$
case studied by \citet{gleich2012vertex}.  In fact, I came up with the
definition for the higher-order clustering coefficient when I was trying to
generalize a lemma from \citet{gleich2012vertex} to the $\ell = 3$ case (for
triangle conductance).

The definitions and analysis of the higher-order clustering coefficient is based
on a paper with Hao Yin and Jure Leskovec that has been submitted for
publication and is currently on arXiv~\cite{yin2017higher}.  The connection to
motif conductance is based on part of a paper with Hao, Jure, and
David Gleich that will appear in the proceedings of the 2017 KDD
conference~\cite{yin2017local}.  Hao helped a substantial amount on the proofs
of \cref{prop:ccf_er,prop:ccf_sw,thm:nbrhd_main}.

\subsection{\cref{ch:tm} -- \tmtitle}

The third tool is a definition of motifs in temporal networks along with
algorithms for efficiently counting them.  The goal of the research in
\cref{ch:tm} is to provide the foundations of higher-order network analysis---a
definition for higher-order structures and efficient algorithms for finding
them---for a broader class of network datasets, i.e., those that contain temporal
information.  We consider a temporal network to be a collection of $(u, v, t)$
tuples (temporal edges), where $u$ and $v$ are elements of a node set $V$ and $t
\in \mathbb{R}_{+}$ is a timestamp.  We define a temporal motif by a multigraph,
an ordering on the edges in the multigraph, and a time window $\delta$, and we
define an instance of the temporal motif in a temporal network to be a subset of
the temporal edges that match the edge pattern of the multigraph, appear in the
specified order, and all occur within $\delta$ time units of each other.  We
provide a general algorithm for efficiently counting the number of instances of
temporal motifs in a given temporal network along with specialized fast
algorithms for certain classes of motifs.  We also show some basic higher-order
analyses on several temporal network datasets.

\Cref{ch:tm} is based on a paper with Ashwin Paranjape and Jure Leskovec that
appeared in the proceedings of the 2017 WSDM
conference~\cite{paranjape2017motifs}.  Ashwin helped with the design,
implementation, and analysis of the algorithms and executed the experiments in
\cref{sec:tm_empirical}.

\subsection{Additional artifacts and impact}

In addition to the content of this thesis and the associated publications, other
artifacts of this research include the following.

\begin{itemize}
\item Implementations of the motif-based spectral clustering algorithm
 \begin{itemize}
 \item for the SNAP software library, which is available at\\
   \url{https://github.com/snap-stanford/snap};
 \item in Julia, which is available at\\
   \url{https://github.com/arbenson/higher-order-organization-julia}; and
 \item in MATLAB, which is available at\\
   \url{https://github.com/arbenson/higher-order-organization-matlab}.
 \end{itemize}
\item Implementations of the temporal motif counting algorithms for the
  SNAP software library, which is available at\\
  \url{https://github.com/snap-stanford/snap}. 
 \item Julia notebooks to reproduce the results in the main text of
   \citet{benson2016higher}, which are available at\\
   \url{https://github.com/arbenson/higher-order-organization-julia}.
\item Metadata for the nodes in the transportation reachability network of
  \citet{frey2007clustering} (city latitudes, longitudes, and metropolitan
  populations), which is available at\\
  \url{http://snap.stanford.edu/data/reachability.html}.
\item Metadata for the nodes in the Florida Bay food web
  \citet{frey2007clustering} (group classification), which is available at\\
  \url{http://snap.stanford.edu/data/Florida-bay.html}.  
\item The $\emaileucore$ network dataset with department
labels for all of the nodes, which is available at\\
  \url{http://snap.stanford.edu/data/email-Eu-core.html}.  
\item The $\wikicat$ network dataset with article names and category
  classifications for all of the nodes, which is available at\\
  \url{http://snap.stanford.edu/data/wiki-topcats.html}.  
\item The $\stackoverflow$ temporal network dataset, which is available at\\
  \url{http://snap.stanford.edu/data/sx-stackoverflow.html}.
\item The $\mathoverflow$ temporal network dataset, which is available at\\
  \url{http://snap.stanford.edu/data/sx-mathoverflow.html}.
\item The $\superuser$ temporal network dataset, which is available at\\
  \url{http://snap.stanford.edu/data/sx-superuser.html}.
\item The $\askubuntu$ temporal network dataset, which is available at\\
  \url{http://snap.stanford.edu/data/sx-askubuntu.html}.
\item The $\emaileu$ temporal network dataset and 4 department-level
  subnetworks, which are available at\\
  \url{http://snap.stanford.edu/data/email-Eu-core-temporal.html}.
\item The $\wikitalk$ temporal network dataset, which is available at\\
  \url{http://snap.stanford.edu/data/wiki-talk-temporal.html}.
\end{itemize}

Furthermore, the ideas of this thesis have already had impact on the broader
research community.  As an example, \citet{meier2016motif} used the motif-based
clustering algorithm in functional connectivity networks of brains and found
symmetry patterns between the hemispheres of the brain.
In their article ``Network analytics in the age of big data''
published in \emph{Science}, \citet{prvzulj2016network} contextualize the
framework of \cref{ch:honc} as an important step towards understanding
the large-scale datasets coming from a ``complex world of interconnected
entities.''  As discussed at the beginning of this introduction, network models
and analysis are fundamental because of their wide applicability, and
\Citet{prvzulj2016network} emphasize this point when (generously) discussing our work:
``The importance of this result lies in its applicability to a broad range of
networked data that we must understand to answer fundamental questions facing
humanity today, from climate change and impacts of genetically modified
organisms, to the environment, to food security, human migrations, economic and
societal crises, understanding diseases, aging, and personalizing medical
treatments.''
Indeed, I hope that the ideas that follow are a step towards tackling these problems.

\chapter{\honctitle}
\label{ch:honc}

\section{Organization via higher-order structures}

Networks are a standard representation of data throughout the sciences, and
higher-order connectivity patterns are essential to understanding the
fundamental structures that control and mediate the behavior of many complex
systems~\cite{milo2002network,mangan2003coherent,yang2014overlapping,holland1970method,rosvall2014memory,prvzulj2004modeling}.
The most common higher-order structures are small network subgraphs, which we
refer to as network motifs (\cref{fig:honc_basics_motifs}).  Network motifs are
considered building blocks for complex
networks~\cite{milo2002network,yaverouglu2014revealing}.  For example,
feedforward loops (\cref{fig:honc_basics_motifs}, $M_{5}$) have proven
fundamental to understanding transcriptional regulation
networks~\cite{mangan2003structure}, triangular motifs
(\cref{fig:honc_basics_motifs}, $M_{1}$--$M_{7}$) are crucial for social
networks~\cite{holland1970method,davis1971structure}, open bidirectional wedges
(\cref{fig:honc_basics_motifs}, $M_{13}$) are key to structural hubs in the
brain~\cite{honey2007network}, and two-hop paths (\cref{fig:honc_basics_motifs},
$M_{8}$--$M_{13}$) are essential to understanding air traffic
patterns~\cite{rosvall2014memory}.  While network motifs have been recognized as
fundamental units of networks, the higher-order \emph{organization} of networks
at the level of network motifs has largely remained an open question.

In this chapter, we use higher-order network structures to gain new insights into the
organization of complex systems.  We develop a framework that identifies
clusters of network motifs.  For each network motif
(\cref{fig:honc_basics_motifs}), a different higher-order clustering may be
revealed (\cref{fig:honc_basics_cluster}), which means that different
organizational patterns are exposed depending on the chosen motif.

\begin{figure}[h]
\centering
\phantomsubfigure{fig:honc_basics_motifs}
\phantomsubfigure{fig:honc_basics_cluster}
\phantomsubfigure{fig:honc_basics_framework}
\includegraphics[width=\textwidth]{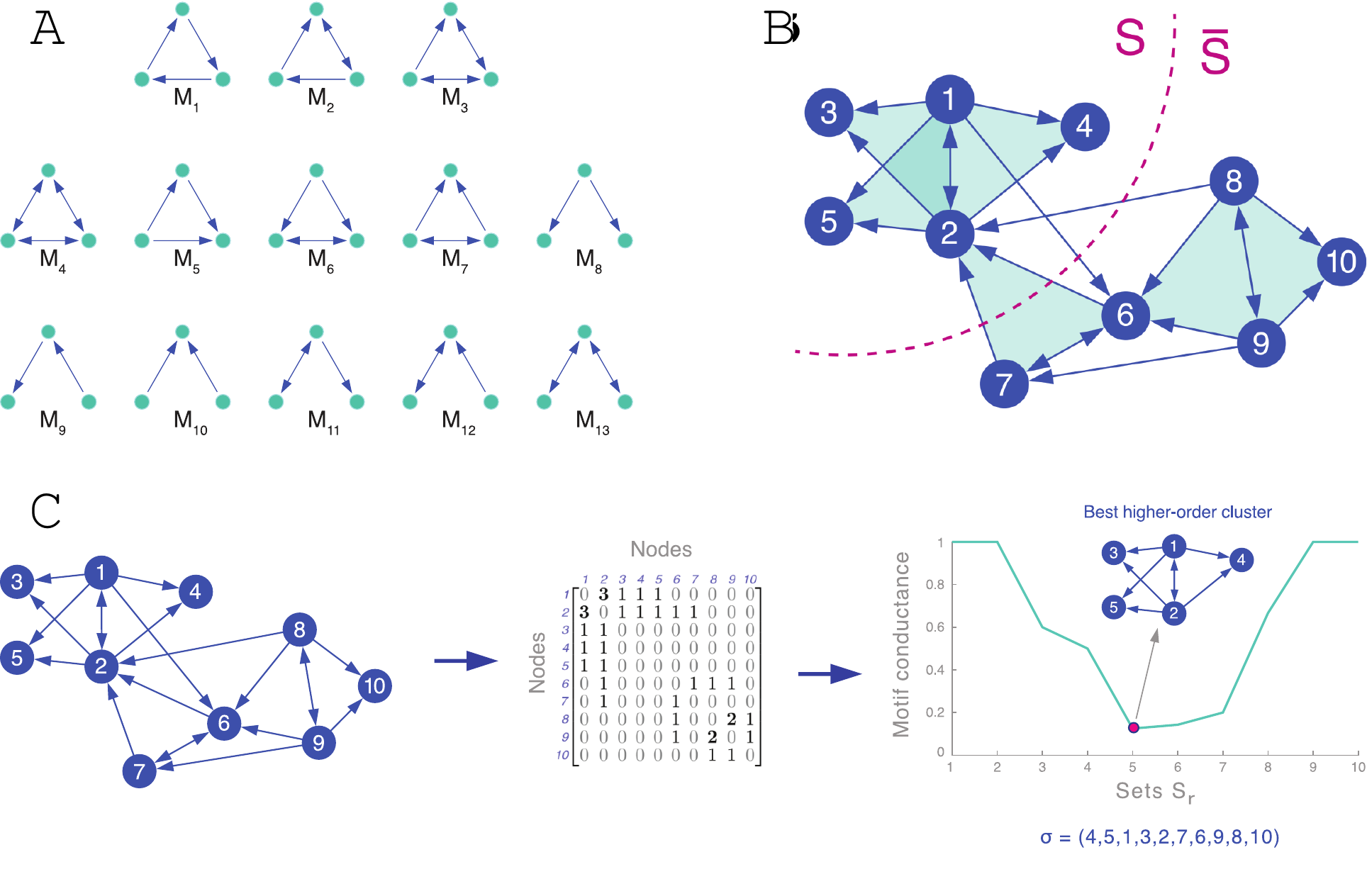}
\dualcaption{Higher-order network structures and the higher-order network clustering framework}{%
{\bf A:}
Higher-order structures are captured by network motifs.  For example, all 13
connected three-node directed motifs are shown here.
{\bf B:} 
Clustering of a network based on motif $M_7$.  For a given motif $M$, our
framework aims to find a set of nodes $S$ that minimizes motif conductance,
$\phi_M(S)$, which we define as the ratio of the number of motifs cut (filled
triangles cut) to the minimum number of nodes in instances of the motif in
either $S$ or $\bar{S}$ (\cref{sec:honc_derivation} provides the formalities).
In this case, there is one motif cut.
{\bf C:}
The higher-order network clustering framework.  Given a graph and a motif of
interest (in this case, $M_7$), the framework forms a motif adjacency matrix
($W_M$) by counting the number of times two nodes co-occur in an instance of the
motif.  An eigenvector of a Laplacian transformation of the motif adjacency
matrix is then computed.  The ordering $\sigma$ of the nodes provided by the
components of the eigenvector produces nested sets $S_r = \{\sigma_1, \ldots,
\sigma_r\}$ of increasing size $r$.
We prove in \cref{sec:motif_cheeger} that the set $S_r$ with the smallest
motif-based conductance, $\phi_M(S_r)$, is a near-optimal higher-order cluster.
}
\label{fig:honc_basics}
\end{figure}

\clearpage

Conceptually, given a network motif $M$, our framework searches for a cluster of
nodes $S$ with two goals.  First, the nodes in $S$ should participate in many
instances of $M$.  Second, the set $S$ should avoid cutting instances of $M$,
which occurs when only a subset of the nodes from a motif are in the set $S$
(\cref{fig:honc_basics_cluster}).  More precisely, given a motif $M$, the
higher-order clustering framework aims to find a cluster (defined by a set of
nodes $S$) that minimizes the following ratio:
\begin{equation}\label{eqn:mcond}
\mmcond{S} = \mmcut{S, \bar{S}} / \min(\mmvol{S}, \mmvol{\bar{S}}),
\end{equation}
where $\bar{S}$ denotes the remainder of the nodes (the complement of $S$),
$\mmcut{S, \bar{S}}$ is the number of instances of motif $M$ with at least one
node in $S$ and one in $\bar{S}$, and $\mmvol{S}$ is the number of nodes in
instances of $M$ that reside in $S$.  \Cref{eqn:mcond} is a
generalization of the conductance metric in spectral graph theory, one of the
most useful graph partitioning scores~\cite{schaeffer2007graph}. We refer to
$\mmcond{S}$ as the motif conductance of $S$ with respect to $M$.
(The formal mathematics of these definitions are discussed later in \cref{sec:motif_def,sec:motif_cond}).

Finding the exact set of nodes $S$ that minimizes the motif conductance is
NP-hard~\cite{wagner1993between}.  To approximately minimize \cref{eqn:mcond}
and hence identify higher-order clusters, we develop an optimization framework
that provably finds near-optimal clusters (\cref{sec:honc_derivation}).  We
extend the spectral graph clustering methodology, which is based on the
eigenvalues and eigenvectors of matrices associated with the
graph~\cite{schaeffer2007graph}, to account for higher-order structures in
networks.  The resulting method maintains the properties of traditional spectral
graph clustering: computational efficiency, ease of implementation, and
mathematical guarantees on the near-optimality of obtained clusters.
Specifically, a cluster $S$ identified by our higher-order clustering framework
satisfies a motif Cheeger inequality:
\begin{equation}
\mmcond{S} \le 2\sqrt{\phi_M^*},
\end{equation}
where $\phi_M^* = \min_{T \subset V} \mmcond{T}$ is the optimal motif
conductance over all possible sets $T$.  In other words, our optimization
framework can find a cluster that is at most a quadratic factor away from
optimal.  We prove this inequality in \cref{sec:motif_cheeger}.

The algorithm (illustrated in \cref{fig:honc_basics_framework}) efficiently
identifies a cluster of nodes $S$ as follows:

\begin{enumerate}
    \item Given a network and a motif $M$ of interest, form the motif
          adjacency matrix $W_M$ whose entries $(i,j)$ are the
          co-occurrence counts of nodes $i$ and $j$ in the motif $M$
          \begin{equation}\label{eqn:weighting}
            (W_M)_{ij}
            = \text{number of instances of $M$ that contain nodes $i$ and $j$}.
          \end{equation}
                   
    \item Compute the spectral ordering $\sigma$ of the nodes from the
          normalized motif Laplacian matrix constructed via $W_M$.  The
          normalized motif Laplacian matrix is
          $\normmotiflap = D^{-1/2} (D - W_M) D^{-1/2}$,
          where $D$ is a diagonal matrix
          with the row-sums of $W_M$ on the diagonal ($D_{ii} = \sum_j (W_M)_{ij}$),
          and $D^{-1/2}$ is the same matrix with the
          inverse square-roots on the diagonal
          ($D^{-1/2}_{ii} = 1/\sqrt{\sum_{j} (W_M)_{ij}}$).
          The spectral ordering $\sigma$ is the by-value ordering of
          $D^{-1/2}z$, where $z$ is the eigenvector corresponding to
          the second smallest eigenvalue of $\normmotiflap$, i.e.,
          $\sigma_i$ is the index of $D^{-1/2}z$ with the $i$th
          smallest value.
 	
    \item Find the prefix set of
          $\sigma$ with the smallest motif conductance, formally:\
          $S := \arg\min_{r} \mmcond{S_r}$, where
          $S_r = \{\sigma_1, \ldots, \sigma_r\}$.
\end{enumerate}

For triangular motifs, the algorithm scales to networks with billions of edges
and typically only takes several hours to process graphs of such size
(Section~\ref{sec:honc_scalability}).  On smaller networks with hundreds of
thousands of edges, the algorithm can process some motifs up to size 9
(specifically, we look at cliques up to size 9).  While the worst-case
computational complexity of the algorithm for triangular motifs scales
as $m^{1.5}$, where $m$ is the number of edges in the network, in practice
the algorithm is much faster.  Analyzing 16 real-world networks where the
number of edges $m$ ranges from 159,000 to 2 billion we found the computational
complexity to scale as $m^{1.2}$.  Moreover, the algorithm can easily be
parallelized and sampling techniques can be used to further improve
performance~\cite{seshadhri2014wedge}.

The framework can be applied to directed, undirected, and weighted networks as
well as motifs.  Moreover, it can also be applied to networks with positive and
negative signs on the edges, which are common in social networks (friend vs.~foe
or trust vs.~distrust edges) and metabolic networks (edges signifying activation
vs.~inhibition).  For example, in \cref{sec:honc_yeast}, we analyze a transcriptional regulation network
with signed edges where the sign corresponds to activation or suppression of genetic
transcription.  The framework can be used to
identify higher-order structure in networks where domain knowledge suggests the
motif of interest.  This is the case for our analysis of a transcriptional
regulation network (\cref{sec:honc_yeast}), a network of neural connections in
the nematode worm \emph{C. elegans} (\cref{sec:honc_celegans}), and a
transportation reachability network (\cref{sec:honc_airports}).  We also show
that when a domain-specific higher-order pattern is not known in advance, the
framework can also serve to identify which motifs are important for the modular
organization of a given network; for example, in \cref{sec:honc_foodweb}, we
identify a motif that organizes a food web.

Such a general framework allows for a study of complex higher-order
organizational structures in a number of different networks using individual
motifs and sets of motifs.  The framework and mathematical theory immediately
extend to other spectral methods such as algorithms for finding
overlapping clusters~\cite{whang2015non} and localized algorithms that find clusters
around a seed node~\cite{andersen2006local} (we develop the localized clustering extension
in \cref{sec:local}).  To find several clusters, one can use
embeddings from multiple eigenvectors and $k$-means
clustering~\cite{ng2001spectral} or apply recursive
bi-partitioning~\cite{boley1998principal} (see \cref{sec:honc_multiple_clusters}
for details).

Our higher-order network clustering framework unifies motif analysis and network
partitioning---two fundamental tools in network science---and reveals new
organizational patterns and modules in complex systems.  Prior efforts along
these lines do not provide worst-case performance guarantees on the obtained
clustering~\cite{serrour2011detecting}, do not reveal which motifs organize the
network~\cite{michoel2011enrichment}, or rely on expanding the size of the
network~\cite{benson2015tensor,krzakala2013spectral}.  Our theoretical results
in \cref{sec:motif_cheeger} also explain why classes of hypergraph partitioning
methods are more general than previously assumed and how motif-based
clustering provides a rigorous framework for the special
case of partitioning directed graphs (\cref{sec:honc_discussion}).

\section{The motif-based spectral clustering algorithm}
\label{sec:honc_derivation}

We now cover the background and theory for deriving and understanding our
method.  We start by reviewing the graph Laplacian and cut, volume, and
conductance measures for sets of vertices in a graph.  We then define network
motifs in and generalize the notions of cut, volume, and conductance to motifs.
From this, we present the spectral algorithm for finding sets with small motif
conductance and the associated theory (motif Cheeger inequality).  We then
analyze the complexity of the algorithm and discuss some extensions of the
framework.

\subsection{Review of cuts, volumes, conductance and the graph Laplacian for weighted, undirected graphs}

Consider a weighted, undirected graph $G = (V, E)$, with $\lvert V \rvert = n$.
Further assume that $G$ has no isolated nodes.  Let $W$ encode the weights of
the graph:
\begin{equation}
W_{ij} = W_{ji} = \text{weight of edge (i, j)}.
\end{equation}
The diagonal degree matrix $D$ is defined by $D_{ii} = \sum_{j=1}^{n} W_{ij}$,
and the graph Laplacian is defined as $L = D - W$.  We now relate these matrices
to the conductance $\cond{S}$ of a subset $S \subset V$ of nodes in $G$.
\begin{align}
\cond{S} &= \cut{S, \bar{S}} / \min(\vol{S}, \vol{\bar{S}}), \\
\cut{S, \bar{S}} &= \sum_{i \in S,\;j \in \bar{S}} W_{ij}, \\
\vol{S} &= \sum_{i \in S} D_{ii}
\end{align}
Here, $\bar{S} = V \backslash S$ is the complement of the set $S$.
(Note that conductance is a symmetric measure in $S$ and $\bar{S}$, i.e.,
$\cond{S} = \cond{\bar{S}}$.)
Conceptually, the cut and volume measures are defined as follows:
\begin{align}
\cut{S, \bar{S}} &= \text{weighted sum of edges that are cut} \label{eqn:cut_words} \\
\vol{S} &= \text{weighted number of edge end points in $S$} \label{eqn:vol_words}.
\end{align}
In general, sets with low conductance are considered good
clusters~\cite{schaeffer2007graph}.  They model the conceptual definition that
good clusters are isolated from the rest of the network (there are
few edges leaving $S$ when the cut is small) and not too small (there is enough
volume in $S$).

Since we have assumed $G$ has no isolated nodes, $\vol{S} > 0$.  If $G$ is
disconnected, it is easy to see that $\cond{C} = 0$ for any connected component
$C$.  Thus, we usually consider breaking $G$ into connected components as a
pre-processing step for algorithms that try to find low-conductance sets.

We now relate the cut metric to a quadratic form on $L$.  Later, we will derive
a similar form for a motif cut measure.  For any vector $y \in \mathbb{R}^{n}$,
\begin{equation}
y^TLy = \sum_{(i, j) \in E}w_{ij}(y_i - y_j)^2.
\end{equation}
Now, define $x$ to be an indicator vector for a set of nodes $S$
i.e., $x_i = 1$ if node $i$ is in $S$ and $x_i = -1$ if node $i$ is in
$\bar{S}$.  If an edge $(i, j)$ is cut, then $x_i$ and $x_j$ take
different values and $(x_i - x_j)^2 = 4$; otherwise, $(x_i - x_j)^2 = 0$.  Thus,
\begin{equation}\label{eqn:quad_cut}
x^TLx = 4 \cdot \cut{S, \bar{S}}.
\end{equation}

Over the next few sections, we generalize these ideas for motifs.

\subsection{Definition of network motifs}
\label{sec:motif_def}

We now define network motifs as used in this chapter.  We consider
motifs to be a pattern of edges on a small number of nodes.  Formally, we define
a motif on $k$ nodes by a tuple $(B, \anchorset)$, where $B$ is a $k \times k$
binary matrix and $\anchorset \subset \{1, 2, \ldots, k\}$ is a set of anchor
nodes.  The matrix $B$ encodes the edge pattern between the $k$ nodes, and
$\anchorset$ labels a relevant subset of nodes for defining motif conductance.
(More specifically, the anchored nodes are the ones that contribute to the cut
and volume counts; we make this formal later.)
In many cases, $\anchorset$ is just the entire set of nodes. As an example, motif $M_4$ (the directed triangle with all bidirectional links;
\cref{fig:honc_basics_motifs}), where all nodes are anchors is denoted by
\[
(B, \anchorset) =
\left(
\begin{bmatrix} 0 & 1 & 1 \\ 1 & 0 & 1 \\ 1 & 1 & 0 \end{bmatrix},
\{1, 2, 3\}
\right)
\]

Let $\chi_{\anchorset}$ be a selection function that takes the subset of a $k$-tuple
indexed by $\anchorset$, and let $\text{set}(\cdot)$ be the operator that takes
an (ordered) tuple to an (unordered) set.  Specifically,
\[
\text{set}((v_1, v_2, \ldots, v_k)) = \{v_1, v_2, \ldots, v_k\}.
\]
The set of all motifs in an unweighted (possibly directed) graph with node set $V$
and adjacency matrix $A$, denoted $\mathcal{M}(B, \anchorset)$, is defined by
\begin{align}
\mathcal{M}(B, \anchorset) =
\{
&(\setof{v}, \setof{\chi_{\anchorset}(v)}) \;\mid\; \label{eqn:motif_set} \\
& v \in V^k,\; v_1, \ldots, v_k \text{ distinct},\; A_{v, v} = B \nonumber
\},
\end{align}
where $A_{v, v}$ is the $k \times k$ adjacency matrix on the subgraph induced by
the $k$ nodes of the \emph{ordered} vector $v$.  Sometimes, a distinction is
made between a \emph{functional} and a \emph{structural}
motif~\cite{sporns2004motifs} (or a subgraph and an induced
subgraph~\cite{inokuchi2000apriori}) to distinguish whether a motif specifies
simply the existence of a set of edges (functional motif or subgraph) or the
existence and non-existence of edges (structural motif or induced subgraph).  By
asserting the equivalency $A_{v, v} = B$, we refer to structural motifs in this
work.  However, a function motif is just a union of structural motifs.  Our
clustering framework allows for the simultaneous consideration of several motifs
(see \cref{sec:honc_extensions}), so we have not lost any generality and can
cluster based on structural motifs.

\Cref{fig:motifs_formal} illustrates these definitions.  The set operator is
a convenient way to avoid duplicates when defining $\mathcal{M}(B, \anchorset)$
for motifs exhibiting symmetries.  Henceforth, we will just use $\minstance$ to
denote $(\setof{v}, \setof{\chi_{\anchorset}(v)})$ when discussing elements of
$\mathcal{M}(B,\anchorset)$.  Furthermore, we call any
$\minstance \in \mathcal{M}(B, \anchorset)$ a \emph{motif instance}.  When $B$
and $\anchorset$ are arbitrary or clear from context, we will simply denote the
motif set by $\mathcal{M}$.

We call motifs where $\anchornodes = v$ \emph{simple motifs}
(\cref{fig:motifs_formalB}) and motifs where $\anchornodes \neq v$ \emph{anchored
motifs} (\cref{fig:motifs_formalC}).  Existing motif analysis has only analyzed
simple motifs.  However, the anchored motif provides us with a more general
framework, and we use an anchored motif for the analysis of the transportation
reachability network in \cref{sec:honc_airports}.


\definecolor{mygreen}{RGB}{81,194,166}
\begin{figure}[tb]
\centering
\phantomsubfigure{fig:motifs_formalA}
\phantomsubfigure{fig:motifs_formalB}
\phantomsubfigure{fig:motifs_formalC}
\includegraphics[width=\columnwidth]{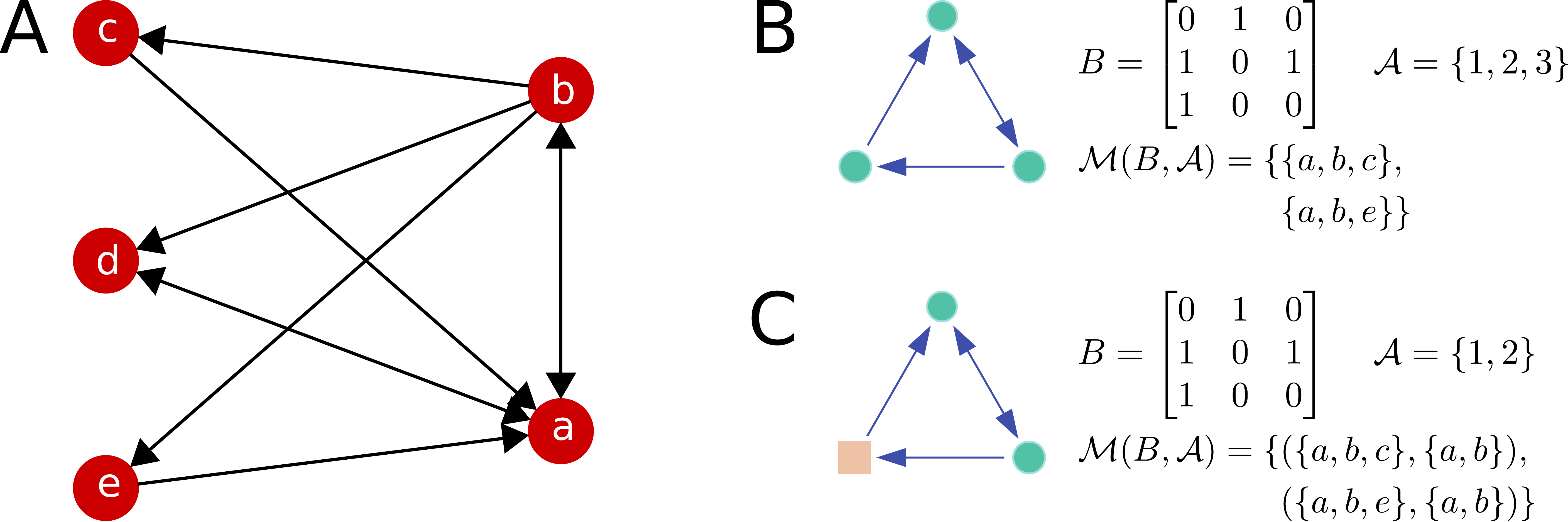}
\dualcaption{Formal motif definition}{%
A motif consist of a $k \times k$ binary matrix
$B$ and a set of anchor nodes $\mathcal{A} \subseteq \{1, \ldots, k\}$.
{\bf A:}
Example network.
{\bf B:}
Example simple motif, where the anchor nodes are the entire index set.  In this
case, we can think of motifs as just subsets of nodes.  The set
$\mathcal{M}(B, \mathcal{A})$ lists all of the subsets of nodes where
the induced subgraph is isomorphic to the graph defined by the matrix $B$.
{\bf C:}
Example anchored motif, where the \textcolor{mygreen}{green} nodes are the
anchor nodes.  There are
two instances of the anchored motif in the graph in A.  The tuple
$(\{a, b, d\}, \{a, b\})$
is not included in the set of motif instances because the induced subgraph on
the nodes $a$, $b$, and $d$ is not isomorphic to the graph defined by the matrix $B$.
}
\label{fig:motifs_formal}
\end{figure}

\clearpage

\subsection{Motif conductance}
\label{sec:motif_cond}

Recall that the key definitions for defining conductance are the notions of cut
and volume.  For an unweighted graph, these are
\begin{align}
\cond{S, \bar{S}} &= \cut{S, \bar{S}} / \min(\vol{S}, \vol{\bar{S}}), \\
\cut{S, \bar{S}} &= \text{number of edges cut}, \\
\vol{S} &= \text{number of edge end points in $S$}.
\end{align}

Our syntactic definition of motif conductance simply replaces an edge with a
motif instance of type $M$:
\begin{align}
\mmcond{S} &= \mmcut{S, \bar{S}} / \min(\mmvol{S}, \mmvol{\bar{S}}), \\
\mmcut{S, \bar{S}} &= \text{number of motif instances cut}, \\
\mmvol{S} &= \text{number of motif instance end points in $S$}.
\end{align}

We say that a motif instance is cut if there is at least one anchor node in $S$
and at least one anchor node in $\bar{S}$ and we say that an ``end point in $S$''
is any anchor node of a motif instance in $S$.  We can formalize this when given a
motif set $\mathcal{M}$ as in \cref{eqn:motif_set}:
\begin{align}
\mmcut{S, \bar{S}} &= 
\sum \nolimits_{\minstance \in \mathcal{M}}\indicator{\exists\; i, j  \in \anchornodes \;\mid\; i \in S, j \in \bar{S}} \label{eqn:motif_cut},  \\
\mmvol{S} &= \sum \nolimits_{\minstance \in \mathcal{M}} \sum \nolimits_{i \in \anchornodes} \indicator{i \in S},
\end{align}
where $\indicator{s}$ is the truth-value indicator function on $s$, i.e.,
$\indicator{s}$ takes the value $1$ if the statement $s$ is true and $0$
otherwise.  The motif cut measure only counts an instance of a
motif as cut if the anchor nodes are separated, and the motif volume counts the
number of anchored nodes in the set.  However, two nodes in an anchor set may be a
part of several motif instances.  Specifically, following the definition in
\cref{eqn:motif_set}, there may be many different $v$ with the
same $\anchornodes$, and the nodes in $\anchornodes$ still get counted
proportional to the number of motif instances.

\subsection{The motif adjacency matrix and the motif Laplacian}
\label{sec:motif_adjacency}

Given an unweighted, directed graph and a motif set $\mathcal{M}$, we conceptually
define the motif adjacency matrix by
\begin{align}
\motifweightedij = &\text{ number of motif instances in $\mathcal{M}$ where} \label{eqn:informal_weighting} \\
&\text{ $i$ and $j$ participate as anchor nodes in the motif}. \nonumber
\end{align}
Or, formally,
\begin{equation}\label{eqn:formal_weighting}
\motifweightedij = 
\begin{cases}
\displaystyle \sum_{\minstance \in \mathcal{M}} \indicator{\{i, j\} \subset \chi_{\anchorset}(v)}, & i \neq j \\ 
0 & \text{otherwise} \\
\end{cases}
\end{equation}
Note that weight is added to $\motifweightedij$ only if $i$ and
$j$ appear in the anchor set.  This is important for the transportation
reachability network analyzed in \cref{sec:honc_airports}, where weight is
added between cities $i$ and $j$ based on the number of intermediary cities that
can be traversed between them.

Next, we define the motif diagonal degree matrix as
\[
D_M = \text{diag}(W_Me),
\]
where $e$ is the vector of all ones and the motif Laplacian as
\[
L_M = D_M - W_M.
\]
Finally, we define the normalized motif Laplacian as
\[
\normmotiflap = D_M^{-1/2}L_MD_M^{-1/2} = I - D_M^{-1/2}W_MD_M^{-1/2}.
\]
In the next section, we develop our clustering algorithm based on computing
an eigenvector of $\normmotiflap$.

\subsection{The spectral algorithm for finding a single cluster}
\label{sec:honc_alg}

We are now ready to describe the algorithm for finding a single cluster in a
graph, which we present in \cref{alg:motif_fiedler}.
The algorithm finds a partition of the nodes into $S$ and $\bar{S}$.
The motif conductance is symmetric in the sense that
$\mmcond{S} = \mmcond{\bar{S}}$, so either set of nodes ($S$ or $\bar{S}$) could be
interpreted as a cluster.  However, in practice, it is common that one set is
substantially smaller than the other.  We consider this smaller set to represent
a module in the network.

\begin{algorithm}[h]\algoptions
  \KwIn{Directed, unweighted graph $G$ and motif $M$}
  \KwOut{Motif-based cluster (subset of nodes in $G$)}
   $(W_M)_{ij} \leftarrow \text{number of instances of $M$ that contain nodes $i$ and $j$}.$\;
   $D_M \leftarrow \diag{W_M \allones}$.\;
   $\normmotiflap \leftarrow I - D_M^{-1/2}W_MD_M^{-1/2}$.\;
   $z \leftarrow $ eigenvector of second smallest eigenvalue for $\normmotiflap$\;
  $\sigma_i \leftarrow$ index of $D_M^{-1/2}z$ with $i$th smallest value\;
  \mycomment{Sweep procedure}\;
   $S_l^* \leftarrow \arg\min_{S_l} \cond{S_l}$, where $S_l = \{\sigma_1, \ldots, \sigma_l\}$ \label{line:sweep}\;
  \eIf{$\lvert S_l^* \rvert < \lvert \bar{S}_l^* \rvert$}{
      \Return $S_l^*$\;
  }{
      \Return $\bar{S}^*_l$\;
  }      
  \dualcaption{Motif-based clustering algorithm for finding a single
    cluster}{The algorithm forms a weighted undirected graph and
    then runs a classical spectral partitioning method on the weighted
    graph.  Our theory in \cref{sec:motif_cheeger} gives bounds on
    the cluster quality in terms of motif conductance.}
  \label{alg:motif_fiedler}
\end{algorithm}

The algorithm itself is simple.  First, the algorithm pre-processes the
(possibly) directed graph into the weighted adjacency matrix.  Second, it uses a
well-known spectral graph partitioning method on the weighted graph
graph~\cite{alon1985lambda1,alon1986eigenvalues,mihail1989conductance,fiedler1973algebraic,chung2007four}.
The novel insight here is that this process provides theoretical results in
terms of motif conductance, as we will show in \cref{sec:motif_cheeger}.  Specifically,
we show that when the motif $M$ has three nodes, the output set $S$ satisfies
\begin{equation}
\mmcond{S} \le 2\sqrt{\phi_M^*},
\end{equation}
where $\phi_M^* = \min_{T \subset V}\mmcond{T}$ is the smallest motif
conductance over all possible sets of nodes.  This is a generalization of
the celebrated Cheeger inequality for graphs~\cite{cheeger1970lower,alon1985lambda1,mihail1989conductance}.
Hence, we call this bound a \emph{motif Cheeger inequality}, which we
will prove and discuss in \cref{sec:motif_cheeger}

Finally, we find that it is often informative to look at all conductance values found from
the ``sweep'' procedure in this algorithm (\cref{line:sweep}), where the
algorithm sweeps over several sets and picks out the one with smallest motif
conductance.  We refer to a plot of $\cond{S_l}$ versus $l$ as a \emph{sweep
profile plot}.  In the following subsection, we show that when the motif has
three nodes, the conductance of any set $S$ in the weighted graph equals the
motif conductance in the original graph $G$.  Thus, in this case, the sweep
profile shows how motif conductance varies with the size of the sets in
\cref{alg:motif_fiedler}.

\subsection{Interlude for matrix computations}
\label{sec:honc_matrix}

Before providing the theoretical analysis for \cref{alg:motif_fiedler}, we use
this subsection to show how, for several motifs, the motif adjacency matrix
$W_M$ (\cref{eqn:informal_weighting}) has a clean formula in terms of simple
matrix computations.  Let $A$ be the adjacency matrix for the original graph,
and let $U$ and $B$ be the adjacency matrix of the unidirectional and
bidirectional links of $G$.  Formally, $B = A \circ A^T$ and $U = A - B$, where
$\circ$ denotes the Hadamard (entry-wise) product.
\Cref{tab:matrix_interp} lists the formula of $W_M$ for all 7 of
the directed triangle motifs in terms of matrix computations involving $U$ and $B$.

\begin{table}[h]
\centering
\dualcaption{Matrix computations for the weighted motif adjacency matrix}{For
the adjacency matrix $A$, $B = A \circ A^T$ and $U = A - B$.}
%
%
\begin{tabular}{c @{\hskip 1cm} l @{\hskip 1.25cm} l}
\toprule
Motif & Matrix computations & $W_M =$ \\ \midrule
$M_{1}$  & $C = (U \cdot U) \circ U^T$ & $C + C^T$ \\
$M_{2}$  & $C = (B \cdot U) \circ U^T + (U \cdot B) \circ U^T + (U \cdot U) \circ B$ & $C + C^T$ \\
$M_{3}$  & $C = (B \cdot B) \circ U + (B \cdot U) \circ B + (U \cdot B) \circ B$ & $C + C^T$ \\
$M_{4}$  & $C = (B \cdot B) \circ B$ & $C$ \\
$M_{5}$  & $C = (U \cdot U) \circ U + (U \cdot U^T) \circ U + (U^T \cdot U) \circ U$ & $C + C^T$ \\
$M_{6}$  & $C = (U \cdot B) \circ U + (B \cdot U^T) \circ U^T + (U^T \cdot U) \circ B$ & $C$ \\
$M_{7}$  & $C = (U^T \cdot B) \circ U^T + (B \cdot U) \circ U + (U \cdot U^T) \circ B$ & $C$ \\
\bottomrule
\end{tabular}
\label{tab:matrix_interp}
\end{table}

The central computational kernel in these computations is $(X \cdot Y) \circ Z$.
\Citet{azad2015parallel} developed and analyzed parallel algorithms for
these computations when $X$, $Y$, and $Z$ are sparse (and if the
graph is sparse, these matrices will be sparse).  This matrix-based
formulation has worse computational complexity than the fast algorithms
discussed later in \cref{sec:honc_basic_complexity}---the computation $X \cdot Y$
counts length-2 paths, which is prohibitively expensive for large, sparse
real-world networks.  However, the implementation is simple and elegant and
works well for small graphs.  To demonstrate this, \cref{fig:julia_spectral}
contains a complete implementation of \cref{alg:motif_fiedler} for motifs
$M_{1}$--$M_{7}$ in fewer than 50 lines of readable Julia code.


\lstdefinelanguage{Julia}%
  {morekeywords={abstract,break,case,catch,const,continue,do,else,elseif,%
      end,export,false,for,function,immutable,import,importall,if,in,%
      macro,module,otherwise,quote,return,switch,true,try,type,typealias,%
      using,while},%
   sensitive=true,%
   alsoother={$},%
   morecomment=[l]\#,%
   morecomment=[n]{\#=}{=\#},%
   morestring=[s]{"}{"},%
   morestring=[m]{'}{'},%
}[keywords,comments,strings]%

\lstset{%
    language         = Julia,
    basicstyle       = \footnotesize \ttfamily,
    keywordstyle     = \bfseries\color{blue},
    stringstyle      = \color{magenta},
    commentstyle     = \color{ForestGreen},
    showstringspaces = false,
    numbers=left,
}
\begin{figure}[tb]
\centering
\scalebox{0.92}{
\begin{lstlisting}
# Find a motif-based cluster for any directed triangle motif.
function MotifSpectralClust(A::SparseMatrixCSC{Int64,Int64}, motif::AbstractString)
    # Form motif adjacency matrix
    B = min.(A, A')  # bidirectional links
    U = A - B        # unidirectional links
    if     motif == "M1"
        C = (U * U) .* U'
        W = C + C'
    elseif motif == "M2"
        C =  (B * U) .* U' + (U * B) .* U' + (U * U) .* B
        W = C + C'
    elseif motif == "M3"
        C = (B * B) .* U + (B * U) .* B + (U * B) .* B
        W = C + C'
    elseif motif == "M4"
        W = (B * B) .* B
    elseif motif == "M5"
        C = (U * U) .* U + (U * U') .* U + (U' * U) .* U
        W = C + C'
    elseif motif == "M6"
        W = (U * B) .* U + (B * U') .* U' + (U' * U) .* B    
    elseif motif == "M7"                
        W = (U' * B) .* U' + (B * U) .* U + (U * U') .* B
    else
        error("Motif must be one of M1, M2, M3, M4, M5, M6, or M7.")
    end

    # Get Fiedler eigenvector
    dinvsqrt = spdiagm(1.0 ./ sqrt.(vec(sum(W, 1))))
    NM = I - dinvsqrt * W * dinvsqrt
    lambdas, evecs = eigs(NM, nev=2, which=:SM)
    z = dinvsqrt * real(evecs[:, 2])
    
    # Sweep cut
    sigma = sortperm(z)
    C = W[sigma, sigma]
    Csums = sum(C, 1)'
    motifvolS = cumsum(Csums)
    motifvolSbar = sum(W) * ones(length(sigma)) - motifvolS
    conductances = cumsum(Csums - 2 * sum(triu(C), 1)') ./ min.(motifvolS, motifvolSbar)
    split = indmin(conductances)
    if split <= length(size(A, 1) / 2)
        return sigma[1:split]
    else
        return sigma[(split + 1):end]
    end
end
\end{lstlisting}}
\dualcaption{Julia implementation of the motif-based spectral clustering algorithm}{Using
the matrix computations in \cref{tab:matrix_interp}, \cref{alg:motif_fiedler}
for motifs $M_{1}$--$M_{7}$ can be implemented in fewer than 50 lines of Julia.
This code is available at\\
\url{https://gist.github.com/arbenson/a7bb06fb74977cddbb455e824519a55e}.}
\label{fig:julia_spectral}
\end{figure}

\clearpage

Another matrix-based interpretation comes from the inner product of the
motif-node incidence matrix.  Specifically, let
$\mathcal{M}(B, \anchorset)$ be a motif set and number the instances of the motif
$1, \ldots, \lvert \mathcal{M} \rvert$,
so that $(v_i, \chi_{\mathcal{A}}(v_i))$ is the $i$th motif.
Define the $\lvert M \rvert \times n$ motif-node incidence matrix by 
\[
(A_M)_{ij} = \indicator{j \in \chi_{\mathcal{A}}(v_i)}.
\]
Then the motif adjacency matrix may be written as
\begin{equation*}
W_M = A_M^TA_M - \text{diag}(A_M^TA_M)
\end{equation*}
This provides a convenient algebraic formulation for defining and thinking about
the weighted motif adjacency matrix.  However, in practice, we do not use this formulation
for any computations.

\subsection{Motif Cheeger inequality for network motifs with three nodes}
\label{sec:motif_cheeger}

In spectral graph theory, the Cheeger inequality is a well-known relationship
between the second-smallest eigenvalue of the normalized Laplacian and the
conductance of a graph.  The name of the inequality is attributed to Jeff
Cheeger, who actually proved a result about the eigenvalues of Laplacians on
manifolds~\cite{cheeger1970lower}.  Later on, the bound was generalized to the
normalized Laplacian on graphs, which we state here.
\begin{theorem}[Cheeger inequality for graphs~\cite{alon1985lambda1,alon1986eigenvalues}]\label{thm:cheeger}
Let $\lambda_2$ be the second smallest eigenvalue of the normalized laplacian $N$
for any weighted graph $W$ with node set $V$.  Then
\begin{equation}\label{eq:cheeger}
\lambda_2 / 2 \le \phi^* \le \sqrt{2\lambda_2},
\end{equation}\label{eq:cheeger}
where $\phi^* = \min_{S \subset V}\cond{S}$.
\end{theorem}

While we cannot expect to find the set with minimium conductance due to the
NP-hardness of minimizing conductance~\cite{wagner1993between}, we might hope to
realize the upper bound of the Cheeger inequality by finding some set $S$ for
which $\cond{S} \le \sqrt{2\lambda_2}$.  In other words, if we could construct
some $S$ with this property, then we get the upper bound in \cref{eq:cheeger}.
By the lower bound of \cref{eq:cheeger}, this set $S$ would then satisfy
$\cond{S} \le \sqrt{4\phi^*} = 2 \sqrt{\phi^*}$.  The sweep procedure
in \cref{alg:motif_fiedler} finds such a set $S$.\footnote{For
a clean proof of this result, see Dan Spielman's lecture notes:
\url{http://www.cs.yale.edu/homes/spielman/561/lect06-15.pdf}.}
\begin{theorem}[Sweep Cheeger inequality~\cite{mihail1989conductance}]\label{thm:sweep}
The set $S$ output by \cref{alg:motif_fiedler} satisfies
\begin{equation}
\cond{S} \le \sqrt{2\lambda_2} \le 2\sqrt{\phi^*},
\end{equation}
where conductance is measured in terms of the weighted graph $W$.
\end{theorem}

\Cref{alg:motif_fiedler} has guarantees on the quality of the output set
$S$, but these guarantees are in terms of the (edge) conductance of sets in the
weighted adjacency graph.  Instead, we would like guarantees in terms of motif
conductance.  In the rest of this section, we derive the same guarantees for
motif conductance with simple three-node motifs, or in general, motifs with 2 or
3 anchor nodes.  The crux of this result is deriving a relationship between the
motif conductance function and the weighted motif adjacency matrix, from which
the motif Cheeger inequality is essentially a corollary.

For the rest of this section and the next, we use the following notation.  Given an
unweighted, directed graph $G$ and a motif $M$, the corresponding weighted graph
defined by \cref{eqn:formal_weighting} is denoted by $G_M$.  For conductance,
cut, and volume, we will also specify the graph in which these measurements are
made.  Specifically, $\gmmcond{G}{S}$, $\gmmcut{G}{S}$, and $\gmmvol{G}{S}$
denote the motif conductance, motif cut, and motif volume of set $S$ in graph
$G$ with respect to motif $M$.

The following Lemma relates the motif volume to the volume in the weighted
graph.  This lemma applies to any anchor set $\anchorset$ consisting of at least
two nodes.  For our main result, we will apply the lemma assuming
$\vert \anchorset \vert = 3$.  However, we will apply the lemma more generally
when discussing four-node motifs in \cref{sec:fournode}.
\begin{lemma}\label{lem:motif_vol}
Let $G = (V, E)$ be a directed, unweighted graph and let $G_M$ be the
corresponding weighted graph for a motif $M$ on $k$ nodes and $\lvert  \anchorset \rvert$
anchor nodes, $2 \le \lvert \anchorset \rvert \le k$.  Then for any subset $S \subset V$,
\[
\gmmvol{G}{S} = \frac{1}{\vert \anchorset \vert  - 1}\gvol{G_M}{S}
\]
\end{lemma}
\begin{proof}
Consider an instance $\minstance$ of a motif.  Let
$(v_1, \ldots, v_{\vert \anchorset \vert}) = \chi_{\anchorset}(v)$.
By \cref{eqn:formal_weighting}, $(W_M)_{v_1, j}$ is incremented by one
for $j = v_2, \ldots, v_{\vert \anchorset \vert}$.  Since 
$(D_M)_{v_1, v_1} = \sum_{j} {(W_M)_{v_1,j}}$,
the motif end point $v_1$ is counted $\lvert \anchorset \rvert - 1$ times.
\end{proof}

The next lemma states that the truth value for determining whether three binary
variables in $\{-1, 1\}$ are not all equal is a quadratic function of the
variables.  Because this function is quadratic, we will be able to relate motif
cuts on three nodes to a quadratic form on the motif Laplacian.

\begin{lemma}\label{lem:ind3}
Let $x_i, x_j, x_k \in \{-1, 1\}$.  Then
\begin{equation}\label{eqn:ind3}
4\cdot\indicator{x_i, x_j, x_k \text{ not all the same}}
= x_i^2 + x_j^2 + x_k^2 - x_ix_j - x_jx_k - x_kx_i. \nonumber
\end{equation}
\end{lemma}

%

The next lemma contains the essential result that relates motif cuts in the
original graph $G$ to weighted edge cuts in $G_M$.  In particular, the lemma
shows that the motif cut measure is proportional to the cut on the weighted
graph defined in \cref{eqn:informal_weighting} when there are three anchor
nodes.
\begin{lemma}\label{lem:motif_cut}
Let $G = (V, E)$ be a directed, unweighted graph and let $G_M$ be the weighted
graph for a motif with $\vert \anchorset \vert = 3$.  Then for any $S \subset
V$,
\[
\gmmcut{G}{S, \bar{S}} = \frac{1}{2}\gcut{G_M}{S, \bar{S}}
\]
\end{lemma}
\begin{proof}
Let $x \in \{-1, 1\}^{n}$ be an indicator vector of the node set $S$.
\begin{align*}
4 \cdot \mmcut{S, \bar{S}}
&= \sum_{(v, \{i, j, k\}) \in \mathcal{M}} 4\cdot\indicator{x_i, x_j, x_k \text{ not all the same}} & \text{by definition}\\
&= \sum_{(v, \{i, j, k\}) \in \mathcal{M}} \left(x_i^2 + x_j^2 + x_k^2\right) - \left(x_ix_j + x_jx_k + x_kx_i\right) & \text{by \cref{lem:ind3}} \\
&= -\frac{1}{2}x^TW_Mx + \sum_{(v, \{i, j, k\}) \in \mathcal{M}} x_i^2 + x_j^2 + x_k^2  
& \text{by \cref{eqn:formal_weighting}} \\
&= -\frac{1}{2}x^TW_Mx + \mmvol{V} & \text{since $x_i^2 = x_j^2 = x_k^2 = 1$}\\
&= -\frac{1}{2}x^TW_Mx + \frac{1}{2}x^TD_Mx & \text{by \cref{lem:motif_vol}}   \\
&= \frac{1}{2}x^TL_Mx & \text{by definition}\\
&= 2\cdot\gcut{G_M}{S, \bar{S}} & \text{by \cref{eqn:quad_cut}}.
\end{align*}
\end{proof}

We are now ready to prove our main result, namely that motif conductance on the
original graph $G$ is equivalent to (edge) conductance on the weighted graph
$G_M$ when there are three anchor nodes.  The result is a consequence of the
volume and cut relationships provided by \cref{lem:motif_vol,lem:motif_cut}.

\begin{theorem}\label{thm:motif_cond}
Let $G = (V, E)$ be a directed, unweighted graph and let $G_M$ be the weighted
graph corresponding to the motif adjacency matrix for any motif
with $\vert \anchorset \vert = 3$.  Then for any set $S \subset V$,
\[
\gmmcond{G}{S} = \gcond{G_M}{S}
\]
In other words, when there are three anchor nodes, the motif conductance is
equal to the conductance on the weighted graph defined by
\cref{eqn:informal_weighting}.
\end{theorem}
\begin{proof}
When $\vert \anchorset \vert = 3$, the motif cut and motif volume are both equal
to half the motif cut and motif volume measures by
\cref{lem:motif_vol,lem:motif_cut}.
\end{proof}

For any motif with three anchor nodes, conductance on the weighted graph is
equal to the motif conductance.  Thus, we can re-interpret the Cheeger
inequality on the weighted graph in terms of motif conductance.  This leads to
our result.
\begin{theorem}[Motif Cheeger inequality]\label{thm:motif_cheeger}
Let $G$ be an unweighted, directed graph and let $M$ be a motif with two or
three anchor nodes.  Let
\begin{itemize}
\item  $S$ be the output of \cref{alg:motif_fiedler} with input $G$ and $M$,
\item $\lambda_2$ be the second smallest eigenvalue of $\normmotiflap$, and
\item $\phi_M^* = \min_{T \subset V} \mmcond{T}$ be the optimal motif conductance over
all sets of nodes $T$.
\end{itemize}
Then
\begin{enumerate}
\item $\lambda_2/2 \le \phi_M^* \le \sqrt{2\lambda_2}$, and
\item $\mmcond{S} \le 2\sqrt{\phi_M^*}$.
\end{enumerate}
\end{theorem}
\begin{proof}
The result follows from \cref{thm:motif_cond,thm:cheeger,thm:sweep}.
\end{proof}

The first result provides a lower bound on the optimal motif conductance in
terms of the eigenvalue $\lambda_2$.  We use this bound in our analysis of a
food web (see \cref{sec:honc_foodweb}) to show that certain motifs do not
provide good clusters, regardless of the procedure used to find a cluster.  The
second part of the result says that \cref{alg:motif_fiedler} outputs a cluster
that is within a quadratic factor of optimal.  This provides the mathematical
guarantees that our procedure finds a good motif-based cluster in a graph, if
one exists.

\subsection{Motif Cheeger inequality for network motifs with four or more nodes}
\label{sec:fournode}

Analogs of the indicator function in \cref{lem:ind3} for four or more variables
are not quadratic~\cite{ihler1993modeling}.  Subsequently, for motifs with
$\vert \anchorset \vert > 3$, we no longer get the motif Cheeger inequalities guaranteed by
\cref{thm:motif_cheeger}.  That being said, solutions found by
motif-based partitioning approximate a related value of conductance.  We now
provide the details.

We begin with a lemma that shows a functional form for four binary variables
taking values in $\{-1, 1\}$ to not all be equal.  We see that it is quartic,
not quadratic.
\begin{lemma}\label{lem:ind4_1}
Let $x_i, x_j, x_k, x_l \in \{-1, 1\}$.
Then the indicator function on all four elements not being equal is
\begin{align}
& 8 \cdot \indicator{x_i, x_j, x_k, x_l \text{ not all the same}} \\
&= \left(7 - x_ix_j - x_ix_k - x_ix_l - x_jx_k - x_jx_l - x_kx_l - x_ix_jx_kx_l\right) \nonumber.
\end{align}
\end{lemma}

We \emph{almost} have a quadratic form, if not for the quartic term
$x_ix_jx_kx_l$.  However, we could use the following related quadratic form:
\begin{align}
& 6 - x_ix_j - x_ix_k - x_ix_{l} - x_jx_k - x_jx_{l} - x_kx_{l} \nonumber \\
&= \left\{
     \begin{array}{ll}
       0 &  x_i, x_j, x_k, x_l \text{ are all the same} \\
       6 &  \text{exactly three of } x_i, x_j, x_k, x_l \text{ are the same} \\
       8 &  \text{exactly two of } x_i, x_j, x_k, x_l \text{ are $-1$}.
     \end{array}
   \right.\label{eqn:motif4_quadratic}
\end{align}


\definecolor{myblue}{RGB}{166,206,227}
\definecolor{mygreen}{RGB}{27,158,119}
\begin{figure}[tb]
\centering
\phantomsubfigure{fig:motif_cutsA}
\phantomsubfigure{fig:motif_cutsB}
\newcommand{\scalesize}{0.78}
\begin{tabular}{l @{\hspace{4pt}} l @{\hspace{4pt}} l @{\hspace{4pt}} l @{\hspace{4pt}} l}
\multicolumn{5}{l}{{\small$f(x) = x_i^2 + x_j^2 + x_k^2 - x_ix_j - x_jx_k - x_kx_i$}} \\
 \scalebox{\scalesize}{\begin{tikzpicture}[baseline=(current bounding box.center)] \input{CH2-TKZ-tri1} \end{tikzpicture}}
&\scalebox{\scalesize}{\begin{tikzpicture}[baseline=(current bounding box.center)] \input{CH2-TKZ-tri2} \end{tikzpicture}}
&\scalebox{\scalesize}{\begin{tikzpicture}[baseline=(current bounding box.center)] \input{CH2-TKZ-tri3} \end{tikzpicture}}
&\scalebox{\scalesize}{\begin{tikzpicture}[baseline=(current bounding box.center)] \input{CH2-TKZ-tri4} \end{tikzpicture}} \\ \\
\multicolumn{5}{l}{{\small $f(x) = \frac{3}{2}(x_i^2 + x_j^2 + x_k^2 + x_l^2) - x_ix_j - x_ix_k - x_ix_{l} - x_jx_k - x_jx_{l} - x_kx_{l}$}} \\
 \scalebox{\scalesize}{\begin{tikzpicture}[baseline=(current bounding box.center)] \input{CH2-TKZ-quad1} \end{tikzpicture}}
&\scalebox{\scalesize}{\begin{tikzpicture}[baseline=(current bounding box.center)] \input{CH2-TKZ-quad2} \end{tikzpicture}} 
&\scalebox{\scalesize}{\begin{tikzpicture}[baseline=(current bounding box.center)] \input{CH2-TKZ-quad3} \end{tikzpicture}} 
&\scalebox{\scalesize}{\begin{tikzpicture}[baseline=(current bounding box.center)] \input{CH2-TKZ-quad4} \end{tikzpicture}} 
&\scalebox{\scalesize}{\begin{tikzpicture}[baseline=(current bounding box.center)] \input{CH2-TKZ-quad5} \end{tikzpicture}} 
\end{tabular}
\dualcaption{Quadratic forms on indicator functions for set assignment}{%
The \textcolor{myblue}{blue} nodes have assignment to set $S$ and the \textcolor{mygreen}{green}
nodes have assignment to
set $\bar{S}$.  The quadratic function gives the penalty for cutting that motif.
{\bf Top:}
Illustration of \cref{eqn:ind3}.  The quadratic form is proportional to the
indicator on whether or not the motif is cut.
{\bf Bottom:}
Illustration of \cref{eqn:motif4_quadratic}.  The quadratic form is equal to
zero when all nodes are in the same set.  However, the form penalizes 2/2 splits
more than 3/1 splits.
}
\label{fig:quad_form}
\end{figure}

The quadratic still takes value $0$ if all four entries are the same, and takes
a non-zero value otherwise.  However, the quadratic takes a larger value if
exactly two of the four entries takes the value $-1$ (see \Cref{fig:quad_form}).
From this, we can provide an analogous statement to
\cref{lem:motif_cut} for motifs with $\vert \anchorset \vert = 4$.

\begin{lemma}\label{lem:motif_cut4}
Let $G = (V, E)$ be a directed, unweighted graph and let $G_M$ be the weighted
graph for a motif $M$ with $\vert \anchorset \vert = 4$.  Then for any $S \subset V$,
\[
\gmmcut{G}{S, \bar{S}} =
\frac{1}{3}\gcut{G_M}{S, \bar{S}} - 
\sum_{(v, \{i, j, k, l\}) \in \mathcal{M}}
\frac{1}{3} \cdot \indicator{\textnormal{exactly two of $i, j, k, l$ in $S$}}.
\]
\end{lemma}
\begin{proof}
Let $x \in \{-1, 1\}^n$ be an indicator vector of the node set $S$.
\begin{align*}
& 6 \cdot \gmmcut{G}{S, \bar{S}} +
\sum_{(v, \{i, j, k, l\}) \in \mathcal{M}}
2 \cdot \indicator{\text{exactly two of } i, j, k, l \text{ in } S} \\
&= \sum_{(v, \{i, j, k, l\}) \in \mathcal{M}} 6 - x_ix_j - x_ix_k - x_ix_l - x_jx_k - x_jx_l - x_kx_l \\
&= \sum_{(v, \{i, j, k, l\}) \in \mathcal{M}} \frac{3}{2}
\left(x_i^2 + x_j^2 + x_k^2 + x_l^2\right) -
\left(x_ix_j + x_ix_k + x_ix_l + x_jx_k + x_jx_l + x_kx_l\right) \\
&= \frac{1}{2}x^TD_Mx  - \frac{1}{2}x^TW_Mx \\
&= \frac{1}{2}x^TL_Mx \\
&= 2\cdot\gcut{G_M}{S, \bar{S}}.
\end{align*}
The first equality follows from \cref{eqn:motif_cut,eqn:motif4_quadratic}.  The
third equality follows from \cref{lem:motif_vol}.  The fourth equality follows
from the definition of $L_M$.  The fifth equality follows from
\cref{eqn:quad_cut}.
\end{proof}

With four anchor nodes, the motif cut in $G$ is slightly different than the
weighted cut in the weighted graph $G_M$. However, by
\cref{lem:motif_vol},
\[
\gmmvol{G}{S} = \frac{1}{3}\gvol{G_M}{S},
\]
so the motif volume in $G$ is still proportional to the weighted volume
in $G_M$.  We use this to derive the following result.

\begin{theorem}\label{thm:motif_cond4}
Let $G = (V, E)$ be a directed, unweighted graph and let $W_M$ be the weighted
adjacency matrix for any motif $M$ with $\vert \anchorset \vert = 4$.  Then for any
$S \subset V$,
\[
\gmmcond{G}{S} =
\gcond{G_M}{S} -
\frac{
\sum_{(v, \{i, j, k, l\}) \in \mathcal{M}} \indicator{\textnormal{exactly two of } i, j, k, l \textnormal{ in } S}
}{
\gvol{G_M}{S}
}
\]
In other words, when there are four anchor nodes, the weighting scheme in
\cref{eqn:informal_weighting} models the exact conductance with an
additional penalty for splitting the four anchor nodes into two groups of two.
\end{theorem}
\begin{proof}
The result follows from \cref{lem:motif_vol,lem:motif_cut4}.
\end{proof}

To summarize, we still get a motif Cheeger inequality from the weighted graph, but it
is in terms of a penalized version of the motif conductance $\gmmcond{G}{S}$.
However, the penalty makes sense---if the group of four nodes is ``more split''
(2 and 2 as opposed to 3 and 1), the penalty is larger.  When $\vert \anchorset
\vert > 4$, we can derive similar penalized approximations to $\gmmcond{G}{S}$.
Developing approximation algorithms for the minimization of motif conductance
with motifs consisting of four or more anchor nodes is an open question.  Preliminary
progress on this problem is presented in \cref{sec:nbrhood_cond}.  There, we show
that if the network has certain measurable clustering structure, then
we can find sets with low motif conductance for clique motifs.  However, we still desire
general theory applicable to any constant-size motif.

\subsection{Analysis of computational complexity}
\label{sec:honc_basic_complexity}

We now analyze the computational complexity of \cref{alg:motif_fiedler}.
Overall, the complexity of the algorithm is governed by:
\begin{compactenum}
\item forming the motif adjacency matrix $W_M$
\item computing an eigenvector of $\normmotiflap$
\item the sweep procedure
\end{compactenum}

We address these in reverse order.  Let $m$ and $n$ denote the number of edges
in the graph.  For the sweep cut, it takes $O(n\log n)$ to sort the indices
given the eigenvector using a standard sorting algorithm such as merge sort.
Computing motif conductance for each set $S_r$ in the sweep also takes linear
time.  In practice, the sweep cut step takes a small fraction of the total
running time of the algorithm.  The time to compute an eigenvector of
$\normmotiflap$ is not understood as well.  In a theoretical sense, there are a
number of ``fast Laplacian
solvers'' for solving the system of equations $N_My = x$ in nearly linear
time~\cite{spielman2004nearly,koutis2011nearly,kelner2013simple,kyng2016approximate}.
With such a solver, we can run a shifted inverse power method to compute the
eigenvector.  The recent algorithm by \citet{kyng2016approximate} in this space
is specifically designed to be practical and implementable.

For the remainder of the analysis, we consider the issue of the
time to compute $W_M$.  We assume that the size of the motif is
constant.  The time to compute $W_M$ is bounded by the time to find all
instances of the motif in the graph, assuming we can update an edge weight in
$O(1)$ time (for each instance of the motif, we increment the edge weight of all
pairs of nodes in the motif by 1).  Naively, for a motif on $k$ nodes, we can
compute $W_M$ in $O(n^k)$ time by checking each $k$-tuple of nodes.
Furthermore, there are cases where there are $O(n^k)$ motif instances in the
graph, e.g., there are $O(n^3)$ triangles in a complete graph.  However, since
most real-world networks are sparse, we instead focus on the complexity of
algorithms in terms of the number of edges and the maximum degree in the graph.
For this case, there are several efficient practical algorithms for real
networks with available software%
~\cite{demeyer2013index,houbraken2014index,wernicke2006efficient,wernicke2006fanmod,aberger2016emptyheaded}.

Here we will consider three classes of motifs:
\begin{enumerate}
\item triangles and more generally $k$-cliques,
\item wedges (connected, non-triangle three-node motifs), and
\item general four-node motifs
\end{enumerate}
\Citet{latapy2008main} analyzed a number of algorithms for listing all triangles in an undirected network, including an algorithm that has computational complexity
$O(m^{1.5})$.  For a directed graph $G$, we can use the following algorithm:
\begin{enumerate}
\item form a new graph $G_{\textnormal{undir}}$ by
removing the direction from all edges in $G$
\item find all triangles in $G_{\textnormal{undir}}$
\item for every triangle in $G_{\textnormal{undir}}$,
check which directed triangle motif it is in $G$.
\end{enumerate}
The first step takes linear time and the third step is linear assuming we can
determine the existence of an edge in $O(1)$ time.  Thus, the same
$O(m^{1.5})$ complexity holds for directed networks.  This analysis holds
regardless of the structure of the network.  However, additional properties of
the network can lead to improved algorithms.  For example, in networks with a
power law degree sequence with exponent greater than $7/2$, Berry et al.\
provide a randomized algorithm with expected running time
$O(m)$~\cite{berry2014why}.  In the case of a bounded degree graph, enumerating
over all nodes and checking all pairs of neighbors takes time $O(nd_{\max}^2)$,
where $d_{\max}$ is the maximum degree in the graph.  We note that with
triangular motifs, the number of non-zeros in $W_M$ is less than the number of
non-zeros in the original adjacency matrix.  Thus, we do not have to worry about
additional storage requirements.
\Citet{chiba1985arboricity} present an algorithm for $k$-clique enumeration with
complexity dependent on the arboricity of the graph.  Specifically, their algorithm
enumerates all $k$-cliques in $O(ka^{k-2}m)$ time, where $a$ is the
arboricity of the graph.  The arboricity of any connected graph is bounded by
$O(m^{1/2})$, so this algorithm runs in time $O(m^{1.5})$ for triangles.
Finally, note that the nodes of any connected component of the motif adjacency
matrix for $k$-clique motifs is a $k$-truss in the original network.
Thus, in practice, we can use pruning techniques for $k$-trusses as a
pre-computation to (possibly) reduce the size of the input
graph~\cite{cohen2008trusses}.

Next, we consider wedges (open triads).  We can list all wedges by looking at
every pair of neighbors of every node.  This algorithm has $O(nd_{\max}^2)$
computational complexity, where $n$ is the number of nodes and $d_{\max}$ is
again the maximum degree in the graph (a more precise bound is $O(\sum_{j}
d^2_j)$, where $d_j$ is the degree of node $j$.)  If the graph is sparse, the
motif adjacency matrix will have more non-zeros than the original adjacency
matrix, so additional storage is required.  Specifically, there is fill-in for
all two-hop neighbors, so the motif adjacency matrix has $O(\sum_{j}d^2_j)$
non-zeros.  This is impractical for large real-world networks but manageable for
modestly sized networks.

\Citet{marcus2010efficient} present an algorithm for listing all four-node motifs in an
undirected graph in $O(m^2)$ time.  We can employ the same edge direction check
as for triangles to extend this result to directed
graphs. \Citet{chiba1985arboricity} develop an algorithm for finding a
representation of all quadrangles (motif on four nodes that contains a four-node
cycle as a subgraph) in $O(am)$ time and $O(m)$ space, where $a$ is the
arboricity of the graph.

Finally, we note that the computation of $W_M$ and the computation of the
eigenvector are suitable for parallel computation.  There are already
distributed algorithms for triangle enumeration~\cite{cohen2009graph}, and the
parallel computation of eigenvectors of a sparse matrix is a classical problem
in scientific computing~\cite{maschhoff1996p_arpack}.

\subsection{Methods for simultaneously finding multiple clusters}
\label{sec:honc_multiple_clusters}

For clustering a network into $k > 2$ clusters based on motifs, we could
recursively cut the graph using the sweep procedure with some stopping
criterion~\cite{boley1998principal,kannan2004clusterings}.  For example, we
could continue to cut the largest remaining cluster until the graph is
partitioned into some pre-specified number of clusters.  We refer to this method
as recursive bi-partitioning.

In addition, there are well-known spectral clustering approaches that compute
multiple eigenvectors of the (normalized) Laplacian, use the eigenvectors to
embed the nodes into Euclidean space, and run a point cloud clustering algorithm
(typically $k$-means) on the embedded nodes~\cite{von2007tutorial}.  We can
adapt these methods for motif-based clustering by simply running the algorithms
on the weighted adjacency matrix.  For our work, we use the following adaptation
of the method of \citet{ng2001spectral}.

\begin{algorithm}[H]\algoptions
  \KwIn{Directed, unweighted graph $G$, motif $M$, number of clusters $k$}
  \KwOut{$k$ disjoint motif-based clusters}
   $(W_M)_{ij} \leftarrow \text{number of instances of $M$ that contain nodes $i$ and $j$}.$\;
   $D_M \leftarrow \diag{W_M e}$\;
   $z_1, \ldots, z_k \leftarrow $ eigenvectors of $k$ smallest eigenvalues for $\normmotiflap = I - D_M^{-1/2}W_MD_M^{-1/2}$\;
   $Y_{ij} \leftarrow z_{ij} / \sqrt{\sum_{j=1}^{k} z_{ij}^2}$ \quad \mycomment{row normalize} \;
   $E_{i} \leftarrow Y_{i,:}$ \quad \mycomment{Embed node $i$ into $\mathbb{R}^k$}\;
   \mycomment{Run $k$-means clustering on the point cloud $\{E_i\}$}\;
  \dualcaption{Motif-based clustering algorithm for finding several clusters}{}
  \label{alg:motif_ngetal}
\end{algorithm}

This method does not have the same Cheeger-like guarantee on quality.  However,
recent theory shows a cluster quality guarantee when $k$-means is replaced
with a different clustering algorithm~\cite{lee2014multiway}.\footnote{The
bounds for multi-way clustering are actually called \emph{higher-order Cheeger
inequalities}, where ``higher-order'' means more than two clusters.  We emphasize that we are using the term \emph{higher-order}
clustering to mean that our clustering objective is based on higher-order
structures.}  We still use $k$-means for its simplicity and empirical success.
Due to our theoretical results equating motif cuts and motif volumes to edge
cuts and edge volumes in the weighted graph
(\cref{lem:motif_vol,lem:motif_cut}), we can automatically apply the newly
developed theory of \citet{lee2014multiway} to the case of motifs with three anchor nodes.

\subsection{Extensions for multiple motifs, weighted motifs, and weighted, signed, and colored networks}
\label{sec:honc_extensions}

We now discuss some easy extensions of the method such as looking for clusters
based on several motifs or handling cases where the graph carries additional
information such as edge weights or node colors.

\xhdr{Simultaneously clustering based on several motifs}
All of our results carry through when considering several motifs simultaneously.
In particular, suppose we are interested in clustering based on $q$ different
motifs $M_1, \ldots, M_q$.  Further suppose that we want to weight the impact of
some motifs more than other motifs.  Let $W_{M_j}$ be the weighted adjacency
matrix for motif $M_j$ and let $\alpha_j \ge 0$ be the weight of motif $M_j$, $j
= 1, \ldots, q$.  Then we can form the weighted adjacency matrix
\begin{equation}
W_M = \sum_{j=1}^{q} \alpha_j W_{M_j}.
\end{equation}
(We will use this approach in \cref{sec:honc_yeast} when clustering based on the
four ``coherent feedforward loops'' in the \emph{S.~cerevisiae} transcriptional regulation network.)

Now, the cut and volume measures are simply weighted sums by linearity.  Suppose
that the $M_j$ all have three anchor nodes and let $G_M$ be the weighted graph
corresponding to $W_M$.  Then
\begin{align}
\gcut{G_M}{S, \bar{S}} &= \sum_{j=1}^{q}\alpha_j\mcut{M_j}{S, \bar{S}}, \\
\gvol{G_M}{S} &= \sum_{j=1}^{q}\alpha_j\mvol{M_j}{S},
\end{align}
and \cref{thm:motif_cheeger} applies to a weighted motif conductance
equal to
\[
\frac{
\sum_{j=1}^{q}\alpha_j\mcut{M_j}{S, \bar{S}}
}{
\min\left(
\sum_{j=1}^{q}\alpha_j\mvol{M_j}{S},
\sum_{j=1}^{q}\alpha_j\mvol{M_j}{\bar{S}}
\right)
}.
\]

\xhdr{Weighted motifs and weighted graphs}
We can also generalize the notions of motif cut and motif volume
for \emph{weighted motifs}, i.e., each motif instance $\minstance$ has an
associated nonnegative weight $\omega_{\minstance}$.  Our cut and volume metrics
are then
\begin{align}
\mmcut{S, \bar{S}} &=
\sum_{\minstance \in \mathcal{M}}
\omega_{\minstance}\indicator{\exists\; i, j  \in \anchornodes \;\mid\; i \in S, j \in \bar{S}},  \\
\mmvol{S} &=
\sum_{\minstance \in \mathcal{M}}
\omega_{\minstance}
\sum_{i \in \anchornodes} \indicator{i \in S}.
\end{align}
Subsequently, we adjust the motif adjacency matrix as follows:
\begin{equation}
\motifweightedij = \sum_{\minstance \in \mathcal{M}}
\omega_{\minstance} \indicator{\{i, j\} \subset \chi_{\anchorset}(v)}
\end{equation}

For weighted networks (i.e., the original network is weighted and we want to
cluster based on motifs), the algorithm user needs to decide how to weight a
motif instance given the weights of the edges in the instance itself.
Possibilities include:
\begin{enumerate}
\item the maximum or minimum edge weight;
\item the arithmetic or geometric mean of the edge weights; or
\item the product of the edge weights.
\end{enumerate}
The weighting scheme would depend on the particular applications and motivations
of the user.  The important idea is that the motif-based clustering scheme makes
it clear that the weights must be specified \emph{for each motif instance}, and
there is no canonical weighting scheme.  We note that the general problem of
constructing weights for higher-order structures in a weighted graph has been
studied in generalizations of the clustering coefficient to weighted
networks~\cite{opsahl2009clustering}.

\xhdr{Signed and colored networks}
Our results easily generalize for signed networks.  We only have to generalize
\cref{eqn:motif_set} by allowing the adjacency matrix $B$ to be signed.
This allows us to seamlessly analyze the signed \emph{S.~cerevisiae}
transcriptional regulation network in \cref{sec:honc_yeast}.
Extending the method for motifs where the edges or nodes are ``colored'' or
``labeled'' is similar.  If the edges are colored, then we again just allow the
adjacency matrix $B$ to capture this information, where the entries in the
matrix are indices for edge colors.  If the nodes in the motif are
colored, we only count motif instances with the specified pattern.

\section{Case studies}

We now use motif-based clustering to analyze several real-world networks.  Our
main goal is to show that motif-based clusters find more meaningful and markedly
different structures in many real-world networks compared to edge-based
clusters.  To this end, we first discuss several edge-based clustering
algorithms that will be used for comparison.

\subsection{Alternative clustering algorithms for evaluation}

For our experiments, we compare our motif-based spectral custering to the
following methods:
\begin{itemize}
\item Standard, edge-based spectral clustering, which is a special
case of motif-based clustering.  In particular, the motifs
\begin{equation}
B_1 = \begin{bmatrix} 0 & 1 \\ 1 & 0 \end{bmatrix},\;
B_2 = \begin{bmatrix} 0 & 1 \\ 0 & 0 \end{bmatrix},\;
\anchorset = \{1, 2\}
\end{equation}
correspond to removing directionality from a directed graph.  We refer to the
union of these two motifs as $\medge$ (there is a weight of 1 between nodes
$i$ and $j$ if there is some edge connecting $i$ and $j$).
\item %
Infomap, an information-theoretic approach designed to optimally compress network dynamics~\cite{rosvall2008maps}.\footnote{Software
for Infomap was downloaded from \url{http://mapequation.org/code.html}.
We run the ``directed'' algorithm for directed links when the network under
consideration is directed.  For all other parameters, we use the default values.}
\item %
The Louvain method, which is a hierarchical greedy algorithm for
modularity maximization~\cite{blondel2008fast}.\footnote{Software for the Louvain method was
downloaded from \url{https://perso.uclouvain.be/vincent.blondel/research/louvain.html}.
We use the ``oriented'' version of the Louvain method for directed graphs.}
\end{itemize}

Infomap and the Louvain method are take as input the graph and produce
as output a set of labels for the nodes in the graph.  In contrast to the
spectral methods, we do not have control over
the number of clusters.  Also, only the spectral methods provide embeddings of
the nodes into Euclidean space, which is useful for visualization.  Thus, for
our analysis of the transportation reachability network in
\cref{sec:honc_airports}, we only compare spectral methods.

\subsection{Comparing motif conductance and edge conductance}
\label{sec:comparing_cond}

\begin{figure}[t]
  \centering
  \phantomsubfigure{fig:conductancesA}
  \phantomsubfigure{fig:conductancesB}  
  \includegraphics[width=\columnwidth]{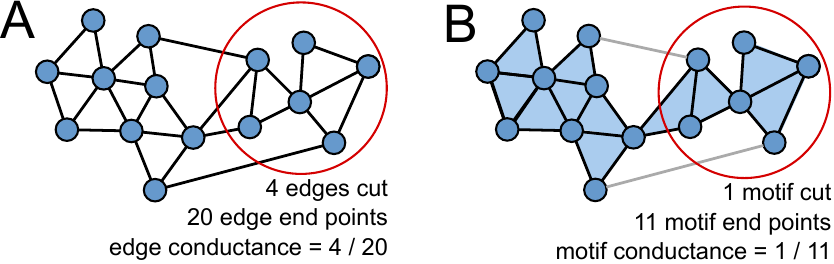}
  \dualcaption{Edge conductance and motif conductance}{The edge
    conductance (A) and the motif conductance (B) is different for the same set
    of nodes in the same graph, where the motif $M$ is the triangle.  Our
    methods finds clusters of nodes based on the motif conductance, where the
    user can decide which motif $M$ to use for the clustering.  Comparing edge
    and motif conductance is meaningful because of the probabilistic interpretation of conductance.}
   \label{fig:conductances}
\end{figure}

When analyzing data with motif-based clustering, we often compare values of
motif conductance to edge conductance (see \cref{fig:conductances}).  Although
these two objective functions measure different (but related) quantities, we
argue that comparing them is meaningful because both have a comparable
probabilistic interpretation.  Let $W$ be a connected, non-bipartite, and
(possibly weighted) undirected graph, and let $\{Z_t\}$ be a stochastic process
with transition probabilities given by the random walk transition matrix $P = WD^{-1}$.
Then it is well-known that
\begin{equation}\label{eq:escape}
\cond{S} = \max\{\prob{Z_1 \in \bar{S} \given Z_0 \in S}, \prob{Z_1 \in S \given Z_0 \in \bar{S}}\},
\end{equation}
where the initial state $Z_0$ is chosen randomly from the stationary distribution of
the random walk (i.e., from the unique vector $\pi$ satisfying $P\pi = \pi$)~\cite{meila2001random}.

When $W$ is a motif-adjacency matrix for a simple motif, then the
random walk is equivalent to the following stochastic process $\{Y_t\}$.
\begin{enumerate}
\item when $Y_{t} = j$, choose a motif instance containing node $j$ as an anchor node uniformly at random, and
\item transition to $Y_{t+1} = i$, where $i$ is an anchor node selected uniformly at random from
the nodes in the motif instance excluding $j$.
\end{enumerate}
To see this, let $r$ be the number of anchor nodes in the motif.  Then,
\begin{align*}
& \prob{\text{transition $j \to i$}} \\
&= P_{ij} \\
&= \frac{W_{ij}}{\sum_{k} W_{kj}} \\
&= \frac{\text{\# motifs containing $i$ and $j$ as anchors}}{(r - 1) \cdot (\text{\# motifs containing $j$ as an anchor})} \qquad \text{by \cref{lem:motif_vol}} \\
&= \frac{1}{r - 1}\prob{\text{random motif that contains $j$ as an anchor also contains $i$ as an anchor}}
\end{align*}
Consequently,
\begin{equation}\label{eq:escape2}
\mmcond{S} = \max\{\prob{Y_1 \in \bar{S} \given Y_0 \in S}, \prob{Y_1 \in S \given Y_0 \in \bar{S}}\},
\end{equation}
where $Y_0$ is chosen uniformly at random from the stationary distribution
(when there are $2$ or $3$ anchor nodes).
\Cref{eq:escape2} is comparable between motifs (including the edge motif).
Here, small motif conductance of a set $S$ just means that there is a stochastic
process that transitions between nodes based on motifs and that this process
tends to stay contained in $S$ or $\bar{S}$.

\subsection{Motif $M_{6}$ in the Florida Bay food web}
\label{sec:honc_foodweb}

We now apply the higher-order clustering framework on the Florida Bay
food
web~\cite{ulanowicz1998network}.\footnote{\url{http://vlado.fmf.uni-lj.si/pub/networks/data/bio/foodweb/Florida.paj}}
In this network, the nodes are compartments (roughly, organisms and species) and
the edges represent directed carbon exchange---an edge $(i, j$) means that
carbon flows from $i$ to $j$.  Often, this means that species $j$ eats species
$i$.  In this domain, motifs model energy flow patterns between several species.

In this case study, we use the framework to identify higher-order modular
organization of the network.  We focus on three motifs: $M_{5}$ corresponds to a
hierarchical flow of energy where species $i$ and $j$ are energy sources (prey)
for species $k$, and $i$ is also an energy source for $j$; $M_{6}$ models two
species that prey on each other and then compete to feed on a common third
species; and $M_{8}$ describes a single species serving as an energy source for
two non-interacting species.  Prevalence of motif $M_{5}$ has been found in food
webs~\cite{bascompte2005interaction,bascompte2009disentangling}, and motif
$M_{6}$ is predicted by a niche model~\cite{stouffer2007evidence}.

We first briefly discuss the connectivity of the motif adjacency matrix with
respect to these motifs, as we will perform analysis on the largest connected
component of these graphs.  The original network is weakly connected with 128
nodes and 2,106 edges.  The largest connected component of the motif adjacency
matrix for motif $M_{5}$ contains 127 of the 128 nodes (the compartment of
``roots'' becomes isolated).  The two largest connected components of the motif
adjacency matrix for motif $M_{6}$ contain 12 and 50 nodes
(see \Cref{tab:foodweb_m6_components}; we will also use this connectivity
information later in our analysis), and the remaining 66 nodes are isolated.
The motif adjacency matrix for $M_{8}$ is connected.  The original network is
weakly connected, so the motif adjacency matrix for $\medge$ is also connected.


\begin{table}[t]
\centering
\dualcaption{Connected components of the Florida Bay food web motif
adjacency matrix for motif $M_{6}$}
{
There are 50 nodes in component 1, 12 nodes in component 2, and 66 isolated
nodes.  Component 2 consists of microfauna and detritus (see
\cref{tab:foodweb_classification}).
}
\scalebox{0.57}{
\begin{tabular}{l c @{\hskip 5cm} l c}
\toprule
\multicolumn{2}{l}{Two largest components} & Isolated nodes \\
Compartment (node) & Component index & Compartment (node) & \\
\midrule
Benthic Phytoplankton & 1 & Barracuda & \\
Thalassia & 1 & \SI{2}{\micro\metre} Spherical Phytoplankton & \\
Halodule & 1 & Synedococcus & \\
Syringodium & 1 & Oscillatoria & \\
Drift Algae & 1 & Small Diatoms ($<$\SI{20}{\micro\metre}) & \\
Epiphytes & 1 & Big Diatoms ($>$\SI{20}{\micro\metre}) & \\
Predatory Gastropods & 1 & Dinoflagellates & \\
Detritivorous Polychaetes & 1 & Other Phytoplankton & \\
Predatory Polychaetes & 1  & Roots & \\
Suspension Feeding Polychaetes & 1 & Coral & \\
Macrobenthos & 1  & Epiphytic Gastropods & \\
Benthic Crustaceans & 1  & Thor Floridanus & \\
Detritivorous Amphipods & 1 & Lobster & \\
Herbivorous Amphipods & 1 & Stone Crab & \\
Isopods & 1 & Sharks & \\
Herbivorous Shrimp & 1 & Rays & \\
Predatory Shrimp & 1 & Tarpon & \\
Pink Shrimp & 1 & Bonefish & \\
Benthic Flagellates & 1 & Other Killifish & \\
Benthic Ciliates & 1 & Snook & \\
Meiofauna & 1 & Sailfin Molly & \\
Other Cnidaridae & 1 & Hawksbill Turtle & \\ 
Silverside & 1 & Dolphin & \\
Echinoderma & 1 & Other Horsefish & \\
Bivalves & 1 & Gulf Pipefish & \\
Detritivorous Gastropods & 1 & Dwarf Seahorse & \\
Detritivorous Crabs & 1 & Grouper & \\
Omnivorous Crabs & 1 & Jacks & \\
Predatory Crabs & 1 & Pompano & \\
Callinectes sapidus (blue crab) & 1 & Other Snapper & \\
Mullet & 1 & Gray Snapper & \\
Blennies & 1 & Mojarra & \\
Code Goby & 1 & Grunt & \\
Clown Goby & 1 & Porgy & \\
Flatfish & 1 & Pinfish & \\
Sardines & 1 & Scianids & \\
Anchovy & 1 & Spotted Seatrout & \\
Bay Anchovy & 1 & Red Drum & \\
Lizardfish & 1 & Spadefish & \\
Catfish & 1 & Parrotfish & \\
Eels & 1 & Mackerel & \\
Toadfish & 1 & Filefishes & \\
Brotalus & 1 & Puffer & \\
Halfbeaks & 1 & Loon & \\
Needlefish & 1 & Greeb & \\
Goldspotted killifish & 1 & Pelican & \\
Rainwater killifish & 1 & Comorant & \\
Other Pelagic Fishes & 1 & Big Herons and Egrets & \\
Other Demersal Fishes & 1 & Small Herons and Egrets & \\
Benthic Particulate Organic Carbon (Benthic POC) & 1 & Ibis & \\
Free Bacteria & 2 & Roseate Spoonbill & \\
Water Flagellates & 2 & Herbivorous Ducks & \\
Water Cilitaes & 2 & Omnivorous Ducks & \\
Acartia Tonsa & 2 & Predatory Ducks & \\
Oithona nana & 2 & Raptors & \\
Paracalanus & 2 & Gruiformes & \\
Other Copepoda & 2 & Small Shorebirds & \\
Meroplankton & 2 & Gulls and Terns & \\
Other Zooplankton & 2 & Kingfisher & \\
Sponges & 2 & Crocodiles & \\
Water Particulate Organic Carbon (Water POC) & 2 & Loggerhead Turtle & \\
Input & 2 & Green Turtle & \\
& & Manatee & \\
& & Dissolved Organic Carbon (DOC) & \\
& & Output & \\
& & Respiration & \\
\bottomrule
\end{tabular}
}
\label{tab:foodweb_m6_components}
\end{table}

\clearpage

\definecolor{myblue}{RGB}{69,122,189}
\definecolor{myred}{RGB}{230,10,143}
\definecolor{mygreen}{RGB}{76,191,115}
\begin{figure}[tb]
\centering \includegraphics[width=0.75\columnwidth]{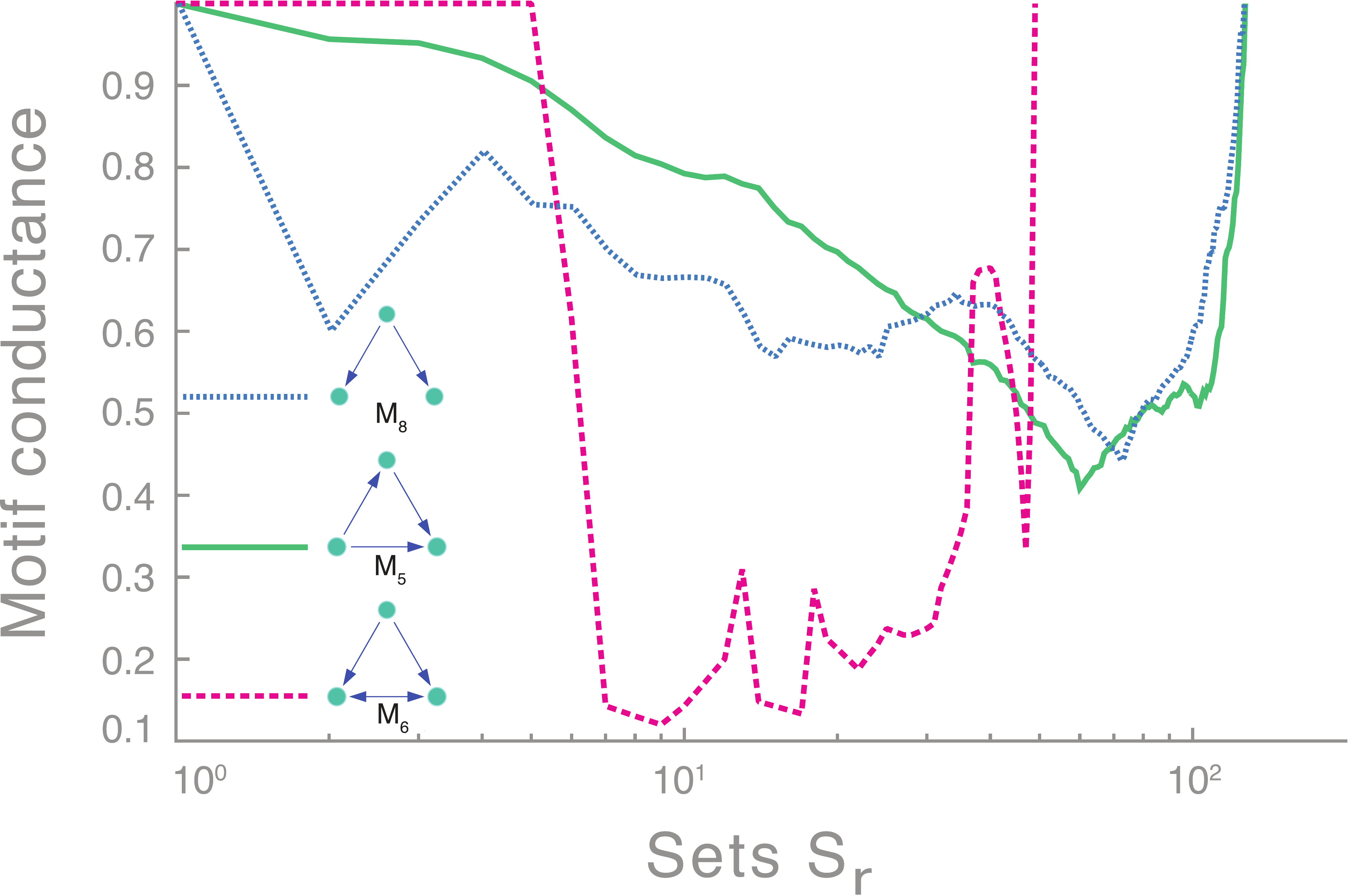}
\dualcaption{Sweep profile plot on the Florida Bay food web}{The sweep profile
  plot measures $\mmcond{S}$ as $S$ varies in the
  sweep procedure of \cref{alg:motif_fiedler}.
  Here, we look at the sweep profile for three motifs on the Florida Bay ecosystem food web.  A priori it is
  not clear whether the network is organized based on a given motif.  For
  example, motifs $M_{5}$ (\textcolor{mygreen}{green}) and $M_{8}$ (\textcolor{myblue}{blue})
  do not reveal any
  higher-order organization (motif conductance has high values). However, the
  downward spikes of the \textcolor{myred}{red} curve show that $M_{6}$ reveals rich higher-order
  clusters~\cite{leskovec2009community}.  
  Ecologically, motif $M_{6}$
  corresponds to two species mutually feeding on each other and also preying on
  a common third species.
  The motif cheeger inequality guarantees
  that the motif conductance for the best $M_{6}$ found by \cref{alg:motif_fiedler} is better
  than the motif conductance of every set for motifs $M_{5}$ and $M_{8}$.}
\label{fig:foodweb_sweep}
\end{figure}

By examining the sweep profile plots of these motifs
in \cref{fig:foodweb_sweep}, we see that low motif conductance clusters are only
found for motif $M_{6}$, whereas clusters based on motifs $M_{5}$ or $M_{8}$
have high motif conductance.  In fact, the motif Cheeger inequality
(\cref{thm:motif_cheeger}) guarantees that clusters based on motif $M_{5}$ or
$M_{8}$ will always have larger motif conductance than clusters based on
$M_{6}$.  Specifically, \cref{thm:motif_cond} says that the motif conductance of
any set is bounded below by $\lambda_2 / 2$, where $\lambda_2$ is the second
smallest eigenvalue of the motif normalized Laplacian $\normmotiflap$.  The
lower bounds $(\lambda_2 / 2)$ on motif conductance for motifs $M_{5}$, $M_{6}$,
$M_{8}$ are are 0.22, 0.03, and 0.22, and the clusters found
by \cref{alg:motif_fiedler} have motif conductances of 0.44, 0.12, and 0.41.
Thus, the cluster $S$ found by the algorithm for $M_{6}$ has smaller motif
$M_{6}$ conductance (0.12) than any possible cluster's $M_{5}$ or
$M_{8}$ conductance.  The same conclusions hold for edge-based clustering.
For motif $\medge$, the lower bound on conductance is 0.2194, and the cluster
found by the algorithm has conductance 0.4083.

Subsequently, we use motif $M_{6}$ and \cref{alg:motif_ngetal} to cluster the food web, revealing four
clusters (\cref{fig:foodweb_org,fig:foodweb_comms}).  Three represent well-known
aquatic layers: (i) the pelagic system; (ii) the benthic predators of eels,
toadfish, and crabs; and (iii) the sea-floor ecosystem of macroinvertebrates.  The
fourth cluster identifies microfauna supported by particulate organic carbon in
water and free bacteria (this cluster is the second largest connected component
of the motif adjacency matrix---component 2
in \cref{tab:foodweb_m6_components}). \Cref{tab:foodweb_classification} lists
the nodes in each cluster.

We also measure how well the motif-based clusters correlate to known ground
truth system subgroup classifications of the nodes~\cite{ulanowicz1998network}.
These classes are microbial, zooplankton, and sediment organism microfauna;
detritus; pelagic, demersal, and benthic fishes; demseral, seagrass, and algae
producers; and macroinvertebrates (\cref{tab:foodweb_classification},
Classification 1).  We also consider a set of labels which does not include the
subclassification for microfauna and producers.  In this case, the labels are
microfauna; detritus; pelagic, demersal, and benthic fishes; producers; and
macroinvertebrates (\cref{tab:foodweb_classification}, Classification 2).

To quantify how well the clusters found by motif-based clustering reflect the
ground truth labels, we use several standard evaluation criteria: adjusted rand
index, F1 score, normalized mutual information, and
purity~\cite{manning2008introduction}.  We compare these results to the
clusters of several methods using the same evaluation criteria.  In total, we
evaluate six methods:
\begin{enumerate}
\item Motif-based clustering with the embedding + k-means algorithm
(\cref{alg:motif_ngetal}) with 500 iterations of k-means.
\item Motif-based clustering with recursive bi-partitioning (repeated application of
\cref{alg:motif_fiedler} on the largest remaining compoennt).  The
process continues to cut the largest cluster until there are 4 total.
\item Edge-based clustering with the embedding + k-means algorithm, again with 500
iterations of k-means.
\item Edge-based clustering with recursive bi-partitioning with the same
partitioning process.
\item The Infomap algorithm.
\item The Louvain method.
\end{enumerate}
We control the number of clusters in the first four algorithms, which set to four.
We cannot control the number of clusters in the last two algorithms, but both
methods happen to find four clusters.

\Cref{tab:foodweb_performance} shows that the motif-based clustering by
embedding + k-means has the best performance for each classification criterion
on both classifications.  We conclude that the organization of compartments in
the Florida Bay food web are better described by motif $M_{6}$ than by edges.

\definecolor{myyellow}{RGB}{217,138,5}
\definecolor{myred}{RGB}{189,20,104}
\definecolor{myblue}{RGB}{46,156,189}
\definecolor{mygreen}{RGB}{28,174,117}
\begin{figure}[tb]
\centering \includegraphics[width=\columnwidth]{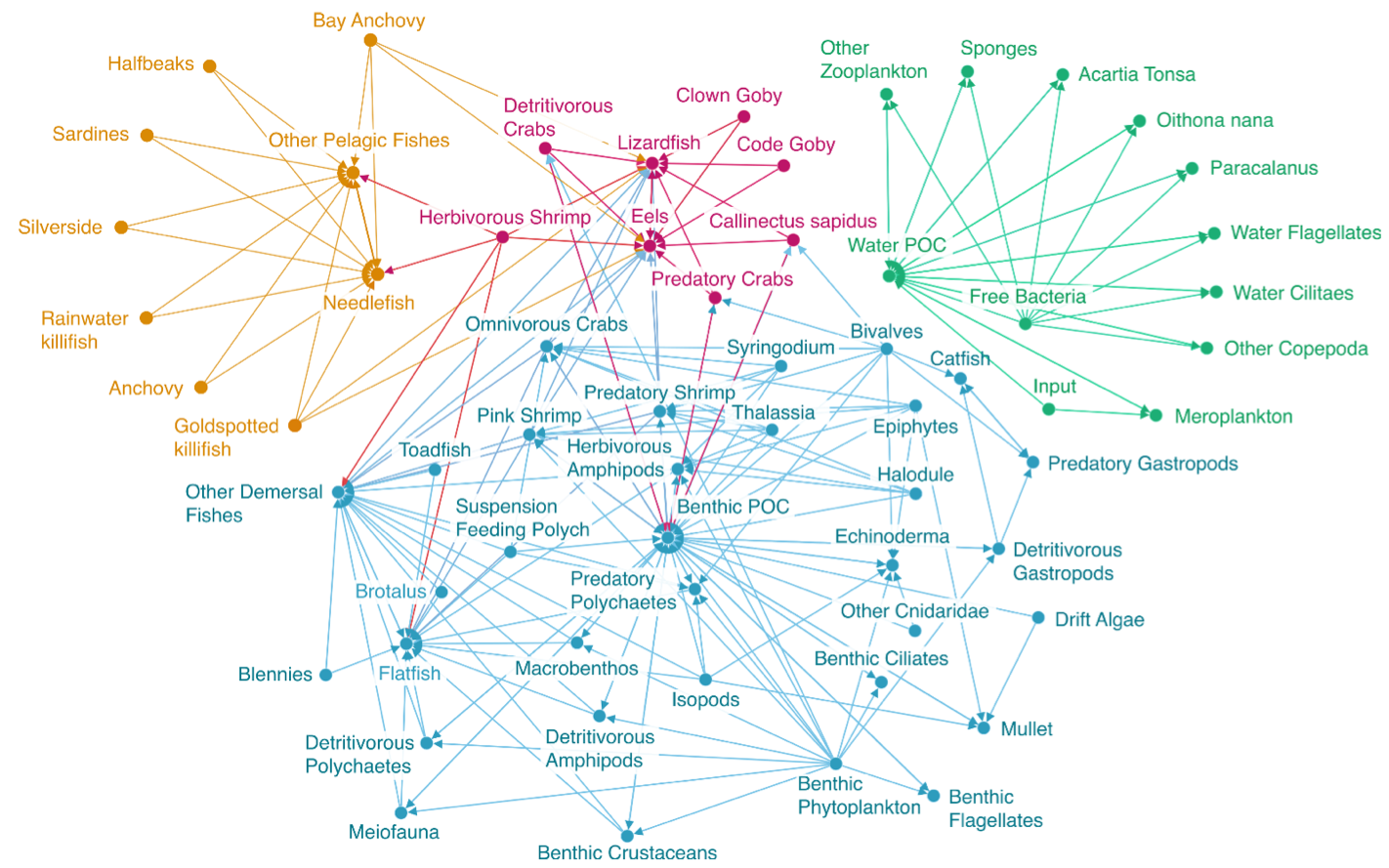}
\dualcaption{Higher-order organization of the Florida Bay food web}{Clustering
  of the food web based on motif $M_{6}$.  (For illustration, edges not
  participating in at least one instance of the motif are omitted.)  The
  clustering reveals three known aquatic layers: pelagic fishes
  (\textcolor{myyellow}{yellow}), benthic fishes and crabs
  (\textcolor{myred}{red}), and sea-floor macroinvertebrates
  (\textcolor{myblue}{blue}) as well as a cluster of microfauna and detritus
  (\textcolor{mygreen}{green}).  Our framework identifies these modules with
  higher accuracy (61\%) than existing methods
  (48--53\%)---see \cref{tab:foodweb_performance}.  \Cref{fig:foodweb_comms}
  examines the \textcolor{myyellow}{yellow} and \textcolor{mygreen}{green}
  higher-order clusters in more detail.}
\label{fig:foodweb_org}
\end{figure}

\clearpage


\definecolor{myyellow}{RGB}{217,138,5}
\definecolor{mygreen}{RGB}{28,174,117}
\begin{figure}[tb]
\centering
\includegraphics[width=\columnwidth]{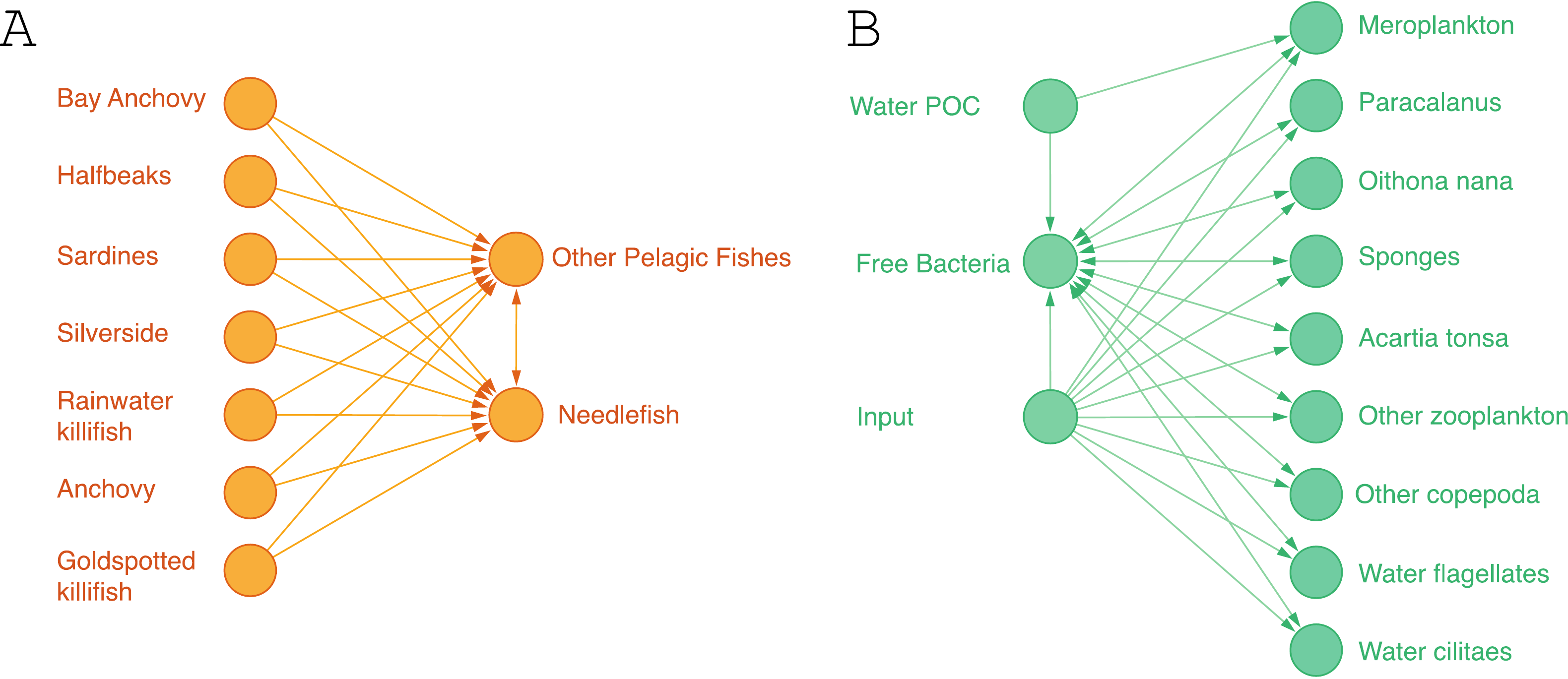}
\dualcaption{Higher-order clusters in the Florida Bay food web}{%
A closer look at two of the higher-order clusters from \cref{fig:foodweb_org}.
{\bf A:}
A higher-order cluster (the \textcolor{myyellow}{yellow} nodes
in \cref{fig:foodweb_org}) shows how motif $M_{6}$ occurs in the pelagic layer.
The needlefish and other pelagic fishes eat each other while several other
fishes are prey for these two species.
{\bf B:}
Another higher-order cluster (the \textcolor{mygreen}{green} nodes
in \cref{fig:foodweb_org}) shows how motif $M_{6}$ occurs between
microorganisms.  Here, several microfauna decompose into Particulate Organic
Carbon in the water (water POC) but also consume water POC.  Free bacteria
serves as an energy source for both the microfauna and water POC.
}
\label{fig:foodweb_comms}
\end{figure}

\clearpage


\definecolor{myyellow}{RGB}{217,138,5}
\definecolor{myred}{RGB}{189,20,104}
\definecolor{myblue}{RGB}{46,156,189}
\definecolor{mygreen}{RGB}{28,174,117}
\begin{table}[t]
\centering
\dualcaption{Ecological classification of nodes in the Florida Bay foodweb}{
Colors correspond to the visualization in \cref{fig:foodweb_org}.
}
\scalebox{0.62}{
\begin{tabular}{l l l @{\hskip 1.75cm} l}
\toprule
Compartment (node) & Classification 1 & Classification 2 & Assignment \\
\midrule
Free Bacteria             & Microbial microfauna         & Microfauna         & \textcolor{mygreen}{green}  \\
Water Flagellates         & Microbial microfauna         & Microfauna         & \textcolor{mygreen}{green}  \\
Water Cilitaes            & Microbial microfauna         & Microfauna         & \textcolor{mygreen}{green}  \\
Acartia Tonsa             & Zooplankton microfauna       & Microfauna         & \textcolor{mygreen}{green}  \\
Oithona nana              & Zooplankton microfauna       & Microfauna         & \textcolor{mygreen}{green}  \\
Paracalanus               & Zooplankton microfauna       & Microfauna         & \textcolor{mygreen}{green}  \\
Other Copepoda            & Zooplankton microfauna       & Microfauna         & \textcolor{mygreen}{green}  \\
Meroplankton              & Zooplankton microfauna       & Microfauna         & \textcolor{mygreen}{green}  \\
Other Zooplankton         & Zooplankton microfauna       & Microfauna         & \textcolor{mygreen}{green}  \\
Sponges                   & Macroinvertebrates           & Macroinvertebrates & \textcolor{mygreen}{green}  \\
Water POC                 & Detritus                     & Detritus           & \textcolor{mygreen}{green}  \\
Input                     & Detritus                     & Detritus           & \textcolor{mygreen}{green}  \\
Sardines                  & Pelagic Fishes               & Pelagic Fishes     & \textcolor{myyellow}{yellow} \\
Anchovy                   & Pelagic Fishes               & Pelagic Fishes     & \textcolor{myyellow}{yellow} \\
Bay Anchovy               & Pelagic Fishes               & Pelagic Fishes     & \textcolor{myyellow}{yellow} \\
Halfbeaks                 & Pelagic Fishes               & Pelagic Fishes     & \textcolor{myyellow}{yellow} \\
Needlefish                & Pelagic Fishes               & Pelagic Fishes     & \textcolor{myyellow}{yellow} \\
Goldspotted killifish     & Demersal Fishes              & Demersal Fishes    & \textcolor{myyellow}{yellow} \\
Rainwater killifish       & Demersal Fishes              & Demersal Fishes    & \textcolor{myyellow}{yellow} \\
Silverside                & Pelagic Fishes               & Pelagic Fishes     & \textcolor{myyellow}{yellow} \\
Other Pelagic Fishes      & Pelagic Fishes               & Pelagic Fishes     & \textcolor{myyellow}{yellow} \\
Detritivorous Crabs       & Macroinvertebrates           & Macroinvertebrates & \textcolor{myred}{red}    \\
Predatory Crabs           & Macroinvertebrates           & Macroinvertebrates & \textcolor{myred}{red}    \\
Callinectus sapidus       & Macroinvertebrates           & Macroinvertebrates & \textcolor{myred}{red}    \\
Lizardfish                & Benthic Fishes               & Benthic Fishes     & \textcolor{myred}{red}    \\
Eels                      & Demersal Fishes              & Demersal Fishes    & \textcolor{myred}{red}    \\
Code Goby                 & Benthic Fishes               & Benthic Fishes     & \textcolor{myred}{red}    \\
Clown Goby                & Benthic Fishes               & Benthic Fishes     & \textcolor{myred}{red}    \\
Herbivorous Shrimp        & Macroinvertebrates           & Macroinvertebrates & \textcolor{myred}{red}    \\
Benthic Phytoplankton     & Producer Demersal            & Producer           & \textcolor{myblue}{blue}   \\
Thalassia                 & Producer Seagrass            & Producer           & \textcolor{myblue}{blue}   \\
Halodule                  & Producer Seagrass            & Producer           & \textcolor{myblue}{blue}   \\
Syringodium               & Producer Seagrass            & Producer           & \textcolor{myblue}{blue}   \\
Drift Algae               & Producer Algae               & Producer           & \textcolor{myblue}{blue}   \\
Epiphytes                 & Producer Algae               & Producer           & \textcolor{myblue}{blue}   \\
Benthic Flagellates       & Sediment Organism microfauna & Microfauna         & \textcolor{myblue}{blue}   \\
Benthic Ciliates          & Sediment Organism microfauna & Microfauna         & \textcolor{myblue}{blue}   \\
Meiofauna                 & Sediment Organism microfauna & Microfauna         & \textcolor{myblue}{blue}   \\
Other Cnidaridae          & Macroinvertebrates           & Macroinvertebrates & \textcolor{myblue}{blue}   \\
Echinoderma               & Macroinvertebrates           & Macroinvertebrates & \textcolor{myblue}{blue}   \\
Bivalves                  & Macroinvertebrates           & Macroinvertebrates & \textcolor{myblue}{blue}   \\
Detritivorous Gastropods  & Macroinvertebrates           & Macroinvertebrates & \textcolor{myblue}{blue}   \\
Predatory Gastropods      & Macroinvertebrates           & Macroinvertebrates & \textcolor{myblue}{blue}   \\
Detritivorous Polychaetes & Macroinvertebrates           & Macroinvertebrates & \textcolor{myblue}{blue}   \\
Predatory Polychaetes     & Macroinvertebrates           & Macroinvertebrates & \textcolor{myblue}{blue}   \\
Suspension Feeding Polych & Macroinvertebrates           & Macroinvertebrates & \textcolor{myblue}{blue}   \\
Macrobenthos              & Macroinvertebrates           & Macroinvertebrates & \textcolor{myblue}{blue}   \\
Benthic Crustaceans       & Macroinvertebrates           & Macroinvertebrates & \textcolor{myblue}{blue}   \\
Detritivorous Amphipods   & Macroinvertebrates           & Macroinvertebrates & \textcolor{myblue}{blue}   \\
Herbivorous Amphipods     & Macroinvertebrates           & Macroinvertebrates & \textcolor{myblue}{blue}   \\
Isopods                   & Macroinvertebrates           & Macroinvertebrates & \textcolor{myblue}{blue}   \\
Predatory Shrimp          & Macroinvertebrates           & Macroinvertebrates & \textcolor{myblue}{blue}   \\
Pink Shrimp               & Macroinvertebrates           & Macroinvertebrates & \textcolor{myblue}{blue}   \\
Omnivorous Crabs          & Macroinvertebrates           & Macroinvertebrates & \textcolor{myblue}{blue}   \\
Catfish                   & Benthic Fishes               & Benthic Fishes     & \textcolor{myblue}{blue}   \\
Mullet                    & Pelagic Fishes               & Pelagic Fishes     & \textcolor{myblue}{blue}   \\
Benthic POC               & Detritus                     & Detritus           & \textcolor{myblue}{blue}   \\
Toadfish                  & Benthic Fishes               & Benthic Fishes     & \textcolor{myblue}{blue}   \\
Brotalus                  & Demersal Fishes              & Demersal Fishes    & \textcolor{myblue}{blue}   \\
Blennies                  & Benthic Fishes               & Benthic Fishes     & \textcolor{myblue}{blue}   \\
Flatfish                  & Benthic Fishes               & Benthic Fishes     & \textcolor{myblue}{blue}   \\
Other Demersal Fishes     & Demersal Fishes              & Demersal Fishes    & \textcolor{myblue}{blue}   \\
\bottomrule
\end{tabular}
}
\label{tab:foodweb_classification}
\end{table}

\clearpage


\begin{sidewaystable}[t]\def\arraystretch{1.4}
\centering
\dualcaption{Evaluation of clustering algorithms in the Florida Bay food web}
{
We compare the motif-based algorithms against other methods for finding the
``ground truth'' classifications listed in \cref{tab:foodweb_classification}.
Performance is evaluated based on Adjusted Rand Index (ARI), F1 score, Normalized
Mutual Information (NMI), and Purity.  In all cases, the motif-based methods
have the best performance.
}
\scalebox{0.93}{
\begin{tabular}{l @{\hskip 0.75cm} l c @{\hskip 0.2cm} c @{\hskip 0.2cm} c @{\hskip 0.2cm} c @{\hskip 0.2cm} c @{\hskip  0.2cm} c } %
\toprule
& Evaluation & \multicolumn{1}{l}{Motif embedding} & \multicolumn{1}{l}{Motif recursive}
& \multicolumn{1}{l}{Edge embedding}  & \multicolumn{1}{l}{Edge recursive}
& \phantom{XX}Infomap\phantom{XX}         & Louvain \\
& 
& \multicolumn{1}{l}{+ k-means}       & \multicolumn{1}{l}{bi-partitioning}
& \multicolumn{1}{l}{+ k-means}       & \multicolumn{1}{l}{bi-partitioning} 
&                                     & \\ \midrule
\parbox[t]{2mm}{\multirow{4}{*}{\rotatebox[origin=c]{90}{Classification 1}}}
& ARI    & \textbf{0.3005} & 0.2156 & 0.1564 & 0.1226 & 0.1423 & 0.2207 \\
& F1     & \textbf{0.4437} & 0.3853 & 0.3180 & 0.2888 & 0.3100 & 0.4068 \\
& NMI    & \textbf{0.5040} & 0.4468 & 0.4112 & 0.3879 & 0.4035 & 0.4220 \\
& Purity & \textbf{0.5645} & 0.5323 & 0.4032 & 0.4194 & 0.4194 & 0.5323 \\ \midrule
\parbox[t]{2mm}{\multirow{4}{*}{\rotatebox[origin=c]{90}{Classification 2}}}
& ARI    & \textbf{0.3265} & 0.2356 & 0.1814 & 0.1190 & 0.1592 & 0.2207 \\
& F1     & \textbf{0.4802} & 0.4214 & 0.3550 & 0.3035 & 0.3416 & 0.4068 \\
& NMI    & \textbf{0.4822} & 0.4185 & 0.3533 & 0.3034 & 0.3471 & 0.4220 \\
& Purity & \textbf{0.6129} & 0.5806 & 0.4839 & 0.4355 & 0.4677 & 0.5323 \\
\bottomrule
\end{tabular}
}
\label{tab:foodweb_performance}
\end{sidewaystable}

\clearpage

\subsection{Coherent feedforward loops in the \emph{S.~cerevisiae} transcriptional regulation network}
\label{sec:honc_yeast}

The transcription regulation network of the yeast \emph{S.~cerevisiae}\footnote{Fun
fact: \emph{S.~cerevisiae} is commonly used in brewing and winemaking.}
we study here describes how genes regulate genetic transcription in one another.
In the network, each node is an operon (a group of genes in a mRNA molecule),
and a directed edge from operon $i$ to operon $j$ means that $i$ is regulated by
a transcriptional factor encoded by $j$~\cite{alon2007network}.\footnote{The network data was downloaded from
\url{http://www.weizmann.ac.il/mcb/UriAlon/sites/mcb.UriAlon/files/uploads/NMpaper/yeastdata.mat}}
Edges are directed and signed.  A positive sign represents activation and a negative sign
represents repression.

For this case study, we examine the coherent feedforward loop motif
(\cref{fig:yeastA}), which acts as a sign-sensitive delay element in
transcriptional regulation
networks~\cite{mangan2003structure,mangan2003coherent}.  Formally, the coherent
feedforward loop is represented by the following simple signed motifs
\begin{align}\label{eqn:cffls}
B_1 = \begin{bmatrix} 0 & + & + \\ 0 & 0 & + \\ 0 & 0 & 0 \end{bmatrix},\;
B_2 = \begin{bmatrix} 0 & - & - \\ 0 & 0 & + \\ 0 & 0 & 0 \end{bmatrix},\;
B_3 = \begin{bmatrix} 0 & + & - \\ 0 & 0 & - \\ 0 & 0 & 0 \end{bmatrix},\;
B_4 = \begin{bmatrix} 0 & - & + \\ 0 & 0 & - \\ 0 & 0 & 0 \end{bmatrix}.
\end{align}
These motifs have the same edge pattern and only differ in sign.  All of the
motifs are simple ($\anchorset = \{1, 2, 3\}$).
\Citet{mangan2003structure} identified the functionality of every coherent
feedforward loop instance, so we have ground truth on the role that these motifs
play in the network.  For our analysis, we consider all coherent feedforward
loops that are subgraphs on the induced subgraph of any three nodes.  However,
there is only one instance where the coherent feedforward loop itself is a
subgraph but not an induced subgraph on three nodes (the induced
subgraph of DAL80, GAT1, and GLN3 contains a bidirectional edge between DAL80
and GAT1 and unidirectional edges from DAL80 and GAT1 to GLN3).


\begin{table}[tb]
  \centering
  \dualcaption{Non-trivial connected components of the motif adjacency
  matrix of the \emph{S.~cerevisiae} network for the coherent feedforward loop}{}
  \scalebox{0.95}{
\begin{tabular}{c@{\hskip 0.5cm} l l}
\toprule
Size & operons  \\ \midrule
18 & ALPHA1, CLN1, CLN2, GAL11, HO, MCM1, MFALPHA1, PHO5, SIN3, \\
     & SPT16, STA1, STA2, STE3, STE6, SWI1, SWI4/SWI6, TUP1, SNF2/SWI1 \\
9 & HXT11, HXT9, IPT1, PDR1, PDR3, PDR5, SNQ2, YOR1, YRR1  \\
9 & GCN4, ILV1, ILV2, ILV5, LEU3, LEU4, MET16, MET17, MET4  \\
6 & CHO1, CHO2, INO2, INO2/INO4, OPI3, UME6 \\
6 & DAL80, DAL80/GZF3, GAP1, GAT1, GLN1, GLN3 &  \\
5 & CYC1, GAL1, GAL4, MIG1, HAP2/3/4/5 \\
3 & ADH2, CCR4, SPT6 \\
3 & CDC19, RAP1, REB1 \\
3 & DIT1, IME1, RIM101 \\
\bottomrule
\end{tabular}
}
\label{tab:yeast_ccs}
\end{table}

Again, we analyze the component structure of the motif adjacency matrix as a
pre-processing step.  The original network consists of 690 nodes and 1082 edges,
and its largest weakly connected component consists of 664 nodes and 1066 edges.
Every coherent feedforward loop in the network resides in the largest weakly
connected component, so we subsequently consider this sub-network in the
following analysis.  Of the 664 nodes in the network, only 62 participate in a
coherent feedforward loop.  The motif adjacency matrix has nine connected
components---of sizes 18, 9, 9, 6, 6, 5, 3, 3, and 3---as well as some isolated
components.  \Cref{tab:yeast_ccs} lists the operons in the components of size at
least 3.

Although the original network is connected, the motif adjacency matrix has
several components, and this already reveals much of the structure in the
network (\cref{fig:yeastB}).  Indeed, this shattering of the graph into
components for the feedforward loop has previously been observed in
transcriptional regulation networks~\cite{dobrin2004aggregation}.  We
additionally use \cref{alg:motif_fiedler} to partition the largest connected
component of the motif adjacency matrix (consisting of 18 nodes).  This reveals
the cluster
\[
\{\text{CLN2, CLN1, SWI4/SWI6, SPT16, HO}\},
\]
which contains three coherent feedforward loops (\cref{fig:yeastD}).  All three
instances of the motif correspond to the function ``cell cycle and mating type
switch''.  The motifs in this cluster are the only feedforward loops for which
the function is described by \citet{mangan2003structure}.  Using the same
procedure on the undirected version of the induced subgraph of the 18 nodes
(i.e., using motif $\medge$) results in the cluster
\[
\{\text{CLN1, CLN2, SPT16, SWI4/SWI6}\}.
\]
This cluster breaks the coherent feedforward loop formed by HO, SWI4/SWI6, and
SPT16.

We also evaluate our method based on the motif functionality labels provided
by \citet{mangan2003structure}.\footnote{This data was downloaded
from \url{http://www.weizmann.ac.il/mcb/UriAlon/sites/mcb.UriAlon/files/uploads/DownloadableData/list_of_ffls.pdf}.}
In total, there are 12 different functionalities and 29 labeled coherent
feedforward loop instances.  We consider the motif-based clustering of the
graph to be the connected components of the motif adjacency matrix with the
additional partition of the largest connected component.  To form an edge-based
clustering, we used the embedding + k-means algorithm on the undirected graph
(i.e., motif $\medge$) with $k = 12$ clusters.  We also clustered the graph
using Infomap and the Louvain method.

We say that a method ``consistently'' labels a motif instance if all three nodes
are assigned to the same cluster, and we would like our clustering algorithms
to consistently label motif instances.  \Cref{tab:ffl_classification} summarizes the
results.  The motif-based clustering consistently labels all 29 motifs, while
the edge-based spectral, Infomap, and Louvain clustering consistently labeled 25,
23, and 23 motif instances, respectively.

We measures the accuracy of each clustering method as the rand
index~\cite{manning2008introduction} on the consistently labeled motifs,
multiplied by the fraction of consistently labeled motifs.  The motif-based
clustering has much highest accuracy (97\%) compared to the other
methods (68--82\%).  We conclude that motif-based clustering
provides an advantage over edge-based clustering methods in identifying
functionalities of coherent feedforward loops in the the \emph{S.~cerevisiae}
transcriptional regulation network.


\begin{figure}[htb]
  \centering
  \phantomsubfigure{fig:yeastA}
  \phantomsubfigure{fig:yeastB}
  \phantomsubfigure{fig:yeastC}
  \phantomsubfigure{fig:yeastD}    
  \includegraphics[width=1\textwidth]{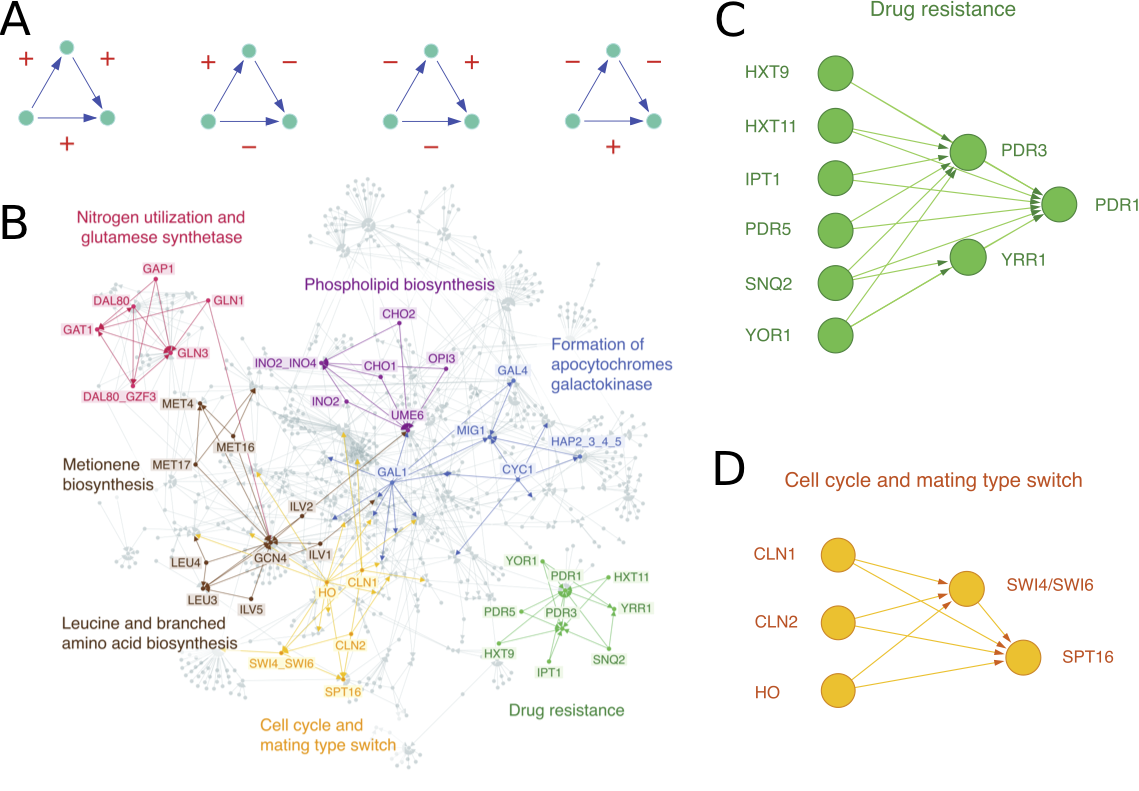}
\dualcaption{Higher-order organization of the \emph{S.~cerevisiae}
transcriptional regulation network}{%
{\bf A:}
The four higher-order structures used by our higher-order clustering method,
which can model signed motifs. These are the coherent feedfoward loop motifs, which
act as sign-sensitive delay elements in transcriptional regulation
networks~\cite{mangan2003coherent}.  The edge signs refer to activation
(positive) or suppression (negative).
{\bf B:}
Six higher-order clusters revealed by the motifs in (A).  Clusters show
functional modules consisting of several motif instances,
which were previously studied individually~\cite{mangan2003structure}.  The
higher-order clustering framework identifies the functional modules with higher
accuracy (97\%) than existing methods (68--82\%).
{\bf C--D:}
Two higher-order clusters from (B). In these clusters, all edges have positive
sign.  The functionality of the motifs in the modules correspond to drug
resistance (C) or cell cycle and mating type switch (D).  The clustering suggests
that coherent feedforward loops function together as a single processing unit
rather than as independent elements.  The cluster in (C) is simply
revealed by a connected component of the motif adjacency matrix (\cref{tab:yeast_ccs}).
}
\label{fig:yeast}
\end{figure}

\clearpage

\begin{sidewaystable}[hb]
\centering
\dualcaption{Clustering of coherent feedforward loops}
{
We compare the classification performance of motif-based clustering to other
methods.  In a given motif instance, we say that it is consistently labeled if the
nodes comprising the motif are in the same cluster.  If a motif is not
consistently labeled, a ``-1'' is listed.  The accuracy is the rand index on the
labels and motif functionality on consistently labeled motifs, multiplied by the
fraction of consistently labeled motifs.  Functions of each feedforward loop were
identified by \citet{mangan2003structure}.
}
\scalebox{0.66}{
\begin{tabular}{l l l l l c c c c}
\toprule
& \multicolumn{3}{c}{Motif nodes} & Function   & \multicolumn{4}{c}{Class label}                                                                          \\
                           &                                 &            &            &                                              & Motif-based & Edge-based & Infomap & Louvain \\ \midrule
                           & GAL11                           & ALPHA1     & MFALPHA1   & pheromone response                           & 1           & 1          & -1      & -1      \\
                           & GCN4                            & MET4       & MET16      & Metionine biosynthesis                       & 2           & 2          & 1       & -1      \\
                           & GCN4                            & MET4       & MET17      & Metionine biosynthesis                       & 2           & 2          & 1       & -1      \\
                           & GCN4                            & LEU3       & ILV1       & Leucine and branched amino acid biosynthesis & 2           & 2          & 1       & 1       \\
                           & GCN4                            & LEU3       & ILV2       & Leucine and branched amino acid biosynthesis & 2           & 2          & 1       & 1       \\          
                           & GCN4                            & LEU3       & ILV5       & Leucine and branched amino acid biosynthesis & 2           & 2          & 1       & 1       \\         
                           & GCN4                            & LEU3       & LEU4       & Leucine and branched amino acid biosynthesis & 2           & 2          & 1       & 1       \\
                           & GLN3                            & GAT1       & GAP1       & Nitrogen utilization                         & 3           & 3          & 1       & 2       \\
                           & GLN3                            & GAT1       & DAL80      & Nitrogen utilization                         & 3           & 3          & 1       & 2       \\
                           & GLN3                            & GAT1       & DAL80/GZF3 & Glutamate synthetase                         & 3           & 3          & 1       & 2       \\
                           & GLN3                            & GAT1       & GLN1       & Glutamate synthetase                         & 3           & 3          & 1       & 2       \\
                           & MIG1                            & HAP2/3/4/5 & CYC1       & formation of apocytochromes                  & 4           & 4          & -1      & -1      \\
                           & MIG1                            & GAL4       & GAL1       & Galactokinase                                & 4           & -1         & -1      & -1      \\
                           & PDR1                            & YRR1       & SNQ2       & Drug resistance                              & 5           & 5          & 2       & 3       \\
                           & PDR1                            & YRR1       & YOR1       & Drug resistance                              & 5           & 5          & 2       & 3       \\
                           & PDR1                            & PDR3       & HXT11      & Drug resistance                              & 5           & 5          & 2       & 3       \\
                           & PDR1                            & PDR3       & HXT9       & Drug resistance                              & 5           & 5          & 2       & 3       \\
                           & PDR1                            & PDR3       & PDR5       & Drug resistance                              & 5           & 5          & 2       & 3       \\
                           & PDR1                            & PDR3       & IPT1       & Drug resistance                              & 5           & 5          & 2       & 3       \\       
                           & PDR1                            & PDR3       & SNQ2       & Drug resistance                              & 5           & 5          & 2       & 3       \\
                           & PDR1                            & PDR3       & YOR1       & Drug resistance                              & 5           & 5          & 2       & 3       \\
                           & RIM101                          & IME1       & DIT1       & sporulation-specific                         & 6           & 6          & 3       & 4       \\
                           & SPT16                           & SWI4/SWI6  & CLN1       & Cell cycle and mating type switch            & 7           & -1         & 4       & 5       \\
                           & SPT16                           & SWI4/SWI6  & CLN2       & Cell cycle and mating type switch            & 7           & -1         & -1      & 5       \\
                           & SPT16                           & SWI4/SWI6  & HO         & Cell cycle and mating type switch            & 7           & -1         & -1      & -1      \\
                           & TUP1                            & ALPHA1     & MFALPHA1   & Mating factor alpha                          & 1           & 1          & -1      & 5       \\
                           & UME6                            & INO2/INO4  & CHO1       & Phospholipid biosynthesis                    & 8           & 6          & 5       & 4       \\
                           & UME6                            & INO2/INO4  & CHO2       & Phospholipid biosynthesis                    & 8           & 6          & 5       & 4       \\
                           & UME6                            & INO2/INO4  & OPI3       & Phospholipid biosynthesis                    & 8           & 6          & 5       & 4       \\ \midrule
Frac. consistently labeled &                                 &            &            &                                              & 29 / 29     & 25 / 29    & 23 / 29  & 23 / 29 \\
Accuracy                   &                                 &            &            &                                              & 0.97        & 0.82       & 0.68    & 0.76    \\
\bottomrule
\end{tabular}
}
\label{tab:ffl_classification}
\end{sidewaystable}

\clearpage

\subsection{Bi-directed length-2 paths in a transportation reachability network}
\label{sec:honc_airports}

Our framework also provides new insights into network organization beyond the
clustering of nodes into modules.  Here we study a transportation rechability
network where the nodes are cities in the United States and Canada, and
there is an edge from city $i$ to city $j$ if the estimated
travel time from $i$ to $j$ is less than some threshold~\cite{frey2007clustering}.\footnote{Data
was downloaded from \url{http://www.psi.toronto.edu/index.php?q=affinity\%20propagation}.}
The network is asymmetric.

Here we show results on a transportation reachability
network~\cite{frey2007clustering} that demonstrate how it finds the essential
hub interconnection airports (\cref{fig:airports}).  These appear as extrema on
the primary spectral direction (\cref{fig:airports_motif_embed}) when two-hop
motifs (\cref{fig:airports_motif}) are used to capture highly connected nodes
and non-hubs.  The first spectral coordinate of the normalized motif Laplacian
embedding is positively correlated with the airport city's metropolitan
population with Pearson correlation 99\% confidence interval [0.33, 0.53].  The
secondary spectral direction identifies the West-East geography in the North
American flight network (it is negatively correlated with the airport city's
longitude with Pearson correlation 99\% confidence interval [-0.66, -0.50]).  On
the other hand, edge-based methods conflate geography and hub structure.  For
example, Atlanta, a large hub, is embedded next to Salina, a non-hub, with an
edge-based method (\cref{fig:airports_edge_embed}).

\definecolor{myred}{RGB}{255,0,13}
\definecolor{mypurple}{RGB}{199,0,255}
\definecolor{mygreen}{RGB}{0,255,0}
\definecolor{myteal}{RGB}{81,194,166}
\begin{figure}[t]
\centering
\phantomsubfigure{fig:airports_motif}
\phantomsubfigure{fig:airports_weighted}
\phantomsubfigure{fig:airports_motif_embed}
\phantomsubfigure{fig:airports_edge_embed}
\includegraphics[width=\textwidth]{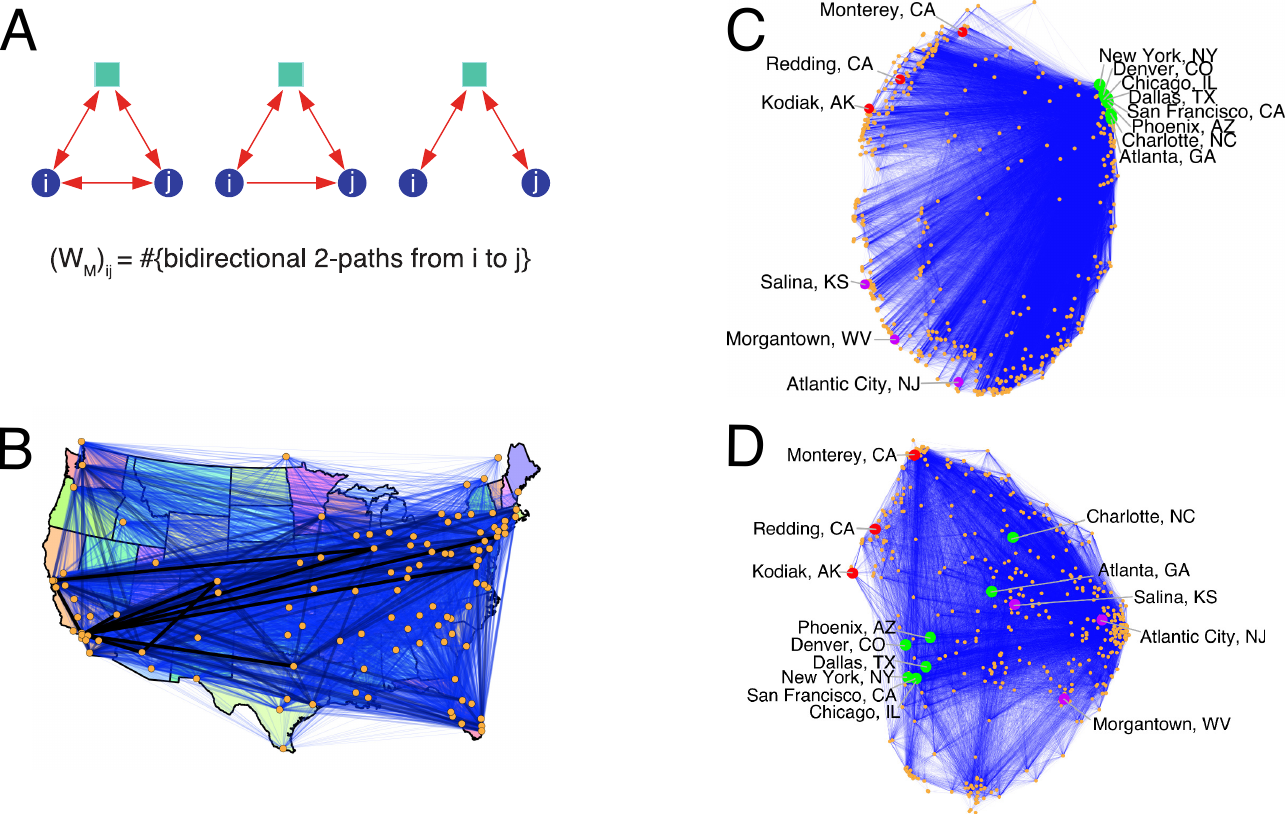}
\dualcaption{Higher-order spectral analysis of a network of airports in Canada and the United States}{%
{\bf A:}
The three higher-order structures used in our analysis.  Each motif is
``anchored'' by the blue nodes $i$ and $j$, which means our framework only seeks
to cluster together the blue nodes.  Specifically, the motif adjacency matrix
adds weight to the $(i, j)$ edge based on the number of third intermediary nodes
(\textcolor{myteal}{green} squares).  The first two motifs correspond to highly-connected cities and
the motif on the right connects non-hubs to non-hubs.
{\bf B:}
The top 50 most populous cities in the United States which correspond to nodes
in the network.  The edge thickness is proportional to the weight in the motif
adjacency matrix $W_M$.  The thick, dark lines indicate that large weights
correspond to popular mainline routes.
{\bf C:}
Embedding of nodes provided by their corresponding components of the first two
non-trivial eigenvectors of the normalized motif Laplacian $\normmotiflap$.  The marked
cities are eight large U.S. hubs (\textcolor{mygreen}{green}), three West coast non-hubs (\textcolor{myred}{red}), and
three East coast non-hubs (\textcolor{mypurple}{purple}).  The primary spectral coordinate (left to
right) reveals how much of a hub the city is, and the second spectral coordinate
(top to bottom) captures West-East geography.
{\bf D:}
Embedding of nodes provided by their corresponding components in the first two
non-trivial eigenvectors of the standard, edge-based normalized Laplacian.
This method does not capture the hub and geography found
by the higher-order method.  For example, Atlanta, the largest hub, is in the
center of the embedding, next to Salina, a non-hub.
}
\label{fig:airports}
\end{figure}

\clearpage

To study these results in more detail, we compare the motif-based spectral
embedding of the transportation reachability network to spectral embeddings from
other connectivity matrices.  The two-dimensional spectral embedding for a graph
defined by a (weighted) adjacency matrix $W \in \mathbb{R}^{n \times n}$ comes
from \cref{alg:motif_ngetal}:
\begin{enumerate}
\item %
Form the normalized Laplacian $N = I - D^{-1/2}WD^{-1/2}$, where $D$
is the diagonal degree matrix with $D_{ii} = \sum_{j}W_{ij}$.
\item %
Compute the first 3 eigenvectors $z_1$, $z_2$, $z_3$ of smallest eigenvalues for
$N$ ($z_1$ has the smallest eigenvalue, which is $0$).
\item %
Form the normalized matrix $Y \in \mathbb{R}^{n \times 3}$ by
$Y_{ij} = z_{ij} / \sqrt{\sum_{j=1}^{3}z_{ij}^2}$.
\item %
Define the primary and secondary spectral coordinates of node $i$ to be $Y_{i2}$
and $Y_{i3}$, respectively.
\end{enumerate}

We use the following three matrices $W$ for embeddings.

\begin{enumerate}
\item {\bf Motif}: %
The sum of the motif adjacency matrix for
three different anchored motifs:
\begin{equation}
B_1 = \begin{bmatrix} 0 & 1 & 1 \\ 1 & 0 & 1 \\ 1 & 1 & 0 \end{bmatrix},\;
B_2 = \begin{bmatrix} 0 & 1 & 1 \\ 1 & 0 & 1 \\ 0 & 1 & 0 \end{bmatrix},\;
B_3 = \begin{bmatrix} 0 & 1 & 0 \\ 1 & 0 & 1 \\ 0 & 1 & 0 \end{bmatrix},\;
\anchorset = \{1, 3\}.
\end{equation}
If $B$ is the matrix of bidirectional links in the graph ($B = A \circ A^T$),
then the motif adjacency matrix for these motifs is $W_M = B^2$.
\Cref{fig:airports_motif_embed} shows the resulting embedding.
This is equivalent to putting a weight between two nodes as the number of
bidirectional length-2 paths between the nodes (\cref{fig:airports_motif}).

\item {\bf Undirected}: %
The adjacency matrix is formed by ignoring edge direction.  In other words, this
is the motif adjacency matrix for motif $\medge$.  This is the standard spectral
embedding.  \Cref{fig:airports_edge_embed} shows the resulting embedding.

\item {\bf Undirected complement}: %
The adjacency matrix is formed by taking the complement of the undirected
adjacency matrix.  This matrix tends to connect non-hubs to each other.
\end{enumerate}
In all three cases, the network is connected.


\begin{table}[tb]
\centering
\dualcaption{Correlations between principal eigenvectors and city metadata}
{%
The table lists the 99\% confidence interval for the Pearson correlation
coefficients for (i) metropolitan population and the primary spectral coordinate
(eigenvector corresponding to the second smallest eigenvalue) and (ii) city
longitude and the secondary spectral coordinate (eigenvector corresponding to
the third smallest eigenvalue).
}
\begin{tabular}{l @{\hskip 0.8cm} l @{\hskip 0.5cm}l}
\toprule
                      & Primary spectral coordinate & Secondary spectral coordinate \\
                      & and metropolitan population & and longitude                 \\ \midrule
Embedding             & 99\% confidence interval    & 99\% confidence interval      \\ \midrule
Motif                 & 0.43 $\pm$ 0.09             & 0.59 $\pm$ 0.08               \\
Undir.            & 0.11 $\pm$ 0.12             & 0.39 $\pm$ 0.11               \\
Undir. comp. & 0.31 $\pm$ 0.11             & 0.10 $\pm$ 0.12               \\
\bottomrule
\end{tabular}
\label{tab:correlations}
\end{table}

We compute 99\% confidence intervals for the Pearson correlation of the primary
spectral coordinate with the metropolitan population of the city using the
Pearson correlation coefficient (\cref{tab:correlations}).  Since
eigenvectors are only unique up to sign, we list the interval with the
largest positive end point under this
symmetry to be consistent across embeddings.  The motif-based primary
spectral coordinate has the strongest correlation with the city populations.
\Cref{tab:correlations} also lists the correlation between the secondary spectral
coordinate and the longitude of the city.  Again, the motif-based clustering has
the strongest correlation.  Furthermore, the lower end of the confidence
interval for the motif-based embedding is above the higher end of the confidence
interval for the other two embeddings.

In order to visualize these relationships, we compute Loess regressions of city
metropolitan population and longitude against the primary and secondary spectral
coordinates for each of the embeddings (\cref{fig:airports_loess}).  The sign of
the eigenvector used in each regression was chosen to match correlation shown
in \cref{fig:airports_motif_embed,fig:airports_edge_embed}.\footnote{The primary
spectral coordinate positively correlated with population and secondary spectral
coordinate negatively correlated with longitude.}  The Loess regressions
visualize the stronger correlation of the motif-based spectral coordinates with
the metropolitan popuatlion and longitude.

\begin{figure}[t]
\centering
\newcommand{\airportfigcolwidth}{0.32\columnwidth}
\includegraphics[width=\airportfigcolwidth]{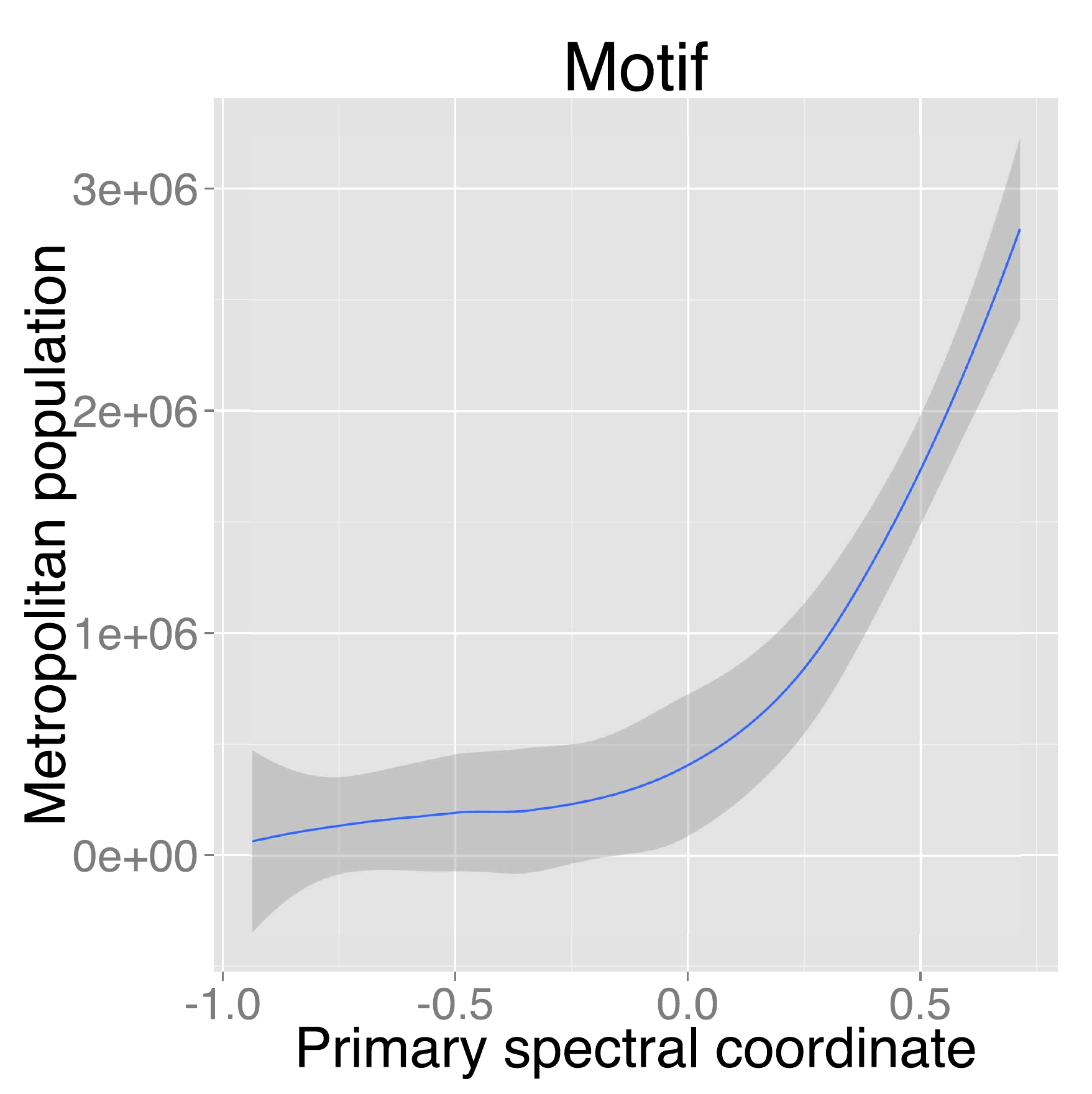}
\includegraphics[width=\airportfigcolwidth]{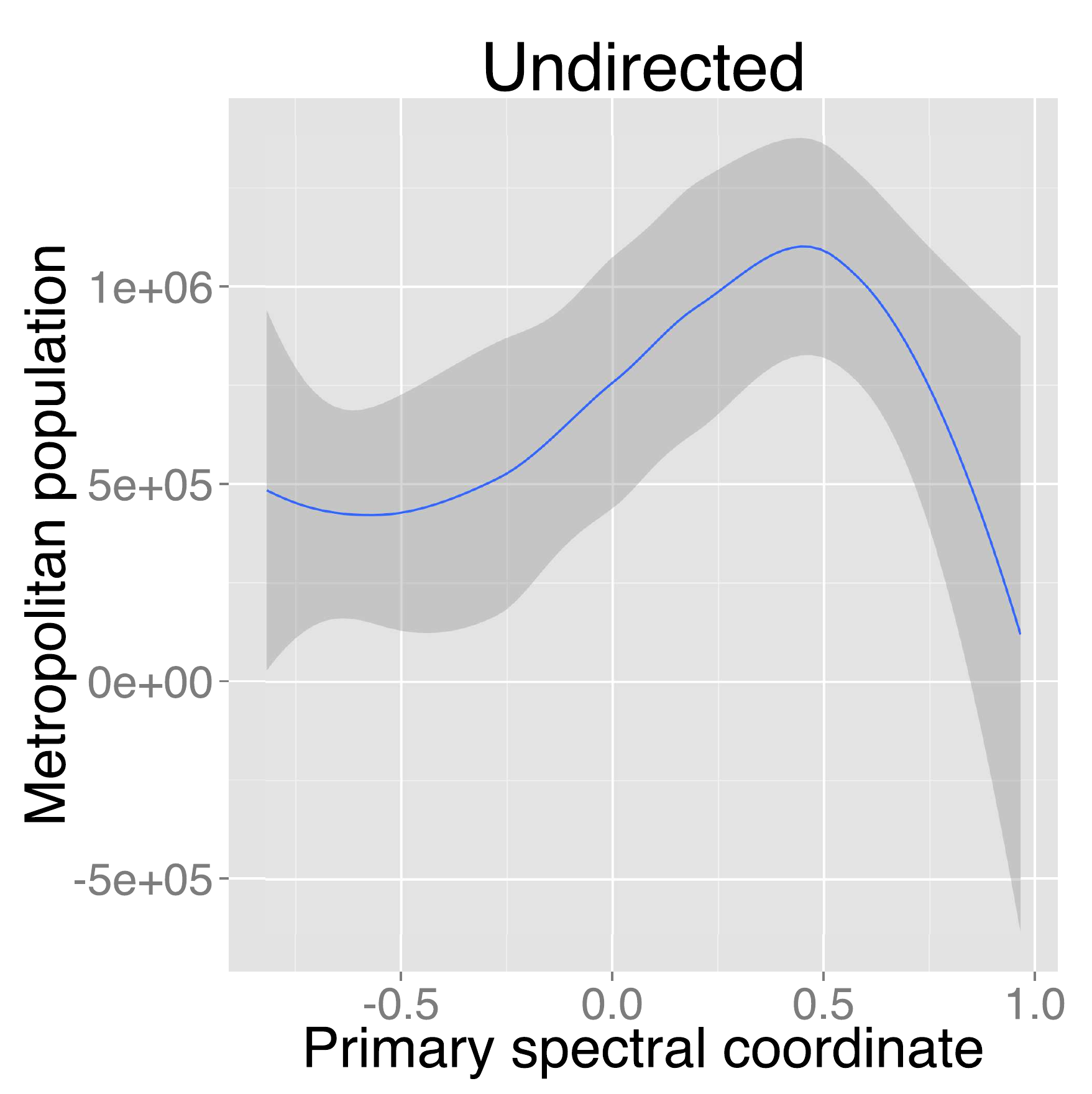}
\includegraphics[width=\airportfigcolwidth]{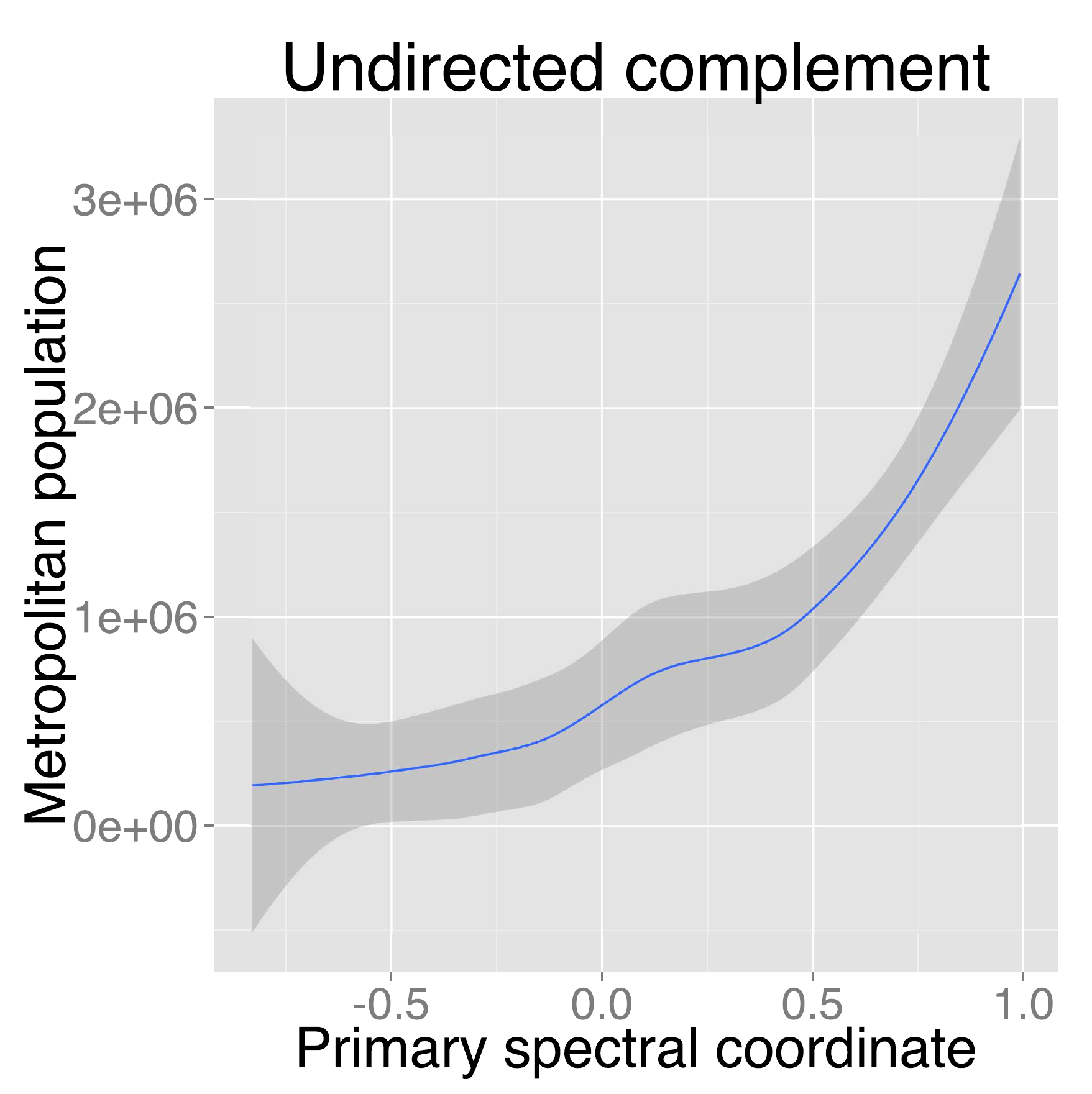}
\includegraphics[width=\airportfigcolwidth]{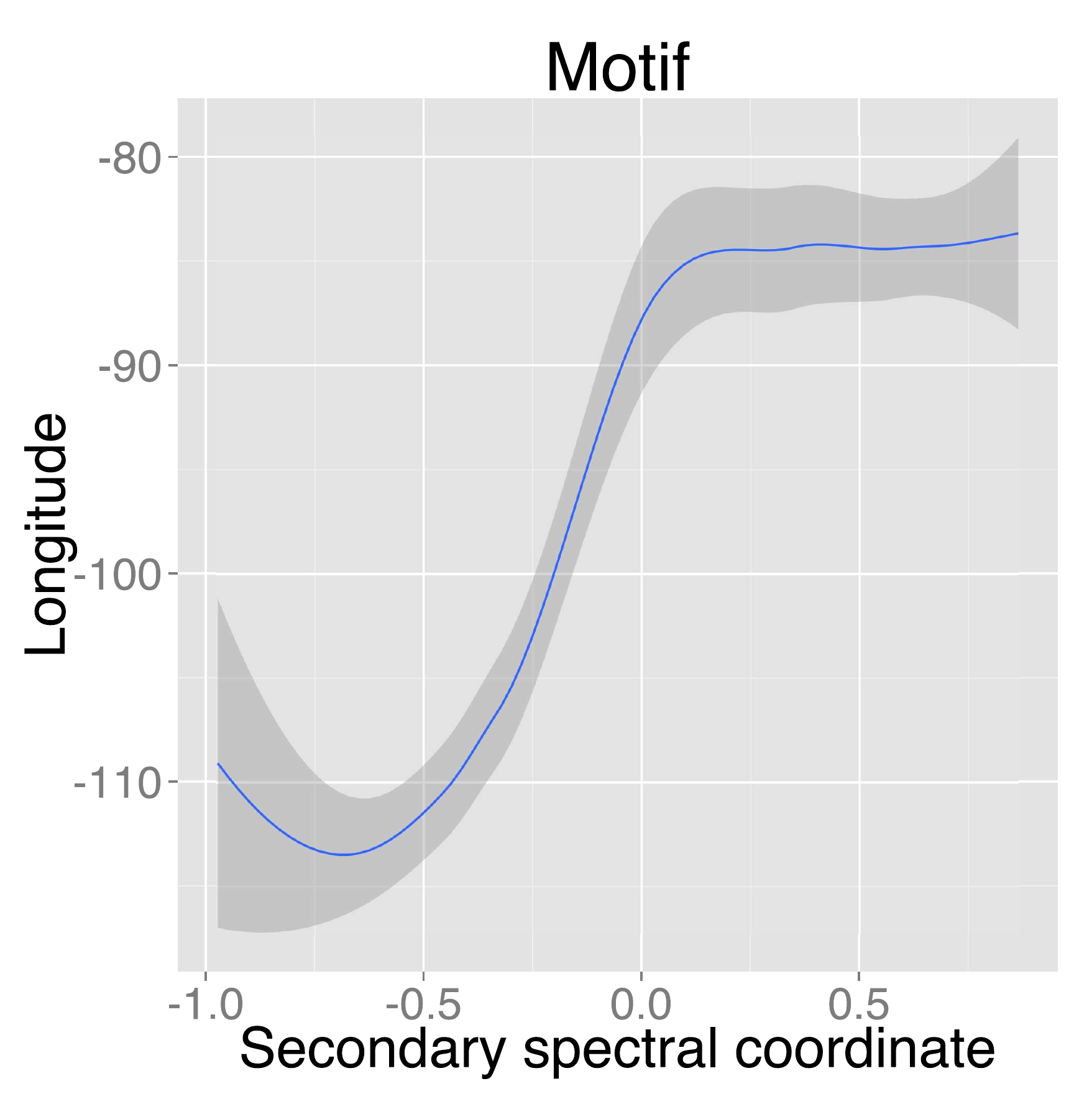}
\includegraphics[width=\airportfigcolwidth]{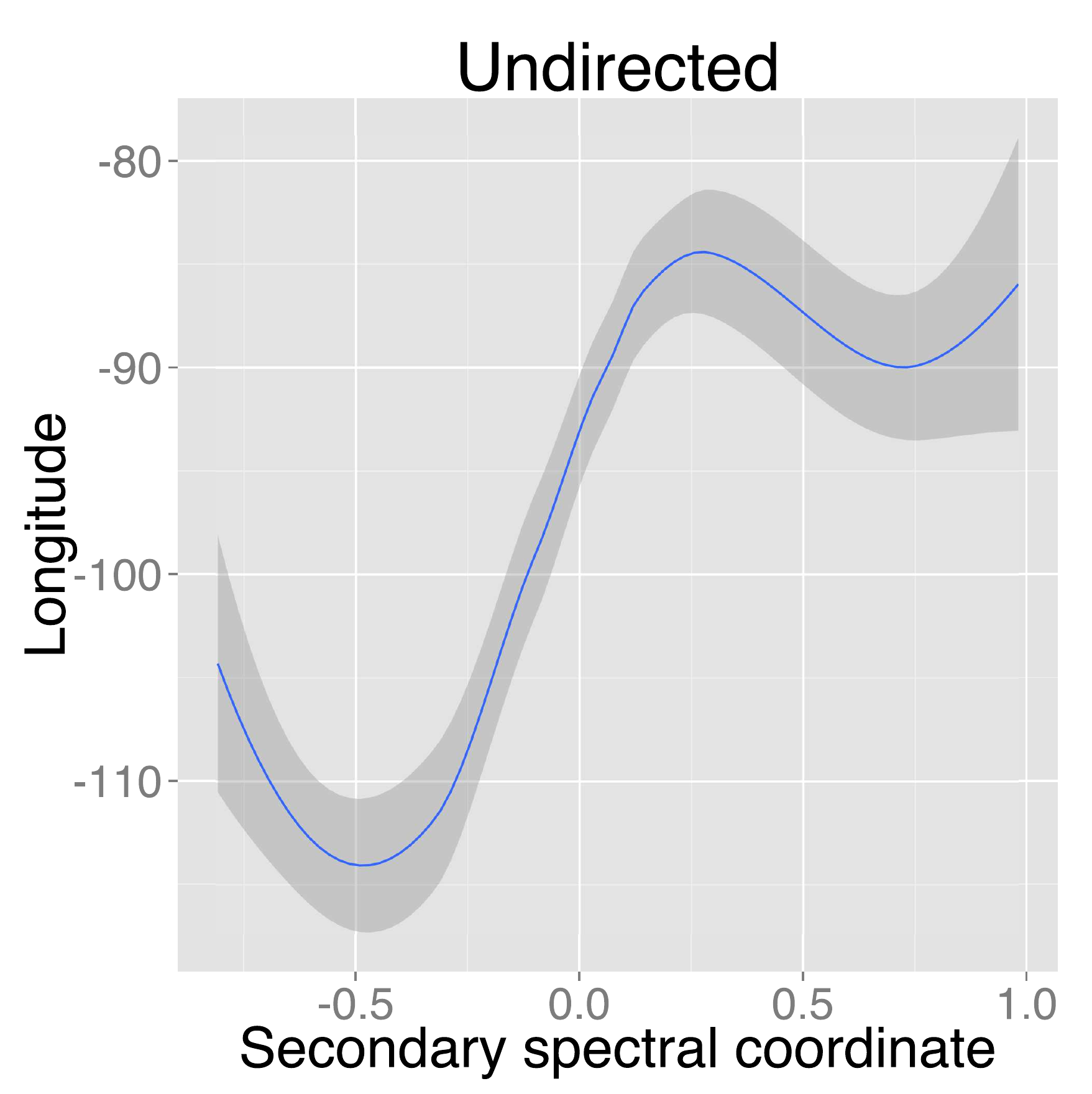}
\includegraphics[width=\airportfigcolwidth]{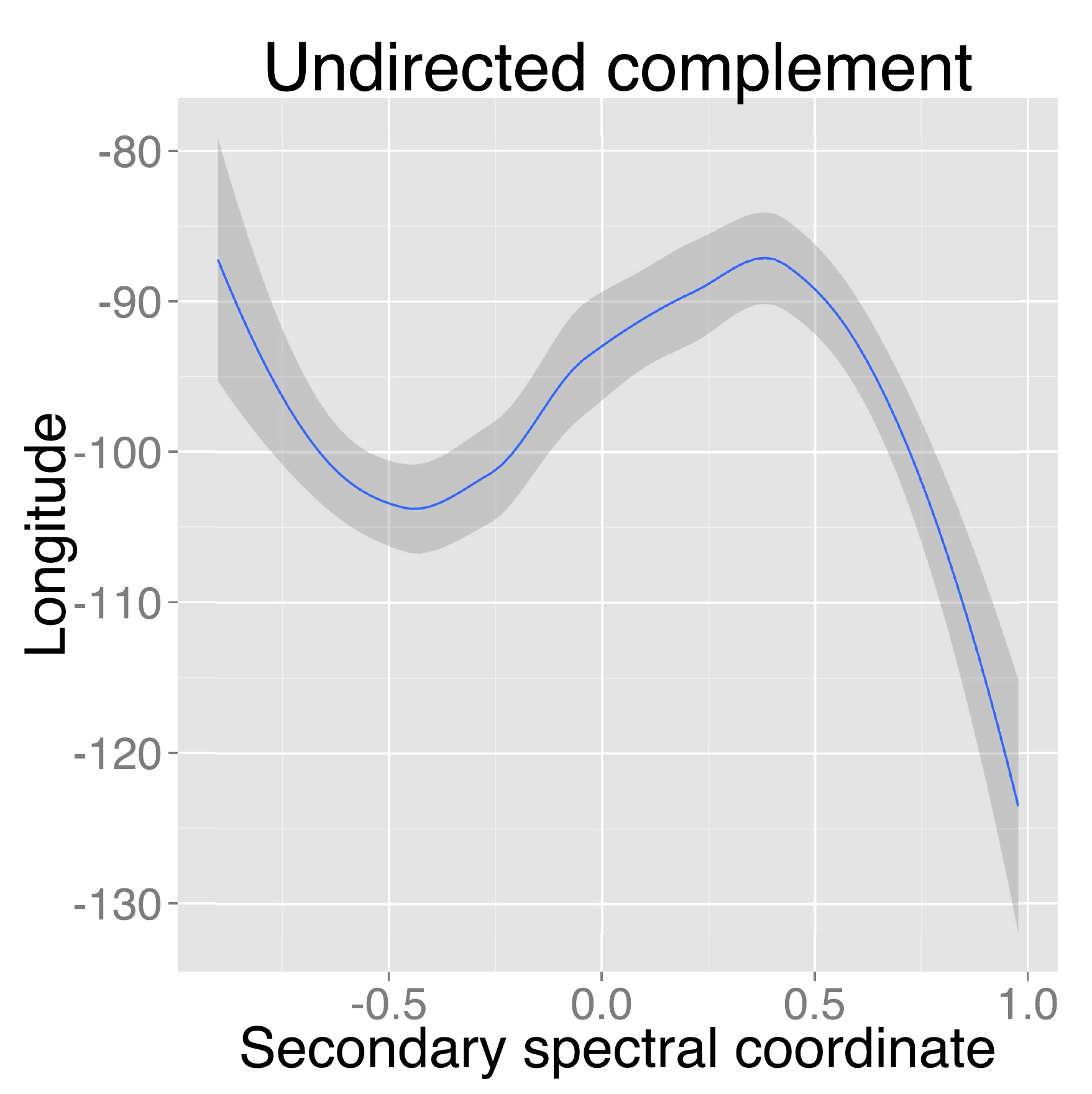}
\dualcaption{Loess regressions of city metadata against principal eigenvectors}{%
Specifically, the figures show the Loess regressions of city metropolitan
population against the primary spectral coordinate (top) and longitude against
secondary spectral coordinate (bottom) for the motif (left), undirected
(middle), and undirected complement (right) motif adjacency matrices.
}
\label{fig:airports_loess}
\end{figure}

\clearpage


\begin{figure}[h]
\centering
\includegraphics[width=4.5cm]{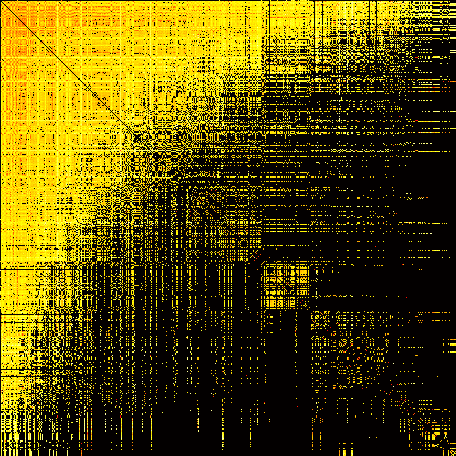}
\dualcaption{Transportation reachability network matrix}{%
Nodes are ordered according to the spectral ordering of the motif adjacency
matrix.  A black dot means no edge exists in the network.  For the edges in the
network, lighter colors correspond longer estimated travel times.
}
\label{fig:spectral_ordered_adj}
\end{figure}

We conclude that the embedding provided by the motif adjacency matrix more
strongly captures the hub nature of airports and West-East geography of the
network.  To gain further insight into the relationship of the primary spectral
coordinate's relationship with the hub airports, we visualize the adjacency
matrix in \cref{fig:spectral_ordered_adj}, where the nodes are ordered by the
spectral ordering.  We see a clear relationship between the spectral ordering
and the connectivity.

\subsection{The bi-fan motif in the \emph{C.~elegans} neuronal network}
\label{sec:honc_celegans}

Next, we apply the higher-order clustering framework to the \emph{C. elegans}
neuronal network with the four-node ``bi-fan'' $\mbifan$
(\cref{fig:celegans_bifan}), which is known to be over-expressed in neural
networks~\cite{milo2002network}.  Specifically, the network is of the frontal
neurons, where edges represent synapses between the
neurons~\cite{kaiser2006nonoptimal}.\footnote{The data was downloaded
from \url{http://www.biological-networks.org/pubs/suppl/celegans131.zip}.}  The
original network is weakly connected with 131 nodes and 764 edges.  The largest
connected component of the motif adjacency matrix for motif $\mbifan$ contains
112 nodes.  The remaining 19 nodes are isolated and correspond to the neurons
AFDL, AIAR, AINR, ASGL/R, ASIL/R, ASJL/R, ASKL/R, AVL, AWAL, AWCR, RID, RMFL,
SIADR, and SIBDL/R.

Using the bi-fan motif and \cref{alg:motif_fiedler}, we find a cluster of 20
neurons in the frontal section with low bi-fan motif conductance
(\cref{fig:celegans_cluster}).  The cluster shows a plausible way for how nictation is
controlled.  Within the cluster, ring motor neurons (RMEL/V/R), proposed
pioneers of the nerve ring~\cite{riddle1997celegans}, propagate information to
IL2 neurons, regulators of nictation~\cite{lee2012nictation}, through the neuron
RIH and several inner labial sensory neurons
(\cref{fig:celegans_cluster_embed}).  Our framework contextualizes the
signifance of the bi-fan motif in this control mechanism.

\definecolor{mycyan}{RGB}{0,202,202}
\definecolor{myorange}{RGB}{240,130,10}
\definecolor{mypurple}{RGB}{104,104,209}
\begin{figure}[t]
\centering
\phantomsubfigure{fig:celegans_bifan}
\phantomsubfigure{fig:celegans_cluster}
\phantomsubfigure{fig:celegans_cluster_embed}
\includegraphics[width=\textwidth]{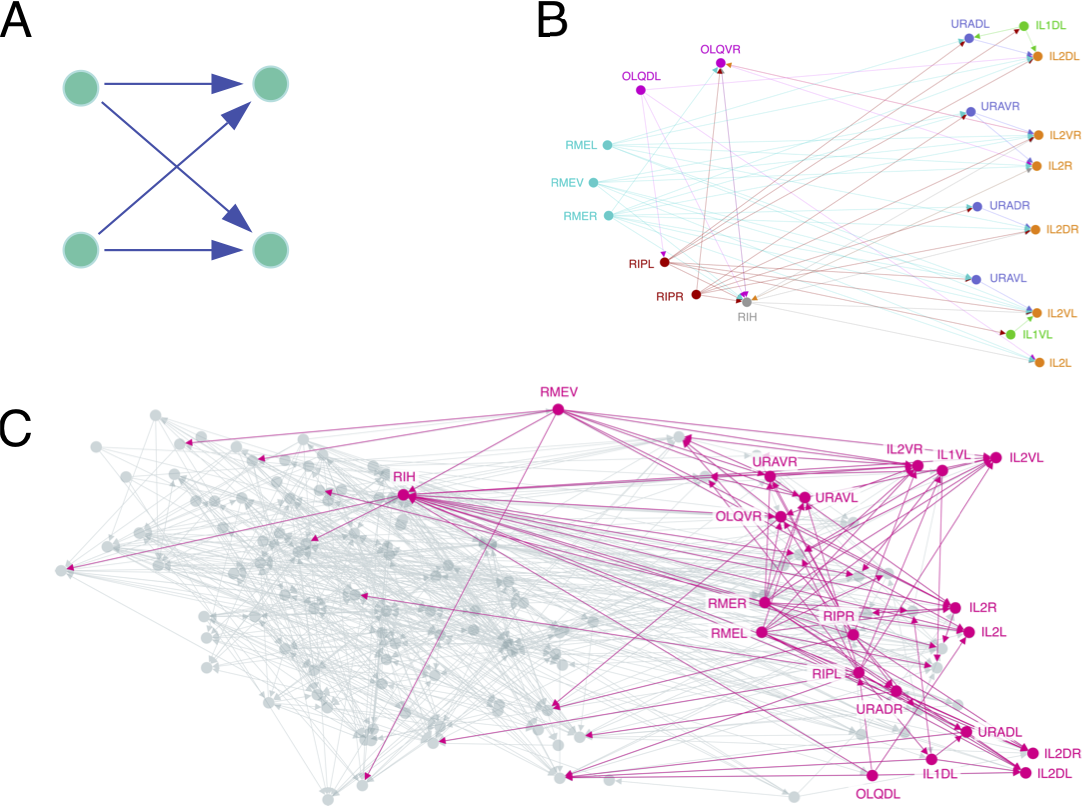}
\dualcaption{Higher-order cluster in the \emph{C.~elegans} neuronal network}{%
{\bf A:}
The 4-node ``bi-fan'' motif, which is over-expressed in neuronal
networks~\cite{milo2002network}.  Conceptually, this motif describes a
cooperative propagation of information from the nodes on the left to the nodes
on the right.
{\bf B:}
The best higher-order cluster in the \emph{C.~elegans} frontal neuronal network
based on the motif in (A).  The cluster contains three ring motor neurons
(RMEL/V/R; \textcolor{mycyan}{cyan}) with many outgoing connections, serving as the source of
information; six inner labial sensory neurons (IL2DL/VR/R/DR/VL; \textcolor{myorange}{orange}) with
many incoming connections, serving as the destination of information; and four
URA neurons (\textcolor{mypurple}{purple}) acting as intermediaries.  These RME neurons are
pioneers of the nerve ring~\cite{riddle1997celegans}, while the IL2
neurons are known regulators of nictation~\cite{lee2012nictation}, and the
higher-order cluster exposes their organization.  The cluster also reveals that
RIH serves as a critical intermediary of information processing.  This neuron
has incoming links from all three RME neurons, outgoing connections to five of
the six IL2 neurons, and the largest total number of connections of any neuron
in the cluster.
{\bf C:}
Illustration of the higher-order cluster in the context of the entire network.
Node locations are the true two-dimensional spatial embedding of the neurons.
RME/V/R/L and RIH serve as sources of information to the neurons on the right.
}
\label{fig:celegans}
\end{figure}

\clearpage

We compare the $\mbifan$-based cluster to the clusters found by edge-based
methods.  First, \cref{fig:celegans_ncp} shows the sweep profile plot from
running \cref{alg:motif_fiedler} for motifs $\mbifan$ and $\medge$.  Both curves
have a similar shape with local minima near the same set sizes---however, the
cluster found with $\mbifan$ is substantially smaller (20 nodes) than for
$\medge$ (64 nodes).  \Cref{fig:celegans_embed} provides a visualization of
these two clusters.  It turns out that the $\medge$-based cluster is a superset
of the $\mbifan$-based cluster, so we can roughly interpret clustering with
$\mbifan$ as identifying a more precise cluster than would be found with
spectral partitioning based on edges.  We also compute clusterings of the
network with Infomap and the Louvain method, to see if this trend holds for
other edge-based clustering methods.  Infomap finds a cluster 
that overlaps with the $\mbifan$-based cluster by 19 nodes but
contains 114 total nodes.  The Louvain method finds a cluster that overlaps by
13 nodes and has 27 total nodes.  We conclude that common edge-based methods do not
capture the same cluster information as our higher-order clustering framework
with the bi-fan motif.

We also compute the exact conductance minimizer of the weighted motif
adjacency matrix for the bi-fan motif, which is identical
to the $\medge$-based cluster.  In other words, the output of \cref{alg:motif_fiedler}
is not the global minimum for the objective function, even though it contains
meaningful structure.  The idea that the spectral method can give more meaningful
results than the conductance minimizer has been observed in other analyses~\cite{gleich2017personal}.


\definecolor{mygreen}{RGB}{27,158,119}
\definecolor{myorange}{RGB}{217,95,2}
\definecolor{mypurple}{RGB}{117,112,179}
\begin{figure}[t]
\centering
\includegraphics[width=9cm]{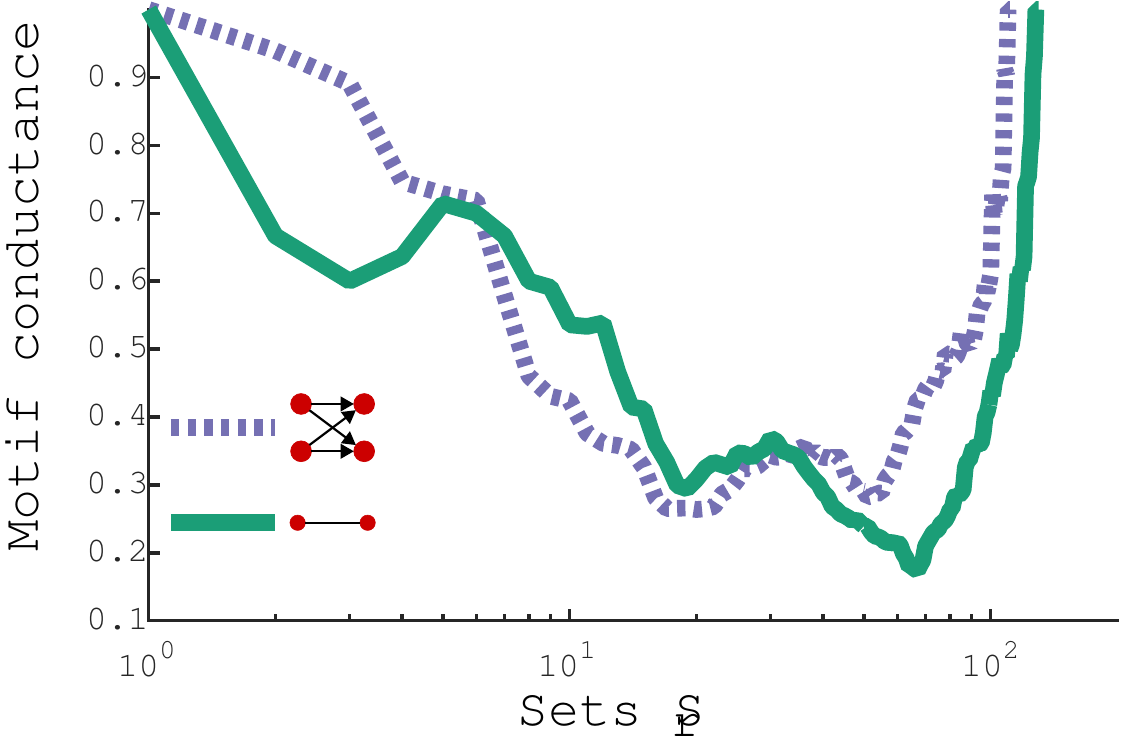}
\dualcaption{Sweep profile plots for \emph{C. elegans}}{%
The sweep profile plot shows $\mmcond{S}$ as a function of set size from the sweep
procedure in \cref{alg:motif_fiedler}.  We show the profiles for $\mbifan$ (\textcolor{mypurple}{purple})
and $\medge$ (\textcolor{mygreen}{green}).
}
\label{fig:celegans_ncp}
\end{figure}


\definecolor{mydarkblue}{RGB}{31,120,180}
\begin{figure}[t]
\centering
\includegraphics[width=\columnwidth]{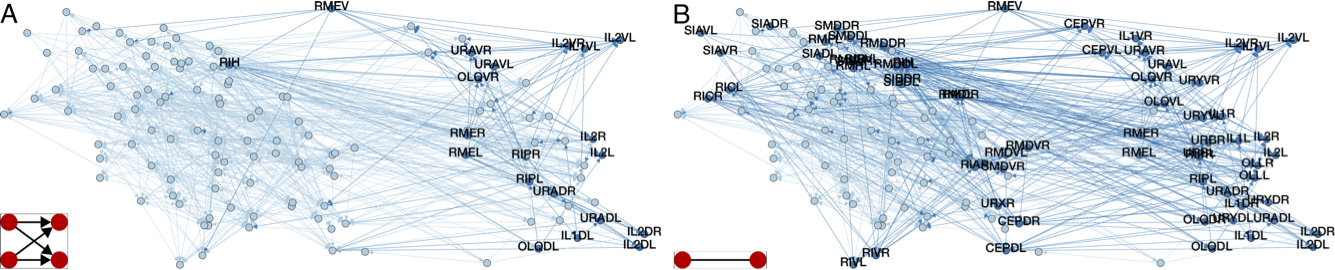}
\dualcaption{Clusters in \emph{C. elegans}}{%
The nodes are embedded according to a true two-dimensional
spatial embedding of the frontal neurons~\cite{kaiser2006nonoptimal}.
{\bf A:}
The $\mbifan$-based cluster consists of the labeled
nodes (same as \cref{fig:celegans_cluster}).
{\bf B:}
The $\medge$-based cluster consists of the labeled nodes.
This cluster is a superset of $\mbifan$-based cluster.
}
\label{fig:celegans_embed}
\end{figure}

\clearpage

\subsection{Motif $M_{6}$ in the English Wikipedia article network}
\label{sec:honc_enwiki}

The English Wikipedia
network~\cite{boldi2004webgraph,boldi2011layered,boldi2004ubicrawler} consists
of 4.21 million nodes (representing articles) and 101.31 million edges, where an
edge from node $i$ to node $j$ means that there is a hyperlink from the $i$th
article to the $j$th article.\footnote{The data was downloaded from
\url{http://law.di.unimi.it/webdata/enwiki-2013/}.}
We use \cref{alg:motif_fiedler} to find a motif-based cluster for motifs $M_{6}$
and $\medge$ (the algorithm was run on the largest connected component of the
motif adjacency matrix).  \Cref{fig:enwiki_comms} shows the clusters.  The nodes
in the motif-based cluster are cities and barangays (small administrative
divisions) in the Philippines.  The cluster has a set of nodes with many
outgoing links that form the source node in motif $M_{6}$.  In total, the
cluster consists of 22 nodes and 338 edges.  The linking pattern appears
anomalous and suggests that perhaps the pages uplinking should receive
reciprocated links.  On the other hand, the edge-based cluster is much larger
cluster and does not have too much structure.  The cluster consists of several
high-degree nodes and their neighbors.

\definecolor{mygreen}{RGB}{27,158,119}
\definecolor{myorange}{RGB}{217,95,2}
\begin{figure}[tb]
\begin{centering}
\includegraphics[width=\columnwidth]{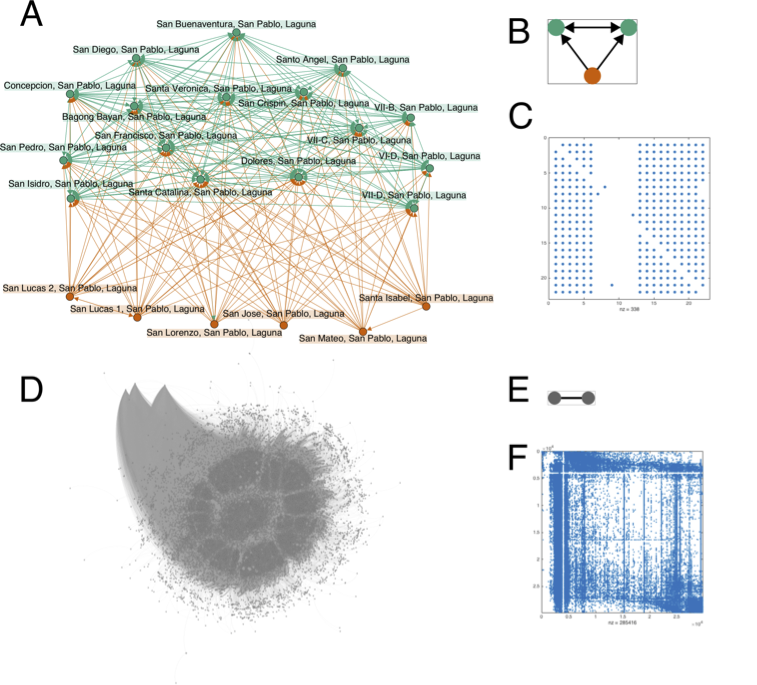}
\dualcaption{Clusters in the English Wikipedia hyperlink network}{%
{\bf A--C}:
Motif-based cluster (A) for motif $M_{6}$ (B).  The cluster consists of cities
and small administrative divisions in the Philippines. The \textcolor{mygreen}{green}
nodes have many bidirectional links with each other and many incoming links
from \textcolor{myorange}{orange} nodes at the bottom of the figure.  The spy
plot illustrates this network structure (C).
{\bf D--F}:
Cluster (D) for undirected edges (E).  The cluster has a few very high-degree
nodes, as evidenced by the spy plot (F).
}
\label{fig:enwiki_comms}
\end{centering}
\end{figure}

\clearpage

\subsection{Motif $M_{6}$ in the Twitter follower network}
\label{sec:honc_twitter}

We also analyze the complete Twitter follower
graph from 2010~\cite{boldi2004webgraph,boldi2011layered,kwak2010twitter}.  The graph
consists of 41.65 million nodes (users) and 1.47 billion edges, where an edge from
node $i$ to node $j$ signifies that user $i$ is followed by user $j$ on the
social network.\footnote{The data was downloaded
from \url{http://law.di.unimi.it/webdata/twitter-2010/}.}  We
use \cref{alg:motif_fiedler} to find a motif-based cluster for motif $M_{6}$
(the algorithm was run on the largest connected component of the motif adjacency
matrix).  The cluster contains 151 nodes and consists of two disconnected
components.  Here, we consider the smaller of the two components, which consists
of 38 nodes.  We also found an edge-based cluster on the undirected graph
(using \cref{alg:motif_fiedler} with motif $\medge$), which consists of 44 nodes.

\Cref{fig:twitter_comms} illustrates the motif-based and edge-based
clusters.  Both clusters capture anomalies in the graph.  The motif-based
cluster consists of holding accounts for a photography company.  The nodes that
form bidirectional links have completed profiles (contain a profile picture)
while several nodes with incomplete profiles (without a profile picture) are
followed by the completed accounts (\cref{fig:twitter_anomaly}).  The edge-based
cluster is a near clique, where the user screen names all begin with ``LC\_''.
We suspect that the similar usernames are either true social communities,
holding accounts, or bots. (For the most part, the tweets of these
accounts are protected, so we cannot verify if any of these scenarios are true.)
Both $M_{6}$ and $\medge$ find anomalous clusters, but
their structures are quite different.  We conclude that $M_{6}$ can lead to the
detection of new anomalous clusters in social networks.


\definecolor{mygreen}{RGB}{27,158,119}
\definecolor{myorange}{RGB}{217,95,2}
\begin{figure}[tb]
  \begin{centering}
    \phantomsubfigure{fig:twitter_commsA}
    \phantomsubfigure{fig:twitter_commsB}
    \phantomsubfigure{fig:twitter_commsC}
    \phantomsubfigure{fig:twitter_commsD}
    \phantomsubfigure{fig:twitter_commsE}
    \phantomsubfigure{fig:twitter_commsF}
\includegraphics[width=\columnwidth]{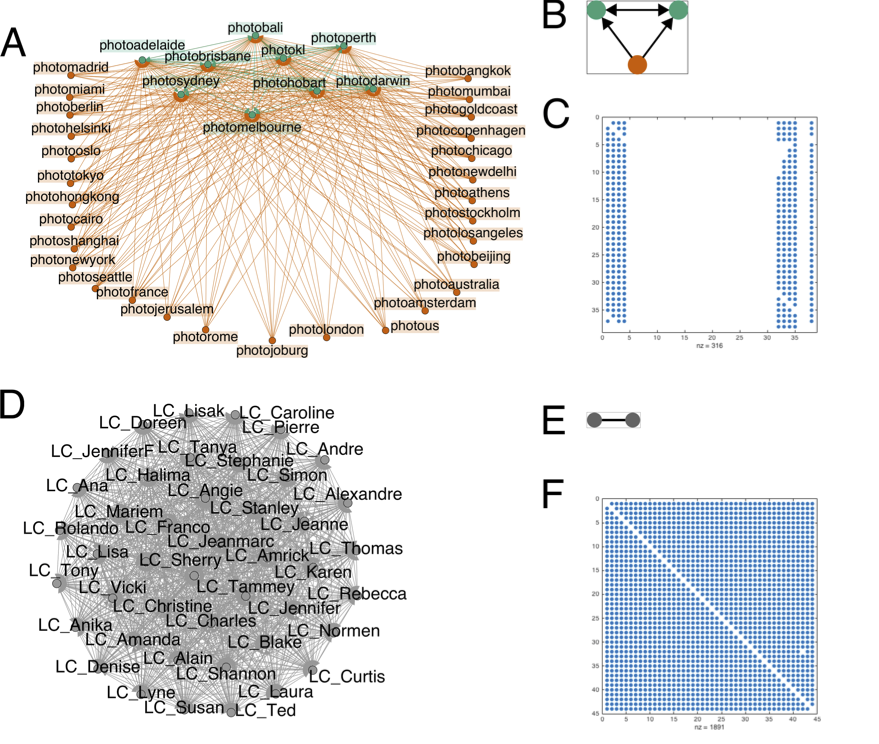}
\dualcaption{Clusters in the 2010 Twitter follower network}{%
{\bf A--C}:
Motif-based cluster (A) for motif $M_{6}$ (B).  All accounts are holding
accounts for a photography company.  The \textcolor{mygreen}{green} nodes
correspond to accounts that have completed profiles, while the accounts
corresponding to the \textcolor{myorange}{orange} nodes have incomplete profiles
(see \cref{fig:twitter_anomaly}).  The spy plot illustrates how the cluster is
formed around this motif (C).
{\bf D--F}:
Cluster (D) for edge-based clustering (E).  The cluster consists of a
near-clique (F) where all users have the prefix ``LC\_''.
}
\label{fig:twitter_comms}
\end{centering}
\end{figure}


\begin{figure}[htb]
\begin{centering}
  \includegraphics[width=0.85\columnwidth]{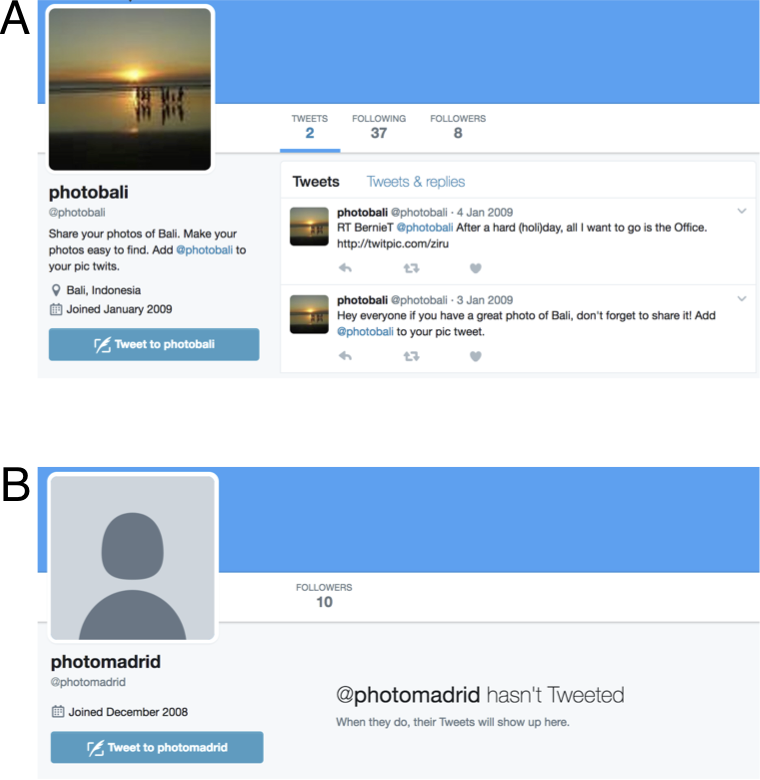}
  \dualcaption{@photobali and @photomadrid Twitter accounts}{The motif-based
    cluster for $M_{6}$ (\cref{fig:twitter_commsA}) consists of holding accounts
    for a photography company.  The nodes forming bidirectional links in motif $M_{6}$
    have complete profiles (A) and those forming unidirectional links have
    incomplete profiles (B).}
\label{fig:twitter_anomaly}
\end{centering}
\end{figure}

\clearpage

\subsection{Motif $M_{7}$ in the Stanford web graph}
\label{sec:honc_stanford}

The Stanford web graph consists of 281,903 nodes and 2,312,497 edges, where an
edge from node $i$ to node $j$ means that there is a hyperlink from the $i$th
web page to the $j$th web page~\cite{hirai2000webbase,kamvar2003extrapolation}.
Here, all of the web pages come from the \texttt{stanford.edu} domain.\footnote{The data was
downloaded from \url{http://snap.stanford.edu/data/web-Stanford.html}.}  We
use \cref{alg:motif_fiedler} to find a motif-based cluster for motif $M_{7}$, a
motif that is over-expressed in web graphs~\cite{milo2002network}.
\Cref{fig:web_stanford} illustrates this cluster
and an edge-based cluster (i.e., using \cref{alg:motif_fiedler} with $\medge$).
Both clusters exhibit a core-periphery structure, albeit markedly different
ones.  The motif-based cluster contains several core nodes with large in-degree.
Such core nodes serve as the sink node in motif $M_{7}$.  On the periphery are
several clusters within which are many bidirectional links (as illustrated by
the spy plot in \cref{fig:web_stanfordC}).  The nodes in these clusters then
up-link to the core nodes.  This type of organizational unit suggests an
explanation for why motif $M_{7}$ is over-expressed: clusters of similar pages
tend to uplink to more central pages.  The edge-based cluster also has a few
nodes with large in-degree, serving as a small core.  On the periphery are the
neighbors of these nodes, which themselves tend \emph{not} to be connected
(\cref{fig:web_stanfordF}).


\definecolor{mygreen}{RGB}{27,158,119}
\definecolor{myorange}{RGB}{217,95,2}
\begin{figure}[thb]
  \begin{centering}
    \phantomsubfigure{fig:web_stanfordA}
    \phantomsubfigure{fig:web_stanfordB}
    \phantomsubfigure{fig:web_stanfordC}
    \phantomsubfigure{fig:web_stanfordD}
    \phantomsubfigure{fig:web_stanfordE}
    \phantomsubfigure{fig:web_stanfordF}
    \scalebox{0.92}{\includegraphics[width=\columnwidth]{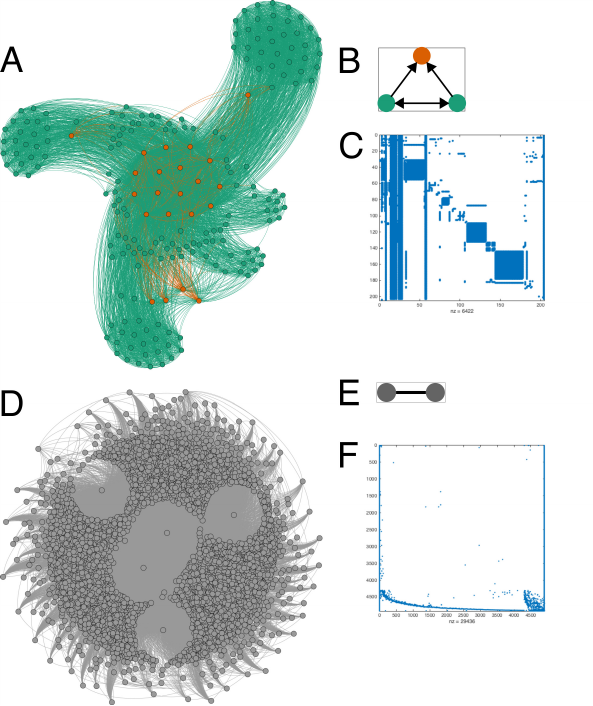}}
\dualcaption{Clusters in the Stanford web graph}{%
{\bf A--C:}
Motif-based cluster (A) for motif $M_{7}$ (B).  The cluster has a core group of
nodes with many incoming links (serving as the sink node in $M_{7}$; shown
in \textcolor{myorange}{orange}) and several periphery groups that are tied
together (the nodes forming the bidirectional link in $M_{7}$; shown in
\textcolor{mygreen}{green}) and also up-link to the core.  This is evident
from the spy plot (C).
{\bf D--F:}
Cluster (C) for undirected edges (B).  The cluster contains a few high-degree
nodes and their neighbors, and the neighbors tend to not be connected (F).
}
\label{fig:web_stanford}
\end{centering}
\end{figure}

\clearpage

\subsection{Semi-cliques in collaboration networks}
\label{sec:honc_collaboration}

We use \cref{alg:motif_fiedler} to identify clusters of the four-node
semi-clique motif that has been studied in conjunction with researcher
productivity in collaboration networks~\cite{chakraborty2014automatic}.
We found a motif-based cluster in two different
collaboration networks (\cref{fig:collabA,fig:collabC}).  Each network
is derived from co-authorship in papers
submitted to the arXiv under a certain category---the ``High Energy
Physics -- Theory'' (HepTh) and ``Condensed Matter Physics'' (CondMat)
categories~\cite{leskovec2007graph}.  The HepTh network has 23,133 nodes and
93,497 edges and the CondMat network has 9,877 nodes and 25,998
edges.\footnote{The HepTh network data was downloaded from
\url{http://snap.stanford.edu/data/ca-HepTh.html}
and the CondMat network data was downloaded from
\url{http://snap.stanford.edu/data/ca-CondMat.html}.}

\Cref{fig:collab} shows the two clusters for each of the collaboration
networks.  In both networks, the motif-based cluster consists of a core group of
nodes and similarly-sized groups on the periphery.  The core group of nodes
correspond to the nodes of degree 3 in the motif and the periphery group nodes
correspond to the nodes of degree 2.  One explanation for this organization is
that there is a small small group of authors that writes papers with different
research groups.  Alternatively, the co-authorship could come from a single
research group, where senior authors are included on all of the papers and
junior authors on a subset of the papers.

On the other hand, the edge-based clusters (i.e., the output of
\cref{alg:motif_fiedler} for $\medge$) are a clique in the HepTh
network and a clique with a few dangling nodes in the CondMat network.  The
dense clusters are quite different from the sparser clusters based on the
semi-clique.  Such dense clusters are not that surprising.  For example, a
clique could arise from a single paper published by a group of authors.


\begin{figure}[htb]
  \begin{centering}
    \phantomsubfigure{fig:collabA}
    \phantomsubfigure{fig:collabB}
    \phantomsubfigure{fig:collabC}
    \phantomsubfigure{fig:collabD}
    \phantomsubfigure{fig:collabE}
    \phantomsubfigure{fig:collabF}
  \includegraphics[width=\columnwidth]{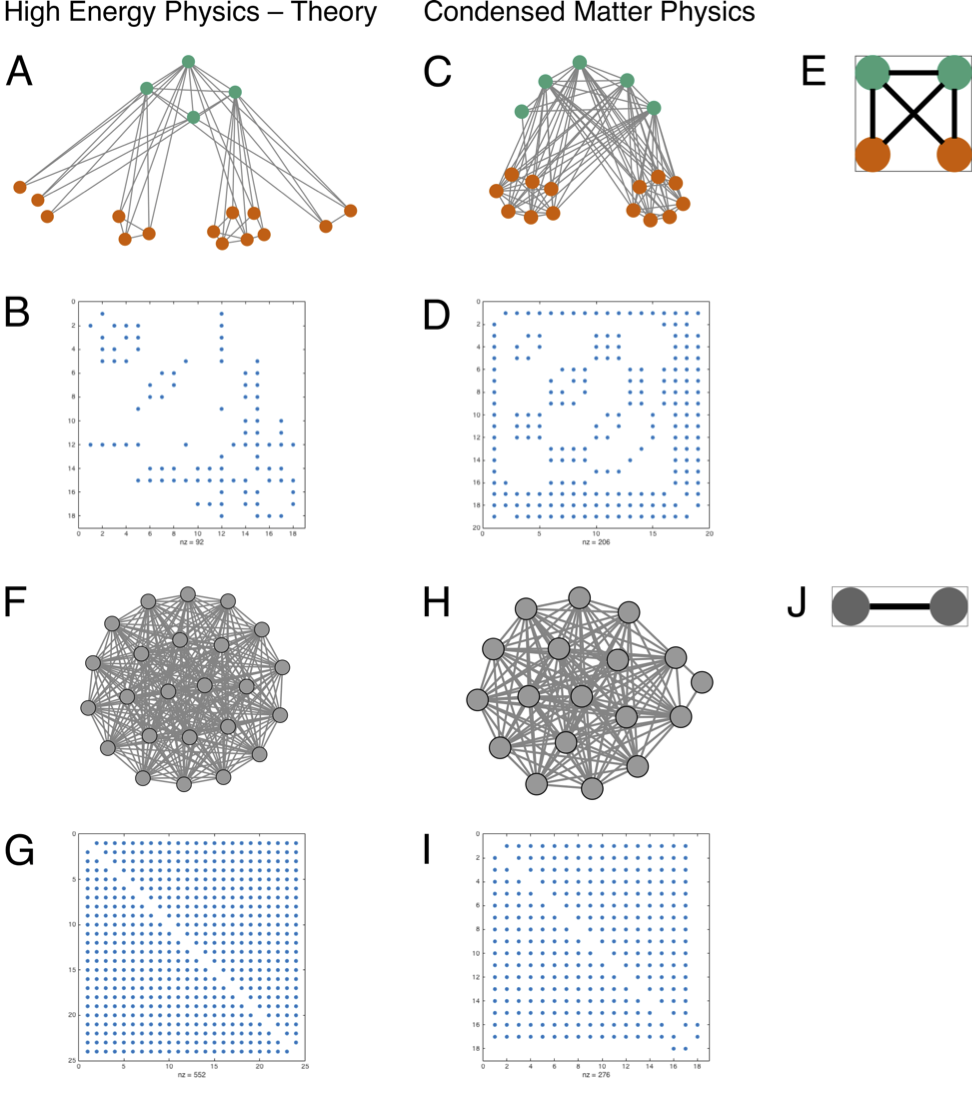}
\dualcaption{Clusters in co-authorship networks}{%
{\bf A--E:}
Best motif-based cluster for the semi-clique motif (E) in the High Energy
Physics--Theory collaboration network (A) and the Condensed Matter Physics
collaboration network (C).  Corresponding spy plots are shown in (B) and (D).
{\bf F--I:}
Best edge-based (I) cluster in the High Energy Physics--Theory collaboration
network (F) and the Condensed Matter Physics collaboration network (H).
Corresponding spy plots are shown in (G) and (I).
}
\label{fig:collab}
\end{centering}
\end{figure}

\clearpage

\section{Scalability experiments}
\label{sec:honc_scalability}

We now empirically analyze the time to find clusters for triangular motifs on a
variety of real-world networks, ranging in size from a few hundred thousand
edges to nearly two billion edges.  Then we show that we can compute
the motif adjacency matrix for cliques up to size 9 on a number of real-world
networks in a reasonable amount of time.

\subsection{Triangular motifs}
\label{sec:scalability}

In this section, we demonstrate that our method scales to real-world networks
with billions of edges.  We tested the scalability of our method on 16 large
directed graphs from a variety of real-world applications.  These networks range
from a couple hundred thousand to two billion edges and from 10 thousand to over
50 million nodes.  We briefly describe the networks here and provide some
summary statistics in \Cref{tab:networks_description}.
\begin{itemize}
\item \dataset{wiki-RfA} represents which users voted for which other users for
  adminship rights on Wikipedia~\cite{west2014exploiting}.
\item \dataset{email-EuAll} represents who emailed whom at a European research
  institution~\cite{leskovec2005graphs}.
\item \dataset{cit-HepPh} represents citations between papers in the ``High
  Energy Physics -- Phenomenology'' category on
  arXiv~\cite{gehrke2003overview}.
\item \dataset{web-NotreDame} is the hyperlink structure of web pages in the
  \texttt{nd.edu} domain~\cite{albert1999internet}.
\item \dataset{amazon0601} consists of connections between frequently
  co-purchased products on Amazon~\cite{leskovec2007dynamics}.
\item \dataset{wiki-Talk} represents which users wrote on which other users' talk page
  on Wikipedia~\cite{leskovec2010governance}.
\item \dataset{ego-Gplus} is a collection of egonetworks (1-hop neighborhoods)
  on the online social network Google+~\cite{leskovec2012learning}.
\item \dataset{uk-2014-tpd} is the hyperlink structure of top private domain
  links in the \texttt{.uk}
  domain~\cite{boldi2004webgraph,boldi2011layered,boldi2014bubing}.
\item \dataset{soc-Pokec} consists of the friendship relationships on the online
  social network Pokec~\cite{takac2012data}.
\item \dataset{uk-2014-host} is the hyperlink structure of host links on the
  \texttt{.uk} domain~\cite{boldi2004webgraph,boldi2011layered,boldi2014bubing}.
\item \dataset{soc-LiveJournal1} consists of the friendships on the online
  social network LiveJournal~\cite{backstrom2006group}.
\item \dataset{enwiki-2013} is the hyperlink structure of articles on English
  Wikipedia~\cite{boldi2004webgraph,boldi2004ubicrawler,boldi2011layered}
\item \dataset{uk-2002} is the hyperlink structure of web pages in the
  \texttt{.uk}
  domain~\cite{boldi2004webgraph,boldi2004ubicrawler,boldi2011layered}.
\item \dataset{arabic-2005} is the hyperlink structure of arabic-language web
  pages~\cite{boldi2004webgraph,boldi2004ubicrawler,boldi2011layered}.
\item \dataset{twitter-2010} represents the followers of users on the online
  social network
  Twitter in 2010~\cite{boldi2004webgraph,boldi2011layered,kwak2010twitter}.
\item \dataset{sk-2005} is the hyperlink structure of web pages in the
  \texttt{.sk}
  domain~\cite{boldi2004webgraph,boldi2004ubicrawler,boldi2011layered}.
\end{itemize}


\begin{table}[hb]
\centering
\dualcaption{Networks for scalability experiments}{The total number of edges is
  the sum of the number of unidirectional edges and twice the number of
  bidirectional edges.}
\begin{tabular}{l c c c c}
\toprule
Dataset                       & \# nodes & \# total edges & \# unidir.\ edges & \# bidir.\ edges \\
\midrule
\dataset{wiki-RfA}         & 10.8K    & 189K           & 175K              & 7.00K            \\
\dataset{email-EuAll}      & 265K     & 419K           & 310K              & 54.5K            \\
\dataset{cit-HepPh}        & 34.5K    & 422K           & 420K              & 657              \\
\dataset{web-NotreDame}    & 326K     & 1.47M          & 711K              & 380K             \\
\dataset{amazon0601}       & 403K     & 3.39M          & 1.50M             & 944K             \\
\dataset{wiki-Talk}        & 2.39M    & 5.02M          & 4.30M             & 362K             \\
\dataset{ego-Gplus}        & 108K     & 13.7M          & 10.8M             & 1.44M            \\
\dataset{uk-2014-tpd}      & 1.77M    & 16.9M          & 13.7M             & 1.58M            \\
\dataset{soc-Pokec}        & 1.63M    & 30.6M          & 14.0M             & 8.32M            \\
\dataset{uk-2014-host}     & 4.77M    & 46.8M          & 33.7M             & 6.55M            \\
\dataset{soc-LiveJournal1} & 4.85M    & 68.5M          & 17.2M             & 25.6M            \\
\dataset{enwiki-2013}      & 4.21M    & 101M           & 82.6M             & 9.37M            \\
\dataset{uk-2002}          & 18.5M    & 292M           & 231M              & 30.5M            \\
\dataset{arabic-2005}      & 22.7M    & 631M           & 477M              & 77.3M            \\
\dataset{twitter-2010}     & 41.7M    & 1.47B          & 937M              & 266M             \\
\dataset{sk-2005}          & 50.6M    & 1.93B          & 1.69B             & 120M             \\
\bottomrule
\end{tabular}
\label{tab:networks_description}
\end{table}

\clearpage

Recall that \Cref{alg:motif_fiedler} consists of two major computational components:
\begin{enumerate}
\item Form the weighted graph $W_M$.
\item Compute the eigenvector $z$ of second smallest eigenvalue of $\normmotiflap$.
\end{enumerate}
After computing the eigenvector, we sort the vertices and loop over prefix sets
to find the lowest motif conductance set.  We consider these final steps as part
of the eigenvector computation for our performance experiments.

For each network in \cref{tab:networks_description}, we ran the method for all
directed triangular motifs ($M_1$--$M_7$).  To compute $W_M$, we used a standard
algorithm that meets the $O(m^{3/2})$
bound~\cite{schank2005finding,latapy2008main} with some additional
pre-processing based on the motif.  Conceptually, the algorithm is as follows:

\begin{enumerate}
\item Take motif type $M$ and graph $G$ as input.

\item (Pre-processing)  If $M$ is $M_1$ or $M_5$, ignore all bidirectional
  edges in $G$ as these motifs only contain unidirectional edges.  If $M$ is
  $M_4$, ignore all unidirectional edges in $G$ as this motif only contains
  bidirectional edges.
  
\item Form the undirected graph $G_{\textnormal{undir}}$ by removing the
  direction of all edges in $G$.
  
\item Let $d_u$ be the degree of node $u$ in $G_{\textnormal{undir}}$.  Order
  the nodes in $G_{\textnormal{undir}}$ by increasing degree, breaking ties
  arbitrarily.  Denote this ordering by $\psi$.
  
\item For every edge undirected edge $\{u, v\}$ in $G_{\textnormal{undir}}$, if
  $\psi_u < \psi_v$, add directed edge $(u, v)$ to $G_{\textnormal{dir}}$;
  otherwise, add directed edge $(v, u)$ to $G_{\textnormal{dir}}$.
  
\item For every node in $u$ in $G_{\textnormal{dir}}$ and every pair of directed
  edges $(u, v)$ and $(u, w)$, check to see if edge $(v, w)$ or $(w, v)$ is in
  $G_{\textnormal{dir}}$.  If so, check if these three nodes form motif $M$ in
  $G$.  If they do, increment the weights of edges $(W_M)_{uv}$,
  $(W_M)_{uw}$, and $(W_M)_{vw}$ by $1$.
  
\item Return $W_M + W_M^T$ as the motif weighted adjacency matrix.  
\end{enumerate}

The algorithm runs in time $O(m^{3/2})$ time and is also
known as an effective heuristic for real-world
networks~\cite{berry2014why,latapy2008main}.  After, we find the largest
connected component of the graph corresponding to the motif adjacency matrix
$W_M$, form the motif normalized Laplacian $\normmotiflap$ of the largest
component, and compute the eigenvector of second smallest eigenvalue of
$\normmotiflap$.  To compute the eigenvector, we use MATLAB's \texttt{eigs}
routine with tolerance 1e-4 and the ``smallest algebraic'' option for the
eigenvalue type.

\Cref{tab:scalability_results} lists the time to compute $W_M$ and the time to
compute the eigenvector for each network.  We omit the time to read the graph
from disk because this time strongly depends on how the graph is compressed.
All experiments ran on a 40-core server with four 2.4 GHz Intel Xeon E7-4870
processors.  All computations of $W_M$ were in serial and the computations of
the eigenvectors were done in parallel (as defaulted to by MATLAB).

Over all networks and all motifs, the longest computation of $W_M$ (including
pre-processing time) was for $M_2$ on \dataset{sk-2005} and took roughly 52.8
hours.  The longest eigenvector computation was for $M_6$ on \dataset{sk-2005}
and took about 1.62 hours. We note that $W_M$ only needs to be computed once per
network, regardless of the eventual number of clusters that are extracted.
Also, the computation of $W_M$ can easily be accelerated by parallel computing
(the enumeration of motifs can be done in parallel over nodes, for example) or
by more sophisticated algorithms~\cite{berry2014why}.  In this work, we perform
the computation of $W_M$ in serial in order to better understand the
scalability.  Our results serve only as a rough baseline.

\begin{table}[h]
\dualcaption{Performance of motif-based clustering}{%
Time to compute the motif adjacency matrix $W_M$ (top) and second eigenvector of
the motif normalized Laplacian $\normmotiflap$ (bottom) in seconds for each
directed triangular motif.
}
\scalebox{0.84}{
\begin{tabular}{l c c c c c c c}
\toprule
Network & $M_{1}$ & $M_{2}$ & $M_{3}$ & $M_{4}$ & $M_{5}$ & $M_{6}$ & $M_{7}$ \\ 
\midrule
\dataset{wiki-RfA}         & 1.19e+00 & 2.67e+00 & 1.71e+00 & 2.06e-02 & 1.79e+00 & 2.42e+00 & 2.35e+00 \\
\dataset{email-EuAll}      & 4.74e-01 & 8.29e-01 & 6.26e-01 & 2.46e-01 & 5.02e-01 & 5.40e-01 & 5.41e-01 \\
\dataset{cit-HepPh}        & 7.65e+00 & 3.36e+00 & 2.73e+00 & 6.22e+00 & 8.20e+00 & 3.29e+00 & 3.35e+00 \\
\dataset{web-NotreDame}    & 9.42e-01 & 2.39e+01 & 2.33e+01 & 2.30e+00 & 1.17e+00 & 8.29e+00 & 8.40e+00 \\
\dataset{amazon0601}       & 2.35e+00 & 8.66e+00 & 6.91e+00 & 1.82e+00 & 2.94e+00 & 5.47e+00 & 5.73e+00 \\
\dataset{wiki-Talk}        & 1.07e+01 & 3.00e+01 & 2.20e+01 & 3.11e+00 & 1.35e+01 & 2.09e+01 & 2.10e+01 \\
\dataset{ego-Gplus}        & 8.55e+02 & 2.42e+03 & 1.73e+03 & 2.08e+01 & 1.63e+03 & 2.07e+03 & 2.17e+03 \\
\dataset{uk-2014-tpd}      & 8.10e+01 & 5.31e+02 & 4.07e+02 & 2.56e+01 & 1.15e+02 & 3.04e+02 & 2.85e+02 \\
\dataset{soc-Pokec}        & 4.17e+01 & 1.34e+02 & 1.21e+02 & 3.04e+01 & 4.88e+01 & 1.00e+02 & 1.04e+02 \\
\dataset{uk-2014-host}     & 9.98e+02 & 4.68e+03 & 2.76e+03 & 8.90e+01 & 1.32e+03 & 2.89e+03 & 2.99e+03 \\
\dataset{soc-LiveJournal1} & 9.08e+01 & 7.66e+02 & 6.24e+02 & 1.24e+02 & 1.24e+02 & 4.41e+02 & 4.49e+02 \\
\dataset{enwiki-2013}      & 8.36e+02 & 9.62e+02 & 7.09e+02 & 3.13e+01 & 9.77e+02 & 8.19e+02 & 8.38e+02 \\
\dataset{uk-2002}          & 1.47e+03 & 8.59e+03 & 5.17e+03 & 2.45e+02 & 1.73e+03 & 4.53e+03 & 5.29e+03 \\
\dataset{arabic-2005}      & 6.51e+03 & 7.64e+04 & 6.05e+04 & 6.08e+03 & 8.39e+03 & 3.59e+04 & 3.69e+04 \\
\dataset{twitter-2010}     & 1.21e+04 & 1.38e+05 & 1.31e+05 & 3.33e+04 & 1.99e+04 & 8.03e+04 & 7.65e+04 \\
\dataset{sk-2005}          & 5.52e+04 & 1.63e+05 & 1.29e+05 & 1.55e+04 & 5.23e+04 & 9.64e+04 & 8.42e+04 \\
\midrule
\dataset{wiki-RfA}         & 1.14e-01 & 2.12e-01 & 1.22e-01 & 2.12e-01 & 2.12e-01 & 2.94e-01 & 2.93e-01 \\
\dataset{email-EuAll}      & 2.29e-01 & 1.62e-01 & 2.43e-01 & 1.62e-01 & 1.62e-01 & 2.35e-01 & 1.92e-01 \\
\dataset{cit-HepPh}        & 2.11e+00 & 2.10e+00 & 2.11e+00 & 2.10e+00 & 2.10e+00 & 2.24e+00 & 2.30e+00 \\
\dataset{web-NotreDame}    & 1.86e-01 & 3.62e-01 & 5.97e-01 & 3.62e-01 & 3.62e-01 & 9.61e-01 & 2.06e+00 \\
\dataset{amazon0601}       & 1.23e-01 & 6.96e-01 & 4.62e+00 & 6.96e-01 & 6.96e-01 & 4.97e+00 & 4.53e+00 \\
\dataset{wiki-Talk}        & 1.28e+00 & 2.40e+00 & 2.51e+00 & 2.40e+00 & 2.40e+00 & 2.54e+00 & 4.52e+00 \\
\dataset{ego-Gplus}        & 4.42e+00 & 1.68e+01 & 2.11e+01 & 1.68e+01 & 1.68e+01 & 2.57e+01 & 4.42e+01 \\
\dataset{uk-2014-tpd}      & 3.59e+00 & 9.66e+00 & 9.92e+00 & 4.35e+00 & 9.66e+00 & 2.10e+01 & 2.16e+01 \\
\dataset{soc-Pokec}        & 1.96e+00 & 1.75e+01 & 3.91e+01 & 1.75e+01 & 1.75e+01 & 2.39e+01 & 2.45e+01 \\
\dataset{uk-2014-host}     & 1.81e+01 & 4.38e+01 & 6.80e+01 & 2.04e+01 & 4.38e+01 & 8.28e+01 & 8.73e+01 \\
\dataset{soc-LiveJournal1} & 2.32e+00 & 2.20e+01 & 1.06e+02 & 2.20e+01 & 2.20e+01 & 4.49e+01 & 6.13e+01 \\
\dataset{enwiki-2013}      & 2.18e+01 & 7.58e+01 & 8.45e+01 & 7.58e+01 & 7.58e+01 & 2.14e+02 & 1.48e+02 \\
\dataset{uk-2002}          & 1.66e+01 & 8.65e+01 & 2.52e+02 & 8.65e+01 & 8.65e+01 & 7.87e+02 & 5.32e+02 \\
\dataset{arabic-2005}      & 1.98e+01 & 1.64e+02 & 4.80e+02 & 3.26e+02 & 1.64e+02 & 1.95e+03 & 1.40e+03 \\
\dataset{twitter-2010}     & 2.23e+02 & 1.23e+03 & 1.95e+03 & 1.23e+03 & 1.23e+03 & 2.22e+03 & 2.18e+03 \\
\dataset{sk-2005}          & 5.73e+01 & 2.94e+02 & 7.98e+02 & 2.94e+02 & 2.94e+02 & 5.83e+03 & 3.81e+03 \\
\bottomrule
\end{tabular}
}
\label{tab:scalability_results}
\end{table}

\clearpage

In theory, the worst-case time for triangle enumeration scales as $m^{1.5}$.
We fit a linear regression of the log of the computation time of the last step of the
enumeration algorithm to the regressor $\log(m)$ and a constant term:
\begin{equation}\label{eqn:regression}
\log(\text{time}) \sim a\log(m) + b
\end{equation}
If the computations truly took $cm^{1.5}$ for some constant $c$, then the
regression coefficient for $\log(m)$ would be $1.5$.  Because of the
pre-processing of the algorithm, the number of edges $m$ depends on the motif.
For example, with motifs $M_1$ and $M_5$, we only count the number of
unidirectional edges.  The pre-processing time, which is linear in the total
number of edges, is not included in the time.  The regression coefficient for
$\log(m)$ (the variable $a$ in \cref{eqn:regression}) is smaller than $1.5$ for
each motif (\cref{tab:scale_conf}).  The largest regression coefficient is
$1.31$ for $M_3$ (with 95\% confidence interval $1.31 \pm 0.19$).  The
regression coefficient over the aggregation of data points (the ``combined''
column in \cref{tab:scale_conf}) is $1.17$ (with 95\% confidence interval $1.17
\pm 0.09$).  We conclude that on real-world datasets, the algorithm for
computing $W_M$ performs much better than the worst-case guarantees.



\begin{table}[h]
\setlength{\tabcolsep}{3pt}
\centering
\dualcaption{Linear models for computation time}{
The 95\% confidence interval for the regression coefficient of the regressor
$\log(m)$ in a linear model for predicting the time to compute $W_M$, based on
the computational results for the 16 networks described at the beginning of
\cref{sec:scalability}.  The algorithm scales as $m^{1.5}$ for the worst-case
input.  ``Combined'' refers to the regression coefficient for
the union of data points of all motifs.
}
\scalebox{0.86}{
\begin{tabular}{c c c c c c c c c}
  \toprule
$M_{1}$ & $M_{2}$ & $M_{3}$ & $M_{4}$ & $M_{5}$ & $M_{6}$ & $M_{7}$ & combined \\ \midrule
$1.20 \pm 0.19$ &  $1.30 \pm 0.20$ & $1.31 \pm 0.19$ & $0.90 \pm 0.31$ & $1.20 \pm 0.20$ & $1.27 \pm 1.21$ & $1.27 \pm 0.21$ &  $1.17 \pm 0.09$ \\
  \bottomrule
\end{tabular}
  }
\label{tab:scale_conf}
\end{table}

\clearpage

\subsection{Larger $k$-clique motifs}

On smaller graphs, we can compute larger motifs.  To demonstrate this point, we
form the motif adjacency matrix $W_M$ based on the $k$-cliques motif for $k = 4,
\dots, 9$ using the algorithm of \citet{chiba1985arboricity} with the additional
pre-processing of computing the ($k-1$)-core of the graph.  (This pre-processing
improves the running time in practice but does not affect the asymptotic
complexity.)  The motif adjacency matrices for $k$-cliques are sparser than the
adjacency matrix of the original graph, so we do not worry about spatial
complexity for these motifs.

We evaluate this procedure on nine real-world networks, ranging from roughly
four thousand nodes and 88 thousand edges to over two million nodes and around
five million edges.  We briefly describe the networks here and list summary
statistics in \cref{tab:cliques_perf}.
\begin{itemize}
\item \dataset{ego-Facebook} is the union of ego networks from the user
  friendship graph of the online social network
  Facebook~\cite{leskovec2012learning}.

\item \dataset{ca-AstroPh} represents scientists who have co-authored a paper
  listed on the AstroPhysics category on
  arXiv~\cite{leskovec2005graphs}.

\item \dataset{soc-Slashdot0811} represents who tagged whom during an event on
  the online social network Slashdot~\cite{leskovec2009community}.  The
  original network data is directed and signed, and we ignore both of these
  properties here.

\item \dataset{com-DBLP} represents scientists who have co-authored a paper
  listed on DBLP~\cite{yang2012defining}.

\item \dataset{com-Youtube} consists of friendships on the online social network
  YouTube~\cite{mislove2007measurement}.
\end{itemize}
We also use the \dataset{wiki-RfA}, \dataset{email-EuAll}, \dataset{cit-HepPh},
and \dataset{wiki-Talk} datasets described earlier, although we now consider the
graphs to be undirected.

Each network contains at least one $9$-clique and hence at least one $k$-clique
for $k < 9$.  All computations ran on the same server as for the triangular
motifs and again there was no parallelism.  We terminated computations after two
hours.  For five of the nine networks, it takes less than two hours to compute
$W_M$ for any $k$-clique motif, $k = 4, \ldots, 9$ (\cref{tab:cliques_perf}).
Furthermore, on all of these networks, the computation takes less than two hours
for $k = 4, 5, 6$.  The smallest network (in terms of number of nodes and number
of edges) is \dataset{ego-Facebook}, where it took just under two hours to
comptue $W_M$ for the $6$-clique motif and over two hours for the $7$-clique
motif.  This network has around 80,000 edges.  On the other hand, for
$\youtube$, which contains nearly 3 million edges, we can compute $W_M$ for the
$9$-clique motif in under a minute.  We conclude that it is possible to use our
framework with motifs much larger than the three-node motifs on which we
performed many of our experiments.  However, the number of edges is not a
good predictor of the running time to compute $W_M$.  This makes sense because the
complexity of the algorithm of \citet{chiba1985arboricity} is $O(a^{k-2}m)$,
where $a$ is the arboricity of the graph.  Hence, the dependence on the number
of edges is always linear, and the arboricity drives the running time.

\begin{table}[h]
\centering
\dualcaption{Time to compute $W_M$ for $k$-clique motifs}{
All times are in seconds.  Only computations that finished within two hours are
listed.
}
\begin{tabular}{l c c @{\qquad} c c c c c c}
\toprule
                      & &  & \multicolumn{6}{c}{Number of nodes in clique ($k$)}        \\
                      \cmidrule(r){4-9} 
Network                    & \# nodes & \# edges & 4  & 5   & 6    & 7    & 8    & 9    \\
\midrule
\dataset{ego-Facebook}     & 4.04K & 88.2K & 14 & 317 & 6816 & --   & --   & --   \\
\dataset{wiki-RfA}         & 10.8K & 182K  & 6  & 22  & 63   & 134  & 218  & 286  \\
\dataset{ca-AstroPh}       & 18.8K & 198K  & 5  & 35  & 285  & 2164 & --   & --   \\
\dataset{email-EuAll}      & 265K  & 364K  & 1  & 2   & 4    & 5    & 6    & 6    \\
\dataset{cit-HepPh}        & 34.5K & 421K  & 3  & 6   & 11   & 18   & 30   & 36   \\
\dataset{soc-Slashdot0811} & 77.4K & 469K  & 3  & 12  & 55   & 282  & 1018 & 2836 \\
$\dblp$                    & 317K  & 1.05M & 9  & 129 & 3234 & --   & --   & --   \\
$\youtube$                 & 1.13M & 2.99M & 12 & 17  & 25   & 33   & 35   & 33   \\
\dataset{wiki-Talk}        & 2.39M & 4.66M & 64 & 466 & 2898 & --   & --   & --   \\
\bottomrule
\end{tabular}
\label{tab:cliques_perf}
\end{table}

\clearpage

\section{An extension to local higher-order clustering}
\label{sec:local}

\subsection{Overview}

Thus far, our algorithms have focused on finding a \emph{global clustering} of
the network---we assign every node in the network to a cluster.  In contrast,
\emph{local} graph clustering methods aim to find a cluster of nodes by exploring a
small region of the graph.  More specifically, the idea is to identify a single
cluster nearby a seed set of nodes without ever exploring the entire graph, which makes
the local clustering methods much faster than their global counterparts.
Because of its speed and scalability, this approach is frequently used in
applications including ranking and community detection on the
Web~\cite{epasto2014reduce,gargi2011large}, social
networks~\cite{jeub2015think}, and bioinformatics~\cite{jiang2009gene}.
Furthermore, the seed-based targeting is also critical to many applications.  In
analysis of protein-protein interaction networks, for instance, local clustering
aids in determining additional members of a protein
complex~\cite{voevodski2009spectral}.  However, current local graph partitioning
methods are also not designed to account for higher-order structures.

A major advantage of the theory we developed for global higher-order clustering
is that it can immediately be adapted for other objectives.  Specifically, we
can re-interpret objective functions based on cuts and volumes of sets (e.g.,
conductance) to motif cuts and motif volumes of sets (e.g., motif conductance).
And in fact, much of the theory and algorithms for local clustering are most
well developed when using conductance as the cluster quality
measure~\cite{andersen2006local,zhu2013local}.  In this section, we adapt one of
these methods---the approximate personalized PageRank (APPR)
algorithm~\cite{andersen2006local}---to find local clusters containing a seed
node with minimal \emph{motif conductance}.  The theory of the approximate
personalized PageRank method is based on the graph Laplacian and conductance, so
we get new theory and algorithms for motif-based local clusters ``for
free''.\footnote{We need to be a little careful with the analysis, but the core
ideas come ``for free''.}  In particular, we get guarantees of fast
running time (independent of the size of the graph, assuming we have
pre-computed the motif adjacency matrix $W$) and cluster quality (in terms of
motif conductance).  For community detection tasks on both synthetic and
real-world networks, our new framework outperforms the current edge-based
personalized PageRank methodology.  From a data mining perspective, the main
advantage of local higher-order clustering is to provide new types of heretofore
unexplored local information based on higher-order structures.

As stated above, our approach to local higher-order clustering is to generalize
the APPR of \citet{andersen2006local} to finding sets of provably small motif
conductance (\cref{Thm:ApproxGuarantee} is the main result).  The APPR method is
a graph diffusion that ``spreads mass'' from a seed set to identify the
cluster. It has an extremely fast running time, which is roughly proportional to
the size of the output cluster.  Our generalization just uses a pre-processing
step that transforms the original (possibly directed) network into the
motif-weighted graph, where the weight on edge $(i, j)$ is the number of times
that nodes $i$ and $j$ co-participate in an instance of the motif.  We show that
running APPR on this weighted network maintains the provably fast running time
and has theoretical guarantees on cluster output quality in terms of motif
conductance.  An additional benefit of our motif-based APPR method is that it
naturally handles directed graphs on which graph clustering has been a
longstanding challenge.  The original APPR method can only be used for
undirected graphs, and existing local approaches for APPR on directed graphs are
challenging to interpret~\cite{andersen2008local}.

We use our motif-based APPR method on a number of community detection tasks and
show improvements over the corresponding edge-based methods.  We show that using
the triangle motif improves the detection of ground truth communities in
synthetic networks.  In addition, we identify important directed triangle motifs
for recovering community structure in directed graphs.  We note that there are
several methods for finding clusters containing a seed node beyond the
personalized PageRank method considered here, including other graph
diffusions~\cite{chung2013solving,kloster2014heat}, local spectral
methods~\cite{li2015uncovering,mahoney2012local}, local modularity
maximization~\cite{clauset2005finding}, and flow-based
algorithms~\cite{orecchia2014flow}.  We focus on generalizing the personalized
PageRank method because of the algorithm's simplicity and scalability.  Finally,
the local method of \citet{rohe2013blessing} is specifically designed around
growing a seed set to avoid cutting the triangle motif, which is most similar to
our ideas, and their analysis is developed from a stochastic blockmodel point of
view.

\subsection{Motif-based personalized PageRank}

We now generalize the classical personalized PageRank method to account for
motifs.  The essential idea of our approach is to transform the input graph,
which is unweighted and possibly directed, into a weighted undirected
graph---the motif adjacency matrix defined in \cref{sec:motif_adjacency}.  We
then show that the fast approximate personalized PageRank method on this
weighted graph will efficiently find a set with small motif conductance that
contains a given seed node.  We also explain how previous theoretical results
are applicable to this approach, which gives us formal guarantees on running
time and cluster output quality in terms of motif conductance.

\xhdr{Background on approximate personalized PageRank (APPR)}
The personalized PageRank ($\ppr$) vector represents the stationary distribution
of a particular random walk.  At each step of the random walk, with a
teleportation parameter $\alpha \in (0, 1)$, the random walker will ``teleport''
back to a specified seed node $u$; and with probability $1 - \alpha$, the walker
will transition uniformly at random to an adjacent node.  The key idea is that
the stationary distribution of this process for a seed node $u$ (the $\ppr$ vector
$p_u$) will have large values for nodes ``close'' to $u$.  We can write the
stationary distribution as the solution to the following system of equations
\begin{equation}\label{eqn:ppr}
(I - \alpha P)p_u = (1 - \alpha)e_u,
\end{equation}
where $P$ is the column-stochastic transition matrix representing the standard
random walk over the graph and $e_u$ is the indicator vector for node $u$.
Formally,
$P = AD^{-1}$,
where $A$ is the adjacency matrix,
$D = \text{diag}(Ae)$
is the diagonal degree matrix, and $e$ is the vector of all ones.

\Citet{andersen2006local} developed a fast algorithm for approximating $p_u$ by a vector
$\tilde{p}_u$ where $0 \le D^{-1}(p_u - \tilde{p}_u) \le \varepsilon$
component-wise.  To obtain a cluster with small conductance from this
approximation, we again use the sweep procedure:
\begin{enumerate}
\item sort the nodes by descending value in the vector $D^{-1}\tilde{p}_u$,
\item compute the conductance of each prefix in this sorted list, and
\item output the prefix set with smallest conductance.
\end{enumerate}
This is the same sweep procedure as in \cref{alg:motif_fiedler}, but the node values
are given by $D^{-1}\tilde{p}_u$ instead of the Fiedler vector.  Overall, this
algorithm is fast (it runs in time proportional to the size of the output
cluster) and is guaranteed to have small conductance provided that node $u$ is
in a set with small conductance.  We will be more specific with the guarantees
in the following section, when we derive the analogous theory for the
motif-based approach.

\xhdr{Motif-weighted APPR}
We now propose an algorithm that finds a cluster with small motif conductance by
finding an approximate $\ppr$ vector on a weighted graph based on the motif.
The algorithm has three steps:
\begin{enumerate}
\item construct the motif adjacency matrix (\cref{eqn:informal_weighting}),
where $W_{ij}$ is the number of instances of $M$ containing nodes $i$ and $j$,
\item compute the approximate \texttt{PPR} vector for this weighted graph,
\item use the sweep procedure to output the set with minimal conductance.
\end{enumerate}
\Cref{alg:Nibble} formally describes this method. Note
that step (i) needs to be done only once, whereas steps (ii) and (iii) would be
repeated for multiple subsequent runs.
This approach is directly motivated by \cref{thm:motif_cond}, which says
that edge-based conductance of sets in the weighted network is equal to
motif-based conductance in the original network.

\begin{algorithm}[tb]\algoptions
    \KwIn{Unweighted graph $G = (V, E)$, motif $M$, seed node $u$, teleportation parameter $\alpha$, tolerance $\varepsilon$}
    \KwOut{Motif-based cluster (set $S \subset V$)}
    $W_{ij} \gets$ \#(instances of $M$ containing nodes $i$ and $j$)\; \label{line:motif_adj}
    $\tilde{p} \gets \text{Approximate-Weighted-PPR}(W, u, \alpha, \varepsilon)$ \quad (\cref{alg:APPR})\;
    $D_W \gets \text{diag}(We)$\;
    $\sigma_i \gets$ $i$th smallest entry of $D_W^{-1}\tilde{p}$\;
    \Return $S \gets  \arg \min_\ell \mmcond{S_\ell}$, where $S_\ell = \{\sigma_1, \dots, \sigma_\ell\}$\;
    \dualcaption{Motif-PageRank-Nibble}{This algorithm finds localized clusters
      with small motif conductance}
\label{alg:Nibble}
\end{algorithm}

\begin{algorithm}[tb]\algoptions
    \KwIn{Undirected edge-weighted graph $G_w = (V_w, E_w, W)$,
    seed node $u$, 
    teleportation parameter $\alpha$, 
    tolerance $\varepsilon$}
    \KwOut{an $\varepsilon$-approximate weighted \ppr vector $\tilde{p}$}
    $\tilde{p} \gets 0$\;
    $r \gets e_u$\;
    $ d_w \gets We$\;
    \While{$r(v) / d_w(v) \geq \varepsilon$ for some node $v \in V_w$} {
      \mycomment{push operation} \; \label{line:push}
      $\tilde{p}(v) \gets \tilde{p}(v) + (1-\alpha) r(u)$ \;
      $r(v) \gets \frac{\alpha}{2}  r(v)$\;
      \For{each $x$ such that $(v, x) \in E_w$} {$r(x) \gets r(x) +  \frac{W_{v,w}}{d_w(v)} \cdot \frac{\alpha}{2}  r(v)$}
      }
    \Return $\tilde{p}$\;
\dualcaption{Approximate-Weighted-PPR}{}
\label{alg:APPR}
\end{algorithm}

The APPR method is designed for \emph{unweighted} graphs, whereas we want to use
the method for weighted graphs.  Mathematically, this corresponds to replacing
the column stochastic matrix $P$ in the linear system with the column stochastic
matrix $P_W = WD_W^{-1}$, where $D_W = \text{diag}(We)$ is the diagonal weighted
degree matrix.  For the purposes of implementation, this modification is simple.
We just need to change the algorithm's \texttt{push} method (\cref{line:push} in \cref{alg:APPR})
to push residual weights to neighbors proportional to edge weights (instead of evenly).  We state
the procedure in \cref{alg:APPR}.

\xhdr{Theory}
For the purposes of theoretical analysis with motifs, it is important that our
edge weights are integers so that we can interpret an edge with weight $k$ as
$k$ parallel edges.  Since all of the analysis of APPR permits parallel edges in
the graph, we can combine previous results for theoretical guarantees on
\cref{alg:Nibble}.  The following theorem says that
our algorithm runs in time proportional to the size of the output set. 
\begin{theorem} \label{Thm:RunningTime}
\Cref{alg:Nibble}, after line 1, runs in $O(\frac{1}{\varepsilon(1-\alpha)})$
time, and the number of nodes with non-zero values in the output
approximated $\ppr$ vector is at most $\frac{1}{\varepsilon (1-\alpha)}$.
\end{theorem}
\begin{proof}
	This follows from \citet[Lemma 2]{andersen2006local},
	where we translate $G_w$ into a unweighted graph with parallel edges. 
\end{proof}
Although APPR with weighted edges has been used
before~\cite{andersen2008local,gleich2015using}, there was never a runtime
bound.  This result is the first (albeit straightforward) theoretical bound on the
runtime of APPR with weighted edges when they arise from integers.

Our next result is a theoretical guarantee on the quality of the output
of \cref{alg:Nibble} in terms of motif conductance.  The proof of the result
follows from combining \cref{thm:motif_cond} and the analysis of 
\citet{zhu2013local} (an improvement over the analysis of \citet{andersen2006local}).
The result says that if there is some set $T$ with small motif conductance, then
there are several nodes in $T$ for which
\cref{alg:Nibble} outputs a set with small motif conductance.  For
notation, let $\eta$ be the inverse mixing time of the random walk on the
subgraph induced by $T$.

\begin{theorem}   \label{Thm:ApproxGuarantee}
Suppose $T \subset V$ is some unknown targeted community we are trying to
retrieve from an unweighted graph using motif $M$.  Then for most seeds $u \in
T$, \cref{alg:Nibble} with
$1-\alpha = \Theta(\eta)$
and
$\varepsilon \in [\frac{1}{10 \mmvol{T}}, \frac{1}{5 \mmvol{T}}]$
runs in time
$O(\mmvol{T} / \mmcond{T})$
and outputs a set $S$ with
\[
\mmcond{S} \leq \tilde O\left( \min\left\{ \sqrt{\mmcond{T}}, \mmcond{T} / \sqrt{\eta} \right\} \right).
\]
\end{theorem}

For the computational complexity, the final piece we need to consider is the
complexity of forming the motif adjacency matrix (\cref{line:motif_adj} of \cref{alg:Nibble}),
which we have analyzed in \cref{sec:honc_basic_complexity}.  For the experiments
in this section, we will only use triangular motifs, so the 
$O(am)$ bound from \citet{chiba1985arboricity} suffices, where $a$ is the arboricity of the
graph and $m$ is the number of edges.

\xhdr{Towards purely local methods}
The graph weighting procedure can also be done locally by having
the \texttt{push} procedure compute $W_{vx}$ ``on the fly'' for all nodes $x$
adjacent to node $v$.  While this suffices for \cref{alg:APPR},
the \texttt{Nibble} method (\cref{alg:Nibble}) needs to know the total volume of the weighted graph
to compute the motif conductance scores.  To address this, one might use recent
techniques for quickly estimating the total $\ell$-clique volume on large
graphs~\cite{jain2017fast}.  We leave these optimizations for future work.

\xhdr{Practical considerations}
The formal theory underlying the methods (\cref{Thm:ApproxGuarantee}) requires
multiple apriori unknowable parameters including the inverse mixing time $\eta$
of the target community and the volume of the output. As practical guidance, we
suggest using $\alpha = 0.99$, computing the \texttt{PPR} vector for
$\varepsilon = 10^{-2}/\bar d_M$, $10^{-3}/\bar d_M$, $10^{-4}/\bar d_M$,
where $\bar d_M =\frac{1}{n} \mmvol{G}$ is the average motif-degree of all
nodes, and outputting the set of best motif conductance.  The reason for scaling
by $\bar d_M$ is as follows.  \Cref{Thm:RunningTime} bounds the running time by
volume as if accessing an edge $(i, j)$ with weight $W_{ij}$ takes
$\Theta(W_{ij})$ time (i.e., as if the edges are parallel).  However, we merely
need to access the value $W_{ij}$, which takes $O(1)$ time.  Scaling
$\varepsilon$ by $\bar d_M$ accounts for the average effect of parallel edges
present due to the weights of the motifs and permits the algorithm to do more
computation with roughly the same running time guarantees.

Rather than using the global minimum in the sweep procedure in the last step of
the Nibble method, we apply the common heuristic of finding the first local
minimum~\cite{yang2012defining}.  The first local minimum is the smallest set
where the PageRank vector suggests a border between the seed and the rest of the
graph. It also better models the small size scale of most ground truth
communities that we encounter in our experiments.

\subsection{Experiments on synthetic networks}
\label{sec:local_synth}

On a network with ground truth communities or clusters, our evaluation procedure
of both edge-based and motif-based APPR method is the following. For each ground
truth community, we use every node as a seed in the APPR method to obtain a set
and pick the set with the highest $F_1$ score for recovering the ground truth.
Then we take the average of the $F_1$ scores over all detected communities for
the detection accuracy of the method.  This measurement captures how well the
communities can possibly be recovered (i.e., given the best seed) and has
previously been used to compare seeded clustering
algorithms~\cite{kloster2014heat}.

We first evaluate our motif-based APPR method for recovering ground truth in two
common synthetic random graph models---the planted partition model and the LFR
model.  In both cases, we find that using the triangle motif increases the range
of parameters in which APPR is able to recover the ground truth communities.

\xhdr{Planted partition model}
The planted partition model generates random, undirected, unweighted graphs with
$kn_1$ nodes.  Nodes are partitioned into $k$ built-in communities, each of size
$n_1$. Between any pair of nodes from the same community, an edge exists with
probability $p$ and between any pair of nodes from different communities, an
edge exists with probability $q$.  Each edge exists independently of all other
edges.

In our experiment, we examine the behavior of the edge-based and motif-based
APPR methods by fixing parameters $n_1 = 50$, $k = 10$, $p = 0.5$, and takings
different values of $q$ such that the community mixing level
$\mu = [(k-1)q] / [p + (k-1)q]$,
which measures the fraction of expected neighbors of a node that cross cluster
boundary, varies from 0.1 to 0.9. For each value of $\mu$, we computed the
average of the ``mean best'' $F_1$ score described above over 20 random instances
of the graph.  For the motif-based APPR method, we used the triangle motif.  We
are motivated in part by recent theoretical results
of \citet{tsourakakis2017scalable} showing that with high probability, the
triangle conductance of a true community in the planted partition model is smaller than
the edge conductance.  Here we take an empirical approach and study recovery
instead of conductance.

\Cref{fig:PPM_F1} illustrates the results.  The triangle-based APPR
significantly outperforms the edge-based APPR method when $\mu \in [0.4, 0.6]$.
In this regime, for any given node, the expected number of intra-community edges
and inter-community edges is roughly the same.  Thus, the edge-based method
degrades in performance.  However, the number of intra-community triangles
remains greater than the number of inter-community triangles, so the
triangle-based method is able to recover the planted partition.

\begin{figure}[tb]
  \centering
  \includegraphics[width=0.7\columnwidth]{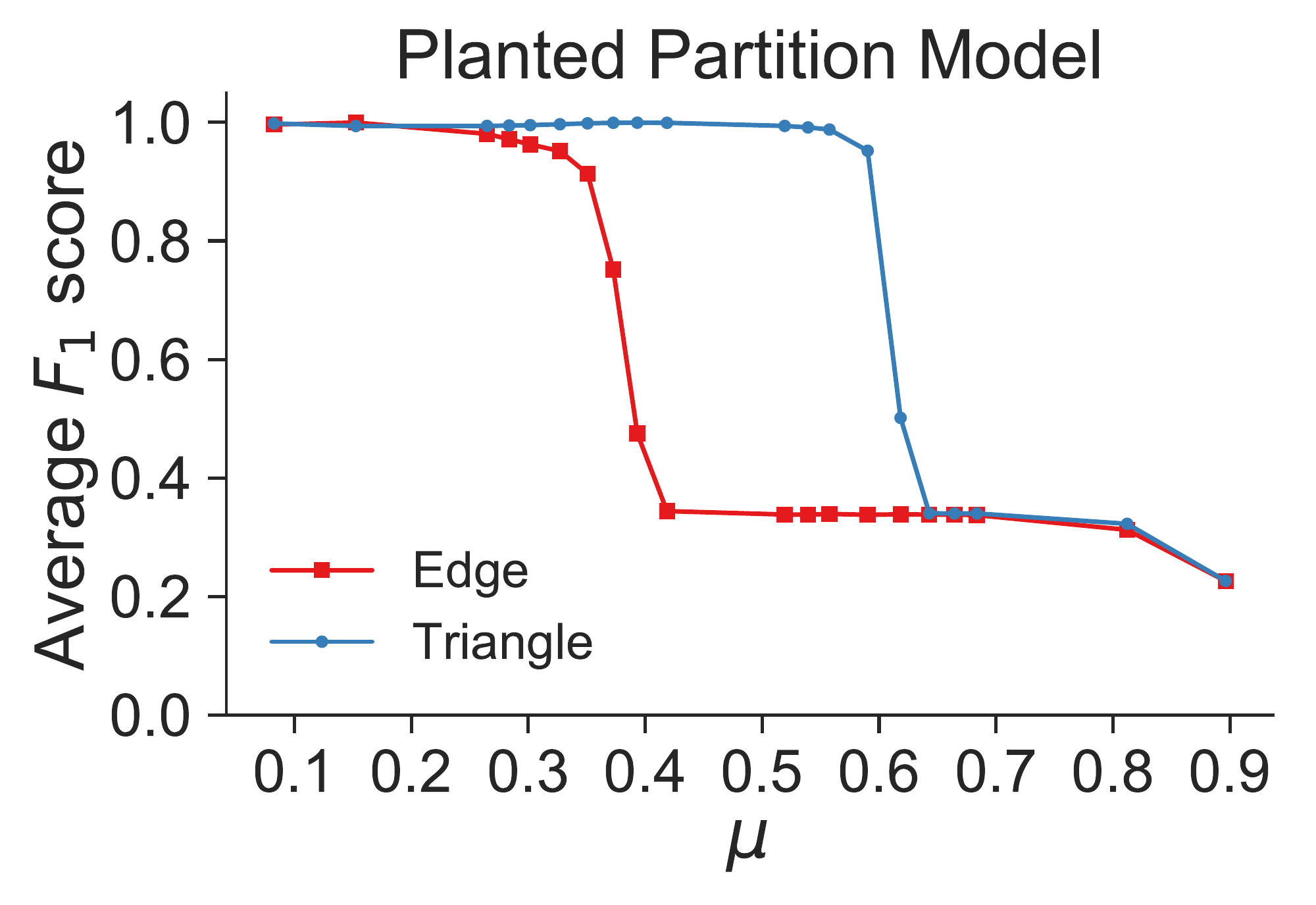}
  \dualcaption{Recovery in the planted partition model}{The plot shows the average
  $F_1$ score of detected clusters in the model as a function of the mixing
  parameter $\mu$ that specifies the fraction of neighbors of a node that cross
  cluster boundaries ($\mu$ varies as the intra-cluster edge probability $p$
  is fixed).  We use edge-based and triangle-based approximate
  personalized PageRank to recover ground truth clusters.  There is a large
  parameter regime where the triangle-based approach significantly out-performs
  the edge-based approach.}
  \label{fig:PPM_F1}
\end{figure}

\clearpage

\xhdr{LFR model}
The LFR model also generates random graphs with planted communities, but the
model is designed to capture several properties of real-world networks with
community structure such as skew in the degree and community size distributions
and overlap in community membership for
nodes~\cite{lancichinetti2008benchmark,lancichinetti2009benchmarks}.  For our purposes, the most important
model parameter is the mixing parameter $\mu$, which is the fraction of a node's
edges that connect to a node in another community.  We fix the other parameters
as follows: $n = 1000$ is the number of nodes, where 500 nodes belong to 1
community and 500 belong to 2; the number of communities is randomly chosen
between $43$ and $50$; the average degree is $20$; and the community sizes range
from 20 to 50.

We again used the edge-based and triangle-based APPR methods, and
\cref{fig:LFR_F1} shows the results.  The
performance of the edge-based method decays as we increase the mixing
parameter $\mu$ from $0.1$ to $0.4$, while the triangle-based method maintains
an $F_1$ score of approximately $0.9$ in this regime.  For mixing parameters as
large as $0.6$, the $F_1$ score for the triangle-based method is still three times
larger than that of the edge-based method, and throughout nearly the entire
parameter space, using triangles improves performance.

To summarize, incorporating triangles into personalized PageRank dramatically
improves the recovery of ground truth community structure in synthetic models.
In the next section, we run experiments on both undirected and directed
real-world networks.


\begin{figure}[tb]
  \centering \includegraphics[width=0.7\columnwidth]{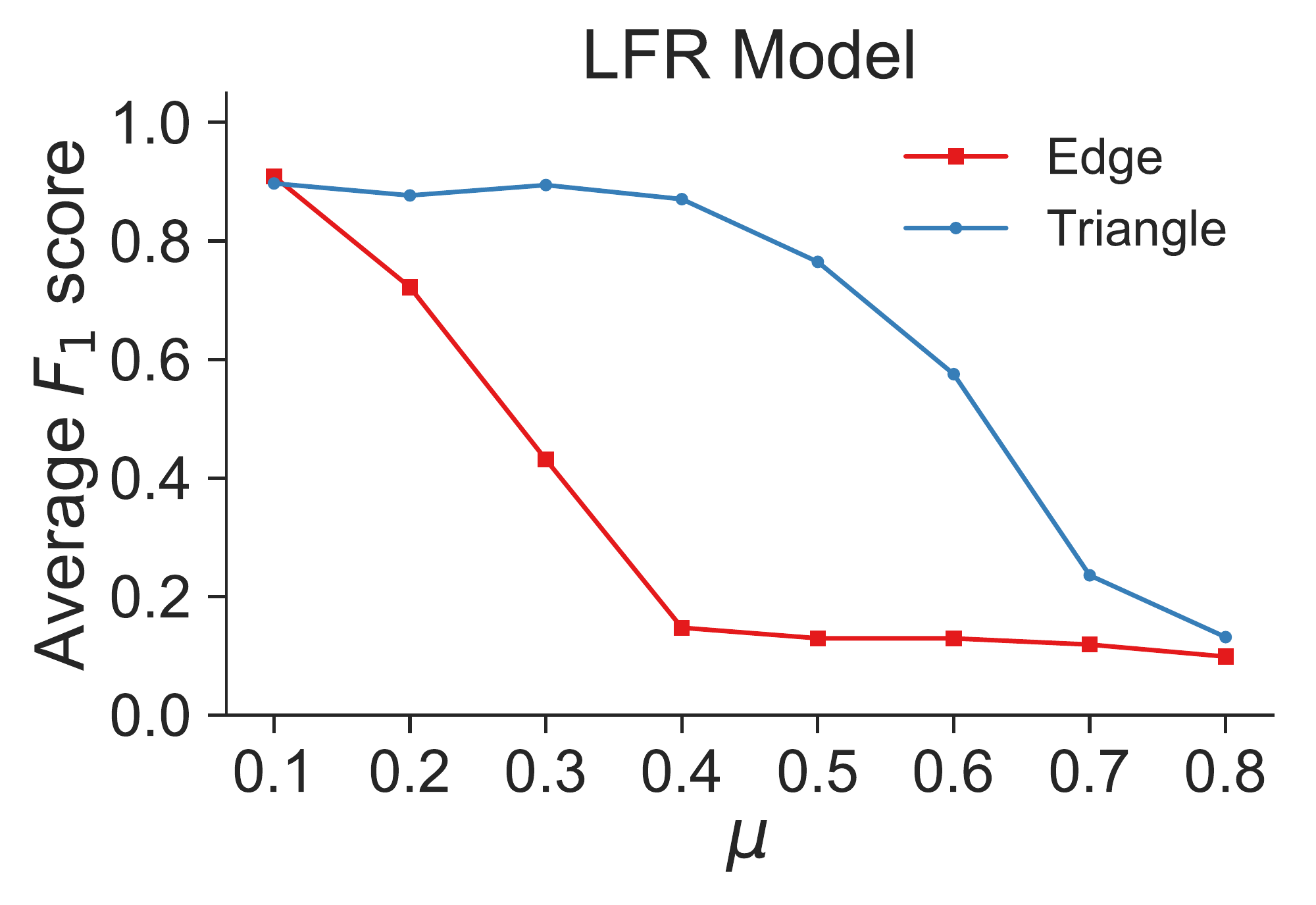}
  \dualcaption{Recovery in the LFR model}{The plot shows the average $F_1$ score
    of detected clusters in the model as a function of the mixing parameter
    $\mu$ that specifies the fraction of neighbors of a node that cross cluster
    boundaries.  We use edge-based and triangle-based approximate personalized
    PageRank to recover ground truth clusters.  As was the case for the planted
    partition model (\cref{fig:PPM_F1}), there is a large parameter regime where
    the triangle-based approach significantly out-performs the edge-based
    approach.}
   \label{fig:LFR_F1}
\end{figure}

\clearpage

\subsection{Experiments on real-world networks}
\label{sec:local_real}

We now compare the edge- and motif-based APPR methods on real-world networks
with ground truth communities.  Although these graphs have as many as 1.8
billion edges, the APPR method takes at most a few seconds per seed once the
graph is in memory and the motif weighted adjacency matrix has been computed.

\xhdr{Undirected graphs}
We analyze several well-known networks with ground truth community structure
constructed from Web data:
\begin{itemize}
\item $\amazon$ represents frequent co-purchasing of items on
the online retailer Amazon.  The communities are the connected components of the subgraph induced
by nodes in a product category~\cite{yang2012defining}.
\item $\dblp$ represents co-authorship on DBLP.  The communities are 
the connected components of the subgraph induced by individuals who have published at a particular
conference or in a particular journal~\cite{yang2012defining}.
\item $\youtube$ represents friendships between users on the online social
network YouTube.  The communities are connected components
of user-defined groups~\cite{mislove2007measurement}.
\item $\lj$ represents friendships between users on the online social
network LiveJournal.  The communities are connected components
of user-defined groups~\cite{mislove2007measurement}.
\item $\orkut$ represents friendships between users on the online social
network Orkut.  The communities are connected components
of user-defined groups~\cite{mislove2007measurement}.
\item $\friendster$ represents friendships between users on the online social
network Friendster.  The communities are connected components
of user-defined groups~\cite{yang2012defining}.
\end{itemize}

For each network, we examine 100 communities whose sizes ranged between 10 and
200 nodes.  We use both edge-based and motif-based APPR for the triangle motif
to recover the known communities.  Summary statistics of the datasets and our
experiment are in \cref{tab:rw_undir_recovery}.  In 5 out of 6 networks,
motif-based APPR achieves a higher $F_1$ score than edge-based APPR.  In 3 of
the 5 networks, the $F_1$ score provides a relative improvement of over 5\%.  In
all 6 networks, the average precision of the recovered clusters is larger, and
in 4 of these networks, the change is greater than 5\%.  We suspect this arises
from triangles encouraging more tight-knit clusters.  For example, dangling
nodes connected by one edge to a cluster are ignored by the triangle-based
method, whereas such a node would increase the edge-based conductance of the
set.  In 4 of the 6 networks, recall in the triangle-based method provides
relative improvements of at least 5\%.

\begin{table}[tb]
\setlength{\tabcolsep}{2pt}
  \centering
  \dualcaption{Recovery of ground truth community structure in undirected
    graphs}{We use edge-based and motif-based APPR for the triangle motif.  Bold
    numbers denote better recovery or smaller conductance with 5+\% relative
    difference.  $F_1$ score, precision, and recall are all averages over the
    100 ground truth communities.}
  \label{tab:rw_undir_recovery}
  \scalebox{0.9}{
  \begin{tabular}{l c c c c c c}
    \toprule
                         & {\footnotesize $\amazon$} & {\footnotesize $\dblp$} & {\footnotesize $\youtube$} & {\footnotesize $\lj$} & {\footnotesize $\orkut$} & {\footnotesize $\friendster$} \\ \midrule
    \# nodes             & 335K                      & 317K                    & 1.13M                      & 4.00M                 & 3.07M                    & 65.6M                         \\
    \# edges             & 926K                      & 1.05M                   & 2.99M                      & 34.7M                 & 117M                     & 1.81B                         \\
    \# comms.            & 100                       & 100                     & 100                        & 100                   & 100                      & 100                           \\
    comm. sizes          & 10--178                   & 10--36                  & 10--200                    & 10--10                & 10--200                  & 10--191                       \\ \midrule
    $F_1$ score                                                                                                                                                                                \\
    \phantom{X} edge     & \textbf{0.620}            & 0.264                   & 0.140                      & 0.255                 & 0.063                    & 0.095                         \\
    \phantom{X} triangle & 0.556                     & 0.269                   & \textbf{0.165}             & \textbf{0.274}        & \textbf{0.078}           & \textbf{0.114}                \\ \midrule
    Precision                                                                                                                                                                                  \\
    \phantom{X} edge     & 0.634                     & 0.342                   & 0.233                      & 0.216                 & 0.072                    & 0.103                         \\ 
    \phantom{X} triangle & 0.660                     & \textbf{0.366}          & \textbf{0.390}             & \textbf{0.280}        & \textbf{0.117}           & \textbf{0.158}                \\ \midrule
    Recall                                                                                                                                                                                     \\
    \phantom{X} edge     & \textbf{0.704}            & 0.310                   & 0.147                      & 0.606                 & \textbf{0.212}                    & 0.204                         \\
    \phantom{X} triangle & 0.567                     & \textbf{0.329}          & \textbf{0.188}             & \textbf{0.672}        & 0.166                    & \textbf{0.234}                \\ \midrule
    Conductance                                                                                                                                                                                \\
    \phantom{X} edge     & 0.163                     & 0.393                   & \textbf{0.536}             & 0.498                 & 0.702                    & 0.747                         \\
    \phantom{X} triangle & \textbf{0.065}            & 0.384                   & 0.739                      & \textbf{0.409}        & \textbf{0.510}           & \textbf{0.622}                \\
    \bottomrule
\end{tabular}
}
\end{table}

\clearpage

\xhdr{Directed graphs}
A major advantage of our motif-based APPR method is that it is straightforward
to analyze directed graphs---we simply need to specify the directed motifs for
our motif-based APPR algorithm.  Here, we use the three different directed
triangle motifs in \cref{fig:local_motifs}: $\mtri$, the undirected triangle; $\mcyc$,
the cycle; and $\mffl$, the feed-forward loop.  We form our motif weighted matrix
$W$ with respect to subgraphs and \emph{not} induced subgraphs.
Thus, a triangle with all six directed edges contains 1 instance of motif $\mtri$,
2 instances of motif $\mcyc$, and 2 instances of motif $\mffl$.  The equivalent
weighting scheme for simultaneously analyzing several motifs is in \cref{fig:local_motifs}.


\begin{figure}[t]
\centering
\includegraphics[width=\columnwidth]{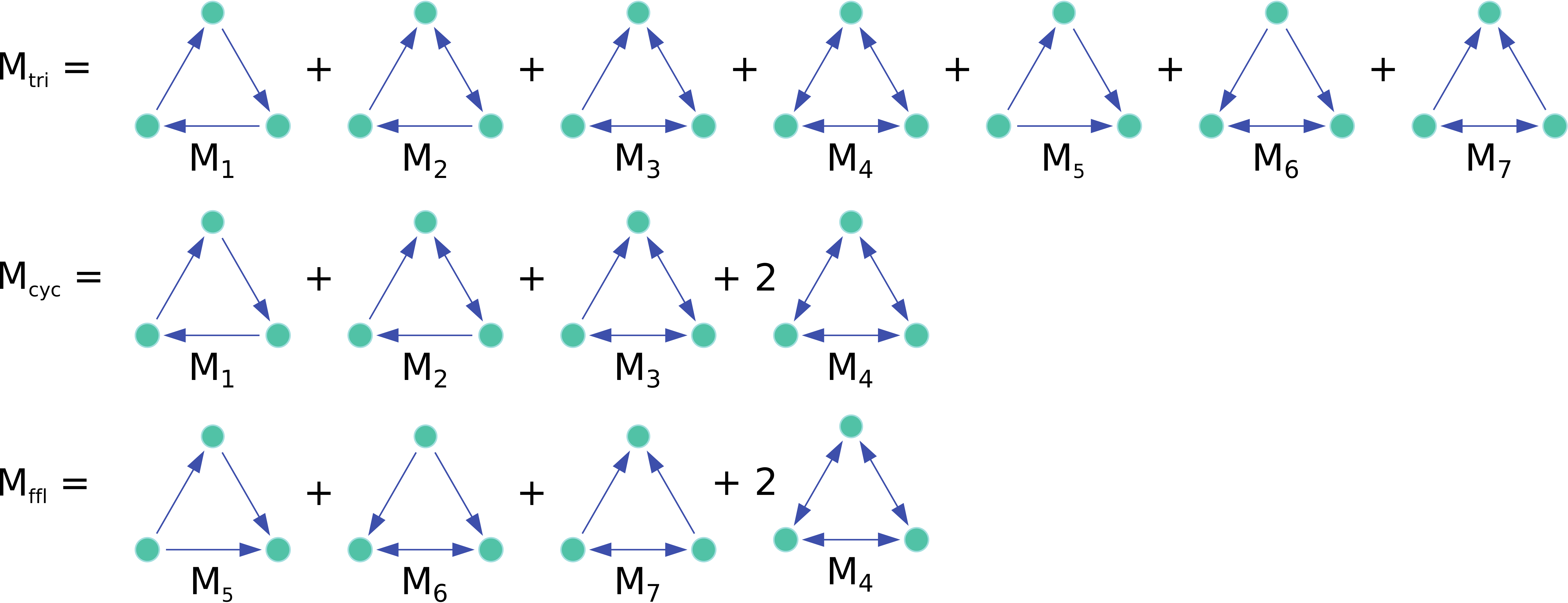}
\dualcaption{Three directed motif groups}{We use these three directed triangle motif
groups in our experiments: a triangle with any edge directions ($\mtri$),
the number of directed 3-cycles between three nodes ($\mcyc$), and the
number of feed-forward loops between three nodes ($\mffl$).}
\label{fig:local_motifs}
\end{figure}

We analyze two directed networks:
\begin{itemize}
\item $\emaileucore$ is an e-mail network between members of a European research institution.
Department membership of researchers determines the ground truth communities.
\item $\wikicat$ is the hyperlink network of English Wikipedia.  The article
categories are the ground truth communities (we only consider 100 categories for our analysis)~\cite{klymko2014using}.
\end{itemize}
The datasets and recovery results are
summarized in \cref{tab:rw_dir_recovery}.  For both networks, using motif $\mtri$
provides an improvement in $F_1$ score over the edge-based method.  The
improvement is drastic in $\emaileucore$ (25\% relative improvement).  In fact, all three
motifs lead to substantial improvements in this network.  We also see that in
both networks, the motifs provide additional precision but sacrifice recall.
These tighter clusters are expected for the same reasons as for the undirected
networks.



\begin{table}[tb]
  \centering
  \dualcaption{Recovery of ground truth communities in directed graphs}{We use
    edge-based and motif-based APPR for the three triangular motifs in
    \cref{fig:local_motifs}.  Bold numbers denote (i) cases where a motif-based
    method's score is a 5+\% relative improvement over the edge-based method and
    (ii) cases where the edge-based method out-performs all 3 motif-based
    methods by 5+\%.}
\label{tab:rw_dir_recovery}
  \scalebox{1.0}{
  \begin{tabular}{l @{\qquad\quad} c @{\qquad} c}
    \toprule
                         & $\emaileucore$ & $\wikicat$     \\ \midrule
    \# nodes             & 1.00K          & 1.79M          \\         
    \# edges             & 25.6K          & 28.5M          \\
    \# comms.            & 28             & 100            \\
    comm. sizes          & 10--109        & 21--192        \\ \midrule
    $F_1$ score                                            \\
    \phantom{XX} edge    & 0.396          & 0.237          \\
    \phantom{XX} $\mtri$ & \textbf{0.496} & 0.245          \\
    \phantom{XX} $\mcyc$ & \textbf{0.447} & 0.233          \\
    \phantom{XX} $\mffl$ & \textbf{0.472} & 0.231          \\ \midrule
    Precision                                              \\ 
    \phantom{XX} edge    & 0.478          & 0.322          \\
    \phantom{XX} $\mtri$ & \textbf{0.584} & \textbf{0.349} \\
    \phantom{XX} $\mcyc$ & \textbf{0.616} & \textbf{0.376} \\
    \phantom{XX} $\mffl$ & \textbf{0.630} & 0.333          \\ \midrule
    Recall                                                 \\ 
    \phantom{XX} edge    & \textbf{0.779} & \textbf{0.380} \\
    \phantom{XX} $\mtri$ & 0.690          & 0.327          \\
    \phantom{XX} $\mcyc$ & 0.577          & 0.227          \\
    \phantom{XX} $\mffl$ & 0.607          & 0.344          \\
    \bottomrule
  \end{tabular}
}
\end{table}

\clearpage


\begin{figure}[tb]
   \centering
   \includegraphics[width=0.75\columnwidth]{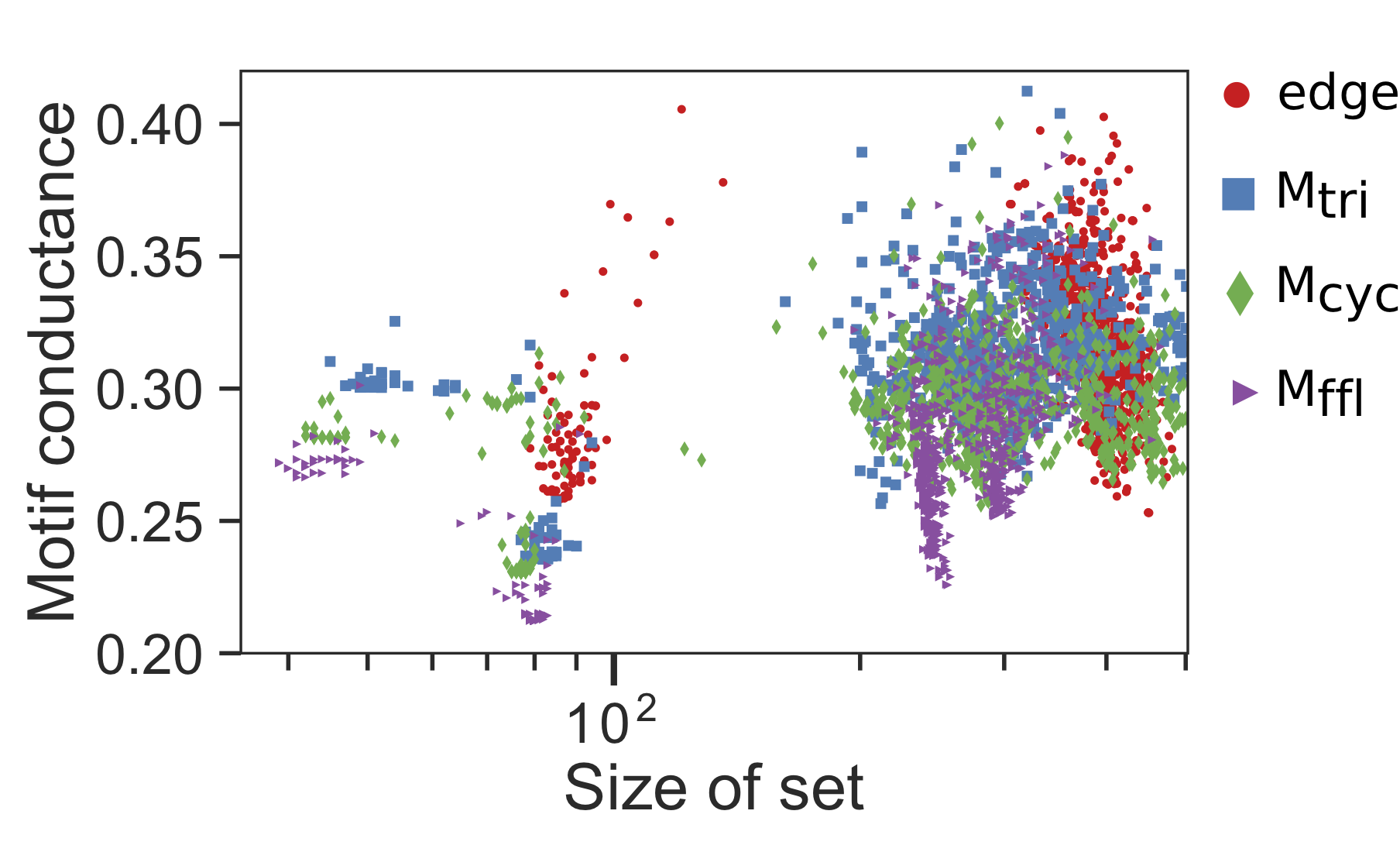}
   \dualcaption{Distribution of set size and conductance for seeded APPR in
     $\emaileucore$}{We use each node in the graph as a seed for the edge-based and
     motif-based APPR methods.  There is a bifurcation into small clusters and
     large clusters.  Small edge-based cluster concentrate in sizes of 70--100,
     a regime in which there are also many motif-based clusters with smaller
     conductance.  Of the larger clusters, edge-based clusters tend to be even
     larger than motif-based clusters.}  \label{fig:email-Eu-core-seeds}
\end{figure}

We investigate the results of the $\emaileucore$ network in more detail, as the use of
motifs dramatically improves the recovery of ground truth clusters with respect
to $F_1$ score.  First, we use every node in the network as a seed for the APPR
method with edges and the three motifs (\cref{fig:email-Eu-core-seeds}).  The
clusters bifurcate into small ($<$ 100 nodes) and large ($>$ 200 nodes) groups.
For the small clusters, the edge-based ones concentrate in sizes of 70--100.  In
this range, there are several clusters with much smaller motif-based conductance
for all three motifs.  This provides evidence that the 3 motifs are better
models for the community structure in the network.  We also see that, of the
large clusters, the edge-based ones tend to be the largest.  Since these sets
are larger than the sizes of the communities in the network, this observation
provides evidence for why precision is better with motif-based APPR.

\begin{figure}[tb]
   \centering
   \includegraphics[width=0.75\columnwidth]{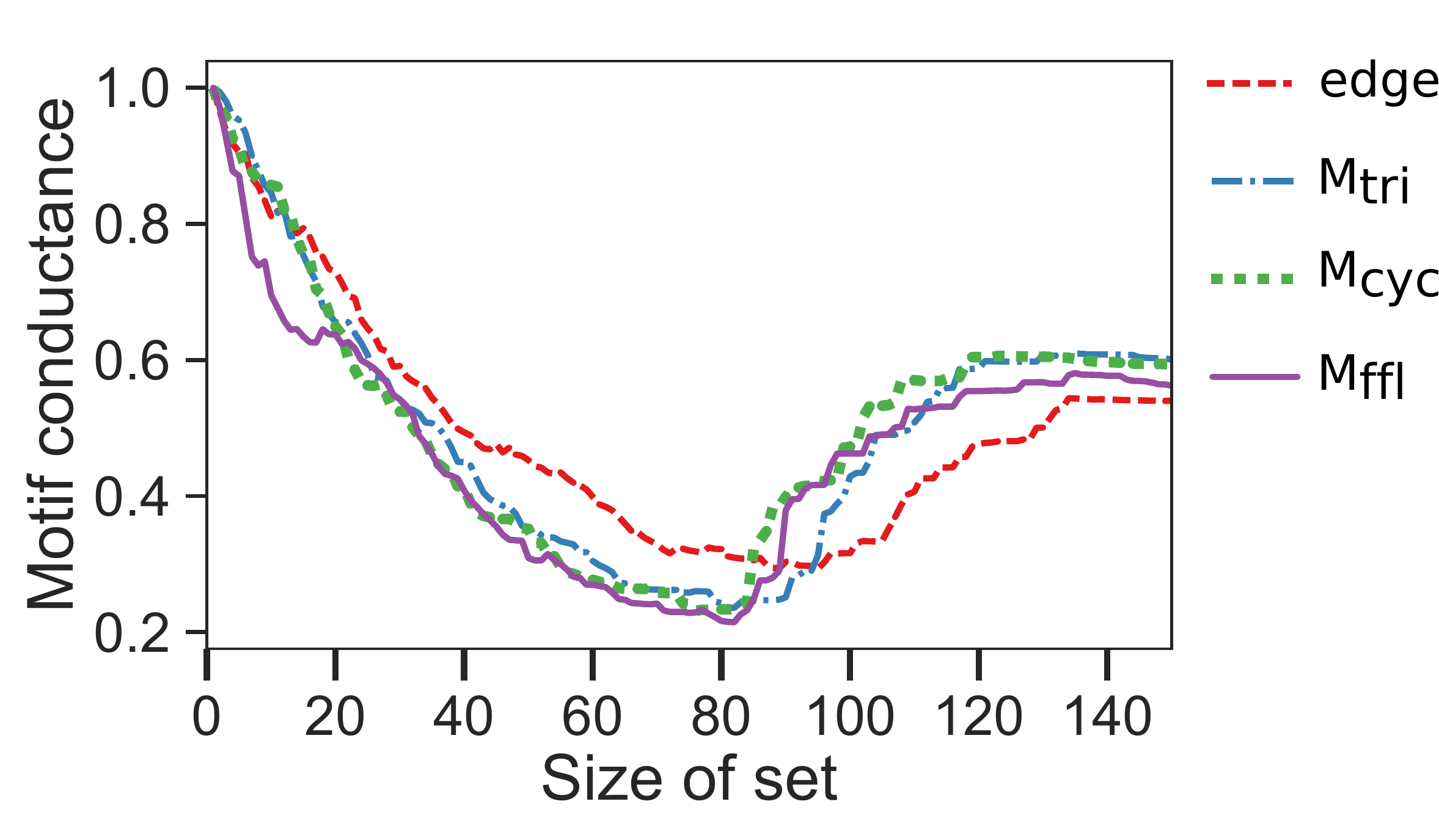}
   \dualcaption{Sweep profile for a single seed in $\emaileucore$}{The shape of the
     curves is similar, but the minimum for the 3 motif-based curves occur for a
     smaller set size and have a smaller conductance compared to the curve for
     edges.}  \label{fig:email-Eu-core-sweep}
\end{figure}

Next, we examine the sweep profile for a single seed node in the $\emaileucore$ network
(\cref{fig:email-Eu-core-sweep}).  The sweep profile highlights key differences
between the output of the motif-based and edge-based algorithm.  Although the
general shape of the sweep profile is the same across the 3 motifs and edges,
the minimum of the curves occurs for a smaller set and at a smaller conductance
value for the motifs.  A plausible explanation is that the edge-based and
motif-based APPR methods are capturing roughly the same set, but the constraint
of triangle participation excludes some nodes.  The smaller motif conductance
values indicate that these motifs are better models for the cluster structure in
the network.

\section{Related work and discussion}
\label{sec:honc_discussion}

The work in this chapter fits into the saga of network clustering, community
detection, and graph partitioning, where the goal is to assign the nodes in a
graph into clusters, communities, or partition components.  The literature in
this space is vast, but there are a few surveys providing an overview of the
topic~\cite{fortunato2010community,malliaros2013clustering,fortunato2016community,schaeffer2007graph}.
In particular, the recent survey by \citet{fortunato2016community}
discusses our motif-based clustering methodology within
a broader context.

We now discuss existing methods for global graph clustering that consider motifs
or higher-order structures.  In my own work, we clustered networks based on
metrics similar to motif conductance using generalizations of spectral methods
to tensor data~\cite{benson2015tensor}.
\Citet{klymko2014using} add weights to the edges of
directed graphs based on the participation of edges in cyclic motifs 3-node
triangle motifs ($M_{4}$, $M_{5}$, $M_{6}$, and $M_{7}$) before plugging the
weighted graph into an existing clustering algorithm.  In this method, the
weight does not depend on the number of instances of the motif---it only matters
whether an edge participates in at least one
instance.  \Citet{gupta2016decompositions} provide an algorithm to decompose
triangle-dense graphs into an ``approximate union of cliques''; more
specifically, most of the nodes can be partitioned into edge- and triangle-dense
subgraphs with radius at most 2.  This approach has been extended to also
provide estimates of clique counts in real-world
networks~\cite{comandur2014counting}.  \Citet{michoel2011enrichment}
and \citet{michoel2012alignment} use a max-cut-like objective to decompose
graphs based on triangular motifs.  More specifically, they aim to find sets of
nodes $X_1$, $X_2$, and $X_3$, such that in many instances of a given triangular
motif, the three nodes in the instance are in the three different sets.
Eigenvectors of the non-backtracking operator, which looks at 3-node length-2
paths as higher-order structures, have been used to get better recovery results
in the stochastic block model~\cite{krzakala2013spectral}.  Finally, there are
studies of the hierarchical composition of complex networks based on
self-similarity and higher-order
structures~\cite{derenyi2005clique,angulo2015network,itzkovitz2005coarse} as
well as generalizations of $k$-core decompositions~\cite{zhang2012extracting,sariyuce2015finding}.

\Citet{arenas2008motif} proposed a ``motif modularity'' that measures how
many more motifs are contained within cluster in a given partition compared to a
random configuration.  A few years later, \citet{serrour2011detecting} showed
that the classical spectral algorithm for modularity
maximization~\cite{newman2006modularity} can be adapted for the case of
triangular motifs and this definition of motif modularity.  The approach simply
uses the same spectral algorithm on the modularity matrix coming from the
weighted motif adjacency matrix we proposed in this chapter.  In fact, this idea
is a simple corollary of our cut and volume equivalence results
(\cref{lem:motif_cut,lem:motif_vol}) and the fact that modularity can be written
in terms of cuts and volumes~\cite{gleich2016mining}.

In addition to the motif modularity ideas
above, \cref{lem:motif_cut,lem:motif_vol} make our theoretical results much more
general than the extension of the Cheeger inequality introduced in this chapter.
Any objective function based on cuts, volumes, and set sizes are immediately
translated, in addition to any algorithms based on conductance.  This was the case for
the localized methods developed in \cref{sec:local}.  We can use the same motif
Laplacian for motif normalized cut or motif ratio cut
objectives~\cite{von2007tutorial} as well as for Laplacian-based methods for
non-exhaustive and overlapping clustering~\cite{whang2015non}.  We can also
generalize more theoretical algorithms for sparsest cut~\cite{arora2009expander}.
However, we focused on the simpler
Cheeger bounds because they are easy to use in practice and actually scale to
large datasets.  The central challenge to these generalizations is that we now
have to deal with weighted graphs and often the theory assumes an unweighted graph.

The work of \citet{serrour2011detecting} is not the only place where our
proposed weighting scheme has shown up.  \Citet{rohe2013blessing} use the same
motif adjacency matrix weights (specifically for the triangle motif) for single
linkage clustering---the resulting dendogram is cut at some threshold level and
the connected components are returned as
clusters.  \Citet{tsourakakis2017scalable} use the same Laplacian (and
independently arrived at some of our same results), but their experiments
analyze undirected graphs.  The motif Laplacian was originally studied
by \citet{rodriguez2002laplacian}, who related eigenvalues of this Laplacian to
bounds on hypergraph generalizations of partition quantities like bisection,
max-cut, and the isoperimetric number.  This same Laplacian was used for
hypergraph-based semi-supervised learning tasks~\cite{zhou2006learning}.  In
this case, hyperedges are constructed from data points that share a feature
(e.g., the nodes are animals and a hyperedge is constructed between all animals
that have a tail).

The general problem of partitioning a graph based on relationships between more
than two nodes has been studied in hypergraph
partitioning~\cite{karypis1999multilevel}.  We can interpret motifs as
hyperedges in a hypergraph.  The key difference of our methods is that \emph{we
are constructing a hypergraph from patterns in a standard graph rather than
receiving the hyperedges a priori}.  The way in which we construct the
hyperedges leads to different information about the underlying dataset.  One
goal with our analysis of the Florida Bay food web, for example, was to find
which hyperedge constructs (induced by different motifs) provide a good
clustering of the network (see \cref{sec:honc_foodweb}).

Our motif-based spectral clustering methodology falls into the area of encoding
a hypergraph partitioning problem by a graph partitioning
problem~\cite{agarwal2006higher}.  The motif Cheeger inequality we proved
(\cref{thm:motif_cheeger}) explains why previous methods (such as clustering
based on the Laplacian proposed by
\citet{rodriguez2002laplacian}) are appropriate for $3$-regular hypergraphs.
Specifically, it respects the standard cut and volume metrics for graph
partitioning.  There are also a couple spectral clustering approaches designed
directly for hypergraphs.  For example, \citet{angelini2015spectral} use a
non-backtracking walk approach for recovering planted clusters in a hypergraph
stochastic block model.  In the theory
community, \citet{louis2014approximation} introduced an \emph{edge expansion}
metric for vertex sets in hypergraphs and a diffusion process whose second
eigenvalue leads to a Cheeger-like bound.  

Lastly, it is worth mentioning that our framework is also a principled approach
for clustering directed graphs, a longstanding problem in network analysis.
Some existing principled generalizations of undirected graph partitioning to
directed graph partitioning proceed from graph
circulations~\cite{chung2005laplacians} or random walks~\cite{boley2011commute}
and are difficult to interpret.  \Citet{yoshida2016nonlinear} uses a nonlinear
Laplacian to derive a Cheeger-like inequality where the ``cut'' measures the
number of outgoing edges from a set (in this case, the cut is asymmetric, so the
measure is actually the minimum over the set $S$ and the complement set
$\bar{S}$.  Our motif-based clustering framework provides a simple,
rigorous framework for directed graph partitioning---the algorithm user simply
needs to specify which motif he or she wants to use.

We can also interpret common heuristics for directed graph clustering in terms
of motifs.  For example, consider the simple approach of clustering the
symmetrized graph $W = A + A^T$, where $A$ is the (directed) adjacency
matrix~\cite{malliaros2013clustering}.  Following \cref{thm:motif_cond},
conductance-minimizing methods for partitioning $W$ are actually trying to
minimize a weighted sum of motif-based conductances for the directed edge motif
and the bidirectional edge motif:
\[
B_1 = \begin{bmatrix} 0 & 1 \\ 0 & 0 \end{bmatrix},\quad
B_2 = \begin{bmatrix} 0 & 1 \\ 1 & 0 \end{bmatrix},
\]
where both motifs are simple ($\anchorset = \{1, 2\}$).  If $W_{B_1}$ and
$W_{B_2}$ are the motif adjacency matrices for motifs $B_1$ and $B_2$, then $A +
A^T = W_M = W_{B_1} + 2W_{B_2}$.  This weighting scheme gives a weight of two to
bidirectional edges in the original graph and a weight of one to unidirectional
edges.  An alternative method for clustering a directed graph is to simply
remove the direction on all edges, treating bidirectional and unidirectional
edges the same.  This is equivalent to using the motif $\medge$
discussed earlier in this chapter.  The resulting
adjacency matrix is equivalent to the motif adjacency matrix for the
bidirectional and unidirectional edges (without any weighting), i.e.,
$W = W_{B_1} + W_{B_2}$.

We conclude this chapter with an executive summary of our contributions.
\begin{enumerate}
\item
We are the first to really demonstrate how clustering based on motifs leads to
new discoveries in network data from several domains.  Moreover, we show how
``first-order'' methods that only consider edge relationships do not make the
same discoveries.  This should encourage practitioners analyzing network data to
think about higher-order structures.

\item
Our methods come with a theoretical guarantees on cluster quality by
generalizing steadfast ideas in spectral graph theory.  The theory explains why
certain hypergraph partitioning methods work for 3-regular hypergraphs.  The
theory also makes the idea of higher-order clustering easily extendible, which
we demonstrate with with the localized methods in \cref{sec:local}.

\item
Our methods provide a simple, easy-to-understand, and principled approach to
clustering directed graphs.
\end{enumerate}

\chapter{\hoccfstitle}
\label{ch:hoccfs}

\section{The clustering coefficient and closure probabilities}

As we have discussed, networks are a fundamental tool for understanding and
modeling complex physical, social, informational, and biological systems.
Although such networks are typically sparse, a recurring trait of networks
throughout all of these domains is the tendency of edges to appear in small
clusters or cliques~\cite{watts1998collective}.  In many cases, such clustering
can be explained by local evolutionary processes.  For example, in social
networks, clusters appear due to the formation of triangles where two
individuals who share a common friend are more likely to become friends
themselves.  This process is known as \emph{triadic
closure}~\cite{rapoport1953spread,granovetter1973strength,simmel1908sociology}.
Similar triadic closures occur in other information networks.  For example, in
citation networks, two references appearing in the same publication are more
likely to be on the same topic and hence more likely to cite each
other~\cite{wu2009modeling}.  And in co-authorship networks, scientists with a
mutual collaborator are more likely to collaborate in the
future~\cite{jin2001structure}.  In other cases, local clustering arises from
highly connected functional units operating within a larger system, e.g.,
metabolic networks are organized by densely connected
modules~\cite{ravasz2003hierarchical}.

The \emph{clustering coefficient} quantifies the extent to which edges of a
network cluster.  The clustering coefficient is defined as the fraction of
length-2 paths, or \emph{wedges}, that are closed with a
triangle~\cite{watts1998collective,barrat2000properties} (\cref{fig:ccf_def},
$\gccf{2}$).  In other words, the clustering coefficient measures the
probability of triadic closure in the network.  However, the clustering
coefficient is inherently restrictive as it measures the closure pattern of just
one simple structure---the triangle.  Higher-order structures such as larger
cliques are crucial to the structure and function of complex
networks~\cite{benson2016higher,yaverouglu2014revealing}.  For example,
4-cliques reveal community structure in word association and protein-protein
interaction networks~\cite{palla2005uncovering} and maximal
cliques of size 5--7 appear more than expected in many real-world networks
compared to a configuration model~\cite{slater2014mid}.  However, the extent of clustering of such
higher-order structures has not been well understood nor quantified.

Here we introduce higher-order clustering coefficients that measure the closure
probability of higher-order network structures (specifically, larger cliques)
and provide a more comprehensive view of how the edges of complex networks
cluster.  Our higher-order clustering coefficients are a natural generalization
of the classical clustering coefficient.

Below, we derive several properties about higher-order clustering coefficients
and analyze them under common random graph models.  We then use higher-order
clustering coefficients to gain new insights into the structure of real-world
networks from several domains.  For example, we show that the
\emph{C. elegans} neural network exhibits clustering in the traditional sense
(i.e., in terms of triadic closure) but not in the higher-order sense.  On the other hand,
co-authorship and social networks exhibit both traditional and higher-order clustering.
Finally, we make a connection between higher-order clustering coefficients and motif
conductance for clique motifs.  More specifically, we show that if the graph has a large higher-order
clustering coefficient, then there is a 1-hop neighborhood of same node $u$ with
small motif conductance for a clique motif.

\section{Definitions of higher-order clustering coefficients}

We first give an alternative interpretation of the clustering coefficient that
will later allow us to generalize it and quantify clustering of higher-order
network structures.  First consider a $2$-clique $K$ in a graph $G$ (that is, a single edge $K$).
Now, ``expand'' the clique $K$ by
considering any edge $e$ adjacent to $K$ (i.e., $e$ and $K$ share exactly one
node).  In our parlance, this expanded subgraph is a \emph{wedge} (length-2 path).
The global clustering coefficient $\gccf{}$ of
$G$~\cite{barrat2000properties,luce1949method} can then be defined as the
fraction of wedges that are \emph{closed}, meaning that the $2$-clique and
adjacent edge induce a $(2 + 1)$-clique, or a triangle (\cref{fig:ccf_def},
$C_2$).  Formally,
\begin{equation}\label{eq:global_ccf}
  \gccf{} = \frac{6 \lvert K_3 \rvert}{\lvert W \rvert},
\end{equation}
where $K_3$ is the set of $3$-cliques (triangles), $W$ is the set of wedges, and
the coefficient $6$ comes from the fact that each $3$-clique closes 6 wedges
(the 6 ordered pairs of edges in the triangle).\footnote{We use an ordered
pair of edges in order to make the constants in our formulas consistent with
the higher-order generalization.}

Our novel interpretation of the global clustering coefficient as a measurement
of $2$-cliques expanding to $3$-cliques through closure of an adjacent edge
naturally carries over to the local clustering coefficient~\cite{watts1998collective}.
Each wedge consists of a $2$-clique and adjacent edge (\cref{fig:ccf_def}),
and we call the unique node in the intersection of the $2$-clique and adjacent
edge the \emph{center} of the wedge.  Under this view, the \emph{local
clustering clustering coefficient} of a node $u$ can be defined as the fraction
of wedges centered at $u$ that are closed:
\begin{equation}\label{eq:local_ccf}
\lccf{}{u} = 
\dfrac{2\lvert K_3(u) \rvert}{\lvert W(u) \rvert},
\end{equation}
where $K_3(u)$ is the set of $3$-cliques containing $u$ and $W(u)$ is the set of
wedges with center $u$ (if $\lvert W(u) \rvert = 0$, we say that $\lccf{}{u}$ is
undefined).  The \emph{average clustering coefficient} $\accf{}$ is the mean of
the local clustering coefficients,
\begin{equation}\label{eq:avg_ccf}
  \accf{} = \frac{1}{\lvert \widetilde V \rvert}\sum_{u \in \widetilde V} \lccf{}{u},
\end{equation}
where $\widetilde V$ is the set of nodes in the network where the local
clustering coefficient is defined~\cite{watts1998collective}.


\begin{figure}[tb]
  \centering
  \begin{tabular}{l c c c}
  & 1.~Start with~ & 2.~Find an adjacent edge & 3.~Check for an \\
  & an $\ell$-clique  & to form an $\ell$-wedge & $(\ell+1)$-clique \\ \\
    $\gccf{2}$
    & \begin{tikzpicture}[baseline=(current bounding box.center)] \input{CH3-TKZ-intro2_part1} \end{tikzpicture}
    & \begin{tikzpicture}[baseline=(current bounding box.center)] \input{CH3-TKZ-intro2_part1}\input{CH3-TKZ-intro2_part2} \end{tikzpicture}
    & \begin{tikzpicture}[baseline=(current bounding box.center)] \input{CH3-TKZ-intro2_part1}\input{CH3-TKZ-intro2_part2}\foreach \w in {w1}{ \draw[densely dotted] (\w) -- (v); } \end{tikzpicture} \\ \\
    $\gccf{3}$
    & \begin{tikzpicture}[baseline=(current bounding box.center)] \input{CH3-TKZ-intro3_part1} \end{tikzpicture}
    & \begin{tikzpicture}[baseline=(current bounding box.center)] \input{CH3-TKZ-intro3_part1}\input{CH3-TKZ-intro3_part2} \end{tikzpicture}
    & \begin{tikzpicture}[baseline=(current bounding box.center)] \input{CH3-TKZ-intro3_part1}\input{CH3-TKZ-intro3_part2}\foreach \w in {w1,w2}{ \draw[densely dotted] (\w) -- (v); } \end{tikzpicture} \\ \\
    $\gccf{4}$
    & \begin{tikzpicture}[baseline=(current bounding box.center)] \input{CH3-TKZ-intro4_part1} \end{tikzpicture}
    & \begin{tikzpicture}[baseline=(current bounding box.center)] \input{CH3-TKZ-intro4_part1}\input{CH3-TKZ-intro4_part2} \end{tikzpicture}
    & \begin{tikzpicture}[baseline=(current bounding box.center)] \input{CH3-TKZ-intro4_part1}\input{CH3-TKZ-intro4_part2}\foreach \w in {w1,w2,w3}{\draw[densely dotted] (\w) -- (v);} \end{tikzpicture} \\        
  \end{tabular}
  \dualcaption{Overview of higher-order clustering coefficients as clique
    expansion probabilities}{The $\ell$th-order clustering coefficient
    $\gccf{\ell}$ measures the probability that an $\ell$-clique and an adjacent
    edge, i.e., an $\ell$-wedge, is closed, meaning that the $\ell - 1$ possible
    edges between the $\ell$-clique and the outside node in the adjacent edge
    exist to form an $(\ell + 1)$-clique.  The case of $\ell = 2$ is the traditional
    (global) clustering coefficient---the fraction of length-2 paths that induce
    a triangle.  We also define the local higher-order clustering
    coefficients for a node $u$ to be the same probability over all $\ell$-wedges
    where $u$ is at the ``center,'' meaning that $u$ is the unique node
    intersecting the $\ell$-clique and the adjacent edge.}\label{fig:ccf_def}
\end{figure}

Our alternative interpretation of the clustering coefficient, described above as
a form of clique expansion, leads to a natural generalization to higher-order
structures.  Instead of expanding $2$-cliques to $3$-cliques, we expand
$\ell$-cliques to $(\ell + 1)$-cliques (\cref{fig:ccf_def}, $C_3$ and
$C_4$).  Formally, we define an $\ell$-wedge to consist of an $\ell$-clique and
an adjacent edge. Then we define the global $\ell$th-order clustering
coefficient $\gccf{\ell}$ as the fraction of $\ell$-wedges that are closed,
meaning that they induce an $(\ell + 1)$-clique in the network.  We can write
this as
\begin{equation}\label{eq:global_ccf_l}
  \gccf{\ell} = \frac{{\ell+1 \choose \ell } {\ell \choose 1 } \lvert K_{\ell + 1} \rvert}{\lvert W_{\ell} \rvert},
\end{equation}
where $K_{\ell + 1}$ is the set of $(\ell + 1)$-cliques, $W_{\ell}$ is the set
of $\ell$-wedges, and the coefficient ${\ell+1 \choose \ell } {\ell \choose 1}$
comes from the fact that each $(\ell + 1)$-clique closes that many wedges.

We also define higher-order local clustering coefficients:
\begin{equation} \label{eq:def_lccf}
  \lccf{\ell}{u} = \frac{{\ell \choose \ell - 1} \lvert K_{\ell + 1}(u) \rvert}{\lvert W_{\ell}(u) \rvert},
\end{equation}
where $K_{\ell + 1}(u)$ is the set of $(\ell + 1)$-cliques containing $u$,
$W_{\ell}(u)$ is the set of $\ell$-wedges with center $u$, and the coefficient
${\ell \choose \ell - 1}$ comes from the fact that each $(\ell + 1)$-clique
containing $u$ closes that many $\ell$-wedges in $W_{\ell}(u)$.

The $\ell$th-order clustering coefficient of a node is defined for any node that
is the center of at least one $\ell$-wedge, and the average $\ell$th-order
clustering coefficient is the mean of the local clustering coefficients:
\begin{equation}
  \accf{\ell} = \frac{1}{\lvert \widetilde V_\ell \rvert}\sum_{u \in \widetilde V_\ell} \lccf{\ell}{u},
\end{equation}
where $\widetilde V_{\ell}$ is the set of nodes that are the centers of at least
one $\ell$-wedge.

We have considered the average (higher-order) clustering coefficient to be the
mean over only nodes at the center of at least one $\ell$-wedge.  This makes
some of the theoretical analysis easier.  We could also consider the clustering
at node $u$ to be $0$ if $u$ does not participate in any $\ell$-wedge.  There is
not a consensus definition in the literature, and this issue is usually not even
discussed (see \citet{kaiser2008mean} for a discussion on how this seemingly
small difference in definition can affect network analyses).  At the very least,
we should report what fraction of nodes in the network participate in at least
one $\ell$-wedge.  We do this in our experiments later and also report the value
of $\accf{\ell}$ for both definitions.

\section{Theoretical Analysis}

\subsection{A method to compute higher-order clustering coefficients}

We first discuss how to compute higher-order clustering coefficients.
Substituting the identity
\begin{equation}\label{eq:wedge_identity}
  \lvert W_{\ell}(u) \rvert = \lvert K_{\ell}(u) \rvert \cdot (d_u - \ell + 1),
\end{equation}
into \cref{eq:def_lccf} gives
\begin{equation}   \label{eq:deriv_lccf}
  \lccf{\ell}{u}
  = \frac{\ell \cdot \lvert K_{\ell+1}(u) \rvert}{(d_u - \ell + 1) \cdot \lvert K_{\ell}(u) \rvert }.
\end{equation}
From \cref{eq:deriv_lccf}, it is easy to see that we can compute all local
$\ell$th-order clustering coefficients, by enumerating all $(\ell + 1)$-cliques
and $\ell$-cliques in the graph.  The computational complexity of the algorithm
is thus bounded by the time to enumerate $(\ell + 1)$-cliques and
$\ell$-cliques.  Using the Chiba and Nishizeki algorithm (discussed
in \cref{sec:honc_basic_complexity}), the complexity is $O(\ell a^{\ell-2}m)$,
where $a$ is the arboricity of the graph, and $m$ is the number of edges.
We note that the arboricity $a$ may be as large as $\sqrt{m}$, so this algorithm
is only guaranteed to take polynomial time if $\ell$ is a constant.  In general,
determining if there exists a single clique with at least $\ell$ nodes is
NP-complete~\cite{karp1972reducibility}.

For the global clustering coefficient, note that
\begin{equation}
\lvert W_{\ell} \rvert = \sum_{u \in V} \lvert W_\ell(u) \rvert.
\end{equation}
Thus, it suffices to enumerate $\ell$-cliques (to compute $\lvert W_{\ell} \rvert$ through \cref{eq:wedge_identity})
and to count the total number of $\ell$-cliques.  In practice, we use
the Chiba and Nishizeki to enumerate cliques and simultaneously compute
$\gccf{\ell}$ and $\lccf{\ell}{u}$ for all nodes $u$.

\subsection{Probabilistic interpretations}

To facilitate understanding of higher-order clustering coefficients, we now
present two probabilistic interpretations of the measurements.  First, we can
interpret $\lccf{\ell}{u}$ as the probability that a wedge $w$ chosen uniformly
at random from all wedges centered at $u$ is closed:
\begin{equation} \label{eq:prob_interp1}
  \lccf{\ell}{u} = \prob{w \in K_{\ell + 1}(u)}.
\end{equation}
We will use this interpretation in \cref{sec:nbrhood_cond} when relating
higher-order clustering coefficients to motif conductance.

For the next interpretation, it is useful to analyze the structure of the 1-hop
neighborhood graph $\onehopnou$ of a given node $u$ (not containing node $u$) in a
graph $G = (V, E)$, $u \in V$.  The vertex set of $\onehopnou$ is the set of
all nodes adjacent to $u$, and the edge set consists of all edges between neighbors
of $u$, i.e., $ \{ (v, w) \;\vert\; (u, v), (u, w), (v, w) \in E \}$.

Any $\ell$-clique in $G$ containing node $u$ corresponds to a unique $(\ell -
1)$-clique in $\onehopnou$, and specifically for $\ell = 2$, any edge $(u, v)$
corresponds to a node $v$ in $\onehopnou$.  Therefore, each $\ell$-wedge
centered at $u$ corresponds to an $(\ell-1)$-clique $K$ and one of the $d_u -
\ell + 1$ nodes outside $K$ (i.e., in $\onehopnou \backslash K$).  Thus,
\cref{eq:deriv_lccf} can be re-written as
\begin{equation}\label{eq:lccf_nbrs}
\frac{\ell \cdot \lvert K_{\ell}(\onehopnou) \rvert}{(d_u - \ell + 1) \cdot \lvert K_{\ell-1}(\onehopnou) \rvert }.
\end{equation}

If we uniformly at random select an $(\ell - 1)$-clique $K$ from $\onehopnou$
and then also uniformly at random select a node $v$ from $\onehopnou$ outside of
this clique, then $\lccf{\ell}{u}$ is the probability that these $\ell$ nodes
form an $\ell$-clique:
\begin{equation} \label{eq:prob_interp2}
  \lccf{\ell}{u} = \prob{K \cup \{ v \} \in K_{\ell}(\onehopnou)}.
\end{equation}

Moreover, if we condition on observing an $\ell$-clique from this sampling
procedure, then the $\ell$-clique itself is selected uniformly at random from
all $\ell$-cliques in $\onehopnou$.  Therefore, $\lccf{\ell-1}{u} \cdot
\lccf{\ell}{u}$ is the probability that an $(\ell - 1)$-clique and two nodes
selected uniformly at random from $\onehopnou$ form an $(\ell + 1)$-clique.
Applying this recursively gives
\begin{equation}
\prod_{j=2}^{\ell}\lccf{j}{u} = \frac{\lvert K_{\ell}(\onehopnou) \rvert}{{d_u \choose \ell}}.
\end{equation}
In other words, the product of the higher-order local clustering coefficients of
node $u$ up to order $\ell$ is the $\ell$-clique density amongst $u$'s
neighbors.

\subsection{Bounds on higher-order clustering coefficients}


\begin{figure}[tb]
\centering
  \begin{tabular}{l @{\hskip 12pt} c @{\hskip 12pt} c @{\hskip 12pt} c}
  & \scalebox{1.7}{\begin{tikzpicture}[baseline=(current bounding box.center)] \input{CH3-TKZ-ccf_equal} \end{tikzpicture}}
  & \scalebox{1.7}{\begin{tikzpicture}[baseline=(current bounding box.center)] \input{CH3-TKZ-ccf_lower} \end{tikzpicture}}
  & \scalebox{1.7}{\begin{tikzpicture}[baseline=(current bounding box.center)] \input{CH3-TKZ-ccf_upper} \end{tikzpicture}} \\ \\
  $\lccf{2}{u}$ & $1$ & $\frac{d_u}{2(d - 1)} \approx \frac{1}{2}$ & $\frac{d_u - 2}{4d_u - 4} \approx \frac{1}{4}$  \\ 
  $\lccf{3}{u}$ & $1$ & $0$ & $\frac{d_u - 4}{2d_u - 4} \approx \frac{1}{2}$ \\
  $\lccf{4}{u}$ & $1$ & $0$ & $\frac{d_u - 6}{2d_u - 6} \approx \frac{1}{2}$    
  \end{tabular}
  \dualcaption{Families of 1-hop neighborhoods of a node $u$ with degree
    $d_u$}{These families illustrate the difference between higher-order
    clustering coefficients of different orders.  Left: For cliques,
    $\lccf{\ell}{u} = 1$ for all $\ell$.  Middle: If node $u$'s neighbors form a
    complete bipartite graph, then $\lccf{2}{u}$ is constant while
    $\lccf{\ell}{u} = 0$ for $\ell \ge 3$.  Right: If half of node $u$'s
    neighbors form a star and the other half form a clique with $u$, then
    $\lccf{\ell}{u} = \sqrt{\lccf{2}{u}}$, reaching the upper bound of
    \cref{eq:Bound_Kappa3}.}  \label{fig:ccf_diffs}
\end{figure}

Next we analyze relationships between local higher-order clustering coefficients
of different orders.  Our technical result is \cref{prop:ccf_bounds}, which
provides tight lower and upper bounds for higher-order local clustering
coefficients in terms of the traditional local clustering coefficient.  The main
ideas of the proof are illustrated in \cref{fig:ccf_diffs}.

\begin{proposition}\label{prop:ccf_bounds}
  For any fixed $\ell \geq 3$,
\begin{equation}   \label{eq:Bound_Kappa3}
0 \leq \lccf{\ell}{u} \leq \sqrt{\lccf{2}{u}}.
\end{equation}
Moreover,
\begin{enumerate}
\item There exists a finite graph $G$ with node $u$ such that the lower bound is
  tight for any prescribed value of $\lccf{2}{u} \in [0, 1]$.
\item There exists a sequence of growing graphs $\{G_i\}_{i=1}^{\infty}$, each
  with node $u$, such that the upper bound is tight in the limit for any
  prescribed limiting value of $\lccf{2}{u} \in [0, 1]$.
\end{enumerate}
\end{proposition}
\begin{proof}
Clearly, $0 \leq \lccf{\ell}{u}$ if the local clustering coefficient is well
defined.  This bound is tight, even if $\lccf{2}{u}$ is constant, when
$\onehopnou$ is $(\ell - 1)$-partite (see \cref{fig:ccf_diffs}, middle).
  
To derive the upper bound, consider the 1-hop neighborhood $\onehopnou$, and let
$\delta_\ell(\onehopnou) = {\lvert K_{\ell}(\onehopnou) \rvert} / {{d_u \choose \ell}}$
denote the $\ell$-clique density of $\onehopnou$.  The Kruskal-Katona
theorem~\cite{kruskal1963number,katona1966theorem} implies that
$\delta_{\ell}(\onehopnou) \leq [\delta_{\ell - 1}(\onehopnou)]^{\frac{\ell}{\ell - 1}}$
and
$\delta_{\ell - 1}(\onehopnou) \leq [\delta_{2}(\onehopnou)]^{\frac{\ell - 1}{2}}$.
Combining this with \cref{eq:deriv_lccf} gives
\begin{equation}
  \lccf{\ell}{u} \leq [\delta_{\ell - 1}(\onehopnou)]^{\frac{1}{\ell - 1}} \leq \sqrt{\delta_{2}(\onehopnou)} = \sqrt{\lccf{2}{u}},
\end{equation}
where the last equality uses the fact that $\lccf{2}{u}$ is the edge density of
$\onehopnou$.

The upper bound is tight if $\onehopnou$ consists of a clique and isolated nodes
(\cref{fig:ccf_diffs}, right) and we take $d_u \to \infty$.  Specifically, let
$N_u$ consist of a clique of size $c$ and $b$ isolated nodes.  When $\ell = 2$,
\begin{align*}
\lccf{\ell}{u} = \frac{{c \choose 2}}{{c + b \choose 2}}
= \frac{(c - 1)c}{(b + c - 1)(b + c)} \to \left(\frac{c}{c + b}\right)^2
\end{align*}
and by \cref{eq:lccf_nbrs}, when $3 \le \ell \le c$,
\begin{align*}
\lccf{\ell}{u} &= \frac{\ell \cdot {c \choose \ell}}{(c + b - \ell + 1) \cdot {c \choose \ell - 1}} 
= \frac{c - \ell + 1}{c + b - \ell + 1} \to \frac{c}{c + b}
\end{align*}

By adjusting the ratio $c / ( b + c)$ in $\onehopnou$, we can
construct a family of graphs such that $\lccf{2}{u}$ may take any value in $[0,1]$
as $d_u \to \infty$ and $\lccf{\ell}{u} = \sqrt{\lccf{2}{u}}$ as $d_u \to \infty$.
\end{proof}

Again, we have an asymptotic result.  However, in \cref{sec:ccfs_empirical}, we
will see that in some real-world data, there are nodes $u$ for which
$\lccf{3}{u}$ is close to $\sqrt{\lccf{2}{u}}$.

\subsection{Analysis for the $G_{n, p}$ model}

Next, we analyze higher-order clustering coefficients in the classical
Erd\H{o}s-R\'enyi model with edge probability $p$, i.e., the $G_{n,p}$
model~\cite{erdos1959random}.  We implicitly assume
that $\ell$ is small in the following analysis so that there should be at
least one $\ell$-wedge in the graph (with high probability and $n$
large, there is no clique of size greater than $(2 + \epsilon)\log n / \log (1 / p)$
for any $\epsilon > 0$~\cite{bollobas1976cliques}).  Our results are
summarized in the following proposition.

\begin{proposition}\label{prop:ccf_er}
Let $G$ be a random graph drawn from the $G_{n, p}$ model.
\begin{enumerate}
\item $\expectover{G}{\gccf{\ell}} = p^{\ell - 1}$, \label{itm:gnp1}
\item $\expectover{G}{\lccf{\ell}{u} \given W_{\ell}(u) > 0} = p^{\ell - 1}$ for any node $u$, \label{itm:gnp2}
\item $\expectover{G}{\accf{\ell}} = p^{\ell - 1}$, and \label{itm:gnp3}
\item For constant $\ell$,
\[
\expectover{G}{\lccf{\ell}{u} \given \lccf{2}{u}, W_{\ell}(u) > 0}
=  \left[\lccf{2}{u} - (1 - \lccf{2}{u}) \cdot O(1 / d_u^2)\right]^{\ell - 1}
\approx (\lccf{2}{u})^{\ell - 1}.
\]  \label{itm:gnp4}
\end{enumerate}
\end{proposition}
\begin{proof}
  For \cref{itm:gnp1}:
  \begin{align*}
    \expect{\gccf{\ell}}
    &= \expectover{G}{\expectover{W_{\ell}}{\gccf{\ell} \given W_{\ell}}} \\
    &= \expectover{G}{\expectover{W_{\ell}}{\dfrac{1}{\lvert W_{\ell}\rvert}\sum_{w \in W_{\ell}}\prob{w \text{ is closed}}}} \\
    &= \expectover{G}{\expectover{W_{\ell}}{\frac{1}{\lvert W_{\ell}\rvert}\sum_{w \in W_{\ell}}p^{\ell - 1}}} \\
    &= \expectover{G}{p^{\ell - 1}} \\
    &= p^{\ell - 1}.
  \end{align*}  
  The second equality is well defined (with high probability) for small $\ell$.
  The third equality comes from the fact that any $\ell$-wedge is closed if and
  only if the $\ell - 1$ possible edges between the $\ell$-clique and the
  outside node in the adjacent edge exist to form an $(\ell+1)$-clique.

  The proof of \Cref{itm:gnp2} is similar.
  \begin{align*}
    \expect{\lccf{\ell}{u} \given W_{\ell}(u) > 0}
    &= \expectover{G}{\expectover{W_{\ell}(u) > 0}{\lccf{\ell}{u} \given W_{\ell}(u)}} \\
    &= \expectover{G}{\expectover{W_{\ell}(u) > 0}{\dfrac{1}{\lvert W_{\ell}(u)\rvert}\sum_{w \in W_{\ell}(u)}\prob{w \text{ is closed}}}} \\
    &= \expectover{G}{\expectover{W_{\ell}(u) > 0}{\frac{1}{\lvert W_{\ell}(u)\rvert}\sum_{w \in W_{\ell}(u)}p^{\ell - 1}}} \\
    &= \expectover{G}{p^{\ell - 1}} \\
    &= p^{\ell - 1}.
  \end{align*}

  For \cref{itm:gnp3}:
  \begin{align*}
    \expectover{G}{\accf{\ell}}
    &= \expectover{G}{\expectover{\tilde{V}}{\accf{\ell} \given \widetilde V }} \\
    &= \expectover{G}{\expectover{\tilde{V}}{\frac{1}{\lvert \widetilde V \rvert} \sum_{u \in \widetilde V} \expect{\lccf{\ell}{u}}}} \\
    &= \expectover{G}{\expectover{\tilde{V}}{\frac{1}{\lvert \widetilde V \rvert} \sum_{u \in \widetilde V} p^{\ell - 1}}} \\
    &= \expectover{G}{p^{\ell - 1}} \\
    &= p^{\ell - 1}.
  \end{align*}
  In the second line, $\tilde{V}$ is nonempty with high probability for $\ell$ small,
  so $1 / \lvert \tilde{V} \rvert$ is well defined.  The third equality comes from \cref{itm:gnp2}.
  
  For \cref{itm:gnp4}, first consider
  \begin{align*}
    & \expectover{G}{\lccf{\ell}{u} \given \lccf{2}{u}, W_{\ell}(u) > 0} \\
    &= \expectover{G}{\expectover{W_{\ell}(u) > 0}{\lccf{\ell}{u} \given \lccf{2}{u},\; W_{\ell}(u)}} \\
    &= \expectover{G}{\expectover{W_{\ell}(u) > 0}{\dfrac{1}{\lvert W_{\ell}(u)\rvert}\sum_{w \in W_{\ell}(u)}\prob{w \text{ is closed} \given \lccf{2}{u}}}}.
  \end{align*}  
  Now, note that $\onehopnou$ has $m = \lccf{2}{u} \cdot {d_u \choose 2}$ edges.
  Knowing that $w \in W_{\ell}(u)$ accounts for ${\ell - 1 \choose 2}$ of these
  edges.  By symmetry, the other $q = m - {\ell - 1 \choose 2}$ edges appear in
  any of the remaining $r = {d_u \choose 2} - {\ell - 1 \choose 2}$ edges
  uniformly at random.  There are ${r \choose q}$ ways to place these edges, of
  which ${r - \ell + 1 \choose q - \ell + 1}$ would close the wedge $w$.  Thus,
  \begin{align}
  \prob{w \text{ is closed} \given \lccf{2}{u}}
  &= \frac{{r - \ell + 1 \choose q - \ell + 1}}{{r \choose q}} \\
  &= \frac{(r - \ell + 1)!q!}{(q - \ell + 1)!r!} \\
  &= \frac{(q - \ell + 2)(q - \ell + 3) \cdots q}{(r - \ell + 2)(r - \ell + 3) \cdots r} \label{eqn:w_closed}
  \end{align}
  Now, for any small nonnegative integer $k$,
  \begin{align*}\label{eqn:closed_limit}
  \frac{q - k}{r - k} &= \frac{
  \lccf{2}{u} \cdot {d_u \choose 2} - {\ell -1 \choose 2} - k
  }{
  {d_u \choose 2} - {\ell - 1 \choose 2} - k
  } \\
  &= \lccf{2}{u} - (1 - \lccf{2}{u})\left[\frac{{\ell -1 \choose 2} + k}{{d_u \choose 2} - {\ell - 1 \choose 2} - k}\right]  \\
  &= \lccf{2}{u} - (1 - \lccf{2}{u}) \cdot O(1 / d_u^2)
  \end{align*}
  We are able to use the big-O notation because $\ell$ is a constant by assumption.
  Thus, the expression in \cref{eqn:w_closed} approaches $(\lccf{2}{u})^{\ell - 1}$
  as $\lccf{2}{u} \to 1$ and as $d_u \to \infty$.
\end{proof}

\Cref{itm:gnp4} of \cref{prop:ccf_er} says that even if the second-order local
clustering coefficient is large, the $\ell$th-order clustering coefficient will
still decay exponentially in $\ell$, at least in the limit as $d_u$ grows.  We
will use this as a reference point in our analysis of real-world networks in
\cref{sec:ccfs_empirical}.

\subsection{Analysis for the small-world model}

\begin{figure}[tb]
  \begin{center}
    \includegraphics[width=0.8\columnwidth]{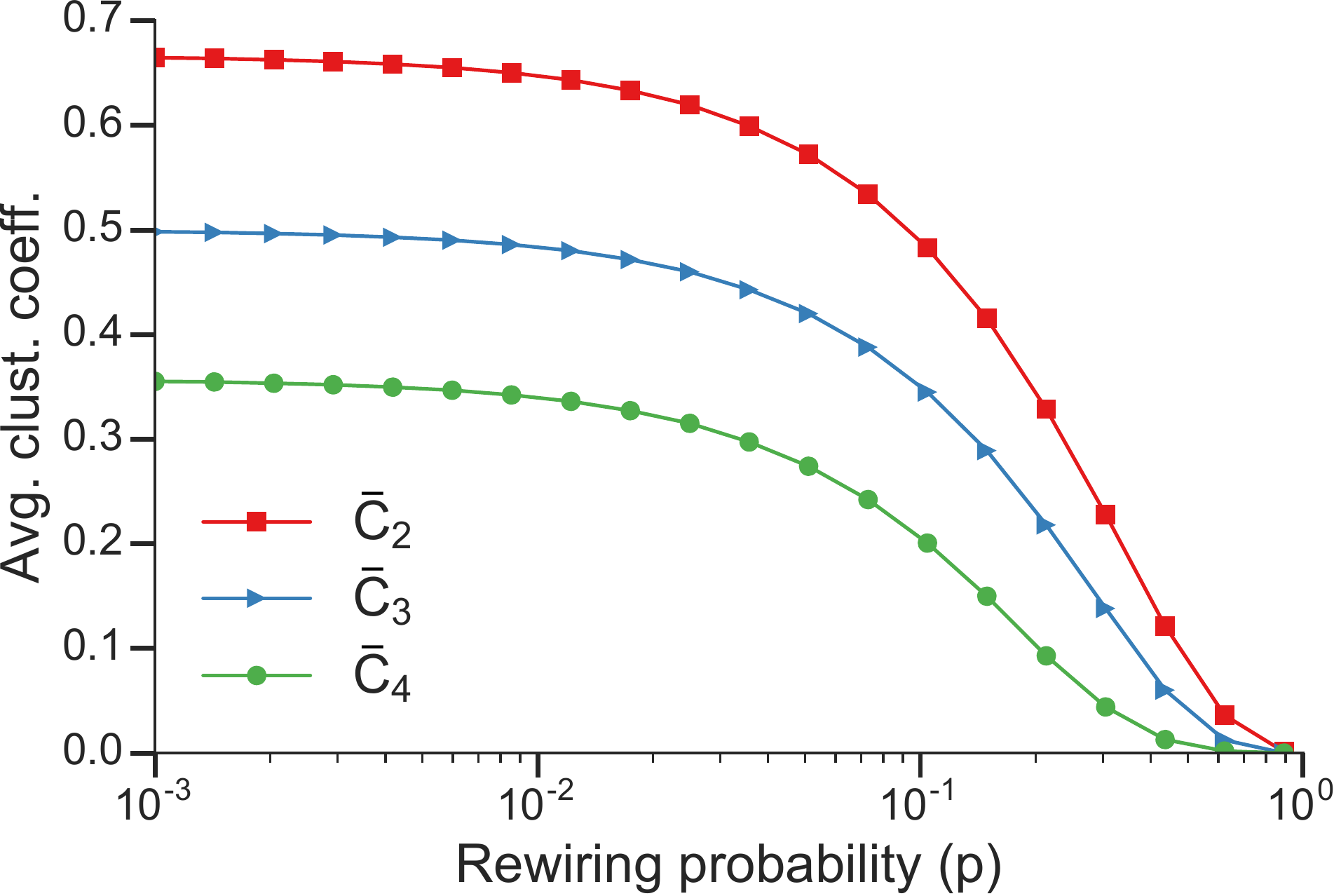}
  \end{center}
  \dualcaption{Higher-order clustering in the small-world model}{The plot shows
    the average higher-order clustering coefficient $\accf{\ell}$ as a function
    of rewiring probability $p$ in small-world networks for $\ell = 2, 3, 4$.
    Each data point is the average over 20 small-world random graph instances
    with 20,000 nodes where each node connects to its $10$ nearest neighbors
    before rewiring.}\label{fig:sw_ccfs}
\end{figure}

We also study higher-order clustering in the small-world random graph
model~\cite{watts1998collective}. The model begins with a ring network where
each node connects to its $2k$ nearest neighbors.  Then, for each node $u$ and
each of the $k$ edges $(u, v)$ with $v$ following $u$ ``clockwise'' in the ring,
the edge is ``rewired'' to $(u, w)$ with probability $p$, where $w$ is chosen
uniformly at random.\footnote{Bret Victor has a wonderful visual explanation of
  the small-world model at\\
  \url{http://worrydream.com/ScientificCommunicationAsSequentialArt/}.}

With no rewiring ($p = 0$) and $k \ll n$, $\accf{2} \approx
3/4$~\cite{watts1998collective}.  As $p$ increases, the average clustering
coefficient $\accf{2}$ slightly decreases until a phase transition near $p =
0.1$, where $\accf{2}$ decays to $0$~\cite{watts1998collective} (also see
\cref{fig:sw_ccfs}).  We generalize these results for higher-order clustering
coefficients.  Specifically, when $p = 0$, we can analytically show that
\begin{equation}  \label{eq:C_SW}
  \accf{\ell} \approx (\ell + 1) / (2\ell)
\end{equation}
for any $\ell \geq 2$. Thus, $\accf{\ell}$ decreases as $\ell$ increases.
Furthermore, via simulation, we observe the same behavior as for $\accf{2}$ when
adjusting the rewiring probability $p$ (\cref{fig:sw_ccfs}).  Regardless of
$\ell$, the phase transition happens near $p = 0.1$ and this is partly
due to the fact that $\lccf{\ell}{u} \to 0$ as $\lccf{2}{u} \to 0$ (\cref{prop:ccf_bounds}).

The following proposition gives an argument for \cref{eq:C_SW}.  We will
not be precise in the proof.
\begin{proposition}\label{prop:ccf_sw}
In the small world model with no rewiring ($p = 0$) and $2k < n$,
for any node $u$, $\lccf{\ell}{u} \approx \frac{\ell + 1}{2\ell}$ for large $k$.
\end{proposition}
\begin{proof}
To derive \Cref{eq:C_SW}, we first label the $2k$ neighbors of $u$ as $1, 2,
\dots, 2k$ by their clockwise ordering in the ring.  Since $2k < n$, these nodes are unique.
Next, define the
\emph{span} of any $\ell$-clique containing $u$ as the difference between the
largest and smallest label of the $\ell-1$ nodes in the clique other than $u$.
The span $s$ of any $\ell$-clique satisfies $s \leq k-1$ since any node is
directly connected with a node of label difference no greater than $k-1$.  Also,
$s \geq \ell-2$ since there are $\ell-1$ nodes in an $\ell$-clique other than
$u$.  For each span $s$, we can find $2k-1-s$ pairs of $(i, j)$ such that $1
\leq i, j \leq 2k$ and $j - i = s$.  Finally, for every such pair $(i, j)$,
there are ${s - 1 \choose \ell - 3}$ choices of $\ell-3$ nodes between $i$ and
$j$ which will form an $\ell$-clique together with nodes $u$, $i$, and
$j$. Therefore,
\begin{align*}
\lvert \cliqueL{\ell}{u} \rvert
&= \sum_{s = \ell-2}^{k-1} (2k-1-s) \cdot {s - 1 \choose \ell - 3} \\
&= \sum_{s = \ell-2}^{k-1} (2k-1-s) \cdot \frac{(s-1)(s - 2) \cdots (s - \ell + 3)}{(\ell - 3)!} \\
&= \sum_{t = 1}^{k-\ell+2} (2k + 2 -t - \ell) \cdot \frac{t (t + 1) \cdots (t + \ell - 4)}{(\ell - 3)!} \\
&\approx \sum_{t = 1}^{k} \frac{(2k - t)t^{\ell - 3}}{(\ell - 3)!} \quad \text{(ignoring lower-order terms of $t$)} \\ 
&= \frac{1}{(\ell - 3)!}\sum_{t = 1}^{k} 2kt^{\ell - 3} - t^{\ell - 2} \\
&= \frac{1}{(\ell - 3)!}\left[2k \cdot \left(\frac{k^{\ell - 2}}{\ell - 2} + O(k^{\ell - 3})\right) - \frac{k^{\ell - 1}}{\ell - 1}  + O(k^{\ell - 2})\right] \\
&= \frac{1}{(\ell - 3)!}\left[\frac{2k^{\ell - 1}}{\ell - 2} - \frac{k^{\ell - 1}}{\ell - 1}\right] \quad \text{(ignoring lower-order terms of $k$)}\\
&= \frac{k^{\ell -1 }}{(\ell - 3)!}\left[\frac{2}{\ell - 2} - \frac{1}{\ell - 1}\right] \\
&= \frac{\ell k^{\ell -1 }}{(\ell - 1)!} 
\end{align*}

Then, following \cref{eq:deriv_lccf},
\begin{align*}
\lccf{\ell}{u} &= \frac{\ell \cdot \lvert K_{\ell+1}(u) \rvert}{(d_u - \ell + 1) \cdot \lvert K_{\ell}(u) \rvert }  \\
&\approx \frac{\ell \cdot \frac{(\ell + 1) k^{\ell}}{\ell!}}{(2k - \ell + 1) \cdot \frac{\ell k^{\ell -1 }}{(\ell - 1)!} } \\
&= \frac{(\ell + 1)k}{(2k - \ell + 1) \cdot \ell} 
\to \frac{\ell + 1}{2\ell} \text{ as } k \to \infty.
\end{align*}
\end{proof}

\cref{prop:ccf_sw} is an imprecise and asymptotic result, but \cref{fig:sw_ccfs} 
at least captures the trend that the clustering decreases with $\ell$.

\section{Empircal Analysis}
\label{sec:ccfs_empirical}
We now measure and analyze higher-order clustering
in five networks:
\begin{enumerate}
\item an $\er$ graph with $n = 1,000$ nodes and edge probability $p = 0.2$;
\item a small-world network with $n = 20,000$ nodes, $k = 10$ edges per node,
and rewiring probability $p = 0.1$;
\item the neural network
of the nematode worm \emph{C.~elegans}~\cite{watts1998collective}, where we take
the edges in this network to be undirected; 
\item the friendships between Stanford
students on Facebook from September 2005~\cite{traud2012social} (denoted $\stanford$); and
\item a co-authorship network constructed from papers posted to the
Astrophysics category on arXiv~\cite{leskovec2007graph} (denoted $\astroph$).
\end{enumerate}

\Cref{tab:data_summary} lists the higher-order clustering coefficients for $\ell
= 2$, $3$, and $4$.  \Cref{prop:ccf_er,fig:sw_ccfs} say that the average
clustering coefficient should decrease as the order increases for the
Erd\H{o}s-R\'enyi and small-world models, and indeed, this is the case.  We see
that $\accf{4} < \accf{3} < \accf{2}$ also holds for the three real-world
networks.  (This does not necessarily have to hold---the right column of
\cref{fig:ccf_diffs} shows an example where $\accf{3} > \accf{2}$.)  Thus, when
averaging over nodes, higher-order cliques are less likely to close in both the
synthetic and real-world networks.  However, the same trends do not hold for the
higher-order global clustering coefficient, which can decrease with $\ell$
(\emph{C. elegans}), increase with $\ell$ ($\astroph$), or non-monotonic with $\ell$
($\stanford$).

For the three real-world networks, we also measure the higher-order clustering
coefficients with respect to two null models (\cref{tab:data_summary}).  First,
we use the Configuration Model (CM) that samples uniformly at random from simple
graphs with the same degree
distribution~\cite{bollobas1980probabilistic,milo2003uniform}.  In the
real-world networks, $\accf{2}$ is much larger than expected with respect to the
CM null model (\cref{tab:data_summary}); this behavior was observed by Watts and
Strogatz~\cite{watts1998collective}.  Unsurprisingly, we find that the same
result holds for $\accf{3}$ and $\accf{4}$.

Second, we use a null model that samples graphs that preserve both the degree
distribution and $\accf{2}$.  Specifically, these are samples from an ensemble
of exponential graphs where the Hamiltonian measures the absolute value of the
difference in $\accf{2}$ between the original network and the sampled
network~\cite{park2004statistical}.  Such samples are referred to as Maximally
Random Clustered Networks (MRCN) and are sampled with a simulated annealing
procedure~\cite{colomer2013deciphering}.  Comparing $\accf{3}$ between the
real-world and the null network, we observe different behavior in higher-order
clustering.  The \emph{C. elegans} neural network has less than expected higher-order
clustering in terms of $\accf{3}$ with respect to the MRCN null model
(\cref{tab:data_summary}).  On the other hand, the Facebook friendship and
co-authorship networks exhibits higher than expected $\accf{3}$ (although
only slightly higher for $\astroph$).  Thus, while all three of the
real-world networks exhibit clustering in the classical sense of triadic
closure, only the friendship and (and to some extent) the co-authorship networks exhibit higher-order
clustering.

The lack of higher-order clustering in the \emph{C. elegans}
network\footnote{Here, we mean the lack of clustering with respect to the
  higher-order clustering coefficient defined in this chapter.  In
  \cref{ch:honc}, we found a meaningful \emph{higher-order cluster} in
  \emph{C. elegans} in terms of the bi-fan motif.  These results are compatible
  and exemplify the different dimensions of higher-order network analysis in
  this thesis.} agrees with previous results that 4-cliques are under-expressed
in parts of \emph{C. elegans}, while open 3-wedges related with cooperative
information propagation are
over-expressed~\cite{benson2016higher,milo2002network,varshney2011structural}.
This also provides credence for the ``3-layer'' model of
\emph{C. elegans}~\cite{varshney2011structural}.  The observed clustering in the
friendship network is consistent with prior work showing the relative
infrequency of open $\ell$-wedges in many Facebook network subgraphs with
respect to a null model accounting for triadic
closure~\cite{ugander2013subgraph}.  Co-authorship networks are known to have
large clustering in the traditional sense, which is partially attributed to
papers with multiple authors that form cliques~\cite{newman2001scientific}.  It
is natural that these cliques contribute to higher-order clustering as well.

\begin{sidewaystable}[t]
\renewcommand{\arraystretch}{1.1}
\setlength{\tabcolsep}{3pt}
  \scalebox{0.8}{
  \begin{tabular}{l   c c  @{\hskip 10pt}  c c c   @{\hskip 10pt}   c c c   @{\hskip 10pt}   c c c}
\toprule
& E-R & S-W 
& \multicolumn{3}{c}{\!\!\!\!\!\!\!\!\!\!\!\!\!\!\!\!\!\!\!\!\!\!\!\!\!\!\!\!\!\!\!\!\emph{C. elegans}} 
& \multicolumn{3}{c}{\!\!\!\!\!\!\!\!\!\!\!\!\!\!\!\!\!\!\!\!\!\!\!\!\!\!\!\!\!\!\!\!$\stanford$} 
& \multicolumn{3}{c}{\!\!\!\!\!\!\!\!\!\!\!\!\!\!\!\!\!\!\!\!\!\!\!\!\!\!\!\!\!\!\!\!$\astroph$} \\
\cmidrule(r){4-6}                                                                                                                                                          
\cmidrule(r){7-9}                                                                                                                                                          
\cmidrule(r){10-12}   
&       &               & orig.\ & CM & MRCN & orig.\ & CM & MRCN & orig.\ & CM & MRCN \\ \midrule
$\lvert V \rvert$ & 1K & 20K & 297 & 297 & 297 & 11.6K & 11.6K & 11.6K & 18.8K & 18.8K & 18.8K  \\ 
$\lvert E \rvert$ & 99.8K & 100K & 2.15K & 2.15K & 2.15K & 568K & 568K & 568K & 198K & 198K & 198K \\ \midrule
$\accf{2}$ & 0.20 & 0.49 & 0.31 & 0.15 [0.13, 0.16] & 0.31 [0.31, 0.31] & 0.25 & 0.03 [0.03, 0.03] & 0.25 [0.25, 0.25]  & 0.68 & 0.01 [0.01, 0.01] & 0.68 [0.68, 0.68] \\
$\accf{3}$ & 0.04 & 0.35 & 0.14 & 0.04 [0.02, 0.05] & 0.17 [0.15, 0.17] & 0.18 & 0.00 [0.00, 0.00] & 0.14 [0.14, 0.14]  & 0.61 & 0.00 [0.00, 0.00] & 0.60 [0.60, 0.60]\\
$\accf{4}$ & 0.01 & 0.20 & 0.06 & 0.01 [0.00, 0.01] & 0.09 [0.07, 0.11] & 0.16 & 0.00 [0.00, 0.00] & 0.09 [0.09, 0.09] & 0.56 & 0.00 [0.00, 0.00] & 0.52 [0.52, 0.53]\\  \midrule
$\frac{\lvert \widetilde V_2 \rvert}{\lvert V \rvert}$  & 1.00 & 1.00 & 0.95 & 0.95 [0.95, 0.95] & 0.95 [0.95, 0.95] & 0.95 & 0.95 [0.95, 0.95] & 0.95 [0.95, 0.95] & 0.93 & 0.93 [0.93, 0.93] & 0.93 [0.93, 0.93] \\
$\frac{\lvert \widetilde V_3 \rvert}{\lvert V \rvert}$  & 1.00 & 1.00 & 0.93 & 0.89 [0.86, 0.91] & 0.94 [0.93, 0.94] & 0.92 & 0.85 [0.85, 0.86] & 0.93 [0.93, 0.93] & 0.84 & 0.46 [0.45, 0.46] &  0.84 [0.84, 0.84] \\
$\frac{\lvert \widetilde V_4 \rvert}{\lvert V \rvert}$  & 1.00 & 1.00 & 0.81 & 0.54 [0.48, 0.60] & 0.83 [0.79, 0.87] & 0.88 & 0.63 [0.63, 0.64] & 0.91 [0.91, 0.91] & 0.74 & 0.04 [0.03, 0.04] & 0.75 [0.75, 0.75] \\ \midrule
$\accf{2}'$ & 0.20 & 0.49 & 0.29 & 0.14 [0.13, 0.15] & 0.29 [0.29, 0.29] & 0.24 & 0.03 [0.03, 0.03] & 0.24 [0.24, 0.24]  & 0.63 & 0.01 [0.01, 0.01] & 0.63 [0.63, 0.63] \\
$\accf{3}'$ & 0.04 & 0.35 & 0.13 & 0.03 [0.02, 0.04] & 0.16 [0.14, 0.16] & 0.17 & 0.00 [0.00, 0.00] & 0.13 [0.13, 0.13] & 0.51 & 0.00 [0.00, 0.00] & 0.50 [0.50, 0.50] \\
$\accf{4}'$ & 0.01 & 0.20 & 0.05 & 0.00 [0.00, 0.01] & 0.08 [0.06, 0.09] & 0.14 & 0.00 [0.00, 0.00] & 0.08 [0.08, 0.08] & 0.42 & 0.00 [0.00, 0.00] & 0.40 [0.39, 0.40] \\ \midrule
$\gccf{2}$ & 0.20 & 0.48 & 0.18 & 0.10 [0.10, 0.11] & 0.14 [0.14, 0.15] & 0.16 & 0.03 [0.03, 0.03] & 0.05 [0.05, 0.05] & 0.32 & 0.01 [0.01, 0.01] & 0.14 [0.13, 0.14] \\
$\gccf{3}$ & 0.04 & 0.36 & 0.08 & 0.02 [0.02, 0.03] & 0.04 [0.04, 0.05] & 0.11 & 0.00 [0.00, 0.00] & 0.02 [0.02, 0.02] & 0.33 & 0.00 [0.00, 0.00] & 0.10 [0.10, 0.10] \\
$\gccf{4}$ & 0.01 & 0.23 & 0.06 & 0.01 [0.00, 0.01] & 0.02 [0.01, 0.03] & 0.12 & 0.00 [0.00, 0.00] & 0.02 [0.02, 0.02] & 0.36 & 0.00 [0.00, 0.00] & 0.09 [0.09, 0.09] \\
\bottomrule
\end{tabular}
  }
  \dualcaption{Higher-order clustering coefficients for several networks}{We
    measure higher-order clustering in the $\er$ (E-R) and small-world (S-W)
    random graph models as well as three real-world networks.  For the
    real-world networks, we also measure the clustering coefficients with
    respect to two null models: a Configuration Model (CM) that produces random
    graphs with the same degree distribution as in the real
    graph~\cite{bollobas1980probabilistic,milo2003uniform} and Maximally Random
    Clustered Networks (MRCN) that preserve the degree distribution as well as
    $\accf{2}$~\cite{park2004statistical,colomer2013deciphering}.  For the
    random networks, we report the sample mean and the minimum and maximum
    values (in brackets) over 100 samples.  Under a CM null model, there is more
    clustering than expected for all orders of the clustering coefficient.
    Under the MRCN null model, the \emph{C. elegans} neural network does not exhibit higher-order
    clustering, while $\stanford$ and $\astroph$ do exhibit higher-order
    clustering (although $\astroph$ has only slightly larger higher-order clustering than expected under the
    null model).  We measure the higher-order average clustering coefficient
    ($\accf{\ell}$), the fraction of nodes that are the center of at least one
    $\ell$-wedge ($\lvert \widetilde V_\ell \rvert / \lvert V \rvert$), the
    alternative higher-order average clustering coefficient that considers the
    clustering of a node not at the center of at least one $\ell$-wedge to be
    $0$, and the higher-order global clustering coefficient ($\gccf{\ell}$).}
  \label{tab:data_summary}
\end{sidewaystable}

\clearpage

\Cref{fig:ccfs23} plots the joint distribution of $\lccf{2}{u}$ and
$\lccf{3}{u}$.  The lower trend line represents random behavior (i.e., the
asymptotic behavior of Erd\H{o}s-R\'enyi networks in expectation; see
\cref{prop:ccf_er}) and the upper trend line denotes the maximum possible value
of $\lccf{3}{u}$ given $\lccf{2}{u}$ (\cref{eq:Bound_Kappa3}).  For many nodes
in the \emph{C. elegans} network, local clustering is nearly random, resembling
the Erd\H{o}s-R\'enyi joint distribution.  This provides further evidence that
the \emph{C. elegans} neural network lacks higher-order clustering.  In the
co-authorship network, there are many nodes $u$ with a large value of
$\lccf{2}{u}$ that have an even larger value of $\lccf{3}{u}$ near the upper
bound of \cref{eq:Bound_Kappa3} (see the inset).  Thus, our upper bound is
nearly tight in practice.  We emphasize that this does not imply that these
nodes are simply members of large cliques---if the 1-hop neighborhood of $u$ is
a clique, then $\lccf{2}{u} = \lccf{3}{u} = 1$.  Instead, some nodes likely
appear in both cliques and as the center of star-like patterns, as in
\cref{fig:ccf_diffs} (right).  This pattern would appear, for example, with an
advisor writing several 2-author papers with students and a multi-author paper
with other collaborators.


\definecolor{myblue}{RGB}{31,120,180}
\definecolor{mylightblue}{RGB}{166,206,227}
\definecolor{mygreen}{RGB}{51,160,44}
\definecolor{mylightgreen}{RGB}{178,223,138}
\begin{figure}[tb]
  \newcommand{\figwidth}{0.32}
  \includegraphics[width=\figwidth\columnwidth]{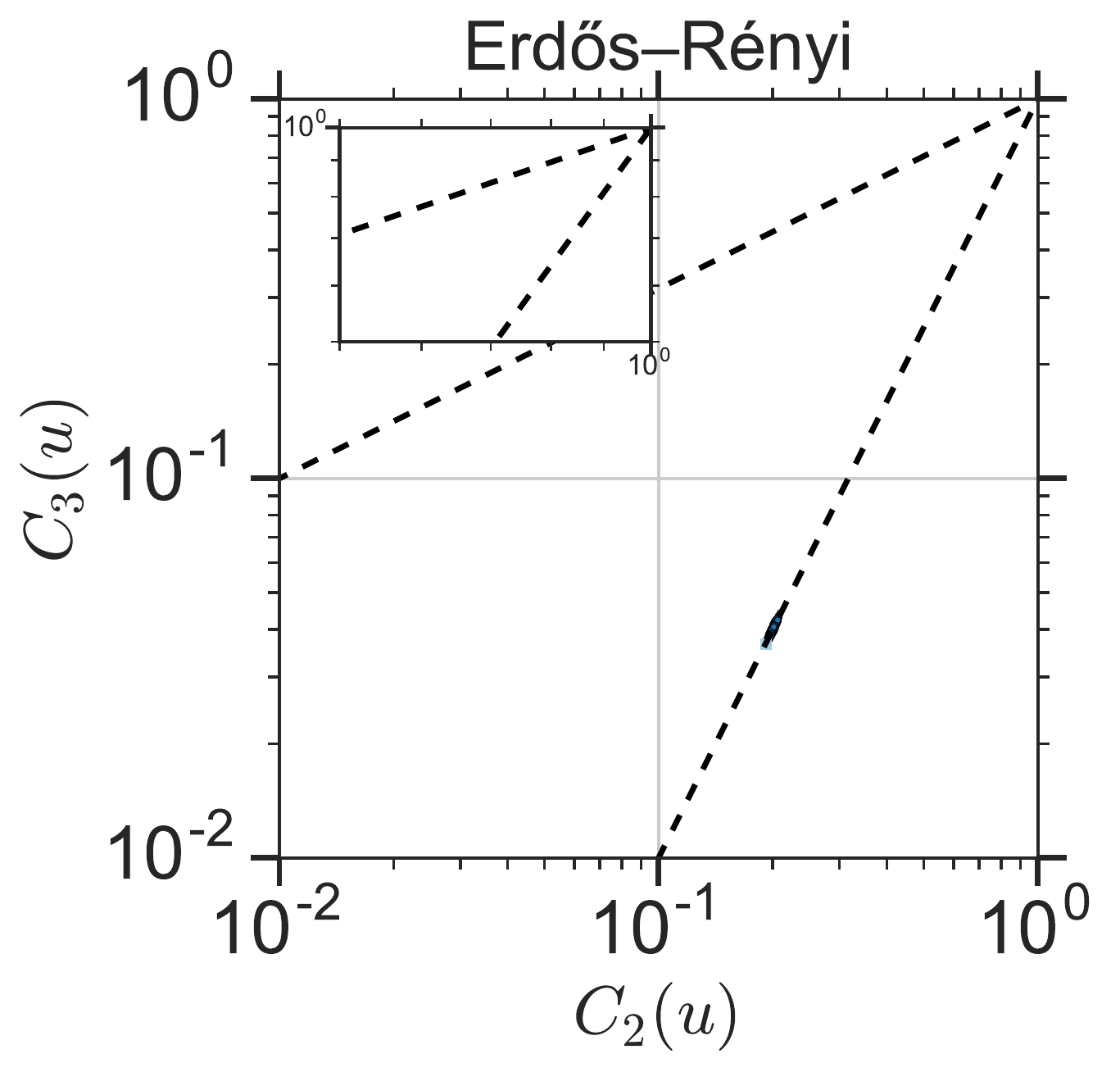}
  \includegraphics[width=\figwidth\columnwidth]{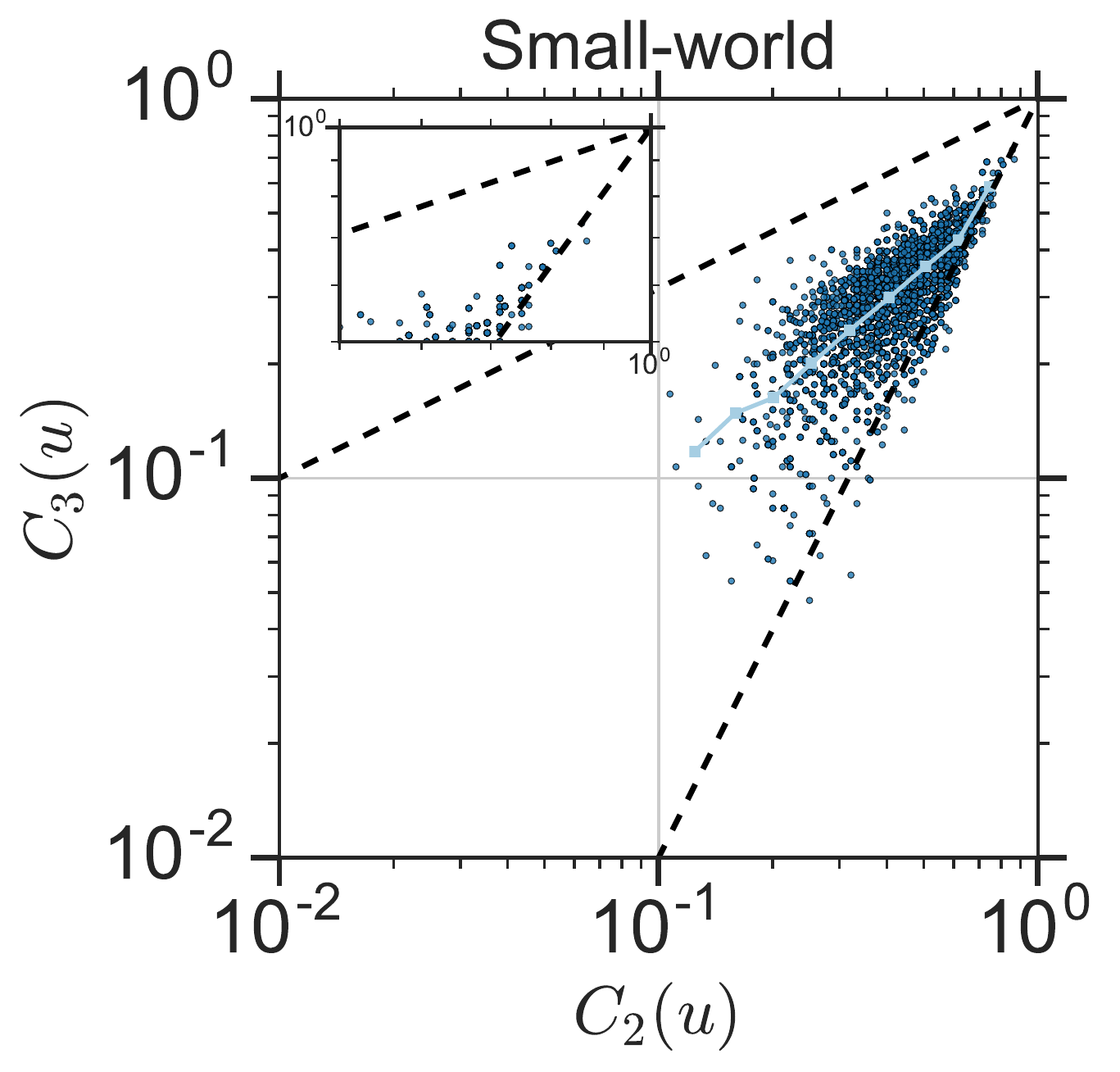} \\
  \includegraphics[width=\figwidth\columnwidth]{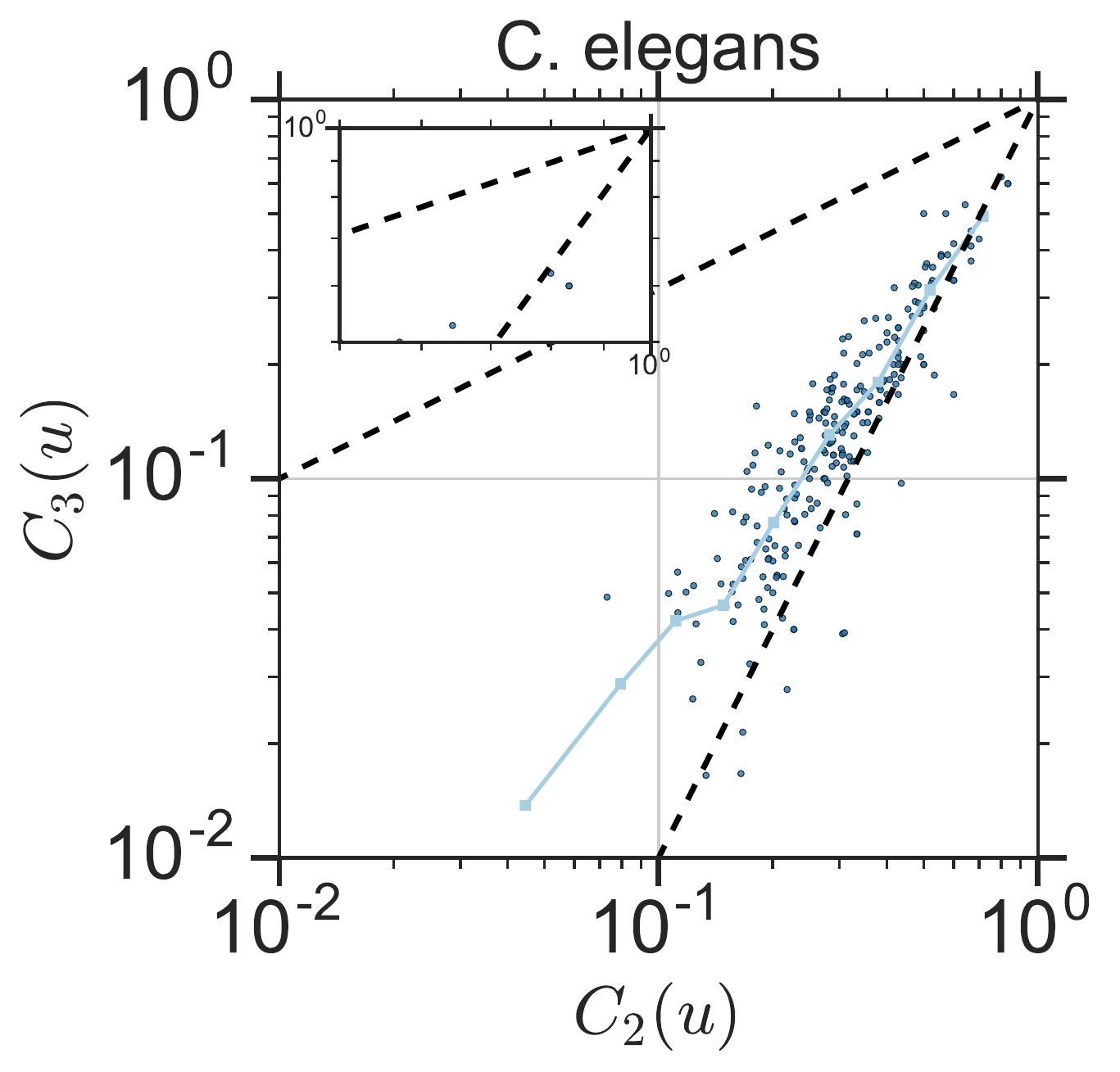}
  \includegraphics[width=\figwidth\columnwidth]{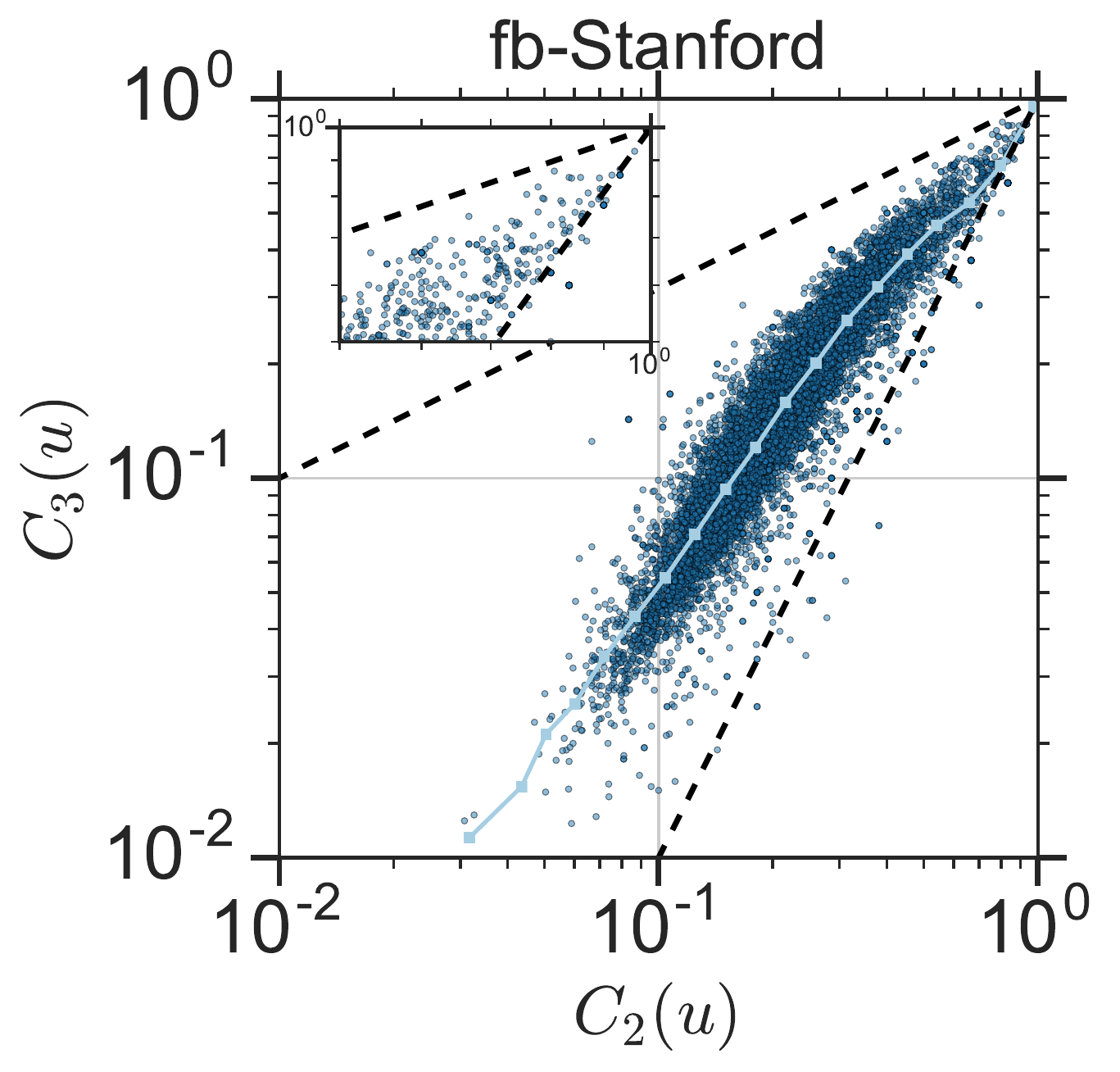}
  \includegraphics[width=\figwidth\columnwidth]{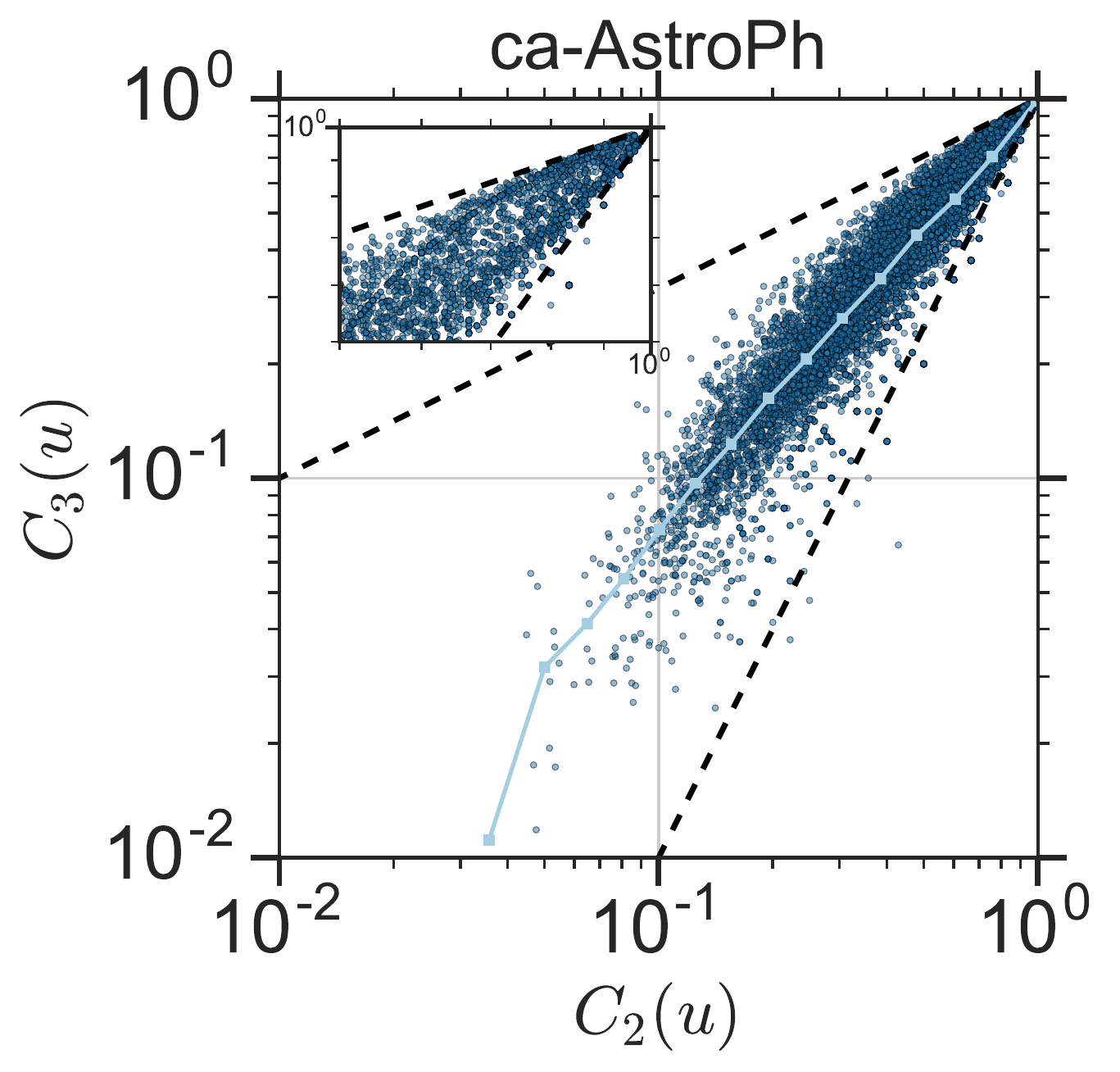} \\
  \includegraphics[width=\figwidth\columnwidth]{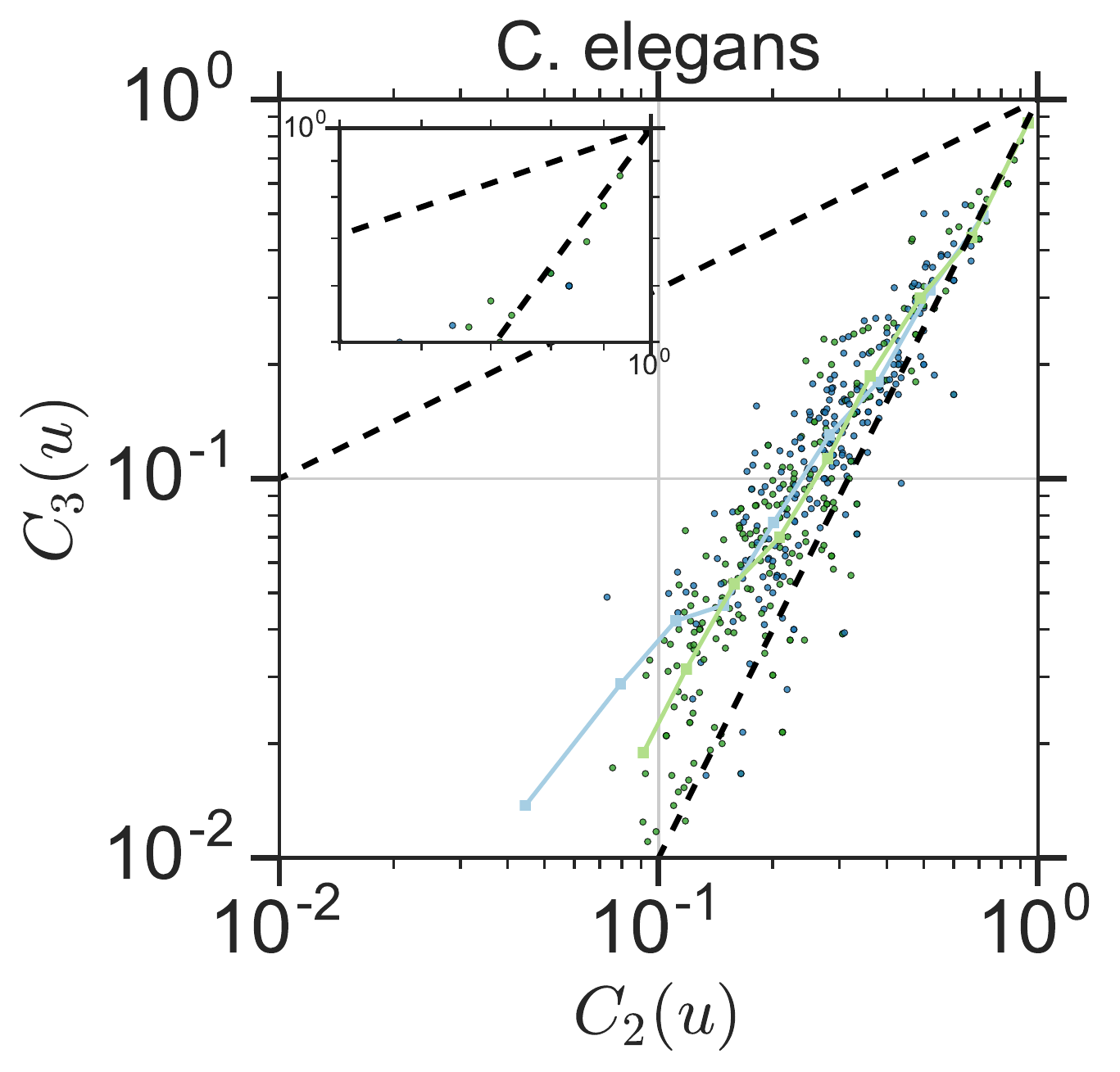}
  \includegraphics[width=\figwidth\columnwidth]{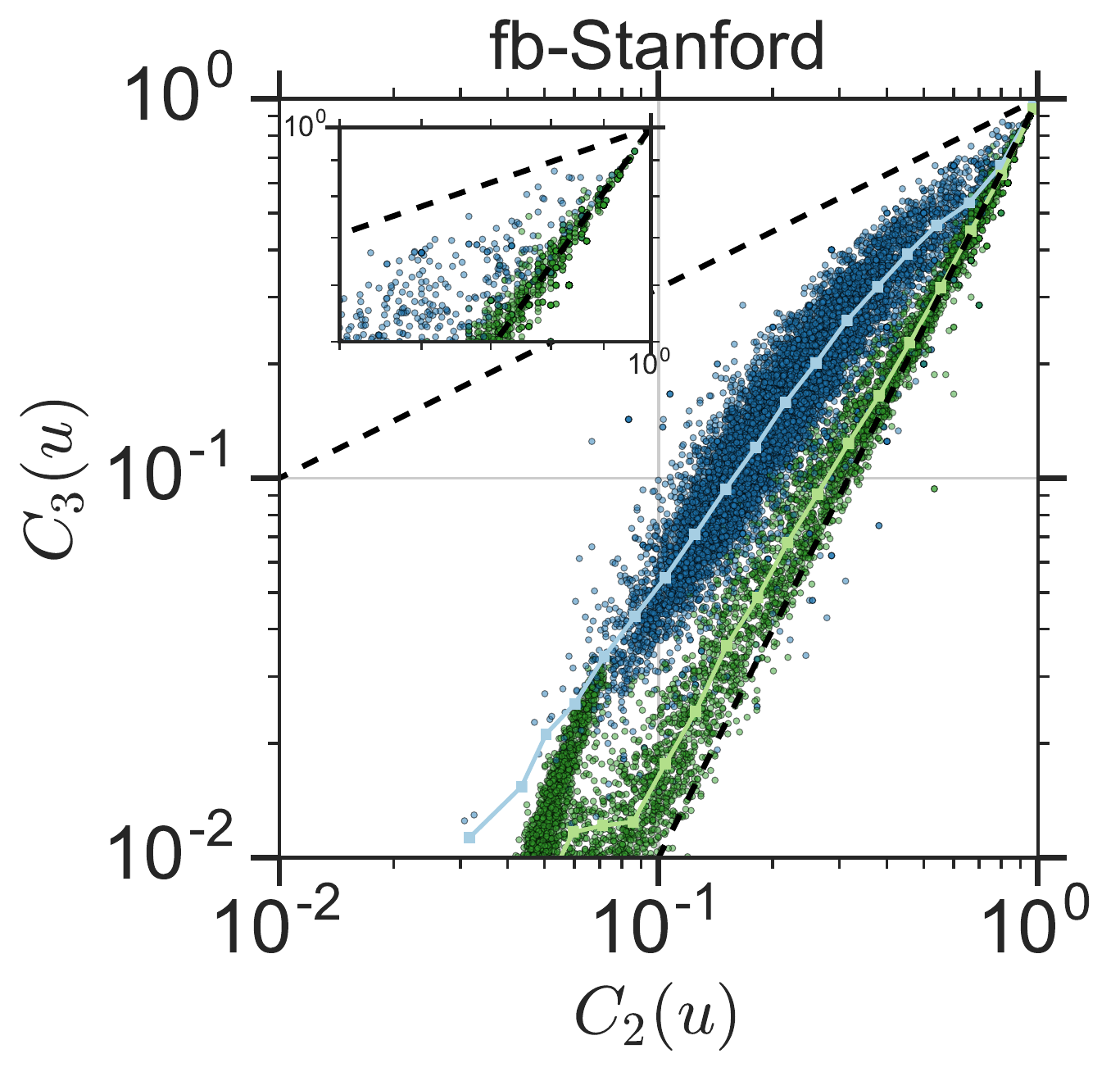}
  \includegraphics[width=\figwidth\columnwidth]{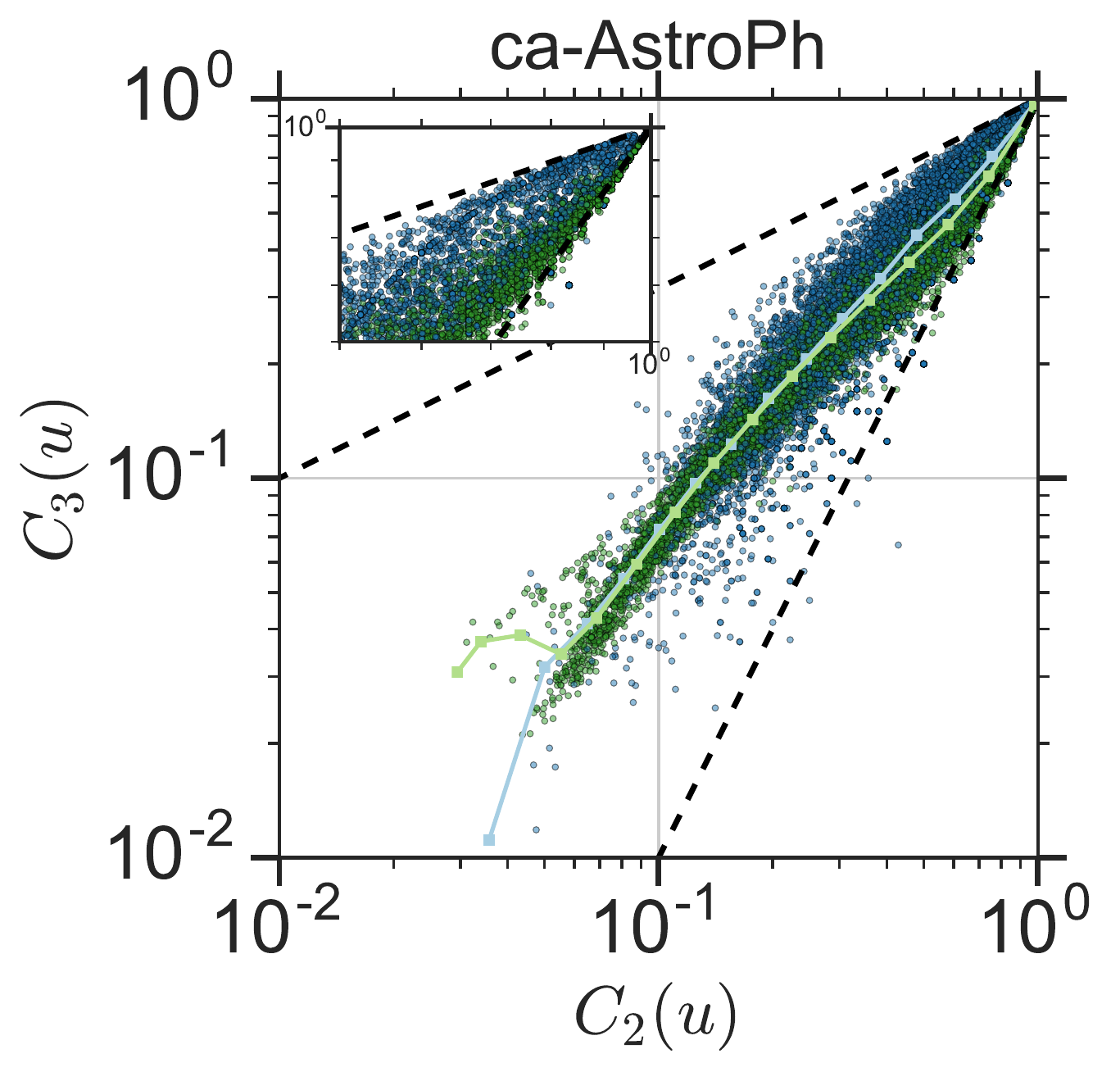}   
  \dualcaption{Joint distributions of $(\lccf{2}{u}$, $\lccf{3}{u})$}{Each
    \textcolor{myblue}{blue} dot corresponds to a node in the original network,
    and the \textcolor{mylightblue}{light blue} curve tracks the average over
    logarithmic bins.  The upper trend line is the upper bound in
    \cref{eq:Bound_Kappa3}, and the lower trend line follows the
    Erd\H{o}s-R\'enyi model where edges appear randomly (\cref{prop:ccf_er}).
    The inset enlarges the portion of the plot where $\lccf{2}{u}$ and
    $\lccf{3}{u}$ are greater than 0.5.  In the bottom row, each \textcolor{mygreen}{green} dot
    corresponds to a node in one of the MRCN random samples, and the
    \textcolor{mylightgreen}{light green} line tracks the average over
    logarithmic bins.}\label{fig:ccfs23}
\end{figure}

\clearpage

We also compute the higher-order clustering coefficient as a function of node
degree (\cref{fig:ccfs_dist}).  In the Erd\H{o}s-R\'enyi, small-world, and \emph{C. elegans}
networks, there is a distinct gap between the average higher-order clustering
coefficients for nodes of all degrees.  Thus, the observed decrease in
clustering as the order increases is independent of degree.  In the Facebook
friendship network, $\lccf{2}{u}$ is larger than $\lccf{3}{u}$ and $\lccf{4}{u}$
on average for nodes of all degrees, but $\lccf{3}{u}$ and $\lccf{4}{u}$ are
roughly the same for nodes of all degrees, which means that 4-cliques and
5-cliques close at roughly the same rate, independent of degree, albeit at a
smaller rate than traditional triadic closure.  In the co-authorship network,
nodes $u$ have roughly the same $\lccf{\ell}{u}$ for $\ell = 2$, $3$, $4$, which
means that $\ell$-cliques close at about the same rate, independent of $\ell$.
We note that the global clustering coefficient $\gccf{\ell}$ slightly increases
with $\ell$ in this network (\cref{tab:data_summary}), 
which probably means there are nodes participating in a \emph{large} clique and
also serving as the center of a star-like pattern (\cref{fig:ccf_diffs}, right),
which causes the global clustering coefficient to increase with the order.

\clearpage
\begin{figure}[tb]
  \newcommand{\figwidth}{0.32}
  \includegraphics[width=\figwidth\columnwidth]{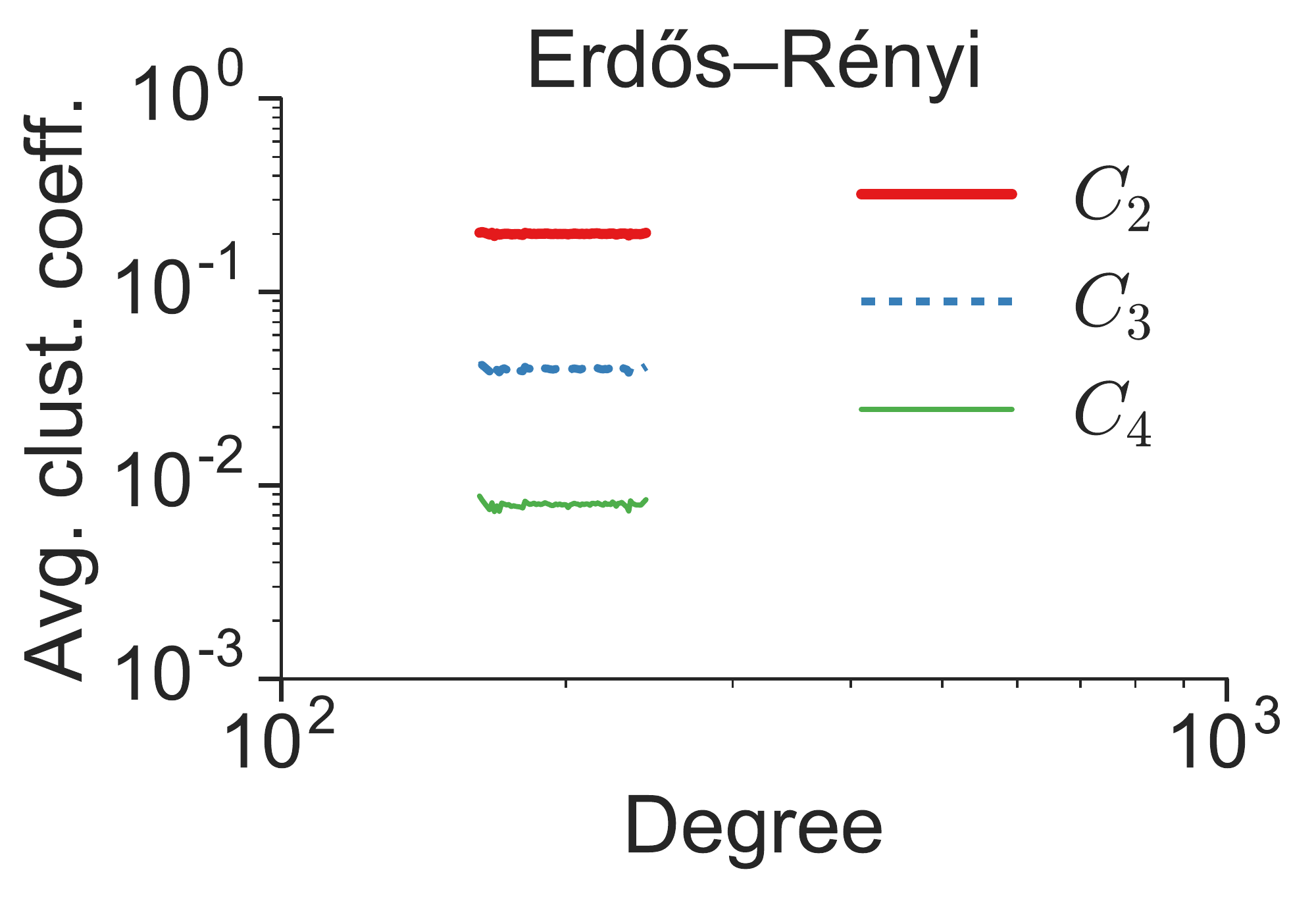}
  \includegraphics[width=\figwidth\columnwidth]{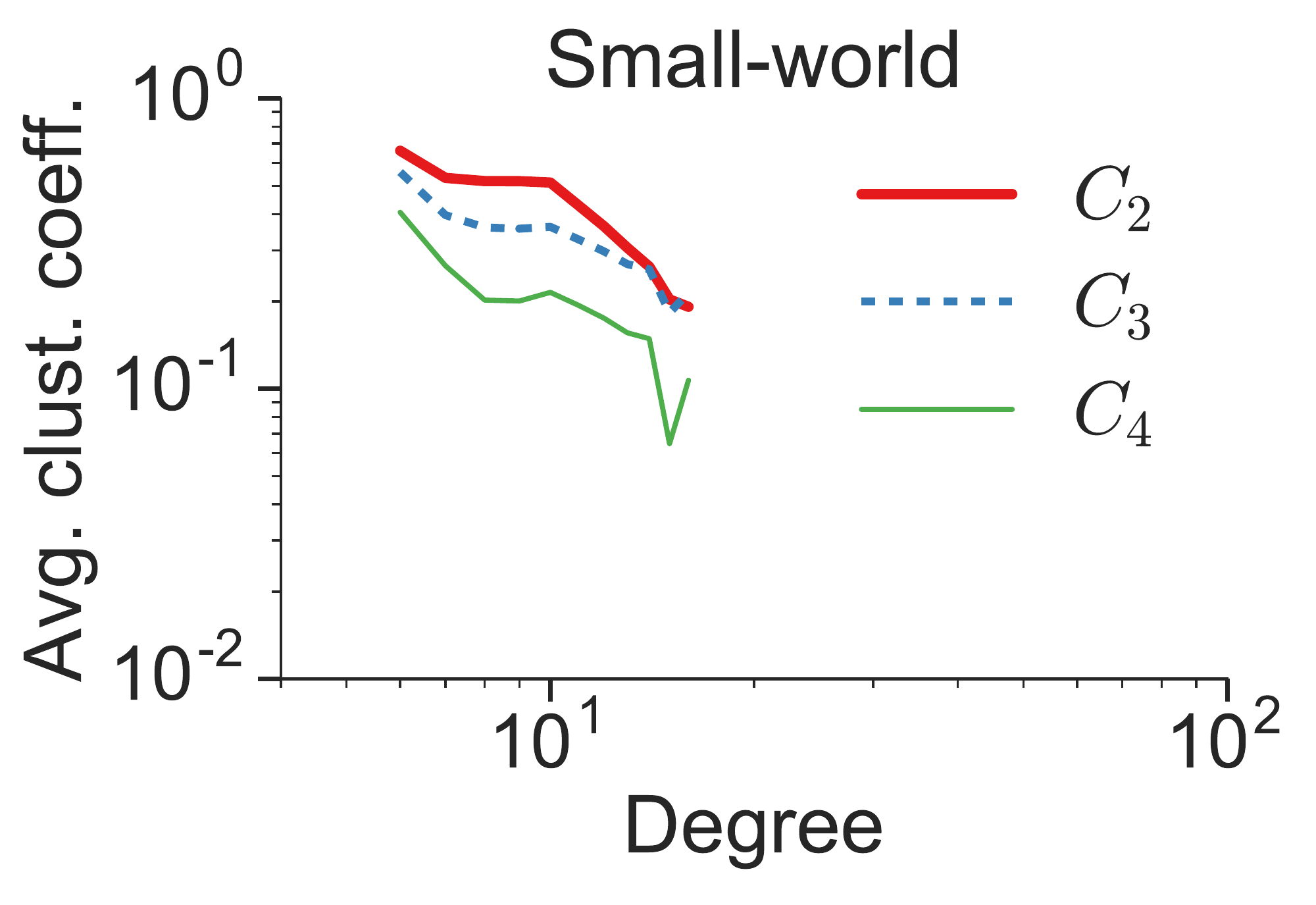}  \\
  \includegraphics[width=\figwidth\columnwidth]{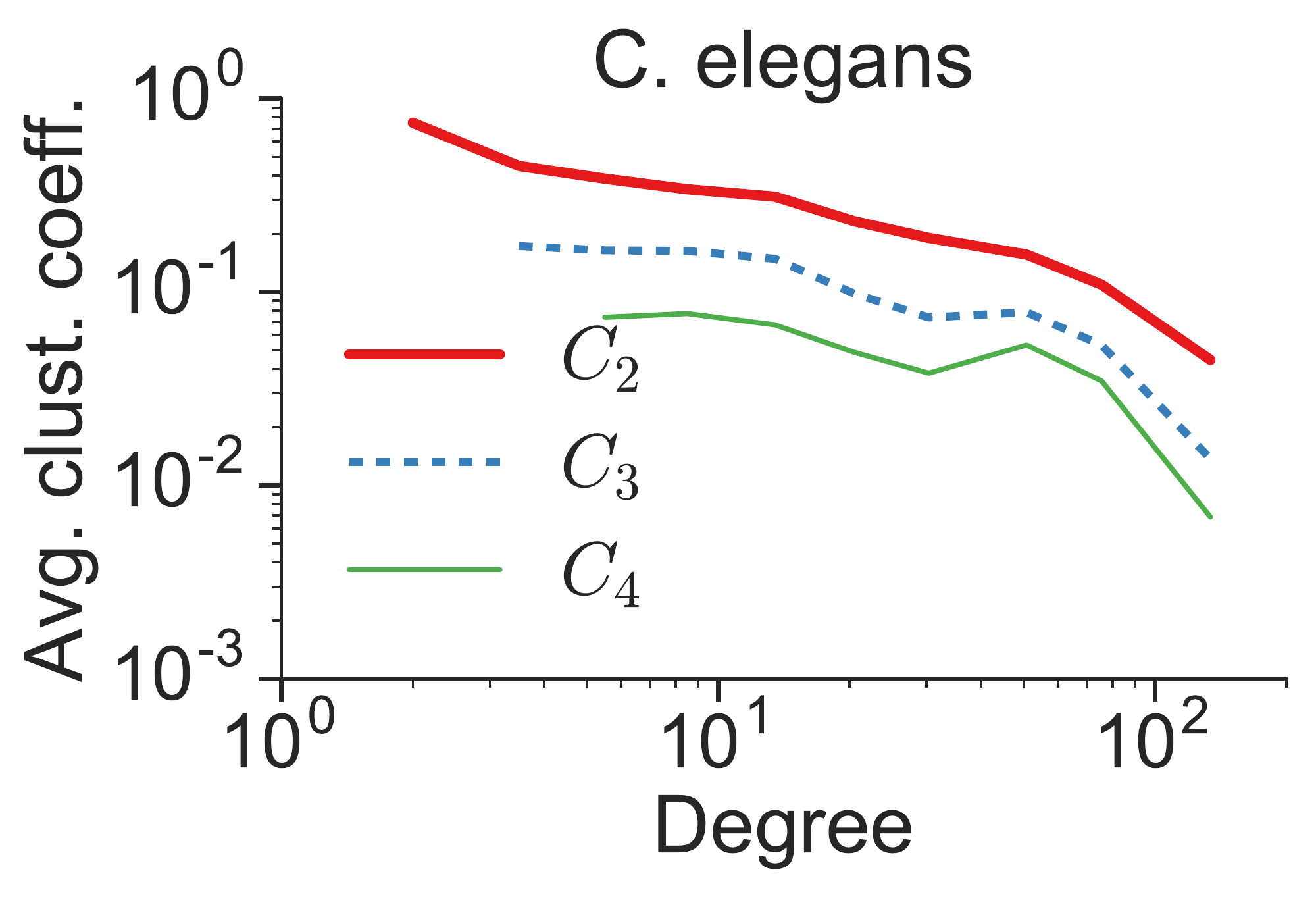} 
  \includegraphics[width=\figwidth\columnwidth]{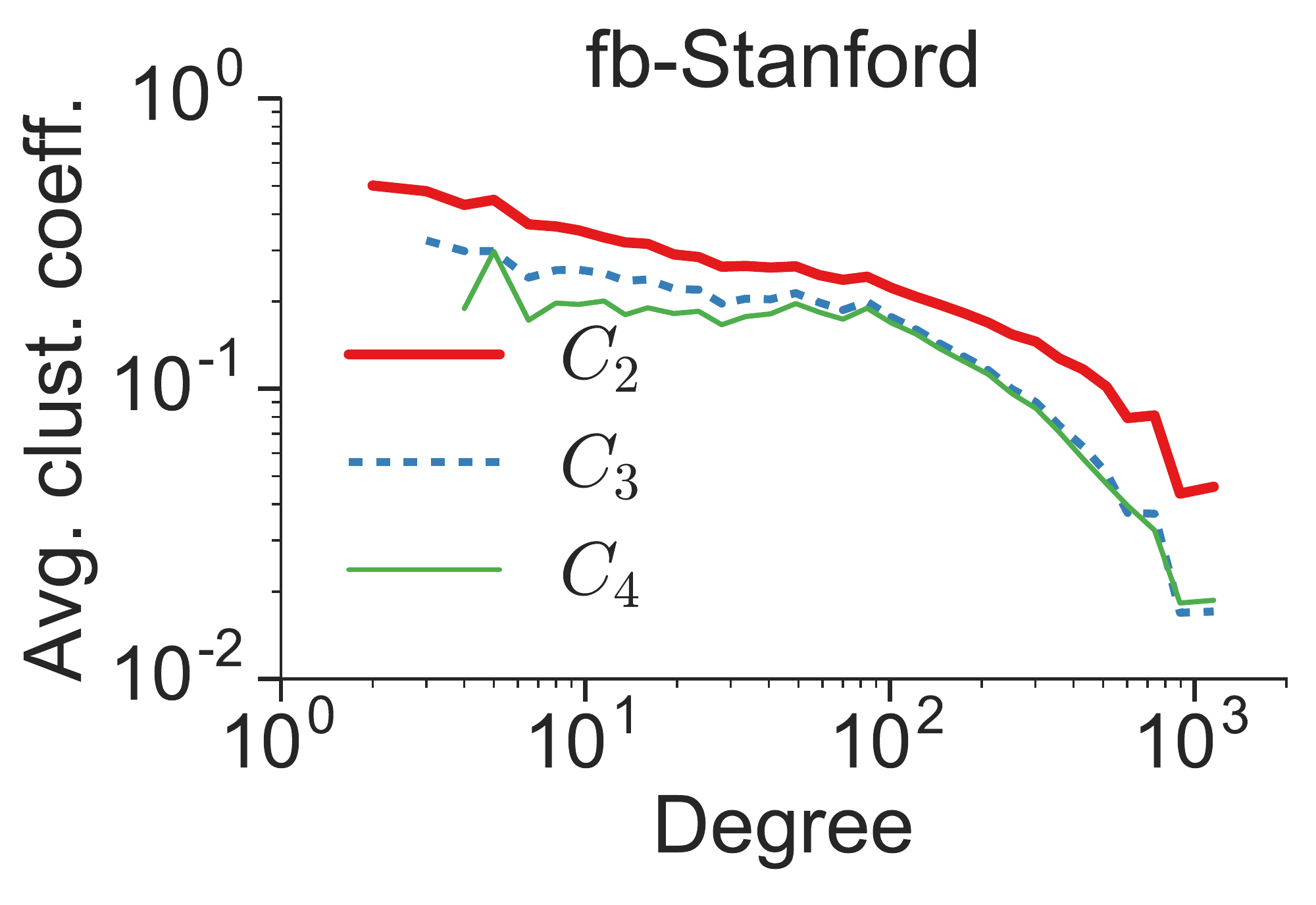}
  \includegraphics[width=\figwidth\columnwidth]{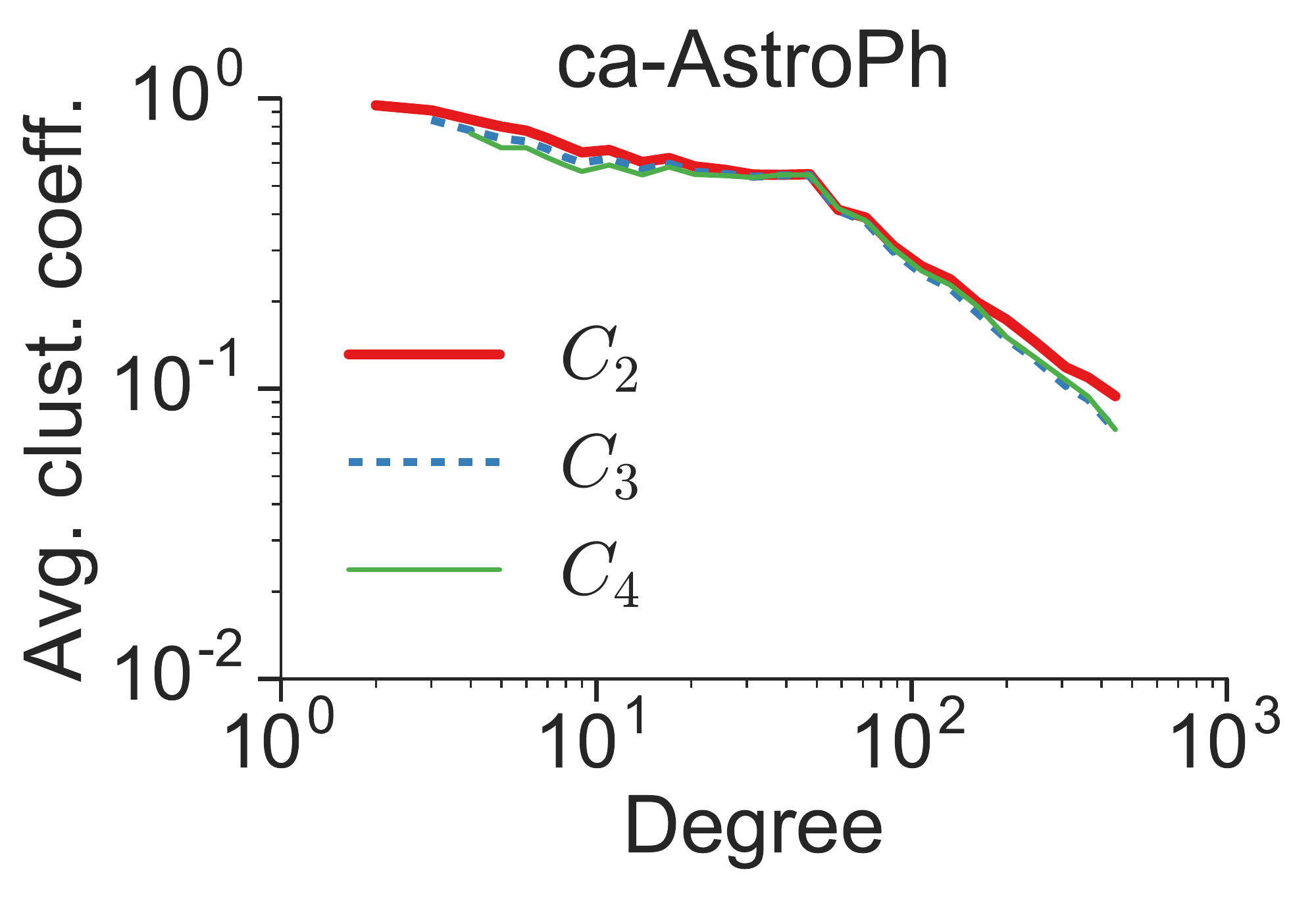}
\dualcaption{Average higher-order clustering coefficients as a function of
  degree}{}\label{fig:ccfs_dist}
\end{figure}

\section{Relating higher-order clustering coefficients to motif conductance through 1-hop neighborhoods}
\label{sec:nbrhood_cond}

In this section, we relate our higher-order clustering coefficients to motif
conductance, which was introduced in \cref{ch:honc} as a measure of how good of
a cluster a subset of nodes is, in terms of some motif.  Here we will look at
motif conductance for clique motifs.  Our main technical result
(\cref{thm:nbrhd_main}) says that if a graph has a large $\ell$th-order global
clustering coefficient, then there is a node whose 1-hop neighborhood has small
motif conductance for the $\ell$-clique motif.  This is a generalization of
similar results from \citet{gleich2012vertex} for edge conductance.  In fact,
these generalizations are what led to us to our definitions for higher-order
clustering coefficients.

\subsection{A formal relationship between 1-hop neighborhoods and motif conductance}
\sectionmark{1-hop neighborhoods and motif conductance}
\label{sec:nbrhd_thy}

We first state our main technical result of this section, which says that if the
network exhibits higher-order clustering, i.e., $\gccf{\ell}$ is large, then
there is a 1-hop neighborhood with small $\ell$-clique conductance.  For
notation, let $\onehopu$ denote the nodes in the 1-hop neighborhood graph
of node $u$ (including $u$):
\begin{equation}
\onehopu = \{v \in V \;\lvert\; (u, v) \in E\} \cup \{u\}.
\end{equation}
Above, we used the notation $\onehopnou$ to denote the graph induced by the
neighborhoods of $u$.  Here, $\onehopu$ includes node $u$ and represents a set
of vertices---$u$ and its neighbors.  Also, let $\clique{\ell}$ denote the
$\ell$-clique motif.

\begin{theorem}\label{thm:nbrhd_main}
Let $G = (V, E)$ with global $\ell$th-order clustering coefficient
$\gccf{\ell}$.  Suppose that
$\mvol{\clique{\ell}}{\onehopu} \le \mvol{\clique{\ell}}{V} / 2$ for each node
$u$.  Then there exists a node $u$ such that
\begin{align}  
\mcond{\clique{\ell}}{\onehopu}
&\leq
\frac{1 - \gccf{\ell}}{1 - \gccf{\ell} + [\gccf{\ell}  / (1 + \sqrt{1 - \gccf{\ell}})]^2 }  \label{eq:cond_upper} \\
& \leq
\min\{2(1 - \gccf{\ell}), 1\}. \label{eq:cond_upper_weaker}
\end{align}
\end{theorem}
The bound in \cref{eq:cond_upper} is monotonically decreasing, approaches $0$ as
$\gccf{\ell}$ approaches 1, and is bounded above by $1$.  The result is a
generalization of a similar statement for edge conductance
by \citet{gleich2012vertex}, but prior results contain only the case of $\ell = 2$,
and only use the weaker bound of \cref{eq:cond_upper_weaker}.

We now prove this result via several technical results relating higher-order
clustering coefficients and motif conductance, where the motif is a clique. We
note that many of the results are generalizations of previous theory developed
by \citet{gleich2012vertex}.

First, we show the extreme-case result that when $\gccf{\ell} = 1$, each
connected component of the graph is a clique.
\begin{proposition}\label{prop:full_clustering}
If a graph has global high-order clustering coefficient $\gccf{\ell} = 1$, then
each connected component of this graph is either complete or $\ell$-clique free.
\end{proposition}
\begin{proof}
We prove this by contradiction.  Suppose a connected component contains an
$\ell$-clique and that the maximum clique of this connected component is of size
$j \geq \ell$.  If this connected component is not a complete graph, there must
be a node connecting to the maximum clique but not forming a bigger clique,
giving an open $\ell$-wedge.
\end{proof}
\begin{corollary}\label{cor:zerocond}
If $\gccf{\ell} = 1$ and there are two $\ell$-cliques, then there exists a node
$u$ with $\mcond{\clique{\ell}}{\onehopu} = 0$.
\end{corollary}
\Cref{cor:zerocond} provides some intuition for the main idea.  If there is
perfect clustering, then there is a perfect cut.  \Cref{thm:nbrhd_main}
is a relaxation of this observation.  We now prove a series of technical lemmas that will be used to
prove \cref{thm:nbrhd_main}.  To begin, the following lemma relates the higher-order
clustering coefficient with neighborhood cuts.
\begin{lemma}\label{lem:sum_nbr_cut}
The sum of the clique-motif cuts of all 1-hop neighborhoods is bounded above by
a $(1 - \gccf{\ell})$ fraction of the total $\ell$-wedges:
\begin{equation}\label{eq:sum_nbr_cut}
\sum_{v \in V} \mcut{\clique{\ell}}{\onehopv} \leq (1 - \gccf{\ell}) \cdot \lvert W_\ell \rvert.
\end{equation}
\end{lemma}
\begin{proof}
If an $\ell$-clique $(u_1, \dots u_\ell)$ gets cut by the partition induced by
$\onehopv$, then $v$ must directly connect with one of $u_1, \dots, u_\ell$, say
$u_1$.  Note that the clique $(u_1, \dots u_\ell)$ and adjacent edge $(u_1, v)$
form an open $(\ell + 1)$-wedge since $v$ cannot connect to all of $u_1, \dots,
u_\ell$ (otherwise, the clique is not cut). Therefore, we have an injective map
from any cut clique on the left-hand side of the inequality to an open
$\ell$-wedge.  Thus, $\sum_{v \in V} \mcut{\clique{\ell}}{\onehopv}$ is no
greater than the number of open $\ell$-wedges, which is exactly
$(1 - \gccf{\ell}) \cdot \lvert W_{\ell} \rvert$.
\end{proof}

Next, we define a probability distribution on the nodes,
\begin{equation}\label{eq:ccf_prob_dist}
\p{\ell}{u} = \lvert W_\ell(u)\rvert / \lvert W_\ell \rvert.
\end{equation}
The following lemma creates a random variable with the probability distribution
in \cref{eq:ccf_prob_dist} whose expectation is bounded by $1- \gccf{\ell}$,
which we use in the proof of \cref{thm:nbrhd_main}.
\begin{lemma}\label{lem:expectation}
Let $X$ be a random variable that takes value
$\frac{\mcut{\clique{\ell}}{\onehopu}}{\lvert W_\ell(u) \rvert}$
with probability $\p{\ell}{u}$.  Then $\expect{X} \leq 1 - \gccf{\ell}$.
\end{lemma}
\begin{proof}
\begin{align*}
\expect{X} &= \sum_{u \in V} \p{\ell}{u} \frac{\mcut{\clique{\ell}}{\onehopu}}{\lvert W_\ell(u) \rvert} \\
&= \frac{1}{\lvert W_\ell \rvert} \sum_{u \in V} \mcut{\clique{\ell}}{\onehopu}  \\
&\leq \frac{1}{\lvert W_{\ell}\rvert}(1 - \gccf{\ell}) \cdot \lvert W_\ell \rvert \quad \text{(by \cref{lem:sum_nbr_cut})}\\
&= 1 - \gccf{\ell}.
\end{align*}
\end{proof}

\begin{lemma}\label{lem:existence}
There exists some node $u$ such that
\begin{equation}\label{eq:special_prop}
\mcut{\clique{\ell}}{\onehopu} \leq a(1 - \gccf{\ell}) \cdot \lvert W_\ell (u)\rvert \text{ and } \lccf{\ell}{u} \geq b
\end{equation}
\end{lemma}
\begin{proof}
Consider the random variable $X$ defined in \cref{lem:expectation} with
$\expect{X} \le 1 - \gccf{\ell}$.  For any constant $a > 1$, by Markov's
inequality, we have that
\begin{equation}\label{eq:markov}
\prob{X > a(1 - \gccf{\ell})} \le 1 / a.
\end{equation}
Let $b = (a \gccf{\ell} - 1) / (a - 1) \le 1$, and
$p = \prob{\lccf{\ell}{u} < b}$.
Now,
\begin{align*}
\gccf{\ell} 
&= \frac{\ell}{\lvert W_\ell \rvert}\cdot (\ell + 1) \lvert \clique{\ell + 1} \rvert \\
&= \frac{\ell}{\lvert W_\ell \rvert}\cdot \sum_{u \in V} \lvert \cliqueL{\ell + 1}{u} \rvert \\
&= \sum_{u \in V} \frac{\lvert W_\ell(u)\rvert }{ \lvert W_\ell \rvert} \cdot \frac{\ell \cdot \lvert \cliqueL{\ell + 1}{u} \rvert }{{\lvert W_\ell(u)\rvert }} \\
&= \sum_{u\in V} \p{\ell}{u} \lccf{\ell}{u}  \\
&= \sum_{\lccf{\ell}{u} < b} \p{\ell}{u} \lccf{\ell}{u} + \sum_{\lccf{\ell}{u} \geq b} \p{\ell}{u} \lccf{\ell}{u} \\
&< b\sum_{\lccf{\ell}{u} < b} \p{\ell}{u}  + \sum_{\lccf{\ell}{u} \geq b} \p{\ell}{u} \\
&= bp + (1 - p).
\end{align*}
Thus,
\begin{equation}\label{eq:bound}
 \prob{\lccf{\ell}{u} < b} = p < (1 - \gccf{\ell}) / (1 - b) = 1 - 1 / a.
\end{equation}
By the union bound with the results in \cref{eq:markov,eq:bound},
\[
1 - \prob{\dfrac{\mcut{\clique{\ell}}{\onehopu}}{\lvert W_\ell(u) \rvert} > a(1 - \gccf{\ell})
\text{ or }
\lccf{\ell}{u} < b} > 0.
\]
The result holds by the probabilistic method.
\end{proof}

We now prove a final lemma, from which \cref{thm:nbrhd_main} will be a
corollary.
\begin{lemma}\label{lem:nbrhood_conductance}
For any node $u$ satisfying the bounds in \cref{eq:special_prop},
\[
\frac{\mcut{\clique{\ell}}{\onehopu}}{\mvol{\clique{\ell}}{\onehopu}} \leq
\frac{1 - \gccf{\ell}}{1 - \gccf{\ell} + (a\gccf{\ell} - 1) / (a-1)}.
\]
\end{lemma}
\begin{proof}
We first provide a lower bound for $\mvol{\clique{\ell}}{\onehopu}$.  First, each
$\ell$-clique cut by the partition induced by $\onehopu$ contributes at least
one to $\mvol{\clique{\ell}}{\onehopu}$.  Second, each $(\ell + 1)$-clique in
$\cliqueL{\ell + 1}{u}$ uniquely corresponds to an $\ell$-clique in $\onehopu$
which is induced by the $\ell$ nodes in the $(\ell + 1)$-clique other than $u$,
thus will contribute $\ell$ into $\mvol{\clique{\ell}}{\onehopu}$.  Note that
each $(\ell + 1)$-clique in $\cliqueL{\ell + 1}{u}$ closes $\ell$ different
$\ell$-wedges in $W_\ell(u)$, and there are $\lccf{\ell}{u} \lvert
W_\ell(u)\rvert$ closed $\ell$-wedges.  Combining these observations,
\begin{align*}
\mvol{\clique{\ell}}{\onehopu} 
&\geq \mcut{\clique{\ell}}{\onehopu} + \ell \cdot \lccf{\ell}{u} \lvert W_\ell(u)\rvert / \ell \\
&\geq \mcut{\clique{\ell}}{\onehopu} + b \lvert W_\ell(u)\rvert \quad \text{(by \cref{eq:special_prop})}.
\end{align*}
Now,
\begin{align*}
\frac{\mcut{\clique{\ell}}{\onehopu}}{\mvol{\clique{\ell}}{\onehopu}}
&\leq
\frac{\mcut{\clique{\ell}}{\onehopu}}{\mcut{\clique{\ell}}{\onehopu} + b \lvert W_\ell(u)\rvert}  \\
&\leq
\frac{a(1 - \gccf{\ell}) \cdot \lvert W_\ell (u)\rvert}{a(1 - \gccf{\ell}) |W_\ell(u)| +  b |W_\ell(u)|} 
\quad \text{(by \cref{eq:special_prop})} \\
&=
\frac{1 - \gccf{\ell}}{1 - \gccf{\ell} + \frac{a \gccf{\ell} - 1}{a(a-1)} }.
\end{align*}
\end{proof}

Finally, \cref{thm:nbrhd_main} follows
from \cref{lem:existence,lem:nbrhood_conductance}, setting
\[
a = (1 + \sqrt{1-\gccf{\ell}}) / \gccf{\ell},
\]
and the fact that
$\mvol{\clique{\ell}}{\onehopu} \le \mvol{\clique{\ell}}{V} / 2$
implies that
\[
\mcond{\clique{\ell}}{\onehopu} = \frac{\mcut{\clique{\ell}}{\onehopu}}{\mvol{\clique{\ell}}{\onehopu}}.
\]

The theory in this section holds for cliques of any size, and the
results do not use the approximation to motif conductance in \cref{sec:fournode}
needed by our motif-based spectral clustering method for motifs with four or
more nodes.  Thus, our definitions of motif conductance are still meaningful for
motifs with at least 4 nodes.

\subsection{Experiments}
\label{sec:nbrhd_exps}

The goals of our experiments are
\begin{enumerate}
\item to demonstrate that there are 1-hop
neighborhood clusters of small motif conductance as a test of how
well \cref{thm:nbrhd_main} holds in practice; and \label{itm:goal1}
\item to use this idea to
quickly find many clusters with small motif conductance by running targeted
cluster expansion around a subset of the 1-hop neighborhood clusters. \label{itm:goal2}
\end{enumerate}
Regarding \cref{itm:goal1}, we find that real-world networks exhibit much better
results than predicted by the theory and the 1-hop neighborhood with minimal
motif conductance is competitive with spectral graph theory
approaches. Regarding \cref{itm:goal2}, we show that a subset of 1-hop
neighborhoods called locally minimal neighborhoods are better seeds than random
nodes. We use this insight to find the global structure of clique conductance
clusters more quickly than exhaustive enumeration.

We evaluate 1-hop neighborhood cluster quality in terms of motif conductance for
$2$-clique (edge), $3$-clique (triangle), and $4$-clique motifs using four
networks where we can exhaustively sample targeted clusters easily: $\condmat$,
a co-authorship network constructed from papers posted to the condensed matter
category on arXiv~\cite{leskovec2007graph}; $\harvard$, a snapshot of the
friendships network between Harvard students on Facebook in September
2005~\cite{traud2012social}; $\enron$, an e-mail communication network of the
employees of Enron Corporation and their contacts~\cite{klimt2004introducing};
and $\google$, a Web graph released by Google for a programming
contest~\cite{leskovec2009community}.  Summary statistics for the networks are
in \cref{tab:nbrhood_data_summary}.


\begin{table}[tb]
\centering
  \dualcaption{Summary statistics of datasets}{The $\enron$ and $\google$ networks
  are treated as undirected, even though the original datasets are directed.}
  \begin{tabular}{l @{\hskip 5mm} l l @{\hskip 6mm} c c c}
    \toprule
    Dataset       & $\lvert V \rvert$ & $\lvert E \rvert$ & $\gccf{2}$ & $\gccf{3}$ & $\gccf{4}$ \\
    \midrule
$\condmat$ & 14,788 & 187,278 & 0.33 & 0.33 & 0.36 \\ 
$\enron$ & 18,561 & 156,139 & 0.11 & 0.06 & 0.05 \\
$\harvard$ & 13,319 & 793,410 & 0.14 & 0.07 & 0.07\\
$\google$ & 393,582 & 2,905,337 & 0.07 & 0.06 & 0.07 \\
\bottomrule
\end{tabular}
\label{tab:nbrhood_data_summary}
\end{table}

\xhdr{There are 1-hop neighborhoods of low motif conductance}
Plugging the higher-order clustering coefficients from
\cref{tab:nbrhood_data_summary} into the bound from \cref{thm:nbrhd_main} yields
weak, albeit non-trivial bounds on the smallest neighborhood conductance (all
bounds are $\ge 0.9$ for the networks we consider).  However, the spirit of the
theorem rather than the bound itself motivates our experiments: with large
higher-order clustering, there should be a neighborhood with small motif
conductance for clique motifs.  We indeed find this to be true in our results.

\Cref{tab:nbrhood_exp_summary} compares the neighborhood with smallest motif
conductance for the $2$-clique, $3$-clique, and $4$-clique motifs with the
Fiedler cluster described in \cref{alg:motif_fiedler} in \cref{ch:honc}.  (For $4$-cliques, we actually
use the quadratic approximation from \cref{eqn:motif4_quadratic} because it
is easier to compute.)  Here, the Fiedler cluster
represents a method that uses the global structure of the network to compare
against the purely local neighborhood clusters.  In all cases, the best neighborhood
cluster has motif conductance far below the upper bound of
\cref{thm:nbrhd_main}.  For all clique orders, the best neighborhood cluster
always has conductance within a factor of 3.5 of the Fiedler cluster in
$\condmat$, $\enron$, and $\harvard$.  With $\google$, the conductances are much
smaller but the best neighborhood still has conductance within an order of
magnitude of the Fiedler set.  We conclude that the best neighborhood cluster in
terms of motif conductance, which comes from purely local constructs, is
competitive with the Fiedler vector that takes into account the global graph
structure.  This motivates our next set of experiments that uses nodes that
induce small neighborhood conductance as seeds for the APPR method developed in
\cref{sec:local}.

\begin{table}[t]
  \dualcaption{Fiedler and 1-hop neighborhood clusters}{We compare the cluster
    found by the spectral sweep cut algorithm (\cref{alg:motif_fiedler}; Fiedler) 
    to the 1-hop neighborhood with
    smallest motif conductance.  A star ($\star$) denotes when 1-hop
    neighborhood clusters are the same across different clique sizes, a dagger
    ($\dagger$) denotes when Fiedler clusters are the same, and a bullet
    ($\bullet$) denotes when the neighborhood and Fiedler clusters are the same.
    In $\condmat$, $\enron$, and $\harvard$, the best neighborhood cluster
    always has conductance within a factor of 3.5 of the Fiedler cluster. \\
    ($^*$) We report the quadratic form approximation to motif conductance
    for $4$-cliques (\cref{eqn:motif4_quadratic}) because it is easier to compute.}
\begin{tabular}{l c c c c}
\toprule
                         & $\condmat$            & $\enron$              & $\harvard$            & $\google$             \\ \midrule
$M$ = $2$-cliques                                                                                                \\ \cmidrule(r){1-1}
Neighborhood                                                                                                             \\
\phantom{XX}  \# nodes   & 23$^{\bullet}$        & 21$^{\star}$          & 4                     & 203                   \\
\phantom{XX} motif cond. & $1.1\!\cdot\!10^{-2}$ & $2.4\!\cdot\!10^{-2}$ & $2.0\!\cdot\!10^{-1}$ & $1.5\!\cdot\!10^{-3}$ \\
Fiedler                                                                                                                  \\
\phantom{XX} \# nodes    & 23$^{\bullet}$        & 25$^{\dagger}$        & 1,470                 & 5,335                 \\
\phantom{XX} motif cond. & $1.1\!\cdot\!10^{-2}$ & $1.7\!\cdot\!10^{-2}$ & $1.1\!\cdot\!10^{-1}$ & $9.2\!\cdot\!10^{-5}$ \\ \midrule 
$M$ = $3$-cliques                                                                                                \\ \cmidrule(r){1-1}
Neighborhood                                                                                                             \\
\phantom{XX} \# nodes    & 17$^{\star}$          & 21$^{\star}$          & 5$^{\star}$           & 62$^{\star}$          \\
\phantom{XX} motif cond. & $2.0\!\cdot\!10^{-3}$ & $1.3\!\cdot\!10^{-2}$ & $1.4\!\cdot\!10^{-1}$ & $1.8\!\cdot\!10^{-4}$ \\
Fiedler                                                                                                                  \\
\phantom{XX} \# nodes    & 18$^{\dagger}$        & 25$^{\dagger}$        & 1,429                 & 10,803                \\
\phantom{XX} motif cond. & $1.5\!\cdot\!10^{-3}$ & $6.7\!\cdot\!10^{-3}$ & $4.9\!\cdot\!10^{-2}$ & $5.9\!\cdot\!10^{-5}$ \\ \midrule       
$M$ = $4$-cliques                                                                                                \\ \cmidrule(r){1-1}
Neighborhood                                                                                                             \\
\phantom{XX} \# nodes    & 17$^{\star}$          & 8                     & 5$^{\star}$           & 62$^{\star}$          \\
\phantom{XX} motif cond.$^*$ & $1.4\!\cdot\!10^{-4}$ & $3.6\!\cdot\!10^{-3}$ & $7.7\!\cdot\!10^{-2}$ & $1.1\!\cdot\!10^{-5}$ \\
Fiedler                                                                                                                  \\
\phantom{XX} \# nodes    & 18$^{\dagger}$        & 25$^{\dagger}$        & 1,299                 & 8,001                 \\
\phantom{XX} motif cond.$^*$ & $1.1\!\cdot\!10^{-4}$ & $1.6\!\cdot\!10^{-3}$ & $2.2\!\cdot\!10^{-2}$ & $2.2\!\cdot\!10^{-6}$ \\ \bottomrule
\end{tabular}\label{tab:nbrhood_exp_summary}
\end{table}

\clearpage

\xhdr{Finding good seeds}
So far, we have used our theory to find a single node whose 1-hop neighborhood
has small motif conductance for clique motifs.  We examine this further by using
nodes whose neighborhoods induce good cuts as seeds for the motif-based APPR
method.  Following the terminology of \citet{gleich2012vertex},
we say that a node $u$ is a \emph{local minimum} if
$\mcond{M}{\onehopu} \leq \mcond{M}{\onehopv}$ for all neighbors $v$ of
$u$.  To test whether local minima are good seeds for
APPR, we first exhaustively compute APPR clusters using every node in each of
our networks as a seed.  Next, we used a one-sided Mann Whitney U test to test
the null hypothesis that the local minima yield APPR clusters with motif
conductances that are \emph{not less than} motif conductances from using
non-local minima as seeds (\cref{tab:mwu_test}).  The $p$-values from these
tests say that we can safely reject the null hypothesis at significance level
$<$ 0.003 for all cliques and networks considered except for $2$-cliques in
$\condmat$.  In other words, local minima are better seeds than non-local
minima.


\begin{table}[t]
\centering
  \dualcaption{Testing neighborhood centers as seeds}{We report the $p$-values
    from Mann-Whitney U tests of the null hypothesis that the motif conductances
    of sets from APPR seeded with local minima are \emph{not less than} the
    motif conductances of sets from APPR seeded with non-local minima.  In all
    but $\condmat$ with the standard edge motif, we reject the null at a
    significance level $<$ 0.003.}
  \begin{tabular}{c c c c c}
    \toprule
 motif   & $\condmat$ & $\enron$ & $\harvard$ & $\google$ \\ \midrule
edge  & 0.87 & $< 1\!\cdot\!10^{-16}$ & $8.17\!\cdot\!10^{-05}$ &  $< 1\!\cdot\!10^{-16}$ \\
triangle   & $2.07\!\cdot\!10^{-03}$ & $< 1\!\cdot\!10^{-16}$ & $4.32\!\cdot\!10^{-04}$ & $< 1\!\cdot\!10^{-16}$ \\
$4$-clique & $7.20\!\cdot\!10^{-10}$ & $< 1\!\cdot\!10^{-16}$ & $1.55\!\cdot\!10^{-05}$ & $< 1\!\cdot\!10^{-16}$ \\
\bottomrule
\end{tabular}
\label{tab:mwu_test}
\end{table}

Finally, we use these local minimum seeds to construct network community profile
(NCP) plots for different motifs.  NCP plots are defined as the optimal
conductance over all sets of a fixed size $k$ as a function of
$k$~\cite{leskovec2009community}.  The shapes of the curves reveal the cluster
structure of the networks.  In practice, these plots are generated by
exhaustively using every node in the network as a seed for the APPR
method~\cite{leskovec2009community}.  Here, we compare this approach to two
simpler ones: (i) using the neighborhood sizes and conductances and (ii) using
only local minima as seeds for APPR.  In the latter case, between 1\% and 15\%
of nodes are local minima, depending on the network, so this serves as an
economical alternative to the typical exhaustive approach.

\Cref{fig:NCPs} shows the NCP plots for $\condmat$ and $\harvard$ with the
triangle and $4$-clique motifs.  Seeding with local minima is sufficient for
capturing the major trends of the NCP plot.  In general, the curves constructed
from neighborhood information capture the first downward spike in the plot, but
do not capture larger sets with small conductance.  Finally, the triangle and
$4$-clique NCP plots are quite similar for both networks.  Thus, we suspect that
local minima for lower-order cliques could also be used as good seeds when
looking for sets based on higher-order cliques.


\definecolor{myblue}{RGB}{55,126,184}
\definecolor{mypurp}{RGB}{152,78,163}
\definecolor{mygreen}{RGB}{77,175,74}
\definecolor{myred}{RGB}{228,26,28}
\begin{figure}[tb]
\centering
\includegraphics[width=0.49\columnwidth]{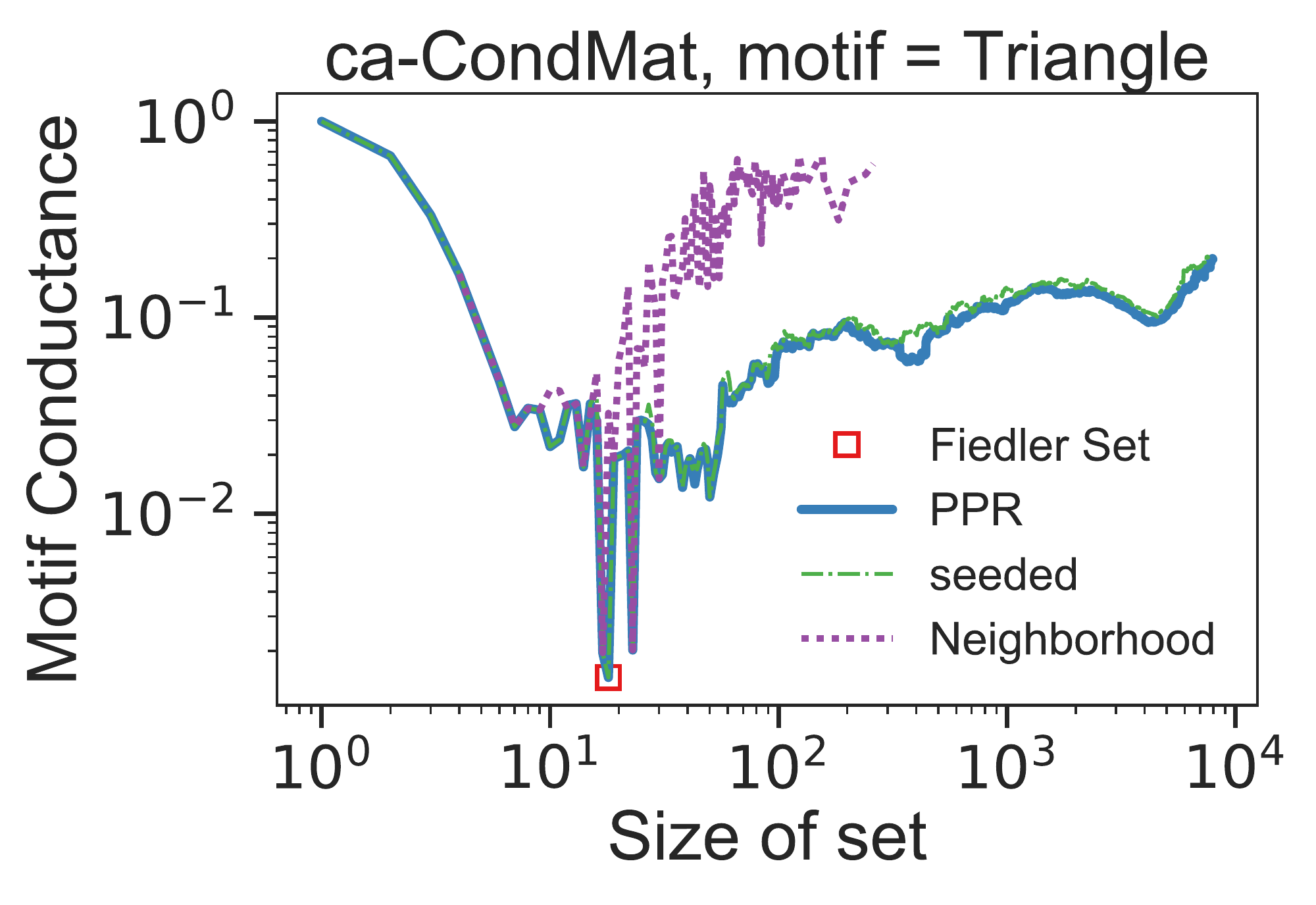}
\includegraphics[width=0.49\columnwidth]{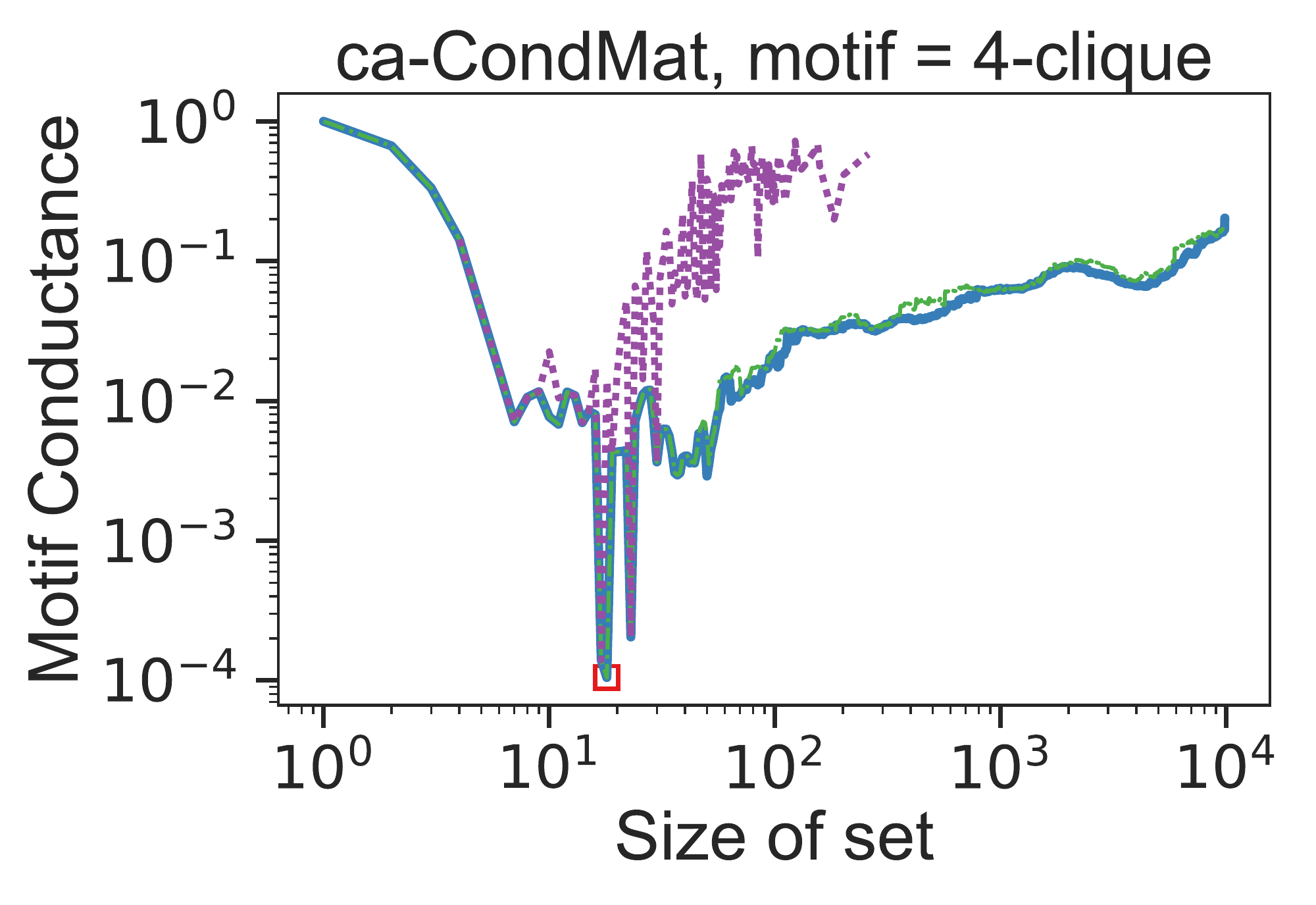}
\includegraphics[width=0.49\columnwidth]{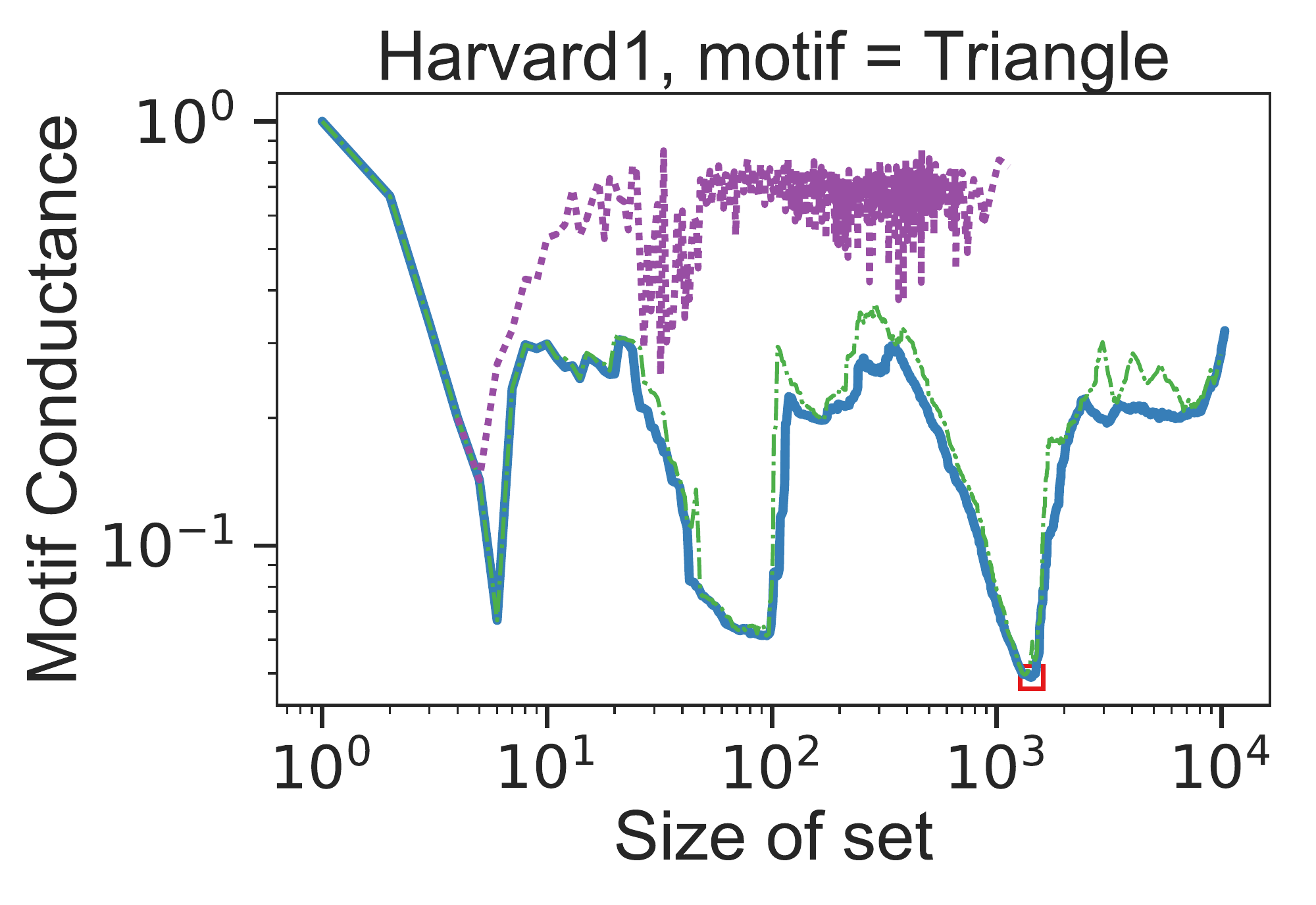}
\includegraphics[width=0.49\columnwidth]{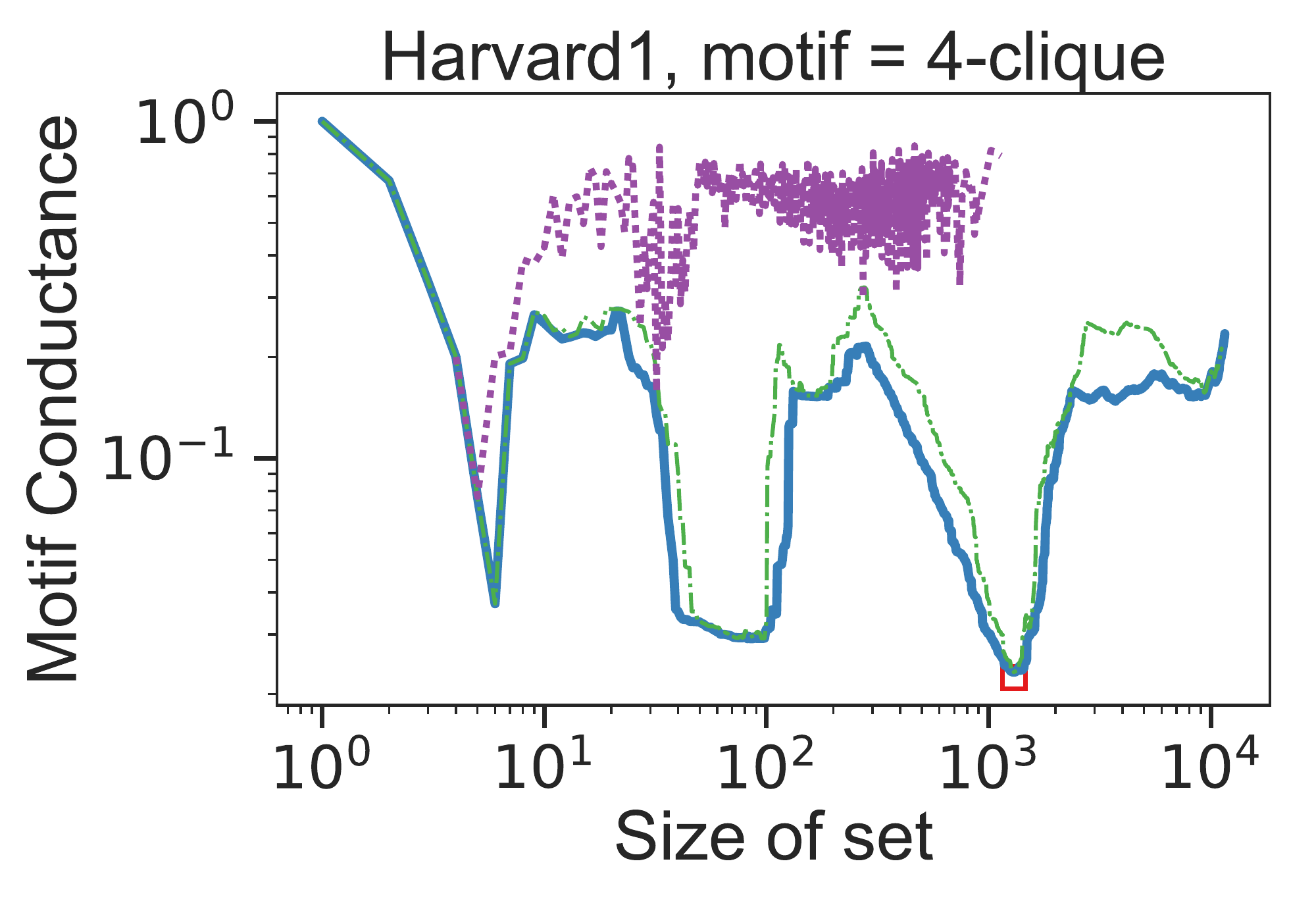}
\dualcaption{NCP plots for two networks with two different clique sizes}{Curves
  are constructed from APPR with all nodes as seeds (\textcolor{myblue}{blue}), APPR with just local
  minima as seeds (``seeded'', \textcolor{mygreen}{green}), and all 1-hop neighborhoods
  (``Neighborhood'', \textcolor{mypurp}{purple}).  Using local minima as seeds captures the trends of
  exhaustive PPR using only a fraction of the seeds.  A local minima seed also captures
  the Fiedler set (marked by \textcolor{myred}{$\square$}).}
\label{fig:NCPs}
\end{figure}

\clearpage

\section{Related work and discussion}

We have proposed a methodology for higher-order clustering coefficients to study
higher-order closure patterns in networks, which generalizes the widely used
clustering coefficient that measures triadic closure.  Our methodology gives new
insights into the clustering behavior of both real-world networks and random
graph models, and our theoretical analysis provides intuition for the way in
which higher-order clustering coefficients describe local clustering in graphs.

There have been a few prior efforts along these
lines.  \Citet{fronczak2002higher} define an $\ell$th-order local clustering
coefficient at node $u$ to be the fraction of neighbors of $u$ that have
shortest path distance $\ell$ in the 1-hop neighborhood graph of $u$ with $u$
removed (i.e., in $\onehopnou$).  In this definition, $\ell = 1$ corresponds to
the traditional local clustering coefficient.  We note that this generalization
does not have a natural global clustering
coefficient.  \Citet{andrade2006neighborhood} also look at shortest paths of
length $\ell$ in an effort to generalize graph measures and mention the
clustering coefficient, although they do not actually compute it.
\Citet{caldarelli2004structure} define the \emph{grid coefficient}
at node $u$ as the number of 4-cycles containing $u$ divided by the total
possible number of 4-cycles $u$ could participate in, given the degrees of $u$
and its neighbors.  This is a generalization of the classical clustering coefficient,
which measures the number of 3-cycles containing $u$ divided by the total
possible number of 3-cycles $u$ could participate in, given the degree of $u$.
\Citet{jiang2004topological} define an $\ell$th-order local clustering coefficient
of node $u$ as the edge density in the $\ell$-hop neighborhood of node $u$
with $u$ removed.  In this case, $\ell = 1$
is the classical clustering coefficient.  

None of the prior work captures the closure patterns of higher-order cliques,
which fits with the notion of triadic closure in sociology.  The prior
generalizations also do not discuss algorithms for the computations of the
proposed measurements.  For example, finding shortest paths in all neighborhood
graphs $\onehopnou$ could be expensive.  A nice property of our generalization
is that we only have to enumerate $\ell$-cliques and $(\ell - 1)$-cliques.

There are also a few other approaches to finding higher-order structure through
generalizations or relaxations of some notion of density.  These include the
$k$-clique densest subgraph problem~\cite{tsourakakis2015k}, triangle
cores~\cite{zhang2012extracting,rossi2013triangle}, trusses (relaxations of
cliques)~\cite{cohen2008trusses}, other generalizations of cores and
trusses~\cite{sariyuce2015finding}, and the $k$-plex~\cite{seidman1978graph}.
We view these analyses as complimentary to our own.  In contrast to other work, our
analysis in \cref{sec:ccfs_empirical}, shows how \emph{combining} existing tools
(the classical clustering coefficient) with new higher-order tools (the
higher-order clustering coefficient) can provide new analysis
(e.g., \cref{fig:ccfs23}).  We hope that this guides future research in the
area.

\chapter{\tmtitle}
\label{ch:tm}
\chaptermark{Motifs in temporal networks}

\section{Analyzing network data with timestamped edges}
\label{sec:tm_introduction}

Thus far, we have used networks to describe static relationships between things
(nodes) and connections between things (edges).  However, many systems are not
static as the links between objects dynamically change over
time~\cite{holme2012temporal}.  Such {\em temporal networks} can be represented
by a series of timestamped edges, or \emph{temporal edges}.  For example, a
network of email or instant message communication can be represented as a
sequence of timestamped directed edges, one for every message that is sent from
one person to another.  Similar representations can be used to model computer
networks, phone calls, financial transactions, and biological signaling
networks.  The goal of this chapter is to bring higher-order analysis to this
type of data---in particular, to develop motifs as a tools for studying temporal
network data.

Data in these systems often have two important properties:
\begin{enumerate}
\item There are often many temporal edges between the same set of nodes.
For example, two individuals may text back and forth several times. \label{itm:tprop1}
\item The time resolution of the edges is small compared to the time span
encompassed by the entire dataset.  In the $\emaileu$ dataset studied later,
we know the exact second that an e-mail is sent, and we have e-mails sent
over more than 2 years apart. \label{itm:tprop2}
\end{enumerate}

While temporal networks are ubiquitous, there are few tools for modeling and
characterizing the underlying structure of such datasets.  Current analytic
methods of temporal network data can be characterized by a few broad categories.
First, there are methods and models that study networks as if they are strictly
growing over time, where a pair of nodes connect once and stay connected
forever~\cite{barabasi1999emergence,jacobs2015assembling,leskovec2007graph,leskovec2008microscopic,lattanzi2009affiliation}.
Such methods fail to address \cref{itm:tprop1}.  The second approach is to
aggregate the temporal network into a sequence of static
snapshots~\cite{araujo2014com2,dunlavy2011temporal,tantipathananandh2007framework}.
These approaches throw away much of the rich, fine-grained temporal information
in the data (\cref{itm:tprop2}).  Third, there are models and analysis that
treat time as continuous.  Examples include studies on macroscopic
characteristics of the data, such as the distribution of the inter-event times
between a pair of nodes or of an individual in a
network~\cite{barabasi2005origin,malmgren2009universality}, and models for
information diffusion~\cite{farajtabar2015coevolve,du2015time}.  Work in this
area has not focused on measuring the interaction of time and structure in
temporal networks, although such joint analysis has been considered
with, e.g., temporal strong components~\cite{bhadra2003complexity,nicosia2012components,rossi2015parallel}.

\begin{figure}[t]
\centering
\phantomsubfigure{fig:intro1A}
\phantomsubfigure{fig:intro1B}
\phantomsubfigure{fig:intro1C}
\includegraphics[width=0.85\columnwidth]{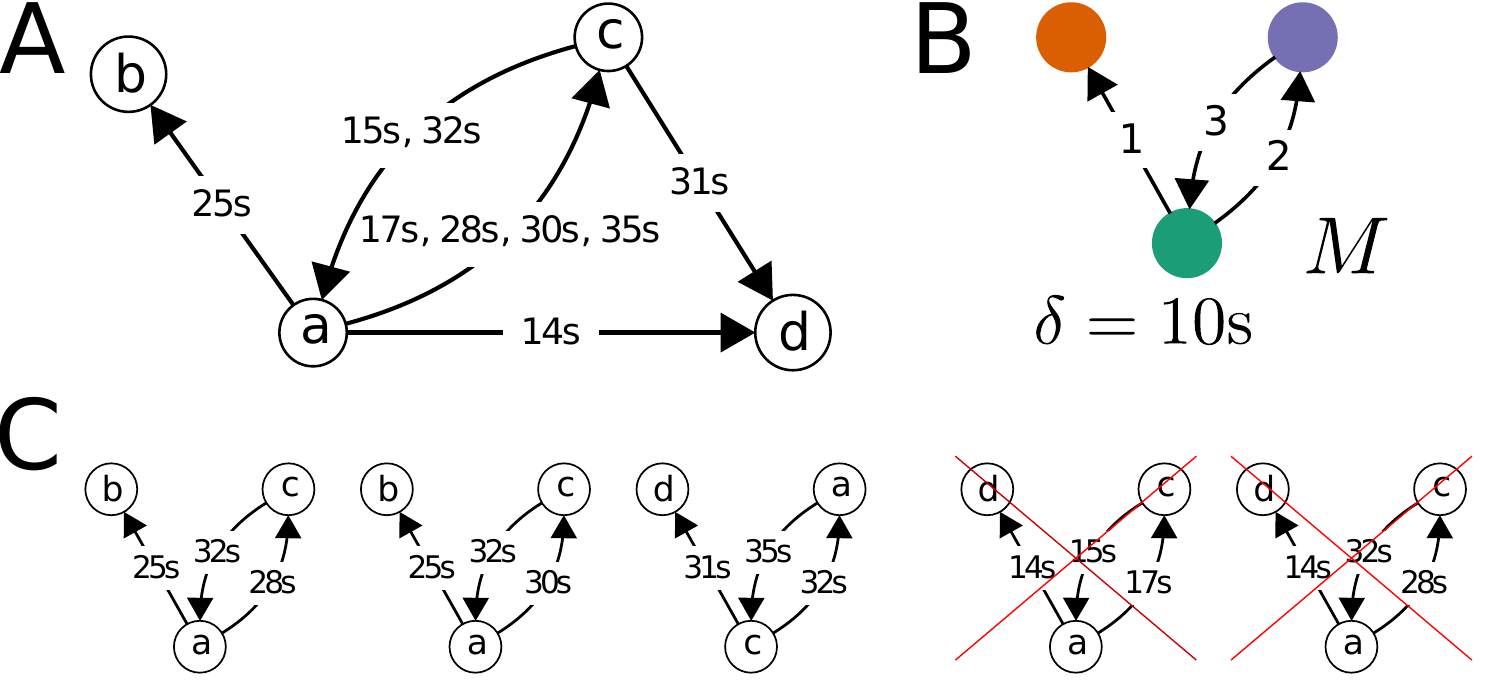}
\dualcaption{Temporal graphs and $\delta$-temporal motifs}{
\textnormal{A}: A temporal graph with nine temporal edges.  Each edge
has a timestamp (listed here in seconds).
\textnormal{B}: Example $3$-node, $3$-edge, $\delta$-temporal motif $M$.
The edge labels correspond to the ordering of the edges.
The time duration of the motif is $\delta = 10$ seconds.
\textnormal{C}:
Instances of the $\delta$-temporal motif $M$.
The crossed-out patterns are not instances of $M$ because either the
edge sequence is out of order or the edges do not all occur within the time
window $\delta$.
}
\label{fig:intro1}
\end{figure}

Our goal in this chapter to develop motifs as a tool for jointly analyzing time
and structure in temporal networks.  We extend static motifs to temporal
networks and define $\delta$-temporal motifs, where all the edges in a given
motif $M$ have to occur inside the time period of $\delta$ time units. These
$\delta$-temporal motifs simultaneously account for ordering of edges and a
temporal window in which edges can occur.  For example,
\cref{fig:intro1B} shows a motif on three nodes and three edges, where the
edge label denotes the order in which the edges appear.  While we focus on
directed edges with a single timestamp in this work, our methodology seamlessly
generalizes to common variations on this model.  For example, our methods can
incorporate timestamps with durations (common in telephone call networks),
colored edges that identify different types of connections, and temporal
networks with undirected edges.

We then consider the problem of counting how many times does each
$\delta$-temporal motif occur in a given temporal network.  We develop a general
algorithm for counting temporal network motifs defined by any number of nodes
and edges that avoids enumeration over subsets of temporal edges and whose
complexity depends on the structure of the static graph induced by the temporal
motif.  For motifs defined by a constant number of temporal edges between $2$
nodes, this general algorithm is optimal up to constant factors---it runs in
$O(m)$ time, where $m$ is the number of temporal edges.

Furthermore, we design fast variations of the algorithm that allow for counting
certain classes of $\delta$-temporal motifs including star and triangle
patterns.  These algorithms are based on a common framework for managing summary
counts in specified time windows.  For star motifs with $3$ nodes and $3$
temporal edges, we again achieve a running time linear in the input, i.e.,
$O(m)$ time. Given a temporal graph with $n$ nodes and $\tau$ induced triangles
in its induced static graph, our fast algorithm counts temporal triangle motifs
with $3$ temporal edges in $O(\min(m\tau^{1/2}, mn))$ worst-case time, after
the $\tau$ triangles have been identified.  In
contrast, any algorithm that processes triangles individually takes $O(\min(m\tau, mn))$
worst-case time.  In practice, our fast temporal triangle counting
algorithm is up to 56 times faster than a competitive baseline and runs in just
a couple of hours on a network with over two billion temporal edges.

Our algorithmic framework enables us to study the structure of several complex
systems.  For example, we explore the differences in human communication
patterns by analyzing motif frequencies in text message, Facebook wall post,
email and private online message network datasets.  Temporal network motif
counts reveal that text messaging and Facebook wall posting are dominated by
one-to-one communication, where a user only engages with one other user at a
time, whereas email is mostly characterized by one-to-many or broadcast
patterns as individuals send out several emails in a row.

Temporal network motifs can also be used to measure the frequency of patterns at
different time scales.  For example, the difference in $\delta$-temporal motif
counts for $\delta = 60$ minutes and $\delta = 30$ minutes counts only the
motifs that take at least 30 minutes and at most 60 minutes to form.  With this
type of analysis, we find that certain question-and-answer patterns on Stack
Overflow need at least 30 minutes to develop.  We also see that in online
private messaging, star patterns constructed by outgoing messages sent by one
user tend to increase in frequency from time scales of 1 to 20 minutes before
peaking and then declining in frequency.

In summary, our work defines a flexible notion of motifs in temporal networks
and provides efficient algorithms for counting them. It enables new analyses in
a variety of scientific domains and paves a new way for modeling dynamic complex
systems.


\section{Definitions of temporal networks and temporal motifs}
\label{sec:tm_preliminaries}

We now provide formal definitions of temporal graphs and $\delta$-temporal
motifs.  In \cref{sec:tm_algorithms}, we provide algorithms for counting the
number of $\delta$-temporal motifs in a given temporal graph.

We define a \emph{temporal edge} to be a timestamped directed edge between an
ordered pair of nodes.  We call a collection of temporal edges a \emph{temporal graph}
(\cref{fig:intro1A}), as formalized in the following definition.
\begin{definition}\label{def:temporal_graph}
A \emph{temporal graph} $T$ on a node set
$V$ is a collection of tuples $(u_i, v_i, t_i)$, $i = 1, \ldots, m$, where each
$u_i$ and $v_i$ are elements of $V$ and each $t_i$ is a timestamp in $\mathbb{R}$.
A tuple $(u_i, v_i, t_i)$ in $T$ is a \emph{temporal edge}.
\end{definition}
Note that there can be many temporal edges directed from $u$ to
$v$.  We assume that
the timestamps $t_i$ are unique so that the tuples may be strictly ordered. This
assumption makes the presentation of the definitions and algorithms clearer, but
our methods can be adapted to the case when timestamps are not unique.
When it is clear from context, we refer to a temporal edge as simply
an \emph{edge}. Finally, by ignoring timestamps and duplicate
edges, the temporal graph induces a standard (static) directed graph.
\begin{definition}
The \emph{static graph} $G$ of a temporal graph $T$ on node set $V$
is a graph with edge set $E = \{ (u, v) \;\vert\; \exists t : (u, v, t) \in T\}$.
An edge in $E$ is called a \emph{static edge} of $T$.
\end{definition}
Sometimes, it will be useful to think about all temporal edges $(u, v, t)$
associated with a specified static edge $(u, v)$.  We will sometimes refer
to these as the timestamps of these temporal edges as the \emph{timestamps
along the static edge $(u, v)$.}  For example, there are 4 timestamps
along the edge $(a, c)$ in \cref{fig:intro1A}.

We now formalize $\delta$-temporal motifs.
\begin{definition}\label{def:temporal_motif}
A \emph{$k$-node, $l$-edge, $\delta$-temporal motif} $M = (\multigraph, \sigma, \delta)$ 
consists of a multigraph $\multigraph = (V, E)$
with $k$ nodes and $l$ edges, an ordering $\sigma$ on the edges of $E$,
and a time span $\delta \in \mathbb{R}_+$.
We typically represent $(\multigraph, \sigma)$ by a sequence
$(u_1, v_1), (u_2, v_2) \ldots, (u_l, v_l)$, $u_i, v_i \in V$, $i = 1, \ldots, l$.
\end{definition}
We often just refer to a motif $M$ by just the multigraph $\multigraph$ and the ordering $\sigma$,
considering $\delta$ to be arbitrary.  \Cref{fig:intro1B} illustrates a particular $3$-node, $3$-edge $\delta$-temporal motif with $\delta$ specified to be a 10 second duration.

\Cref{def:temporal_motif} provides a template for a particular pattern, and we are
interested in how many times a given pattern occurs in a given temporal network.
Intuitively, a collection of edges in a given temporal graph is
an \emph{instance} of a $\delta$-temporal motif $M$ if it matches the same edge
pattern of the multigraph, the temporal edges occur in the specified order, and
all of the temporal edges occur with a $\delta$ time window
(\cref{fig:intro1C}).
\begin{definition}
In a temporal graph $T$, a time-ordered sequence of 
$S = (w_1, x_1, t_1)$, $\ldots$, $(w_l, x_l, t_l)$ of $l$ unique edges in $T$ ($t_1 < t_2 < \ldots < t_l$)
is an \emph{instance} of the temporal motif $M = (\multigraph, \sigma, \delta)$
with $(\multigraph, \sigma) = (u_1, v_1), \ldots, (u_l, v_l)$ if
\begin{enumerate}
\item There exists a bijection $f$ on the vertices such that
$f(w_i) = u_i$ and $f(x_i) = v_i$, $i = 1, \ldots, l$, and
\item the edges all occur within $\delta$ time, i.e., $t_l - t_1 \le \delta$
\end{enumerate}
\end{definition}

The central algorithmic goal of this chapter is to count the number of ordered
subsets of edges from a temporal graph $T$ that are instances of a particular
motif.  In other words, given a $k$-node, $l$-edge $\delta$-temporal motif, we
seek to find how many of the $l! {m \choose l}$ ordered length-$l$ sequences of
edges in the temporal graph $T$ are instances of the motif.  A naive approach to
this problem would be to simply enumerate all ordered subsets and then check if
it is an instance of the motif.  In modern datasets, the number of edges $m$ is
typically quite large (we analyze a dataset in \cref{sec:tm_experiments} with
over $m >$ 2 billion), and this approach is impractical even for $l = 2$.  In
the following section, we discuss several faster algorithms for counting the
number of instances of $\delta$-temporal motifs in a temporal graph.

In contrast to \cref{ch:honc}, here we define motifs by just the existence
of edges (a general subgraph) and not the non-existence of edges (an induced subgraph).
This is crucial for our development of fast counting algorithms.

\Cref{tab:notation} summarizes the notation introduced in this section and 
the next.

\begin{table}[tb]
\centering
\dualcaption{Notation for \cref{ch:tm}}{}
  \begin{tabular}{r l}
    \toprule
    $T$                            & temporal network                                             \\
    $e = ((u, v), t)$              & temporal edge                                                \\
    $m$                            & number of temporal edges in a temporal network               \\
    $n$                            & number of nodes in a temporal network                        \\
    $K$                            & multigraph component of a temporal motif             \\
    $\sigma$                       & ordering on the multigraph edges of a temporal motif \\
    $\delta$                       & time span of a temporal motif                        \\
    $M = (K, \sigma, \delta)$      & temporal motif                                       \\
    $M_{i,j}$ for $1 \le i, j, \le 6$ & one of the motifs illustrated in \cref{fig:three_edge_motifs}            \\
    $k$                            & number of nodes in a temporal motif                  \\
                                   & (number of nodes in the multigraph)                          \\
    $\ell$                         & number of edges in a temporal motif                  \\
                                   & (number of edges in the multigraph)                          \\
    $S$                            & ordered sequence of temporal edges                           \\
    $G$                            & static graph induced by a temporal network                   \\
    $\tau$                         & number of triangles in a static graph                        \\
    $H$                            & motif pattern in a static graph                              \\
    $H'$                           & instance of a static motif in a static graph                 \\
    \bottomrule
    \end{tabular}\label{tab:notation}
\end{table}  

\clearpage

\section{Counting algorithms}
\label{sec:tm_algorithms}

We now present several algorithms for counting the exact number of instances
of temporal motifs in a temporal graph.  We first present a general counting
algorithm in \cref{sec:tm_general_framework}, which can count instances of any
$k$-node, $l$-edge temporal motif faster than simply enumerating over all
size-$l$ ordered subsets of edges.  The computational complexity depends on the
multigraph structure of the motif, but the algorithm is optimal for counting $2$-node, $l$-edge
temporal motifs (for constant $l$) in the sense that it is linear in the number
of edges in the temporal graph.  In \cref{sec:tm_faster_algs}, we provide
faster, specialized algorithms for counting specific types of $3$-node, $3$-edge
temporal motifs (\cref{fig:three_edge_motifs}).  The computational complexities
of all of our algorithms are independent of $\delta$.

\subsection{General counting framework}\label{sec:tm_general_framework}

We begin with a general framework for counting the number of instances of a of a
temporal motif $M = (\multigraph, \sigma, \delta)$.  To start, consider $H$ to
be the static directed graph induced by $\multigraph$ (i.e., $H$ consists of all
unique edges in $\multigraph$).  A sequence of temporal edges is an instance
of $M$ if and only if
\begin{enumerate}
\item the static subgraph induced by the edges in the sequence is isomorphic to $H$
\item the ordering of the edges in the sequence matches $\sigma$,
\item all the edges in the sequence span a time window of at most $\delta$ time units.
\end{enumerate}
This leads to the following general algorithm for counting instances of $M$ in a
temporal graph $T$:
\begin{enumerate}
\item
Identify all instances $H'$ of the static motif $H$ within the static graph $G$
induced by the temporal graph $T$.  For example, there are three instances of
$H$ induced in \cref{fig:intro1} ($\{(a, b), (a, c), (c, a)\}$, $\{(a, d), (a, c), (c, a)\}$, and
$\{(c, d), (c, a), (a, c)\}$).
\item
For each static motif instance $H'$, gather all temporal edges between pairs of
nodes forming an edge in $H'$ into an ordered sequence $S =$ $(u_1, v_1, t_1)$,
$\ldots$, $(u_L, v_L, t_L)$.
\item
Count the number of (potentially non-contiguous) subsequences of edges in $S$
occurring within $\delta$ time units that match $(\multigraph, \sigma)$.
\end{enumerate}

The first step can use known algorithms for enumerating motifs in static graphs
such as those by \citet{wernicke2006fanmod}, and the second step is a simple
matter of fetching the appropriate temporal edges.  To perform the third step
efficiently, we develop a dynamic programming approach for counting the number
of subsequences (instances of motif $M$) that match a particular pattern within
a larger sequence ($S$).  The key idea is that, as we stream through an input
sequence of edges, the count of a given length-$l$ pattern (i.e., motif) with a
given final edge is computed from the current count of the length-($l-1$) prefix
of the pattern.  Inductively, we maintain auxiliary counters of all of the
prefixes of the pattern (motif).  Second, we also require that all edges in the
motif be at most $\delta$ time apart.  Thus, we use the notion of a moving time
window such that any two edges in the time window are at most $\delta$ time
apart.  The auxiliary counters now keep track of only the subsequences occurring
within the current time window.  Last, it is important to note that the
algorithm only \emph{counts} the number of instances of motifs rather
than \emph{enumerating} them.  This is critical to making the algorithm fast.

\begin{algorithm}[tb]\algoptions
  \SetKwFunction{increment}{IncrementCounts}
  \SetKwFunction{decrement}{DecrementCounts}  
  \KwIn{Sequence $S$ of edges $(e_1 = (u_1, v_1), t_1), \ldots,$ $(e_L, t_L)$ with
  $t_1 < \ldots < t_L$, time window $\delta$
  }
  \KwOut{Number of instances of each $l$-edge $\delta$-temporal motif $M$ contained in the sequence}
  $\tstart \leftarrow 1$,  $\counts \leftarrow$ Counter(default = 0)\;
  \For{$\textnormal{end} = 1, \ldots, L$}{
        \mycomment{Remove edges out of $\delta$ window}\;
  	\While{$t_{\tstart} + \delta < t_{\tend}$}{
            \decrement{$e_{\tstart}$},
	        $\tstart \pluseq 1$\;
	}
	\mycomment{Update counts with edge $(e_{\tend}, t_{\tend})$} \;
	\increment{$e_{\tend}$}\; \label{line:endfor}
  }
  \Return{$\counts$}\; \;
  \myproc{\decrement{$e$}}{
    $\counts[e] \minuseq 1$\;
      \For{\textnormal{$\suffix$ \KwTo $\counts.\keys$ of length $< l - 1$}}{
          $\counts[\concat(e, \suffix)] \minuseq \counts[\suffix]$\;
    }
 }
  \myproc{\increment{$e$}}{
      \For{$\prefix$ \KwTo $\counts.\keys.\reverse()$ of length $< l $}{
            $\counts[\concat(\prefix, e)] \pluseq \counts[\prefix]$
        }
      $\counts[e] \pluseq 1$\;
  }
  \dualcaption{General counting method for temporal motifs}{This routine
  counts the number of instances of all possible $l$-edge $\delta$-temporal
  motifs. The input is an ordered sequence of temporal edges.  We assume the
  keys of $\counts[\cdot]$ are accessed in order of length.}
  \label{alg:general}
\end{algorithm}

For simplicity of presentation, \cref{alg:general} counts \emph{all} possible
$l$-edge motifs that occur in a given sequence of edges.  The data structure
$\counts[\cdot]$ maintains auxiliary counts of all (ordered) patterns of length
at most $l$. Specifically, $\counts[e_1\cdots e_r]$ is the number of times the
subsequence $[e_1\cdots e_r]$ occurs in the current time window (if $r < l$) or
the number of times the subsequence has occurred within all time windows of
length $\delta$ (if $r = l$).  Crucially, we also assume the keys of $\counts[\cdot]$ are
accessed in order of length.  Moving the time window forward by adding a new
edge into the window, all edges $(e = (u, v), t)$ farther than $\delta$ time
from the new edge are removed from the window and the appropriate counts are
decremented (the \emph{DecrementCounts()} method)---first, the single edge
counts ($[e]$) are updated. Based on these updates, length-$2$ subsequences
formed with $e$ as its first edge are updated and so on, up through
length-($l-1$) subsequences. On the other hand, when an edge $e$ is added to the
window, similar updates take place, but in reverse order, from longest to
shortest subsequences, in order to increment counts in subsequences where $e$ is
the last edge (the \emph{IncrementCounts()} method). Importantly, length-$l$
subsequence counts are incremented in this step but never decremented.  As the
time window moves from the beginning to the end of the sequence of edges, the
algorithm accumulates counts of all length-$l$ subsequences in all possible time
windows of length $\delta$.

\definecolor{myred}{RGB}{228, 26, 28}
\definecolor{myblu}{RGB}{55, 126, 184}
\definecolor{mypur}{RGB}{152, 78, 163}
\newcommand{\inc}[1]{\textbf{\textcolor{myblu}{#1}}}
\newcommand{\dec}[1]{\textbf{\textcolor{myred}{#1}}}
\newcommand{\incdec}[1]{\textbf{\textcolor{mypur}{#1}}}
\begin{table}[tb]
\centering
\setlength{\tabcolsep}{3pt}
\begin{tabular}{lcccccccc}
  & 15s & 17s & 25s & 28s & 30s & 32s & 35s \\
  & $(c, a)$ & $(a, c)$ & $(a, b)$ & $(a, c)$ & $(a, c)$ & $(c, a)$ & $(a, c)$ \\ \midrule
  $\text{counts}[(a, b)]$ & 0 & 0 & \inc{1} & 1 & 1 & 1 & 1 \\
  $\text{counts}[(a, c)]$ & 0 & \inc{1} & 1 & \incdec{1} & \inc{2} & 2 & \inc{3} \\
  $\text{counts}[(c, a)]$ & \inc{1} & 1 & 1 & \dec{0} & 0 & \inc{1} & 1 \\
  $\text{counts}[(a,b)(a,c)]$ & 0 & 0 & 0 & \inc{1} & \inc{2} & 2 & \inc{3} \\
  $\text{counts}[(a,b)(a,c)(c,a)]$ & 0 & 0 & 0 & 0 & 0 & \inc{2} & 2 \\ \midrule
  start & 1 & 1 & 1 & \inc{3} & 3 & 3 & 3 \\
  end   & \inc{1} & \inc{2} & \inc{3} & \inc{4} & \inc{5} & \inc{6} & \inc{7} \\
\end{tabular}
\dualcaption{Example execution of \cref{alg:general} for counting instances of the
  $\delta$-temporal motif $M$ and network in \cref{fig:intro1}}{%
Each column shows the value of counters at the end of the for loop that
processes temporal edges (after \cref{line:endfor} in \cref{alg:general}).  Color indicates change in the variable: incremented
(\textcolor{myblu}{blue}), decremented (\textcolor{myred}{red}),
incremented and decremented (\textcolor{mypur}{purple}), or no change
(black).  At the end of execution, $\counts[(a,b)(a,c)(c,a)] = 2$ for the two
instances of the temporal motif $M$ with center node $a$.  Here we only show the
counters needed to count $M$, whereas \cref{alg:general} as presented would
count all 3-edge patterns in the sequence (maintaining 39 total counters).  However,
for any particular $l$-edge temporal network motif, \cref{alg:general} only
needs to maintain $O(l^2)$ counters.
}
\label{tab:general_alg_example}
\end{table}

While \cref{alg:general} maintains counts for all $l$-edge motifs, we could
simply maintain the $\counts[\cdot]$ data structure for the contiguous prefixes
of the motif $M$.  As an example, \cref{tab:general_alg_example} shows the
execution of \cref{alg:general} for a particular sequence of edges and a particular
motif, where counts of only the necessary contiguous subsequences of the motif are maintained.
In general, there are $O(l^2)$ contiguous subsequences of an $l$-edge motif $M$,
and there are $O(\lvert H \rvert^l)$ total keys in $\counts[\cdot]$, where
$\lvert H \rvert$ is the number of edges in the static subgraph $H$ induced by
$M$, in order to count all $l$-edge motifs in the sequence (i.e., not just motif
$M$).

We now analyze the complexity of the overall 3-step algorithm.  We assume that
the temporal graph $T$ has edges sorted by timestamps, which is reasonable if
edges are logged in their order of occurrence, and we also assume that we have pre-processed
$T$ so that we can access the sorted list of all edges between $u$ and $v$ in
$O(1)$ time.  Given an instance $H'$ of $H$, constructing the time-sorted
sequence $S$ in step 2 of the algorithm then takes
$O(\log({\lvert H \rvert}) {\lvert S \rvert})$
time (by merging sorted lists).  Each edge inputted to
\cref{alg:general} is processed exactly twice: once to increment counts when
it enters the time window and once to decrement counts when it exits the time
window.  As presented in \cref{alg:general}, each update changes $O(\vert H \vert^l)$
counters resulting in an overall complexity of $O(\vert H \vert^l {\lvert S \rvert})$.
However, one could modify \cref{alg:general} to only update counts for
contiguous subsequences of the sequence $M$, which would change $O(l^2)$
counters and have overall complexity $O(l^2 {\lvert S \rvert})$.  We are typically only
interested in small constant values of $\vert H \rvert$ and $l$ (for our
experiments in \cref{sec:tm_experiments}, $\lvert H \rvert \le 3$ and $l = 3$),
in which case the running time of \cref{alg:general} 
is linear in the size of the input to the algorithm, i.e., $O({\lvert S \rvert})$.
What remains is to identify all instances $H'$ of $H$.  Of course, this depends
on the structure of $H$.  For example, if $H$ is an edge, this can be done in
linear time.  If $H$ is a triangle, this can be done in $O(a \lvert G \rvert)$ time,
where $\lvert G \rvert$ is the number of edges in the static graph $G$ induced
by $T$ and $a$ is the arboricity of $G$~\cite{chiba1985arboricity}.

In the remainder of this section we analyze our 3-step algorithm with respect to
different types of motifs (2-node, stars, and triangles) and argue benefits and
deficiencies of the proposed framework.  We show that for $2$-node motifs, our
general counting framework takes time linear in the total number of edges $m$.
Since all the input data needs to be examined for computing exact counts, this
means the algorithm is optimal for $2$-node motifs. However, we also show that
for star and triangle motifs the algorithm is not optimal, which then motivates
us to design faster algorithms in \cref{sec:tm_faster_algs}.

\xhdr{General algorithm for 2-node motifs}
We first show how to map $2$-node motifs to the framework described above.  Any
induced graph $H$ of a $2$-node $\delta$-temporal motif is either a single or a
bidirectional edge.  In either case, it is straightforward to enumerate over all
instances of $H$ in the induced static graph.  This leads to the following procedure.
For each pair of nodes $u$ and $v$ for which there is at least one edge,
gather and sort the edges in either direction between $u$ and $v$
and call \cref{alg:general} with these edges.  To obtain the total motif count,
we simply need to sum the results from each call to \cref{alg:general}.

We only need to input each edge to \cref{alg:general} once, and
under the assumption that we can
access the sorted directed edges from one node to another in $O(1)$ time, the
merging of edges into sorted order takes linear time (merging 2 sorted lists).
Therefore, the total running time is $O(l^2m)$, where $l$ is the number
of edges in the motif, which is linear in the number of temporal edges $m$.  
We are mostly interested in small patterns, i.e., cases
when $l$ is a small constant.  Thus, this methodology is optimal (linear in the
input size, $m$) for counting $2$-node $\delta$-temporal motif instances.

\xhdr{General algorithm for star motifs}
Next, we consider $k$-node, $l$-edge star motifs $M$, whose induced static graph
$H$ consists of a center node and $k-1$ neighbors, where edges may occur in
either direction between the center node and a neighbor node.  For example, in
the top left corner of \cref{fig:three_edge_motifs}, $M_{1,1}$ is a star motif
with all edges pointing toward the center node (the orange node).  In such
motifs, the induced static graph $H$ contains at most $2(k - 1) = 2k - 2$ static
edges---one incoming and outgoing edge from the center node to each neighbor
node.  We have the following method for counting the number of instances of
$k$-node, $l$-edge star motifs. For each node $u$ in the static graph and for
each unique set of $k - 1$ neighbors, gather and sort the edges in either
direction between $u$ and the neighbors.  Then, count the number of instances of
$M$ using \cref{alg:general}.  The counts from each call to \cref{alg:general}
are summed over all center nodes and sets of $k - 1$ neighbors.

The major drawback of this approach is that we have to loop over each
size-$(k-1)$ neighbor set.  This can be prohibitively expensive even when $k =
3$ if the center node has large degree.  In \cref{sec:tm_faster_algs}, we will
design an algorithm that avoids this issue for the case when the star motif has
$k = 3$ nodes and $l = 3$ edges.


\definecolor{mygreen}{RGB}{27,158,119}%
\definecolor{myorange}{RGB}{217,95,2}%
\definecolor{mypurple}{RGB}{117,112,179}%
\begin{figure}[tb]
\centering
\resizebox {\columnwidth} {!} {\input{CH4-TKZ-embedded-three-edge-all.tex}}
\dualcaption{All $2$-node and $3$-node, $3$-edge $\delta$-temporal motifs}{%
The box in the bottom left contains the four $2$-node motifs, and the
two boxes in the top right contain the eight triangles.  The 24 other motifs are stars.
We index the 36 motifs $M_{i,j}$ by 6 rows and 6 columns.  The first edge in
each motif is from the \textcolor{mygreen}{green} node to the
\textcolor{myorange}{orange} node.  The second edge is the same
along each row, and the third edge is the same along each column.
We develop fast algorithms for counting the number of instances
of each of these motifs.
}
\label{fig:three_edge_motifs}
\end{figure}

\clearpage

\xhdr{General algorithm for triangle motifs}
With triangle motifs, the induced graph $H$ consists of 3 nodes and at least one
directed edge between any pair of nodes (see \cref{fig:three_edge_motifs}
for all eight of the $3$-edge triangle motifs).  The induced static graph $H$ of $M$
contains at least three and at most six static edges.  A straightforward
algorithm for counting $l$-edge triangle motifs in a temporal graph $T$ is:
\begin{enumerate}
\item Use a fast static graph triangle enumeration algorithm to find all
  triangles in the static graph $G$ induced by $T$ (using, e.g.,
  one of the methods analyzed by \citet{latapy2008main}).
\item For each triangle, merge all temporal edges from each pair of nodes to
  get a time-sorted list of edges.  Use \cref{alg:general} to count the
  number of instances of $M$.
\end{enumerate}
This approach is problematic as the edges between a pair of nodes may
participate in many triangles. \Cref{fig:worstcase} shows a worst-case example
for the motif $M = (w, u), (w, v), (u, v)$ with $\delta = \infty$.
In the figure, the timestamps are ordered by their index. There are
$m-2n$ edges between $u$ and $v$, and each of these edges forms an instance of
$M$ with every $w_i$. Thus, the overall worst-case running time of the algorithm
is $O(\triangleenum + \min(m\tau, mn))$, where $\triangleenum$ is the time to enumerate
the number of triangles $\tau$ in the static graph.  (The $mn$ running time comes
from the case that each edge can participate in at most $n$ triangles).
In the following section, we devise an algorithm that significantly reduces the dependency on $\tau$ from
linear to sub-linear (specifically, to $O(\triangleenum + \min(m\sqrt{\tau}, mn)$)
when there are $l = 3$ edges.


\begin{figure}[t]
\begin{center}
\definecolor{mygreen}{RGB}{27,158,119}
\definecolor{myorange}{RGB}{217,95,2}
\definecolor{mypurple}{RGB}{117,112,179}

\renewcommand*{\VertexInterMinSize}{20pt}
\renewcommand*{\VertexSmallMinSize}{20pt}

\begin{tikzpicture}
\renewcommand*{\VertexLineColor}{myorange}
\renewcommand*{\VertexLightFillColor}{myorange}
\renewcommand*{\VertexTextColor}{white}
\Vertex[x=3.5,y=4.5]{u}
\renewcommand*{\VertexLineColor}{mypurple}
\renewcommand*{\VertexLightFillColor}{mypurple}
\Vertex[x=8.5,y=4.5]{v}

\node[draw=none] (vd) at (7,2) {\large $\ldots$};

\renewcommand*{\VertexLineColor}{mygreen}
\renewcommand*{\VertexLightFillColor}{mygreen}
\node[VertexStyle](w1) at (2,2) {$w_1$};
\node[VertexStyle](w2) at (4.5,2) {$w_2$};
\node[VertexStyle](wn) at (10,2) {$w_n$};

\SetUpEdge[style={->,>=triangle 60, black, bend left=10},
  label={\footnotesize $t_{2n+1}, \ldots, t_{m}$}]
\Edge(u)(v)
\SetUpEdge[style={->,>=triangle 60, black}]

\Edge[label=$t_1$](w1)(u)
\Edge[label=$t_2$](w1)(v)
\Edge[label=$t_3$](w2)(u)
\Edge[label=$t_4$, style = {->, >=triangle 60,bend right=10}](w2)(v)
\Edge[label=$t_{2n-1}$, style={->,>=triangle 60,bend right=15}](wn)(u)
\Edge[label=$t_{2n}$](wn)(v)

\node[draw=none] (vd) at (7,2) {\large $\ldots$};

\end{tikzpicture}
\end{center}
\dualcaption{Worst-case example for counting triangular motifs with
\cref{alg:general}}{When using the general algorithm (\cref{alg:general})
 to count instances of motif
$M_{4,5}$ (\cref{fig:three_edge_motifs}) with $\delta = \infty$, we would need to analyze
the edges $t_{2n+1}, \ldots, t_m$ for each of the $n$ nodes $w_1, \ldots, w_n$
forming a triangle with nodes $u$ and $v$.}
\label{fig:worstcase}
\end{figure}

\subsection{Faster algorithms}\label{sec:tm_faster_algs}

The general counting algorithm from the previous section counts the number of
instances of any $k$-node, $l$-edge $\delta$-temporal motif, and is also optimal
for $2$-node motifs.  However, the computational cost may be expensive for other
motifs such as stars and triangles.  We now develop specialized algorithms that
count certain motif classes faster.  Specifically, we design faster algorithms
for counting all 3-node, 3-edge star and 3-node, 3-edge triangle motifs
(\cref{fig:three_edge_motifs} illustrates these motifs).  Our algorithm for
stars is linear in the input size, so it is optimal up to constant factors.

\xhdr{Fast algorithm for 3-node, 3-edge stars}
With $3$-node, $3$-edge star motifs, the drawback of the previous algorithmic
approach is the need to loop over all pairs of neighbors given a center node.
Instead, we will count all instances of star motifs for a given center node in
just a single pass over the edges adjacent to the center node.


\begin{figure}[t]
  \centering
  \begin{tikzpicture}

\colorlet{TufteRed}{red!80!black}
\newcommand{\ax}{0.0}
\newcommand{\ay}{0.0}
\newcommand{\bx}{-1.25}
\newcommand{\by}{1.5}
\newcommand{\cx}{1.25}
\newcommand{\cy}{1.5}
\newcommand{\tfirst}{1}
\newcommand{\tsecond}{2}
\newcommand{\tthird}{3}
\newcommand{\figfontsize}{\large}
\newcommand{\labelfontsize}{\normalsize}
\newcommand{\xshift}{4}

\renewcommand*{\VertexLineColor}{white}
\node[draw=none] (pre) at (\ax,\by) {\Large pre};
\node[draw=none] (pre) at (\xshift+\ax,\by) {\Large post};
\node[draw=none] (pre) at (2*\xshift+\ax,\by) {\Large peri};

\renewcommand*{\VertexInterMinSize}{11pt}
\renewcommand*{\VertexSmallMinSize}{11pt}
\definecolor{mygreen}{RGB}{27,158,119}
\definecolor{myorange}{RGB}{217,95,2}
\definecolor{mypurple}{RGB}{117,112,179}

\SetVertexNoLabel
\renewcommand*{\VertexLineColor}{mygreen}
\renewcommand*{\VertexLightFillColor}{mygreen}
\Vertex[x=\ax,y=\ay]{a}
\renewcommand*{\VertexLineColor}{myorange}
\renewcommand*{\VertexLightFillColor}{myorange}
\Vertex[x=\bx,y=\by]{b}
\renewcommand*{\VertexLineColor}{mypurple}
\renewcommand*{\VertexLightFillColor}{mypurple}
\Vertex[x=\cx,y=\cy]{c}
\SetUpEdge[style={-,ultra thick,black}, label={\labelfontsize $\tfirst,\tsecond$}]\Edge(a)(b)
\SetUpEdge[style={-,ultra thick,black}, label={\labelfontsize $\tthird$}]\Edge(a)(c)

\renewcommand*{\VertexLineColor}{mygreen}
\renewcommand*{\VertexLightFillColor}{mygreen}
\Vertex[x=\xshift+\ax,y=\ay]{a}
\renewcommand*{\VertexLineColor}{myorange}
\renewcommand*{\VertexLightFillColor}{myorange}
\Vertex[x=\xshift+\bx,y=\by]{b}
\renewcommand*{\VertexLineColor}{mypurple}
\renewcommand*{\VertexLightFillColor}{mypurple}
\Vertex[x=\xshift+\cx,y=\cy]{c}
\SetUpEdge[style={-,ultra thick,black}, label={\labelfontsize $\tfirst$}]\Edge(a)(b)
\SetUpEdge[style={-,ultra thick,black}, label={\labelfontsize $\tsecond,\tthird$}]\Edge(a)(c)

\renewcommand*{\VertexLineColor}{mygreen}
\renewcommand*{\VertexLightFillColor}{mygreen}
\Vertex[x=2*\xshift+\ax,y=\ay]{a}
\renewcommand*{\VertexLineColor}{myorange}
\renewcommand*{\VertexLightFillColor}{myorange}
\Vertex[x=2*\xshift+\bx,y=\by]{b}
\renewcommand*{\VertexLineColor}{mypurple}
\renewcommand*{\VertexLightFillColor}{mypurple}
\Vertex[x=2*\xshift+\cx,y=\cy]{c}
\SetUpEdge[style={-,ultra thick,black}, label={\labelfontsize $\tfirst,\tthird$}]\Edge(a)(b)
\SetUpEdge[style={-,ultra thick,black}, label={\labelfontsize $\tsecond$}]\Edge(a)(c)

\end{tikzpicture}
  \dualcaption{The three classes of $3$-node, $3$-edge star temporal motifs}{Each of the
    $3$-edges in each class can be directed outward from or inward towards the
    center node.  Thus, each class consists of $8$ motifs, and there are 24
    total $3$-node, $3$-edge star temporal motifs (see
    \cref{fig:three_edge_motifs} for illustrations of all of them).}
\label{fig:star_classes}
\end{figure}

We use a dynamic programming approach for counting star motifs.  First, note
that every temporal edge in a star with center $u$ is defined by
\begin{enumerate} 
\item a neighbor node,
\item a direction of the edge (outward from or inward toward $u$), and
\item the timestamp.
\end{enumerate}
With this characterization, there are 3 classes of star motifs with 3 nodes and 3 edges
(illustrated in \cref{fig:star_classes}).
\begin{enumerate}
\item ``pre'' -- the first two edges contain node $u$ and neighbor $v$, and the third
edge contains node $u$ and a different neighbor $w$.
\item ``post'' -- the last two edges contain node $u$ and neighbor $w$, and the first
edge contains node $u$ and a different neighbor $v$.
\item ``peri'' -- the first and third edges contain node $u$ and neighbor $v$, and the second
edge contains node $u$ and a different neighbor $w$.
\end{enumerate}
Each class consists of $2^3 = 8$ motifs corresponding to the different combinations
of the direction of the edges from $u$ to the neighbors (see \cref{fig:three_edge_motifs}, which shows all of the $3 \cdot 8 = 24$ 3-node, 3-edge stars).

Now, suppose we process the time-ordered sequence of all edges containing the center
node $u$.  We maintain the following counters when processing an edge with
timestamp $t_j$:
\begin{itemize}
\item $\presum[\dir1, \dir2]$ is the number of sequentially ordered pairs of
  edges in $[t_j - \delta, t_j)$ where the first edge points in direction
    $\dir1$ and the second edge points in direction $\dir2$.  We can think of a
    direction as being outward from $u$ or inward toward $u$.
    
\item $\postsum[\dir1, \dir2]$ is the number of sequentially ordered
pairs of edges in $(t_j, t _j + \delta]$ where the first edge points in
direction $\dir1$ and the second edge points in direction $\dir2$.

\item $\middlesum[\dir1, \dir2]$ is the number of pairs of edges where the first
  edge is in direction $\dir1$ and occurred at time $t < t_j$, the second edge
  is in direction $\dir2$ and occurred at time $t' > t_j$, and $t' - t \le \delta$.
  
\end{itemize}

If we are currently processing an edge, the ``pre'' class gets
$\presum[\dir1, \dir2]$ new motif instances for any choice of directions $\dir1$
and $\dir2$ (specifying the first two edge directions) and the current edge
serves as the third edge in the motif (hence specifying the third edge
direction).  Similar updates are made with the $\postsum[\cdot,\cdot]$ and
$\middlesum[\cdot,\cdot]$ counters, where the current edge serves as the first
or second edge in the motif, respectively.

In order for our algorithm to be fast, we must be able to efficiently update
these counters when processing edges.  To aid in this, we introduce two
additional counters:
\begin{itemize}
\item $\prenodecount[\dir, v_i]$ is the number of times node $v_i$ has appeared
  in an edge with $u$ with direction $\dir$ in the time window $[t_j - \delta, t_j)$.
  We can think of the direction as being outward from $u$ to $v_i$ or
  inward from $v_i$ toward $u$.
  
\item $\postnodecount[\dir, v_i]$ is the number of times node $v_i$ has appeared
  in an edge with $u$ with direction $\dir$ in the time window $(t_j, t_j + \delta]$.
\end{itemize}

Following the ideas of \cref{alg:general}, it is easy to update these counters
when we process a new edge.  Consequently, $\presum[\cdot,\cdot]$,
$\postsum[\cdot,\cdot]$, and $\middlesum[\cdot,\cdot]$ can be maintained when
processing an edge with just a few simple primitives:
\begin{itemize}
\item
  \emph{Push()} and \emph{Pop()} update the counts for
  $\prenodecount[\cdot,\cdot]$, $\postnodecount[\cdot,\cdot]$,
  $\presum[\cdot,\cdot]$ and $\postsum[\cdot,\cdot]$ when edges enter and leave
  the time windows $[t_j - \delta, t_j)$ and $(t_j, t_j + \delta]$.
\item
  \emph{ProcessCurrent()} updates motif counts involving the current edge and
  updates the counter $\middlesum[\cdot,\cdot]$.
\end{itemize}

We describe the general procedure in \cref{alg:fast_framework}, which calls the
subroutines \emph{Push()}, \emph{Pop()}, and
\emph{ProcessCurrent()}, and \cref{alg:fast_wedges} implements
these subroutines.  The $\globalpre[\cdot,\cdot,\cdot]$,
$\globalpost[\cdot,\cdot,\cdot]$, and $\globalmid[\cdot,\cdot,\cdot]$ counters
in \cref{alg:fast_wedges} maintain the counts of the three different classes of
stars described above.  The edges input to each call of \cref{alg:fast_wedges}
assumes that an edge consists of the following two pieces of information:
\begin{enumerate}
 \item the neighbor node $\nbr$ of $u$ and
 \item the direction $\dir$ of the edge.
\end{enumerate}
The timestamps are all handled in \cref{alg:fast_framework}.

We have intentionally separated the logic of
the \emph{Push()}, \emph{Pop()}, and \emph{ProcessCurrent()}
from \cref{alg:fast_framework}.  The reason for this is that our fast $3$-node,
$3$-edge triangle counting procedure follows the same logic
of \cref{alg:fast_framework}, only with different implementations of these
subroutines.  We will get to this idea shortly.

Finally, we note that our counting scheme incorrectly includes instances of
$2$-node motifs such as $M = (u, v_i)$, $(u, v_i)$, $(u, v_i)$, but we can use
our fast $2$-node motif counting algorithm to account for this.  Putting
everything together, we have the following procedure:
\begin{enumerate}
\item For each node $u$ in the temporal graph $T$, get a time-ordered
  list of all edges containing $u$.
  
\item Use \cref{alg:fast_framework,alg:fast_wedges} to count
  star motif instances.
  
\item For each neighbor $v$ of a star center $u$, subtract the 2-node motif
  counts using \cref{alg:general}.
\end{enumerate}

If the $m$ temporal edges of $T$ are time-sorted, the first step takes linear
time.  The second and third steps also run in linear time in the input size.
Each edge is used in steps (ii) and (iii) exactly twice: once for each end point
as the center node.  Thus, the overall complexity of the algorithm is $O(m)$.
This is optimal (linear in the size of the dataset).  We also need $O(n)$
memory (in addition to storing the graph) to maintain the counters in \cref{alg:fast_framework}.


\begin{algorithm}[tb]\algoptions
  \SetKwFunction{push}{Push}
  \SetKwFunction{pop}{Pop}
  \SetKwFunction{process}{ProcessCurrent}
  \KwIn{Sequence of edges $(e_1 = (u_1, v_1), t_1), \ldots, (e_L, t_L)$ with
  $t_1 < \ldots < t_L$, time window $\delta$}
  Initialize counters $\prenodecount$, $\postnodecount$, $\middlesum$, $\presum$, $\postsum$\;
  $\tstart \leftarrow 1$, $\tend \leftarrow 1$\;
  \For{$j = 1, \ldots, L$}{
      \mycomment{Adjust counts in time window $[t_j - \delta, t_j)$} \;
      \While{$t_{\tstart} + \delta < t_j$}{
         \pop{$\prenodecount$, $\presum$, $e_{\tstart}$}\;
         $\tstart \pluseq 1$
       }
      \mycomment{Adjust counts in time window $(t_j, t_j + \delta]$} \;
      \While{$t_{\tend} \le t_j + \delta$}{
          \push{$\postnodecount$, $\postsum$, $e_{\tend}$}\;
           $\tend \pluseq 1$\;
       }
      \mycomment{Handle current edge} \;
      \pop{$\postnodecount$, $\postsum$, $e_j$}\;
       \process{$e_j$}\;
       \push{$\prenodecount$, $\presum$, $e_j$}\;
  }
  \dualcaption{Framework for fast counting of $3$-node, $3$-edge star and
    triangle temporal motifs}{The fast star counting method
    (\cref{alg:fast_wedges}) and triangle counting method
    (\cref{alg:fast_triangles}) implement different versions of the
    \emph{Push()}, \emph{Pop()}, and \emph{ProcessCurrent()} subroutines.}
  \label{alg:fast_framework}
\end{algorithm}


\begin{algorithm}[tb]\algoptions
  \SetKwFunction{push}{Push}
  \SetKwFunction{pop}{Pop}
  \SetKwFunction{process}{ProcessCurrent}
  Initialize counters $\globalpre$, $\globalpost$, $\globalmid$\;
  \myproc{\push{$\nodecount$, $\sumtxt$, $e = (\nbr, \dir)$}}{
    $\sumtxt[:, \dir] \pluseq \nodecount[:, \nbr]$\;
    $\nodecount[\dir, \nbr] \pluseq 1$\;
  }
  \myproc{\pop{$\nodecount$, $\sumtxt$, $e = (\nbr, \dir)$}}{
    $\nodecount[\dir, \nbr] \minuseq 1$\;
    $\sumtxt[\dir, :] \minuseq \nodecount[:, \nbr]$\;
  }
  \myproc{\process{$e = (\nbr, \dir)$}}{
    $\middlesum[:, \dir] \minuseq \prenodecount[:, \nbr]$\;
    $\globalpre[:,:,\dir] \pluseq \presum[:,:]$\;
    $\globalpost[\dir, : , :] \pluseq \postsum[:,:]$\;
    $\globalmid[:, \dir, :] \pluseq \middlesum[:,:]$\;
    $\middlesum[\dir, :] \pluseq \postnodecount[:, \nbr]$\;    
  }
  \Return{$\globalpre$, $\globalpost$, $\globalmid$}
  \dualcaption{Fast counting of $3$-node, $3$-edge temporal star motifs}{These
    are implementations of the subroutines \emph{Push()}, \emph{Pop()}, and
   \emph{ProcessCurrent()} used in \cref{alg:fast_framework}.
   Temporal edges are specified by a neighbor $\nbr$ (the neighbor of
   some given node $u$) and a direction $\dir$ (incoming to $u$ or outgoing from $u$).
   The ``:'' notation represents a selection of all indices in an array.}
  \label{alg:fast_wedges}
\end{algorithm}

\clearpage

\xhdr{Fast algorithm for 3-edge triangle motifs}
While our fast 3-node, 3-edge star counting routine relied on counting motif instances for all
edges adjacent to a given \emph{node}, our fast triangle algorithm is based on
counting instances for all edges adjacent to a given \emph{pair of nodes}.
Specifically, given a pair of nodes $u$ and $v$ and a list of common neighbors
$w_1, \ldots, w_d$, we count the number of motif instances for triangles $(w_i, u, v)$,
$i = 1, \ldots, d$.  Given all of the temporal edges in these $d$ static triangles,
the counting procedures are nearly identical to the case of stars.
We use the same general counting method (\cref{alg:fast_framework}), but the behavior of the
subroutines \emph{Push()}, \emph{Pop()}, and \emph{ProcessCurrent()} depends on
whether or not the edge is between $u$ and $v$.

\begin{algorithm}[tb]\algoptions
  \SetKwFunction{push}{Push}
  \SetKwFunction{pop}{Pop}
  \SetKwFunction{process}{ProcessCurrent}
  Initialize counter $\globalall$\;
  \myproc{\push{$\nodecount$, $\sumtxt$, $e = (\nbr, \dir, \uorv)$}}{
    \lIf{$\nbr \in \{u, v\}$}{\Return}
      $\sumtxt[\muorv, :, \dir] \pluseq \nodecount[\muorv, :, \nbr]$\;
    $\nodecount[\uorv,\dir,\nbr] \pluseq 1$\;
  }
  \myproc{\pop{$\nodecount$, $\sumtxt$, $e = (\nbr, \dir, \uorv)$}}{
    \lIf{$\nbr \in \{u, v\}$}{\Return}
    $\nodecount[\uorv, \dir, \nbr] \minuseq 1$\;
    $\sumtxt[\uorv, \dir, :] \minuseq \nodecount[\muorv, :, \nbr]$
  }
  \myproc{\process{$e = (\nbr, \dir, \uorv)$}}{
    \uIf{$\nbr \notin \{u, v\}$}{
      $\middlesum[\muorv, :, \dir] \minuseq \prenodecount[\muorv, :, \nbr]$\;
    $\middlesum[\uorv, \dir, :] \pluseq \postnodecount[\muorv, :, \nbr]$\;   
    }\Else{
      $\utov = (\nbr==u)$ XOR $\dir$\;
      \For{$0 \le i, j, k \le 1$}{
        $\globalall[i, j, k] \pluseq \middlesum[j \text{ XOR } \utov, i, k]$\;
        \hspace{2.85cm}$+ \postsum[i \text{ XOR } \utov, j, 1 - k]$\;
        \hspace{2.85cm}$+ \presum[k \text{ XOR } \utov, 1- i, 1 - j]$\;
      }
    }
  }
  \mycomment{$\globalall$ key map to \cref{fig:three_edge_motifs}:}\;
  \mycomment{$[0, 0, 0] \mapsto M_{1,3}$, $[0, 0, 1] \mapsto M_{1, 4}$, $[0, 1, 0] \mapsto M_{2, 3}$, $[0, 1, 1] \mapsto M_{2,4}$}\;
  \mycomment{$[1, 0, 0] \mapsto M_{3, 5}$, $[1, 0, 1] \mapsto M_{3, 6}$, $[1, 1, 0] \mapsto M_{4, 5}$, $[1, 1, 1] \mapsto M_{4, 6}$}\;
\Return{$\globalall$}
\dualcaption{Fast temporal triangle motif counting subroutine}{These are implementations of the subroutines \emph{Push()}, \emph{Pop()}, and \emph{ProcessCurrent()} in \cref{alg:fast_framework}
  for counting $3$-edge triangle motifs containing a specified pair of nodes $u$ and $v$.  Temporal
  edges are specified by a neighbor $\nbr$, a direction $\dir$ (incoming or
  outgoing), an indicator ``$\uorv$'' denoting if the edge connects to $u$ or
  $v$.  The ``:'' notation represents a selection of all
  indices in an array.}
\label{alg:fast_triangles}
\end{algorithm}

These methods are implemented in \cref{alg:fast_triangles}.  The input
is a list of edges adjacent to a given pair of neighbors $u$ and $v$, where each
edge consists of three pieces of information:
\begin{enumerate}
 \item a neighbor node $\nbr$,
 \item an indicator of whether or not the node $\nbr$ connects to node $u$ or node $v$, and
 \item the direction $\dir$ of the edge.
\end{enumerate}

The node counters ($\prenodecount[\cdot,\cdot,\cdot]$ and
$\postnodecount[\cdot,\cdot,\cdot]$) in \cref{alg:fast_triangles} have an extra
dimension compared to \cref{alg:fast_wedges} to indicate whether the counts
correspond to edges containing node $u$ or node $v$ (denoted by ``$\uorv$'').
For example, $\prenodecount[u, \dir, w_i]$ is the number of times node $w_i$ has
appeared in an edge with $u$ with directed $\dir$ (incoming to $u$ or outgoing
from $u$) in the time window $[t_j - \delta, t_j)$.  Similarly, the sum counters
($\presum[\cdot,\cdot,\cdot]$, $\middlesum[\cdot,\cdot,\cdot]$ and
$\postsum[\cdot,\cdot,\cdot]$) have an extra dimension to denote if the first
edge is incident on node $u$ or node $v$.  For example,
$\presum[u, \dir1, \dir2]$ is the number of sequentially ordered pairs of edges
containing some node $w_i$ in the time window $[t_j - \delta, t_j)$, where the
first edge is adjacent to $u$ in direction $\dir1$ and the second edge is
adjacent to $v$ in direction $\dir2$.


\begin{algorithm}[t!]\algoptions
   Enumerate all triangles in the undirected static graph $G$ of $T$\;
   $\rho \leftarrow $ number of temporal edges on each static edge $e$ in $G$\;
   \ForEach{static triangle $\Delta = (e_1$, $e_2$, $e_3)$ \KwTo $G$}{
       $e_{\max} = \arg\max_{e \in \Delta}\{\rho_{e}\}$\;
       \ForEach{$e$ \KwTo $\Delta$}{
       \mycomment{$e' \in \ell_{e}$ if $e \in \Delta$ and $\Delta$ assigned to $e'$}\;
       Append $e_{\max}$ to edge list $\ell_{e}$\;
       }
   }
   \ForEach{temporal edge $(e=(u,v), t)$ \KwTo time-sorted $T$}{
       \ForEach{$e'$ \KwTo $\ell_{e}$}{
       \mycomment{$a_{e'}$ contains $(e,t)$ if $e$ belongs to a static triangle assigned to $e'$}\;
       Append $(e, t)$ to temporal-edge list $a_{e'}$\;
       }
   }
   \ForEach{undirected edge $e$ \KwTo $G$}{
       Update counts using \cref{alg:fast_triangles} with input $a_{e}$\;
   }
  \dualcaption{Fast temporal triangle motif counting}{This is a sketch of the fast algorithm
    for counting the number of $3$-edge $\delta$-temporal
    triangle motifs in a temporal graph $T$.}
  \label{alg:triangle}
\end{algorithm}

\clearpage

Recall that the problem with counting triangle motifs by the general framework
in \cref{alg:general} is that a pair of nodes with many edges might have to be
counted for many triangles in the graph.  However, with
\cref{alg:fast_triangles}, we can simultaneously count all triangles
adjacent to a given pair of nodes.  What remains is that we must assign each
triangle in the static graph to a pair of nodes.  Here, we propose to assign
each triangle to the pair of nodes in that triangle containing the largest
number of edges (this is sketched in \cref{alg:triangle}).  The core
idea of this assignment procedure is that we should simultaneously
process as many triangles as possible for pairs of nodes with many
edges. The following theorem says that this is faster than simply counting for
each triangle (described in \cref{sec:tm_general_framework}).  Specifically, we
reduce $O(\triangleenum + \min(m\tau, mn))$ complexity to $O(\triangleenum + \min(m\sqrt{\tau}, mn))$, where
$\tau$ is the number of triangles in the graph.

\begin{theorem}\label{thm:fast_triangles}
In the worse case, \cref{alg:triangle} runs in time
$O(\triangleenum + \min(m\sqrt{\tau}, mn))$
where $\triangleenum$ is the time to enumerate all triangles in the static graph
$G$, $m$ is the total number of temporal edges, $n$ is the number of nodes, and
$\tau$ is the number of static triangles in $G$.
\end{theorem}
\begin{proof}

Let $\rho_e$ be the number of temporal edges along the static edge $e$, and
let $z_e \ge 1$ be the number of times that timestamps along this static edge
are used in a call to \cref{alg:fast_triangles} by
\cref{alg:triangle}.  Since \cref{alg:fast_triangles} runs in
linear time in the number of edges in its input, the total running time is on
the order of $\sum_{e}\rho_ez_e$.  Each static edge appears in at
most $n$ triangles and has at most $m$ temporal edges, so $\sum_{e}\rho_ez_e \le mn$.
We now consider the case when $\sqrt{\tau} \le n$.

The $\rho_e$ are fixed, and we wish to find the values of $z_e$ that maximize
the summation.  Without loss of generality, write $\{\rho_e\}$ by
$\rho_1 \ge \rho_2 \ge \ldots \ge \rho_r$, where $r$ is the number of
static edges in the static graph induced by the temporal graph.  Consequently,
$z_i \le i$, $i = 1, \ldots, r$.  There are $\tau$ static triangles,
so \cref{alg:fast_triangles} will look at the timestamps along $3\tau$ static
edges.  Thus, $\sum_{i}z_i \le 3\tau$.  The summation
$\sum_{i=1}^{r}\rho_iz_i$ is maximized when $z_1 = 1$, $z_2 = 2$, and so on up
to some index $J = O(\sqrt{\tau})$ for which
$\sum_{i=1}^{J}z_i = \sum_{i=1}^{J}i = 3\tau$.
Now given that the $z_i$ are fixed and the $\rho_i$
are ordered, the summation is maximized when
$\rho_1 = \rho_2 = \ldots = \rho_j = m / j$.
In this case,
\[
\sum_{i}\rho_iz_i 
= \sum_{i=1}^{J}\frac{m}{J} \cdot i 
= \frac{m}{J}\sum_{i=1}^{J} 
= O\left(\frac{m}{J} \cdot J^2\right) 
= O(m\sqrt{\tau}).
\]
\end{proof}


\begin{figure}[t]
\begin{center}
\definecolor{mygreen}{RGB}{27,158,119}
\definecolor{myorange}{RGB}{217,95,2}
\definecolor{mypurple}{RGB}{117,112,179}

\renewcommand*{\VertexInterMinSize}{15pt}
\renewcommand*{\VertexSmallMinSize}{15pt}

\begin{tikzpicture}
\renewcommand*{\VertexLineColor}{myorange}
\renewcommand*{\VertexLightFillColor}{myorange}
\renewcommand*{\VertexTextColor}{white}
\Vertex[x=5,y=3.75]{u}

\renewcommand*{\VertexLineColor}{mygreen}
\renewcommand*{\VertexLightFillColor}{mygreen}
\node[VertexStyle](v1) at (1,1.5) {$v_1$};
\node[VertexStyle](v2) at (3,1.5) {$v_2$};
\node[VertexStyle](v3) at (5,1.5) {$v_3$};
\node[VertexStyle](vk) at (9,1.5) {$v_k$};
\node[draw=none] (vd) at (7,1.5) {\LARGE $\ldots$};

\SetUpEdge[style={->, >=triangle 60,black}]

\Edge[label=$\rho_1$](v1)(u)
\Edge[label=$\rho_2$](v2)(u)
\Edge[label=$\rho_3$](v3)(u)
\Edge[label=$\rho_k$](vk)(u)

\Edge(v1)(v2)
\Edge(v2)(v3)

\SetUpEdge[style={->, >=triangle 60,black}]
\Edge[style={->, >=triangle 60, bend right=20}](v1)(v3)
\Edge[style={->, >=triangle 60, bend right=25}](v1)(vk)
\Edge[style={->, >=triangle 60, bend right=20}](v2)(vk)
\Edge[style={->, >=triangle 60, bend right=15}](v3)(vk)

\end{tikzpicture}
\end{center}
\dualcaption{Graph structure of worst-case instance for \cref{alg:triangle}}{There
is one temporal edge along each static edge $(v_i, v_j)$.  Each 
static edge $(v_i, u)$ contains an equal number of temporal edges, and the number
of such temporal edges is strictly greater than 1.  The $\rho_i$
denote the ordering of edges in the proof of \cref{thm:fast_triangles}.}
\label{fig:fast_worstcase}
\end{figure}

We now show a worst-case example for the algorithm, which
is illustrated in \cref{fig:fast_worstcase}.
The graph consists of $k + 1$ nodes $u$, $v_1$, $\ldots$, $v_k$, and the
static edges are $(v_i, v_j)$, $1 \le i < j \le k$, and $(v_r, u)$, $1 \le r \le k$.
Each edge $(v_i, v_j)$ has a single temporal timestamp, and each edge $(v_i, u)$
has an equal number of timestamps.  Without loss of generality, the first
$k$ edges in the ordering of the static edges is $\rho_1 = (v_1, u)$, $\rho_2 = (v_2, u)$,
$\ldots$, $\rho_k = (v_k, u)$ (the remainder of the ordering is arbitrary).
Finally, suppose that we are counting the motif $M = (a, b), (a, c), (b, c)$.  

In this case, the call to the \cref{alg:fast_triangles} subroutine for edge $(v_r, u)$ ($1 \le r \le k$)
in \cref{alg:triangle} processes temporal edges along $(v_r, v_2)$, $\ldots$, $(v_r, v_k)$
and $(v_r, u)$, $\ldots$, $(v_k, u)$.  Each temporal edge along $(v_i, v_j)$ is processed
only once, and the temporal edges along $(v_r, u)$ are processed $r$ times.  Thus, the
total number of temporal edges processed is
\begin{align*}
{k \choose 2} + \sum_{r = 1}^{k}r\frac{m - {k \choose 2}}{k}
= O(k^2) + O(mk)
= O(\tau + m\sqrt{\tau}).
\end{align*}

In this case, \cref{alg:general} has the same asymptotic running time.  Each of
the $O(k^2)$ static triangles contains $O(m / k)$ temporal edges.

Finally, we emphasize that, if we ignore the pre-processing time in \cref{alg:fast_triangles},
the amount of computation in \cref{alg:fast_triangles}
is always less than the computation done by \cref{alg:general}
for counting triangles.  Thus, regardless of the asymptotics,
\cref{alg:fast_triangles} is still a strict improvement over the general algorithm.

\section{Experiments}
\label{sec:tm_experiments}

Next, we use our algorithms to reveal patterns in a variety of temporal network
datasets.  We find that the number of instances of various $\delta$-temporal
motifs reveal basic mechanisms of the networks.

\subsection{Data}\label{sec:tm_data}

\begin{table}[tb]
\centering
\dualcaption{Summary statistics of temporal network datasets}{The time span
is the time between latest and earliest timestamp available in the data.  All timestamps
are recorded at one second resolution.}
  \begin{tabular}{l c c c c}
    \toprule 
    dataset        & \# nodes & \# static edges & \# edges & time span (days) \\
    \midrule
    \emaileu       & 986      & 2.49K     & 332K     & 803       \\
    \phone         & 1.05M    & 2.74M     & 8.55M    & 7         \\
    \sms           & 44.1K    & 67.2K     & 545K     & 338       \\
    \messages      & 1.90K    & 20.3K     & 59.8K    & 193       \\
    \stackoverflow & 2.58M    & 34.9M     & 47.9M    & 2774      \\
    \bitcoin       & 24.6M    & 88.9M     & 123M     & 1811      \\
    \fbwall        & 45.8K    & 264K      & 856K     & 1560      \\
    \wikitalk      & 1.09M    & 3.13M     & 6.10M    & 2277      \\
    \phonecallme   & 18.7M    & 360M      & 2.04B    & 364       \\
    \smsme         & 6.94M    & 51.5M     & 800M     & 89        \\
    \bottomrule
    \end{tabular}\label{tab:datasets}
\end{table}

We gathered a variety of datasets in order to study the patterns of
$\delta$-temporal motifs in several domains.  The datasets are described below
and summary statistics are in \cref{tab:datasets}.  The time resolution of the
edges in all datasets is one second.

\begin{itemize}
\item $\emaileu$\footnote{This is the same network as $\emaileucore$
in \cref{sec:local_real}, but we are now using the temporal information.}
is a collection of emails between members of a European
research institution~\cite{leskovec2007graph}.  An edge $(u, v, t)$ signifies
that person $u$ sent person $v$ an email at time $t$.

\item $\phone$ was constructed from telephone call records for a major European
service provider.  An edge $(u, v, t)$ signifies that person $u$ called person
$v$ starting at time $t$.

\item $\sms$ is a collection of short messaging service (SMS) texting
records from a company charging account~\cite{wu2010evidence}.
In this dataset, an edge $(u, v, t)$ means that person $u$ sent
an SMS message to person $v$ at time $t$.

\item $\messages$ is comprised of private messages sent on an online social network
at the University of California, Irvine~\cite{panzarasa2009patterns}.  Users
could search the network for others and then initiate conversation based on
profile information.  An edge $(u, v, t)$ means that user $u$ sent a private
message to user $v$ at time $t$.

\item $\stackoverflow$ is constructed from communication on Stack Overflow.
On stack exchange web sites, users post questions and receive answers from other
users, and users may comment on both questions and answers.  We derive a
temporal network by creating an edge $(u, v, t)$ if, at time $t$, user $u$: (1)
posts an answer to user $v$'s question, (2) comments on user $v$'s question, or
(3) comments on user $v$'s answer.  We formed the temporal network from the
entirety of Stack Overflow's history up to March 6, 2016.

\item $\bitcoin$ consists of all payments made with the decentralized
digital currency and payment system Bitcoin up to October 19, 2014~\cite{kondor2014rich}.
Nodes in the network correspond to Bitcoin addresses, and an individual may have
several addresses.  An edge $(u, v, t)$ signifies that bitcoin was transferred
from address $u$ to address $v$ at time $t$.  

\item $\fbwall$ consists of wall posts between users on the social network
Facebook located in the New Orleans region~\cite{viswanath2009evolution}.  Any friend of
a given user can see all posts on that user's wall, so communication is public
among friends.  An edge $(u, v, t)$ means that user $u$ posted on user $v$'s
wall at time $t$.

\item $\wikitalk$ represents edits on user talk pages on
Wikipedia~\cite{leskovec2010governance}.  An edge $(u, v, t)$ means that
user $u$ edited user $v$'s talk page at time $t$.

\item $\phonecallme$ is constructed from phone call records of a large
telecommunications service provider in the Middle East.  An edge $(u, v, t)$ 
means that user $u$ initiated a call to user $v$ at time $t$.

\item $\smsme$ is constructed from SMS texting records from the same
telecommunications service provider in the Middle East.
An edge $(u, v, t)$ means that user $u$ sent an
SMS message to user $v$ at time $t$.
\end{itemize}

\subsection{Empirical observations of motif counts}\label{sec:tm_empirical}

We first examine the distribution of 2- and 3-node, 3-edge motif instance counts
from 8 of the datasets described in \cref{sec:tm_data} with $\delta = 1$ hour
(\cref{fig:raw_counts}).  We choose 1 hour for the time window as this is close
to the median time for a node to take part in three edges in most of our
datasets. We make a few empirical observations uniquely available due to
temporal motifs and provide possible explanations for these observations.

\begin{sidewaysfigure}[thb]
\centering \includegraphics[width=\columnwidth]{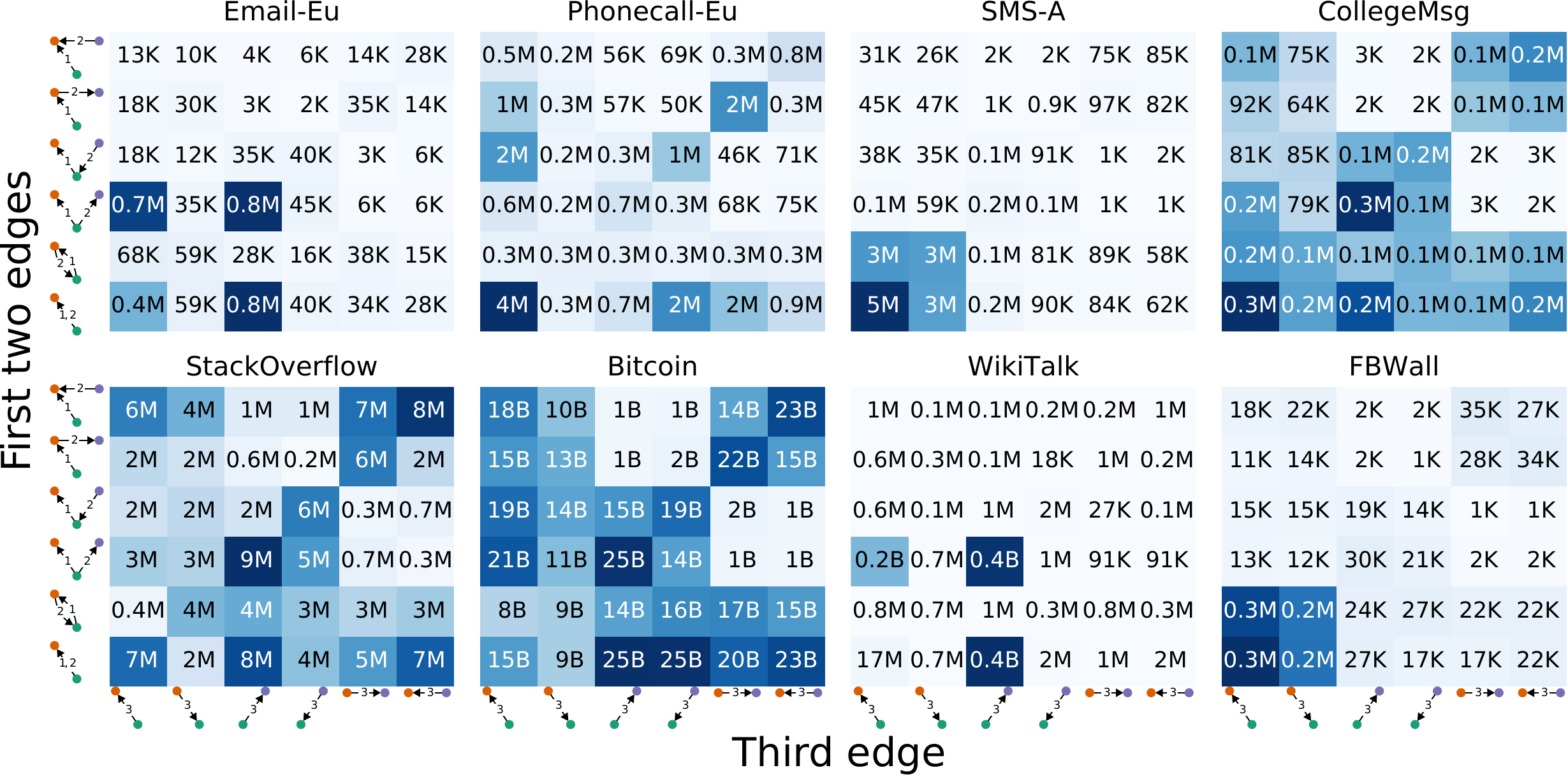}
\dualcaption{Counts of instances of all $2$- and $3$-node, $3$-edge
  $\delta$-temporal motifs with $\delta =$ 1 hour}{For each dataset, counts in
  the $i$th row and $j$th column is the number of instances of motif $M_{i, j}$
  (see \cref{fig:three_edge_motifs}); this motif is the union of the two edges
  in the row label and the edge in the column label.  For example, there are 0.7
  million instances of motif $M_{4,1}$ in the $\emaileu$ dataset. The color for
  the count of motif $M_{i,j}$ indicates the fraction over all $M_{i,j}$ on a
  linear scale---darker blue means a higher count.}
\label{fig:raw_counts}
\end{sidewaysfigure}

\clearpage

\begin{figure}[tb]
\centering
\includegraphics[width=0.9\columnwidth]{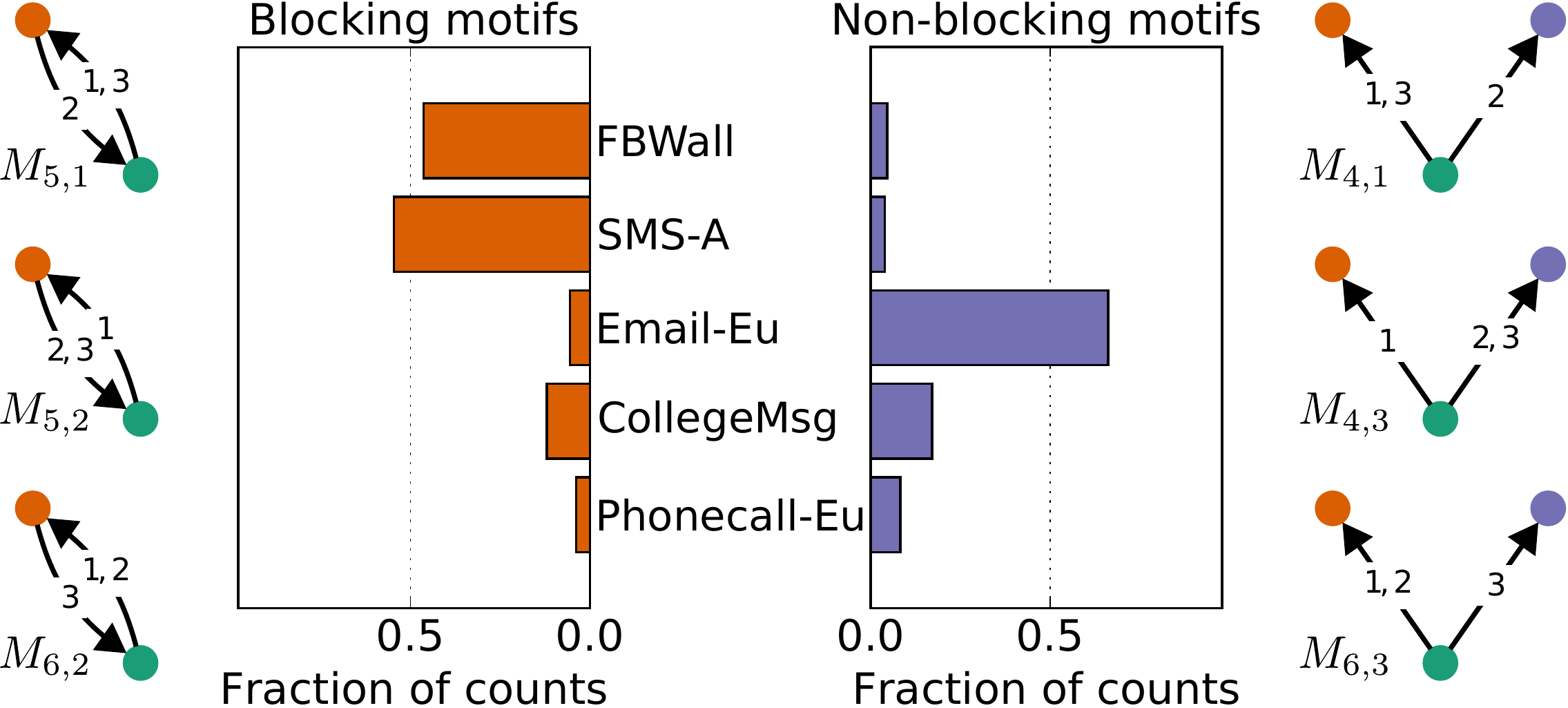}
\dualcaption{Blocking and non-blocking behavior}{%
Fraction of the $2$- and $3$-node, $3$-edge $\delta$-temporal motif instance
counts that correspond to two groups of motifs ($\delta =$ 1 hour).  Motifs on the left
capture ``blocking'' behavior, common in $\sms$ and $\fbwall$
posting, where most communication is one-to-one at this time scale.
Motifs on the right exhibit ``non-blocking'' behavior, common in
$\emaileu$, where broadcast patterns are common.
$\messages$ contains a mixture of the two behaviors at this time scale.
}
\label{fig:blocking}
\end{figure}

\xhdr{Blocking communication}
If an individual typically waits for a reply from one individual before
proceeding to communicate with another individual, we consider it a
\emph{blocking} form of communication. A typical conversation between two
individuals characterized by fast exchanges happening back and forth is blocking
as it requires complete attention of both individuals. We capture this
behavior in the ``blocking motifs'' $M_{5,1}$, $M_{5,2}$ and $M_{6,2}$, which
contain 3 edges between two nodes with at least one edge in either direction
(\cref{fig:blocking}, left).  However, if the reply doesn't arrive soon,
we might expect the individual to communicate with others without waiting for a
reply from the first individual. This is a non-blocking form of communication
and is captured by the ``non-blocking motifs'' $M_{4,1}$, $M_{4,3}$ and $M_{6,3}$
having edges originating from the same source but directed to different
destinations (\cref{fig:blocking}, right)

The fractions of counts corresponding to the blocking and non-blocking motifs
out of the counts for all 36 motifs in \cref{fig:three_edge_motifs}
uncover several interesting characteristics in communication networks ($\delta =
1$ hour; see \cref{fig:blocking}). In $\fbwall$ and $\sms$, blocking
communication is vastly more common, while in $\emaileu$ non-blocking
communication is prevalent.  Email is not a dynamic method of
communication and replies within an hour are rare.  Thus, we would expect
non-blocking behavior.  Interestingly, the $\messages$ dataset shows both
behaviors as we might expect individuals to engage in multiple conversations
simultaneously.  In complete contrast, the $\phone$ dataset shows neither
behavior.  A simple explanation is that that a single edge (a phone call)
captures an entire conversation and hence blocking behavior does not emerge.

\begin{figure}[tb]
\centering
\includegraphics[width=\columnwidth]{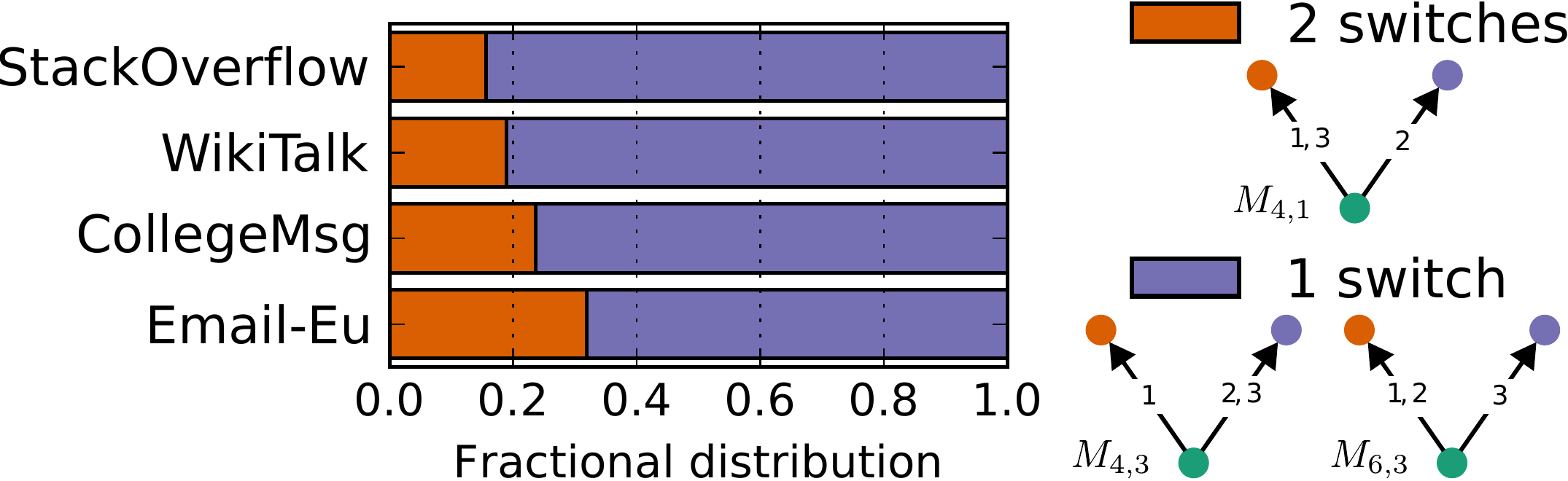}
\dualcaption{Switching behavior in non-blocking motifs}{%
A ``switch'' occurs with two temporal edges $(u, v, t_1)$ and $(u, w, t_2)$, $t_1 < t_2$,
where node $u$ switches an outgoing neighbor.  Switching is least common on
Stack Overflow and most common in email.  We also see differences in switching behavior
as we vary time scales in the $\messages$ dataset (see \cref{fig:IM_over_time}).
}
\label{fig:switching}
\end{figure}

\xhdr{Cost of switching}
Amongst the non-blocking motifs discussed above, $M_{4,1}$ captures two
consecutive ``switches'' between pairs of nodes whereas $M_{4,3}$ and $M_{6,3}$ each
have a single switch (\cref{fig:switching}, right). In communication networks, a switch
corresponds to a change in message destination for a node $u$.  Prevalence of
$M_{4,1}$ indicates a lower cost of switching targets, whereas prevalence of the
other two motifs are indicative of a higher cost.  We observe in
\cref{fig:switching} that the ratio of 2-switch to 1-switch motif counts
is the least in $\stackoverflow$, followed by $\wikitalk$, $\messages$ and then
$\emaileu$. On Stack Overflow and Wikipedia talk pages, there is a high
cost to switch targets because of peer engagement and depth of discussion.  On the other
hand, in the $\messages$ dataset there is less cost to switch because it
lacks depth of discussion within the time frame of $\delta = $ 1 hour.
In $\emaileu$, there is almost no peer engagement, and cost of switching is
negligible.

\begin{figure}[tb]
\centering
\includegraphics[width=\columnwidth]{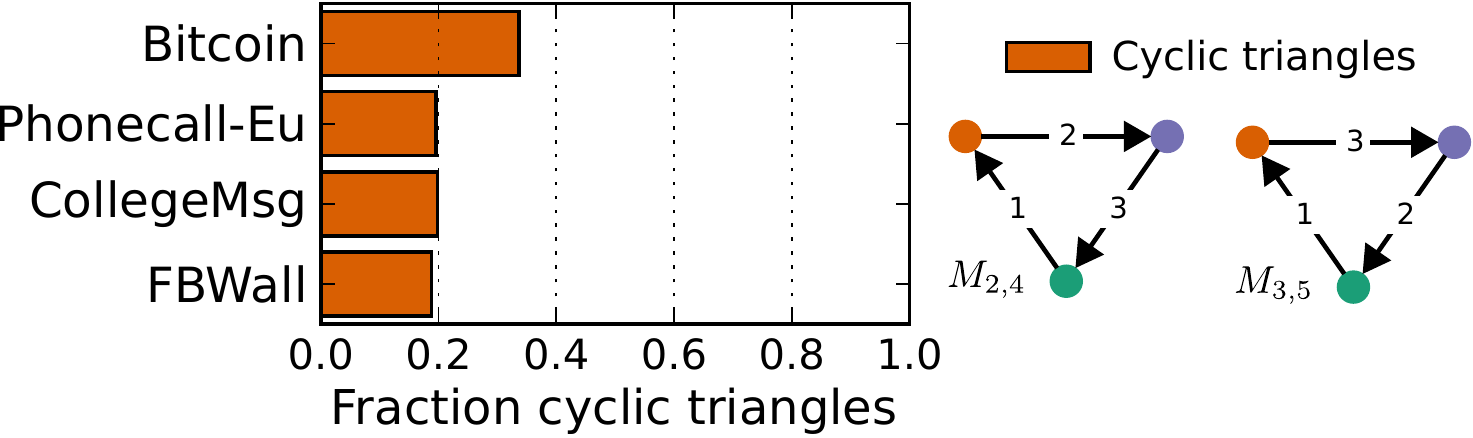}
\dualcaption{Cyclic temporal triangles}{%
Fraction of $3$-edge $\delta$-temporal triangle motif counts ($\delta =$ 1 hour)
corresponding to cyclic triangles (right) in the four datasets for which this
fraction is the largest.  $\bitcoin$ has a much higher fraction compared to all
other datasets.
}
\label{fig:cycles}
\end{figure}

\xhdr{Cycles in $\bitcoin$}
Of the eight $3$-edge triangle motifs, $M_{2,4}$ and $M_{3,5}$ are cyclic, i.e.,
the target of each edge serves as the source of another edge.\footnote{By ``cyclic''
we do not mean that the temporal edges must form a cycle.  Rather, the multigraph
$K$ in the formal motif definition is a cycle.}  We observe in
\cref{fig:cycles} that the fraction of triangles that are cyclic is much higher in
$\bitcoin$ compared to any other dataset.
This can be attributed to the transactional nature of $\bitcoin$ where
the total amount of bitcoin is limited. Since remittance (outgoing edges) is typically associated
with earnings (incoming edges), we should expect cyclic behavior.



\begin{figure}[tb]
\centering \includegraphics[width=1\columnwidth]{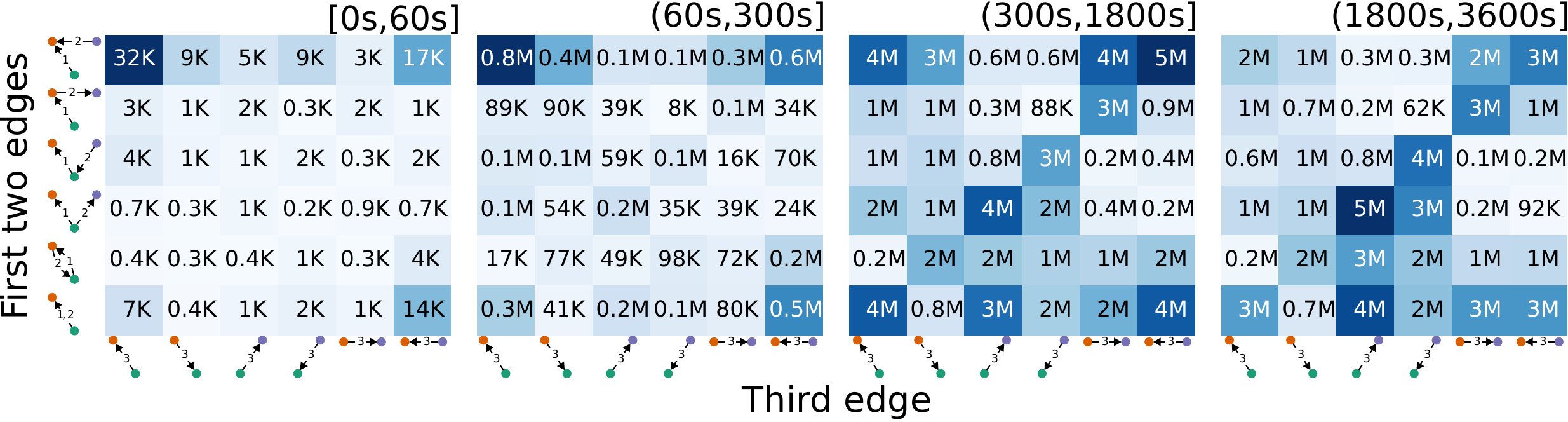}
\dualcaption{Counts for all $2$- and $3$-node, $3$-edge $\delta$-temporal motifs
  in four time intervals for the $\stackoverflow$ dataset}{%
For each interval, the count in the $i$th row and $j$th column is the number of
instances of motif $M_{i, j}$ (see \cref{fig:three_edge_motifs}).
}
\label{fig:so_counts_over_time}
\end{figure}

\xhdr{Motif counts at varying time scales}
We now explore how motif counts change at different time scales.  For the
$\stackoverflow$ dataset we counted the number of instances of $2$- and
$3$-node, $3$-edge $\delta$-temporal motifs for $\delta = $ 60, 300, 1800, and
3600 seconds (\cref{fig:so_counts_over_time}).  These counts determine the
number of motifs that completed in the intervals
[0, 60], (60, 300], (300, 1800s], and (1800, 3600] seconds
(e.g., subtracting 60 second counts from 300 second counts gives the interval (60, 300]).
Observations at smaller timescales reveal phenomenon which start to get eclipsed
at larger timescales.  For instance, on short time scales, motif $M_{1,1}$
(\cref{fig:so_counts_over_time}, top-left corner) is quite common.  We
suspect this arises from multiple, quick comments on the original question, so
the original poster receives many incoming edges.  At larger time scales, this
behavior is still frequent but relatively less so.  Now let us compare counts
for $M_{1,5}$, $M_{1,6}$, $M_{2,5}$, $M_{2,6}$ (the four in the top right
corner) with counts for $M_{3,3}$, $M_{3,4}$, $M_{4,3}$, $M_{4,4}$ (the four in
the center). The former counts likely correspond to conversations with the
original poster while the latter are constructed by the same user interacting with
multiple questions.  Between 300 and 1800 seconds (5 to 30 minutes),
the former counts are relatively more common while the latter counts only become
more common after 1800 seconds. A possible explanation is that the typical
length of discussions on a post is about 30 minutes, and later on, users answer
other questions.

\definecolor{msgblue}{RGB}{0, 114, 178}
\definecolor{msgorange}{RGB}{213,94,0}
\begin{figure}[tb]
\centering
\includegraphics[width=0.75\columnwidth]{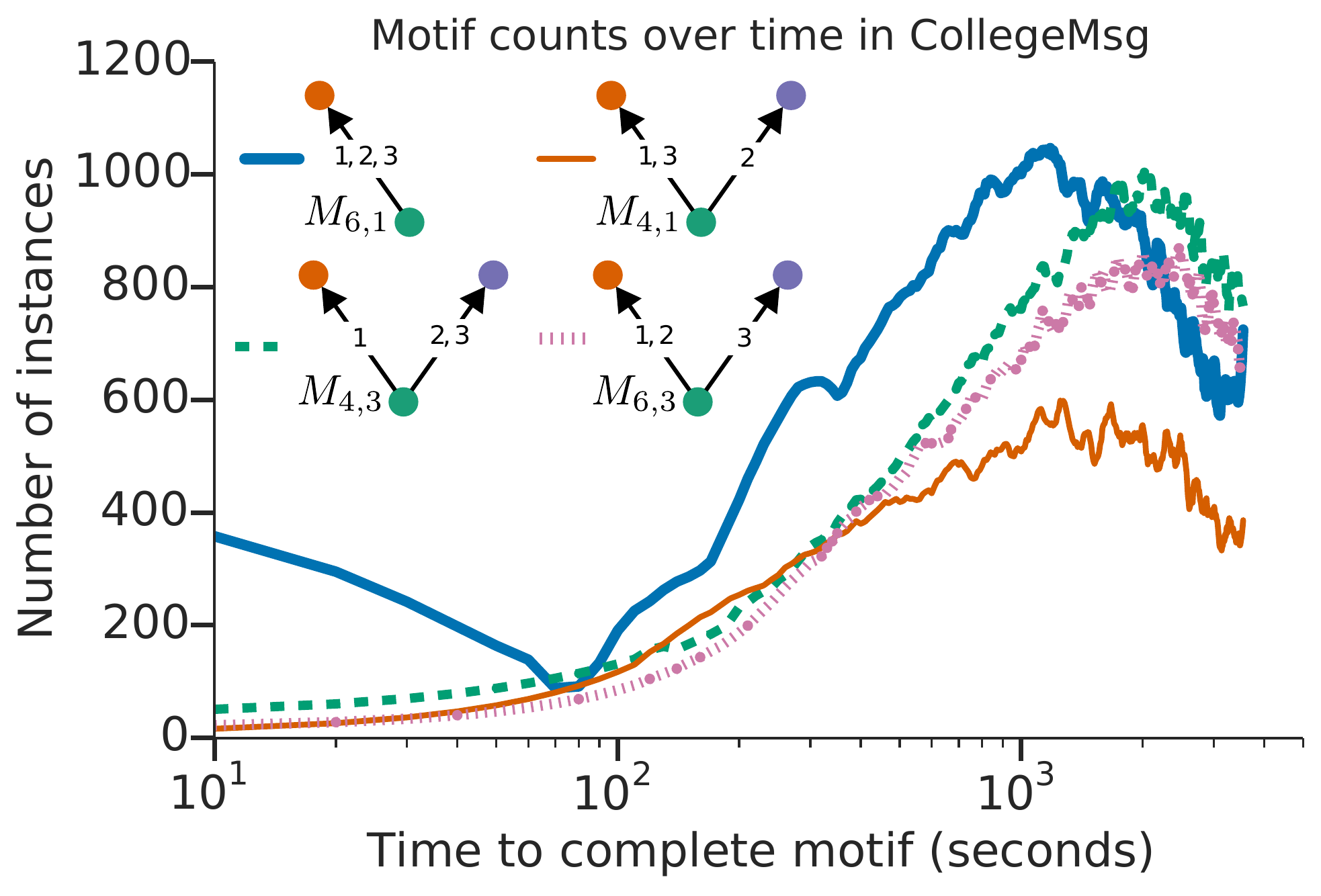}
\dualcaption{Motif counts over various time scales in the $\messages$ dataset}{The four motifs correspond to all motifs where a node sends $3$ outgoing
messages to $1$ or $2$ neighbors.  At early time scales, one-to-one
communication is the most common ($M_{1,1}$, \textcolor{msgblue}{blue} curve).
At later time scales, the motif with two switches
($M_{4,1}$, \textcolor{msgorange}{orange} curve) is less common than the motifs
with one or two switches.}
\label{fig:IM_over_time}
\end{figure}

Next, we examine messaging behavior in the $\messages$ dataset at fine-grained
time intervals.  We counted the number of motifs consisting of a single node
sending three outgoing messages to one or two neighbors (motifs $M_{6,1}$,
$M_{6,3}$, $M_{4,1}$, and $M_{4,3}$) in the time bins $[10(i-1), 10i)$ seconds,
$i = 1, \ldots, 500$ (\cref{fig:IM_over_time}).  We first notice that at
small time scales, the motif consisting of three edges to a single neighbor
($M_{6,1}$) occurs frequently.  This pattern could emerge from a succession of
quick introductory messages.  Overall, motif counts increase from roughly 1
minute to 20 minutes and then decline.  After 5 minutes, counts
for the three motifs with one switch in the target ($M_{6,1}$, $M_{6,3}$, and
$M_{4,3}$) grow at a faster rate than the counts for the motif with two switches
($M_{4,1}$).  As mentioned above, this pattern could emerge from a tendency to
send several messages in one conversation before switching to a conversation
with another friend.

\subsection{Algorithm scalability}\label{sec:tm_scalability}
  
Finally, we performed scalability experiments of our algorithms.  All algorithms
were implemented in C++, and all experiments ran using a single thread of a
2.4GHz Intel Xeon E7-4870 processor.  We did not measure the time to load
datasets into memory, but our timings include all pre-processing time needed by
the algorithms (e.g., the triangle counting algorithms first find triangles in
the static graph).  We emphasize that our implementation is single threaded, and
there is ample room for parallelism in our algorithms.


\begin{table}[tb]
\centering
\dualcaption{Time to count the eight 3-edge $\delta$-temporal triangle motifs}
{%
We compare the time of our general counting method (\cref{alg:general}) and our
fast counting method (\cref{alg:triangle}).  The fast algorithm provides
significant speedups for all datasets.  In these experiments, $\delta = 3600$,
but this does not signifcantly affect the running times.  All times are listed
in seconds.
}
\begin{tabular}{l c c c c}
    \toprule 
    dataset        & \# static triangles & \cref{alg:general} & \cref{alg:triangle} & speedup \\
    \midrule
    \wikitalk      & 8.11M               & 51.1               & 26.6                & 1.92x   \\
    \bitcoin       & 73.1M               & 27.3K              & 483                 & 56.5x   \\
    \smsme         & 78.0M               & 2.54K              & 1.11K               & 2.28x   \\
    \stackoverflow & 114M                & 783                & 606                 & 1.29x   \\
    \phonecallme   & 667M                & 12.2K              & 8.59K               & 1.42x   \\
    \bottomrule
    \end{tabular}\label{tab:scalability}
\end{table}

First, we used both the general counting method (\cref{alg:general}) and the
fast counting method (\cref{alg:triangle}) to count the number of all eight
3-edge $\delta$-temporal triangle motifs in our datasets ($\delta =$ 1 hour).
\Cref{tab:scalability} reports the running times of the algorithms for all
datasets with at least one million triangles in the static graph. For all of
these datasets, our fast temporal triangle counting algorithm provides
significant performance gains over the general counting method, ranging between
a 1.29x and a 56.5x speedup.  The gains of the fast algorithm are the largest
for $\bitcoin$, which is due to some pairs of nodes having many edges between
them and also participating in many triangles.


\begin{figure}[t]
\centering
\includegraphics[width=0.75\columnwidth]{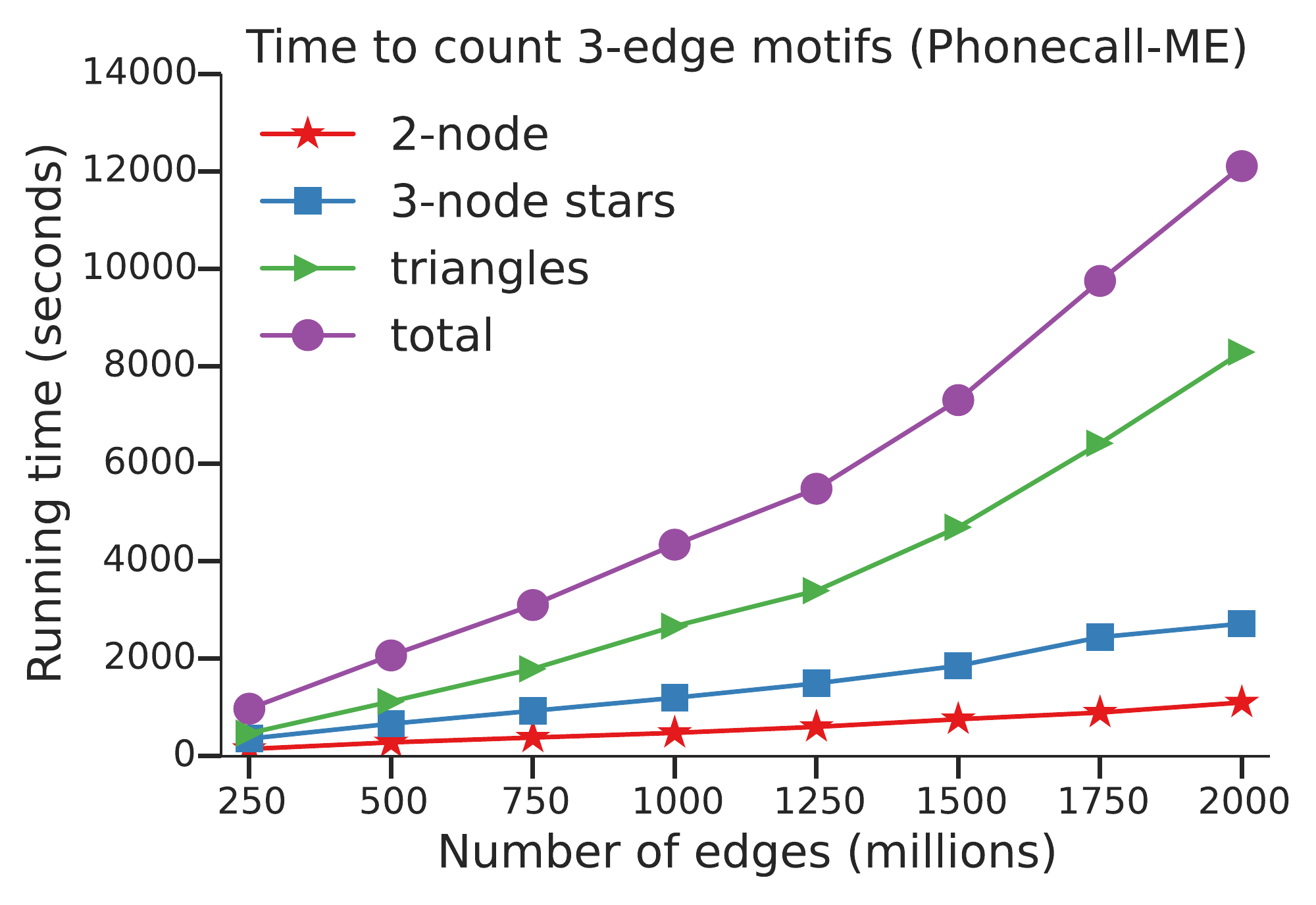}
\dualcaption{Scalability of temporal motif counting in}{The plot
reports the time to count $3$-edge motifs on the first $k$ temporal million edges
in $\phonecallme$ as a function of $k$.  The time to count star and edge patterns
scales linearly with the number of edges in correspondence with our computational
complexity analysis.}
\label{fig:scalability}
\end{figure}

Second, we measured the time to count various $3$-edge $\delta$-temporal motifs
in our largest dataset, $\phonecallme$.  Specifically, we measured the time to
compute (1) $2$-node motifs, (2) $3$-node stars, and (3) triangles on the first
$k$ million edges in the dataset for $k = 250, 500, \ldots, 2000$
(\cref{fig:scalability}).  The time to compute the $2$-node, $3$-edge motifs and
the $3$-node, $3$-edge stars scales linearly, as expected from our complexity
analysis.  The time to count triangle motifs grows superlinearly and becomes the
dominant cost when there is a large number of edges.  For practical purposes,
the running times are quite modest.  With two billion edges, our methods take
less than 3.5 hours to complete (executing sequentially).

\section{Prior definitions of temporal motifs and other related work}
\sectionmark{Prior definitions and related work}
\label{sec:tm_related}

There are several studies on pattern formation in growing networks where one
only considers the addition of edges to a static graph over time.  In this
context, motif-like patterns have been used to create evolution rules that
govern the ways that networks
develop~\cite{berlingerio2009mining,ugander2013subgraph}.  The way we consider
ordering of temporal edges in our definition of $\delta$-temporal motifs is
similar in spirit.  There are also several analyses on the formation of
triangles in a variety of social
networks~\cite{huang2014mining,kossinets2006empirical,leskovec2010signed,zignani2014link}.
In contrast, in the temporal graphs we study here, three nodes may form a
triangle several times.  On the algorithmic side, a large amount of research has
been devoted simply to exact counting and enumeration of triangles in undirected
static graphs (see, e.g., \cite{latapy2008main}) and to approximating the total
number of triangles in static
graphs~\cite{ahmed2015efficient,seshadhri2014wedge,de2016triest}.

There are some definitions of temporal motifs for sequences of snapshot
graphs~\cite{zhang2014dynamic,lahiri2007structure,jin2007trend}, but we
focus our discussion here on alternative definitions that target temporal graphs in
the way that we defined them in \cref{def:temporal_graph}.  To understand prior
definitions of temporal motifs, we need one more definition.  We say that two
temporal edges $(u, v, t_1)$ and $(w, x, t_2)$ are \emph{$\omega$-adjacent} if
$\lvert t_2 - t_1 \rvert \le \omega$ and $(u, v)$ and $(w, x)$ share at least one
node (i.e., ${\lvert \{u, v, w, x\} \rvert < 4}$).  With this connectivity, we
can partition the temporal edges of a temporal graph into $\omega$-connected
components where every edge in a component is able to reach every other by
following $\omega$-adjacent connections.

\Citet{zhao2010communication} induce static graphs from $\omega$-connected components
and then use the standard static graph motifs that we used in
in \cref{ch:honc}.  This definition is used to analyze the structure of communication data.
\Citet{gurukar2015commit,hulovatyy2015exploring}
both propose additional equivalent definitions of temporal motifs, and the
definition is similar to ours.  The motifs consist of ordered directed
multigraphs and a time span $\omega$.  An instance of the motif is a subset of events
that match the ordered multigraph, where every sequential edge in the ordering
is $\omega$-connected.  While the definition is similar, \citet{gurukar2015commit} do
not have theoretical guarantees on running time
and \citet{hulovatyy2015exploring} have running times that are as expensive as
temporal motif enumeration.

\Citet{kovanen2011temporal} use the same notion of $\omega$-adjacency of events
with an ordered multigraph but further restrict motif instances
to cases where the events are consecutive for the nodes.
In other words, in the time spanned
by the motif instance, there can be no other temporal edge adjacent to one of
the nodes.  While this is a severe
restriction, a major benefit of this definition is that it permits fast count algorithms.
For example, triangle motifs may be enumerated in linear time.  This motif
definition has been used to analyze phone call
records~\cite{kovanen2013temporal} and Wikipedia
edits~\cite{jurgens2012temporal}.  While \citet{kovanen2011temporal}
present their work for directed temporal graphs, the ideas have been extended to
undirected temporal networks and applied to the analysis of proximity and
contact networks~\cite{zhang2015human}.

Our definition of temporal motifs is a bit simpler by throwing out any requirements
for $\omega$-adjacency of edges.  The time duration is modeled at the level of the motif, rather
than at the level of individual edges.  This keeps with the spirit of
higher-order modeling and analysis.  

We also note that \citet{viard2016computing} developed algorithms
for counting what they call maximal $\Delta$-cliques.  Their definition requires that
all pairs of nodes in a clique interact at least once in sub-intervals of size $\Delta$.

\section{Discussion}
\label{sec:tm_discussion}

We have developed $\delta$-temporal network motifs as a tool for analyzing
temporal networks.  We introduced a general framework for counting instances of
any temporal motif as well as faster algorithms for certain classes of motifs
and found that motif counts reveal key structural patterns in a variety of
temporal network datasets.  Our fast algorithms are designed for $3$-node,
$3$-edge star and triangle motifs.  We expect that the same general techniques
can be used to count more complex temporal motifs.  There is also a host of
theoretical questions in this area for lower bounds on temporal motif counting.

Our fast algorithms only \emph{count} the number of instances of motifs rather
than \emph{enumerate} the instances, and this was crucial to the design of the
algorithms.  Temporal motif enumeration algorithms provide an additional
algorithmic design challenge.  The concept of avoiding enumeration has 
also been used to accelerate static motif counting~\cite{pinar2017escape}.  A general disadvantage of
counting algorithms is that such statistics cannot immediately be used to equip
nodes or edges with features, on which subsequent analysis can be performed.
However, there are still many features that we can use from the fast algorithms
developed here.  When using the general algorithm, we can still count how often
a node $u$ participates in a given position in a given motif.  For example, we
can easily count in how many motif instances of motif $M_{5,1}$, $M_{5,2}$,
$M_{6,1}$, and $M_{6,2}$ node $u$ serves as the source node (the green node
in \cref{fig:three_edge_motifs}).  With our fast $3$-node, $3$-edge star
counting algorithms, we can compute in how many motif instances a given node $u$
serves as the center for all 24 of these motifs.  Our fast $3$-edge triangle
counting algorithm are not amenable to per-node statistics.  However, one could
use the general algorithm for this purpose.

Although we did not run parallel performance experiments, there is ample room
for parallelism in these algorithms:
\begin{itemize}
\item The general algorithm parallelizes over all instances of the induced static motif.
\item The fast star counting algorithm parallelizes over all center nodes $u$.
\item The fast triangle counting algorithm parallelizes over all edges $(u, v)$
that are assigned at least one triangle.
\end{itemize}

Motif counts can also be measured with respect to a null
model~\cite{kovanen2011temporal,milo2002network} and such analysis may yield
additional discoveries.  Importantly, our algorithms will speed up such
computations, which use motif counts from many random instances of a generative
null model.

Finally, our definition of a temporal network motif and the associated counting
counting algorithms only rely on an \emph{ordering} on the edges in the temporal
graph and some way of measuring the \emph{distance} between edges in the
ordering.  While temporal network data is natural, our data does not necessarily
have to come from time.  We can count motifs for any collection of edges when
each edge is equipped with a (unique) element from an ordered field.  This
could lead to new applications in future work.

\setlength{\bibsep}{3pt}
\bibliographystyle{bib}
{\footnotesize \bibliography{refs}}

                  
\label{LastPage}
\end{document}